\documentclass[a4paper,11pt]{article}
\usepackage{graphicx} 
\usepackage[a4paper,left=2cm,right=2cm,top=3cm,bottom=3cm]{geometry}
\usepackage{natbib}
\usepackage{amsfonts}
\usepackage{amsmath}
\usepackage{amsthm}
\usepackage{import}
\usepackage{enumerate}
\usepackage[dvipsnames]{xcolor}
\usepackage{authblk, bm, bbm, soul}
\usepackage{amssymb}
\usepackage{enumitem}
\usepackage{multirow}
\usepackage{titling}
\usepackage{capt-of}
\usepackage{booktabs}
\usepackage{derivative}

\usepackage{xr-hyper}
\usepackage{hyperref}[]
\hypersetup{
    colorlinks=true,
    linkcolor=blue,
    filecolor=magenta,     
    urlcolor=cyan,
      citecolor=blue,
    }
 
\usepackage{cleveref}

\usepackage{natbib}
\usepackage{stackrel}
\usepackage{algorithm}
\usepackage{bbm}
\usepackage{algpseudocode}
\usepackage[toc,page]{appendix}

\newtheorem{theorem}{Theorem}[]
\newtheorem{definition}{Definition}[]

\newtheorem{assumption}{Assumption}
\newtheorem{proposition}[theorem]{Proposition}
\newtheorem{corollary}{Corollary}
\newtheorem{lemma}[theorem]{Lemma}
\theoremstyle{remark}
\newtheorem{remark}{Remark}

\DeclareMathOperator*{\var}{Var}

\DeclareMathOperator*{\op}{op}
\DeclareMathOperator*{\F}{F}
\DeclareMathOperator*{\opt}{opt}

\title{Change point localisation and inference in fragmented functional data}
\author{Gengyu Xue}
\author{Haotian Xu}
\author{Yi Yu}
\affil{Department of Statistics, University of Warwick}
\date{\today}

\begin{document}

\maketitle
\begin{abstract}
    We study the problem of change point localisation and inference for sequentially collected fragmented functional data, where each curve is observed only over discrete grids randomly sampled over a short fragment. The sequence of underlying covariance functions is assumed to be piecewise constant, with changes happening at unknown time points. To localise the change points, we propose a computationally efficient fragmented functional dynamic programming (FFDP) algorithm with consistent change point localisation rates. With an extra step of local refinement, we derive the limiting distributions for the refined change point estimators in two different regimes where the minimal jump size vanishes and where it remains constant as the sample size diverges. Such results are the first time seen in the fragmented functional data literature. As a byproduct of independent interest, we also present a non-asymptotic result on the estimation error of the covariance function estimators over intervals with change points inspired by \citet{lin2021basis}. Our result accounts for the effects of the sampling grid size within each fragment under novel identifiability conditions. Extensive numerical studies are also provided to support our theoretical results. 
\end{abstract}

\section{Introduction}
Functional data analysis (FDA) has enjoyed increasing popularity over the past few decades owing to the advances in modern technology which enable high-resolution data collection and storage. Modelling of data as realisations of random functions has found great success in various applications, including neuroscience \citep[e.g.][]{hasenstab2017multi,dai2019age}, medicine \citep[e.g.][]{chen2017modelling,crawford2020predicting,zhang2022nonparametric}, climatology \citep[e.g.][]{besse2000autoregressive,fraiman2014detecting}, finance \citep[e.g.][]{fan2014functional}, agriculture \citep[e.g.][]{lei2015functional} and others. A wide range of statistical research has also been carried out in the area of FDA such as classification \citep[e.g.][]{delaigle2012achieving}, estimation \citep[e.g.][]{yao2005functional,cai2011optimal} and functional regression \citep[e.g.][]{cai2012minimax,zhou2023functional}. We refer readers to \citet{wang2016functional} for a comprehensive review.

In this paper, we study the problem of change point detection and inference in fragmented functional data, where the partially observed functional data are collected sequentially as a non-stationary time series, with abrupt changes in covariance functions occurring at unknown time points. To be specific, let the observations $\{(X_{t,j},  Y_{t,j})\}_{t=1,j=1}^{n,m} \subseteq [0,1] \times \mathbb{R}$ satisfy the model
\begin{align}
    \label{model_obs}
    Y_{t,j} = f_{t}(X_{t,j})+ \varepsilon_{t,j}, \quad \text{for} \; t = 1, \ldots, n, \; \text{and} \; j =1, \ldots, m.
\end{align}
In the above notation, $\{X_{t,j}\}_{t=1,j=1}^{n,m} \subseteq [0,1]$ denote the discrete grids where the noisy functional data $\{ Y_{t,j}\}_{t=1,j=1}^{n,m} \subseteq \mathbb{R}$ are observed, $\{f_t(\cdot): [0,1] \rightarrow \mathbb{R}\}_{t=1}^n$ denote a sequence of mutually independent and square-integrable random processes with mean $0$ and covariance function $\Sigma_t^*$, where $\Sigma_t^*(k, l) = \mathbb{E}[f_t(k)f_t(l)]$ for every $k,l \in [0,1]$, and $\{\varepsilon_{t,j}\}_{t=1,j=1}^{n,m}$ denote the measurement errors. Moreover, different from the conventional setup of sparse FDA in the literature where each object can be observed over the entire domain \citep[e.g.][]{cai2012minimax,zhang2016sparse}, we will work under a more challenging scenario where each object can only be observed over a short fragment of length $0 <\delta \leq 1$, namely it holds that $\max_{j,k\in \{1, \ldots, m\}} |X_{t,j}-X_{t,k}| \leq \delta$ for every $t \in \{1, \ldots, n\}$.

To model the non-stationarity of the sequentially observed data, we introduce our change point notation. Across time, we assume that the covariance function possesses a piecewise constant pattern, i.e.~for $K \in \mathbb{N}_{+}$, there exists an increasing sequence of change points $\{\eta_\ell\}_{\ell=1}^{K}$ such that $0 = \eta_0 < \eta_1 < \cdots < \eta_K < n < \eta_{K+1} = n+1$ and
\begin{align}
    \label{model_cp}
    \Sigma^*_t \neq \Sigma^*_{t+1} \quad \text{if and only if} \quad t\in \{\eta_1,\dots, \eta_K\}.
\end{align}
For any $\ell \in \{1, \ldots, K\}$, denote the $\ell^{\mathrm{th}}$ jump function and its normalised version by $\nu_\ell =\Sigma^*_{\eta_{\ell}+1}-\Sigma^*_{\eta_{\ell}} $ and $\Upsilon_\ell = (\Sigma^*_{\eta_{\ell}+1}-\Sigma^*_{\eta_{\ell}})/ \|\Sigma^*_{\eta_{\ell}+1}-\Sigma^*_{\eta_{\ell}}\|_{L^2} =  \nu_\ell/\kappa_\ell$ respectively, where $\|\cdot\|_{L^2}$ is the function $L^2$ norm defined in \Cref{section_notation}. Let the minimal spacing between two consecutive change points and minimal jump size be
\begin{align*}
    \Delta = \min _{\ell=1, \ldots, K+1}|\eta_\ell-\eta_{\ell-1}|, \quad \text{and} \quad \kappa=\min _{\ell=1, \ldots, K} \kappa_\ell = \min _{\ell=1, \ldots, K} \|\Sigma^*_{\eta_{\ell}+1}-\Sigma^*_{\eta_{\ell}}\|_{L^2}.
\end{align*}
Our interest is to consistently estimate both the unknown number $K$ of the change points and also the time points $\{\eta_\ell\}_{\ell=1}^K$. To be specific, we seek estimators $\{\widehat{\eta}_\ell\}_{\ell=1}^{\widehat{K}}$, such that, as the sample size $n$ grows to infinity, it holds with probability going to $1$ that
\begin{align*}
    \widehat{K} = K, \quad \text{and} \quad \frac{\epsilon}{\Delta}= \underset{\ell = 1, \ldots, K}{\max} \frac{|\widehat{\eta}_\ell -\eta_\ell|}{\Delta}\rightarrow 0.  
\end{align*}
We refer to the quantity $\epsilon$ as the localisation error and $\frac{\epsilon}{\Delta}$ as the localisation rate. 
We further derive limiting distributions of the change point
estimators.

\subsection{Connections with relevant literature}
There are two key ingredients inside our problem formulated above: fragmented FDA and change point analysis. We discuss the relevant literature separately in this subsection.

\noindent \textbf{Fragmented functional data analysis.} Fragmented functional data, also known as functional snippets, are often collected in longitudinal studies when subjects enter a study at random times and are followed only over a short period. While this flexible design of studies offers advantages such as shorter completion times and potentially larger sample sizes, it also introduces new challenges for statistical analysis due to the complete missingness of design points in the off-diagonal region $\{(k,l): |l-k| >\delta; l,k \in [0,1]\}$. Consequently, the existing methods using interpolation techniques based on kernel smoothing and splines \citep[e.g.][]{yao2005functional,zhang2016sparse,peng2009geometric,xiao2013fast} fail to produce a consistent estimator of the covariance function in the off-diagonal region, and new exploration techniques have to be developed.

To the best of our knowledge, \citet{delaigle2016approximating} are arguably the first to study this problem. By modelling a discretised version of trajectories using a Markov chain, together with nonparametric smoothing techniques, they extend the observed fragments to the whole domain and construct covariance function estimators accordingly. \citet{descary2019recovering} and \citet{zhang2022nonparametric} handle this problem using matrix completion techniques requiring extra modifications when the observation grids are sampled randomly over the fragment. \citet{delaigle2021estimating} and \citet{lin2021basis} proceed by least square methods to numerically solve this problem via basis expansion. In a setting where the off-diagonal components can be modelled parametrically, \citet{lin2022mean} address the problem by the semiparametric divide-and-conquer method introduced in \citet{fan2007analysis} to overcome the challenges of missing data. The problem has also been studied more recently under a general setting of positive-semidefinite continuation by \citet{waghmare2022completion}, where they develop a kernel version of the canonical extension and consider the problem of extending a partially specified covariance kernel to the entire domain.

Note that the fragmented setting considered here represents a different and more challenging scenario compared to the partially observed functional data considered in \citet{kraus2015components}, \citet{gromenko2017cooling}, and \citet{liebl2019partially}. In their cases, the length of each fragment can be intuitively as large as the full domain $[0,1]$. Therefore a sufficiently large proportion of the curves are observed, indicating that most of the design points in the off-diagonal regions are still available.

\noindent \textbf{Change point analysis.} The considered problem of fragmented functional data with change points falls into the broad class of literature on change point analysis, which has a long history and remains a subject of active investigation till date. In the recent literature, the off-line version of the problem, where a full dataset is collected before analysis is carried out, has been studied on numerous settings including high dimensional linear regression \citep[e.g.][]{xu2022change, cho2022high}, time series models \citep[e.g.][]{chan2021optimal,yau2016inference}, networks \citep[e.g.][]{wang2021optimal,padilla2022change} and nonparametric models \citep[e.g.][]{padilla2021optimal, madrid2024change}, to name but a few. 

In the field of FDA, problems of change point analysis for the covariance functions have also received extensive attention, particularly in testing, localisation, and inference. For example, \citet{jaruvskova2013testing}, \citet{aue2020structural} and \citet{dette2021detect} study change point testing problems when there are changes in the eigensystems of the covariance operators; \citet{dette2022detecting} propose testing procedures to detect changes when jump sizes are measured in the sup-norm; \citet{stoehr2021detecting} localise change points with dimension reduction techniques when there is a shift in the covariance functions under weak dependence assumptions. Inference results on functional covariance changes for fully observed functional data can be found in \citet{horvath2022change} and \citet{jiao2023break}. However, to the best of our knowledge, most prior works focus only on single change point settings. Existing methods and theoretical results rely heavily on fully observed functional data without measurement errors, which may not be realistic in practice. Multiple change point localisation problems for partially observed fragmented functional data remain an open area for investigation, and inference results for change point estimators are yet to be explored. 

\subsection{List of contributions}
We summarise the contributions of this paper as follows.
\begin{itemize}
    \item To the best of our knowledge, this is the first study focusing on the change point analysis in the fragmented functional data. We account for the most general model where we allow model parameters to change with $n$, including the number of change points $K$, the smallest $L^2$-distance between two consecutive change points $\Delta$, the smallest difference between two consecutive covariance functions $\kappa$, and the size of sampling grids $m$. We propose the fragmented functional dynamic programming (FFDP) algorithm detailed in \Cref{DP}, which could be solved in polynomial time. 

    \item Under the regularity conditions in \Cref{section_theory}, we demonstrate that the initial estimators $\{\widehat{\eta}_\ell\}_{\ell=1}^{\widehat{K}}$ output by FFDP are consistent in localising multiple change points. Based on the consistent estimators $\{\widehat{\eta}_\ell\}_{\ell=1}^{\widehat{K}}$, we proceed with an additional step of local refinement, and derive the limiting distributions of the refined estimators $\{\Tilde{\eta}_\ell\}_{\ell=1}^{\widehat{K}}$ in \Cref{t_inference}. Depending on whether the jump size vanishes as $n$ diverges, the limiting distributions are shown to have two different non-degenerate regimes. This result is the first time seen for fragmented functional data, and it echos existing literature on fully observed settings \citep[e.g.][]{aue2018detecting,horvath2022change,jiao2023break}. 

    \item A byproduct of our theoretical analysis is a new non-asymptotic result on the estimation error of the covariance functions, detailed in \Cref{t_estimation}. By introducing a new identifiability condition, our result incorporates the effect of grid sizes within each fragment, which is novel in the fragmented FDA literature and holds independent interest.  

    \item  Results of extensive numerical experiments on both simulated and real data are presented in \Cref{section_numerical}. These results not only further support our theoretical findings but also demonstrate the benefits of our algorithm over current methods and highlight its practicality.
\end{itemize}

\subsection{Notation and organisation} \label{section_notation}
Throughout the paper, we adopt the following notation, some of them are only used in the Appendix. For any positive integer $a$, denote $[a] = \{1, \ldots, a\}$. For any set $\mathcal{S}$, denote $|\mathcal{S}|$ the cardinality of $\mathcal{S}$. For any $a,b \in \mathbb{R}$, let $\lceil a \rceil$ be the smallest integer greater than or equal to $a$, $\lfloor a \rfloor$ be the greatest integer less than or equal to $a$, and $a \vee b = \max\{a,b\}$. For any vectors $v, u \in \mathbb{R}^p$, let $\|v\|_2$ and $\|v\|_{\op}$ be its $\ell_2$- and operator norms respectively and let $v \odot u \in \mathbb{R}^{p^2}$ be the Kronecker product between $v$ and $u$. For any matrix $A \in \mathbb{R}^{p \times p}$, let $\|A\|_{\F}$ denote the Frobenius norm, $\mathrm{vec}(A) \in \mathbb{R}^{p^2}$ be the vectorisation of $A$ by stacking the columns of $A$ together, $(A)_{ij}$ denote the $(i,j)^{\mathrm{th}}$-entry of A, and write $A \succcurlyeq 0$ if $A$ is positive semidefinite.

Let $L^2([0,1])$ be the class of functions $f:\,[0,1] \rightarrow \mathbb{R}$ such that $f$ is square-integrable, i.e. $\{\int_{0}^1 f^2(t)\,\mathrm{d}t\}^{1/2}<\infty$. For any function $g:\, [0,1]^2 \rightarrow \mathbb{R}$, let $\|g\|_{L^2} = \{\int_{0}^1 \int_{0}^1 g^2(s,t) \,\mathrm{d}s\,\mathrm{d}t\}^{1/2}$ be its $L^2$-norm, and denote $L^2([0,1]^2)$ the class of bivariate functions $g:[0,1]^2 \rightarrow \mathbb{R}$ such that $g$ is square-integrable, i.e.~$\|g\|_{L^2} < \infty$.

For a deterministic or random $\mathbb{R}$-valued sequence $a_n$, write that a sequence of random variables $X_n = O_p(a_n)$ if $\lim_{M \rightarrow \infty}\lim\sup_n \mathbb{P}(|X_n|\geq M a_n)=0$, and write $X_n = o_p(a_n)$ if $\lim\sup_n \mathbb{P}(|X_n|\geq M a_n)=0$ for all $M>0$. For two deterministic or random $\mathbb{R}$-valued sequences $a_n,b_n>0$, write $a_n \gg b_n$ if $a_n/b_n \rightarrow \infty$ as $n \rightarrow \infty$. Write $a_n \lesssim b_n$, $a_n \gtrsim b_n$, and $a_n \asymp b_n$ if $a_n/b_n \leq C$, $b_n/a_n \leq C$, and $c \leq a_n/b_n \leq C$ for all  $n \geq 1$, where $c, C >0$ are absolute constants. For an $\mathbb{R}$-valued random variable $X$ and $k \in (0,1)\cup \{1,2\}$, let $\|X\|_{\psi_k}$ denote the Orlicz-$\psi_k$ norm, i.e.~$\|X\|_{\psi_k} = \inf \{t > 0: \mathbb{E}[\exp(\{|X|/t\}^k)]\leq 2\}$. The convergences in distribution and in probability are denoted by $\stackrel{\mathcal{D}}{\rightarrow}$ and $\stackrel{p}{\rightarrow}$ respectively.

The rest of the paper is organised as follows. \Cref{section_algorithm} consists of our two-step estimation procedure for change point localisation. Model assumptions and our main theoretical results are collected in \Cref{section_theory}. The consistency of the initial estimators and the limiting distributions of the refined estimators are detailed in Sections~\ref{section_consistency} and \ref{section_inference} respectively, followed by an intermediate result on the covariance function estimation errors in \Cref{section_byproduct}. Extensive numerical experiments are presented in \Cref{section_numerical}.  The paper is concluded with final discussions in \Cref{section_conclusion}, with proofs and technical details deferred to the Appendix.

\section{The fragmented functional data change point estimators and their refinement} \label{section_algorithm}
To achieve the goal of obtaining consistent change point estimators with tractable limiting distributions, we adopt a two-step procedure where we obtain a preliminary estimator in the first step and refine it using local information in the second step. Such a two-step procedure could be regularly seen in the recent literature \citep[e.g.][]{xu2022change, madrid2024change}. 

In the first step, we propose a fragmented functional dynamic programming (FFDP) algorithm, which is an $\ell_0$-penalised method to obtain the initial estimators. The procedure consists of two layers: estimation of covariance functions and estimation of change points. Motivated by the penalised least square estimators introduced in \citet{lin2021basis} and \citet{delaigle2021estimating}, we estimate the covariance function in the following way. Suppose that $\Phi = \{\phi_i\}_{i=1}^{\infty}$ is a complete orthonormal system of $L^2([0,1])$. Denote the collection of the pairs $\mathcal{O} = \{(2i-1,2i): i \in \{1,\ldots, \lfloor m/2 \rfloor\}\}$ and $\Phi_r(\cdot) = (\phi_1(\cdot), \ldots, \phi_r(\cdot))^\top$.  For any integer interval $I \subseteq [n]$, the estimated covariance function over interval $I$, $\widehat{\Sigma}_I$, is then defined as $\Phi_r^\top\widehat{C}_I\Phi_r$, where $r \in \mathbb{N}_+$ is a pre-specified level of truncation, 
\begin{align}
   \notag
   \widehat{C}_I &= \underset{\substack{C \in \mathbb{R}^{r \times r},\\C \succcurlyeq 0, \; \mathrm{symmetric}}}{\arg \min } \; Q(C)\\
   \label{loss_cov}
   &= \underset{\substack{C \in \mathbb{R}^{r \times r},\\C \succcurlyeq 0,\; \mathrm{symmetric}}}{\arg \min } \; \Bigg\{\frac{1}{|I|\lfloor \frac{m}{2} \rfloor}\sum_{t \in I}\sum_{(j,k) \in \mathcal{O}} \big\{Y_{t,j}Y_{t,k} - \Phi_r^\top(X_{t,j}) C\Phi_r(X_{t,k})\big\}^2 + \frac{\lambda \mathcal{J}(\Sigma_{r,C})}{2\sqrt{|I|\lfloor \frac{m}{2} \rfloor}}\Bigg\},
\end{align}
$\Sigma_{r,C} = \Phi_r^\top C \Phi_r$, $\lambda >0$ is a tuning parameter, and $\mathcal{J}(\cdot)$ is the roughness penalty defined as
\begin{align*}
    \mathcal{J}(\Sigma) =\int_0^1 \int_0^1 \Big(\frac{\partial^2 \Sigma}{\partial s^2}\Big)^2+ 2\Big(\frac{\partial^2 \Sigma}{\partial s \partial t}\Big )^2 + \Big(\frac{\partial^2 \Sigma}{\partial t^2}\Big)^2 \,\mathrm{d}s\,\mathrm{d}t.
\end{align*}

Given the above framework, we now present the $\ell_0$-penalised algorithm to obtain the initial change point estimators. Let $\mathcal{P}$ be an interval partition of $[n]$ into $K_{\mathcal{P}}$ disjoint intervals, i.e.
\begin{align*}
    \mathcal{P} = \Big\{\{1, \ldots, i_1\},\{i_1+1, \ldots, i_2\}, \ldots, \{i_{K_{\mathcal{P}}-1}+1, \ldots, i_{K_{\mathcal{P}}}-1\}\Big\},
\end{align*}
for some integers $1 < i_1 < \cdots < i_{K_{\mathcal{P}}} = n+1$ and $K_{\mathcal{P}} \geq 1$. For a tuning parameter $\xi>0$, let 
\begin{align} \label{loss_l0}
    \widehat{\mathcal{P}}  =\underset{\mathcal{P}}{\operatorname{argmin}}\Big\{\sum_{I \in \mathcal{P}} H(\widehat{C
 }_I, I)+\xi|\mathcal{P}|\Big\},
\end{align}
where $|\mathcal{P}| = K_{\mathcal{P}}$, the minimisation is taken over all possible interval partitions of $\{1, \ldots, n\}$, and the loss function $H(\cdot, \cdot)$ is defined as
\begin{align} \label{loss_localisation}
     H(C, I)= \begin{cases}
         \sum_{t \in I} \sum_{(j,k) \in \mathcal{O}} \big\{Y_{t,j}Y_{t,k} -\Phi_r^\top(X_{t,j}) C\Phi_r(X_{t,k})\big\}^2, \quad & |I| \geq \frac{\xi}{m},\\
         0, \quad & \text{otherwise}.
     \end{cases}
\end{align}
The change point estimators induced by the solution to \eqref{loss_l0} are simply obtained by taking all the right endpoints of the intervals $I \in \widehat{\mathcal{P}}$, except $n$. Even though the optimisation problem \eqref{loss_l0} is non-convex due to the presence of the penalty term $|\mathcal{P}|$, known as the minimal partition problem, it can still be solved in polynomial time by the dynamic programming approach detailed in \Cref{DP} with an overall computational cost of order $O\big(n^2 \mathcal{T}(n)\big)$, where $\mathcal{T}(n)$ denotes the computational cost of computing $\widehat{C}_I$ with $|I| = n$ \citep[see e.g.~Algorithm 1 in][]{friedrich2008complexity}.

\begin{algorithm}
    \caption{Fragmented functional dynamic programming. FFDP ($\{(X_{t,j},  Y_{t,j})\}_{t=1,j=1}^{n,m}, \Phi, r, \lambda, \xi$)}\label{DP}
    \begin{algorithmic}
        \Require Data $\{(X_{t,j},  Y_{t,j})\}_{t=1,j=1}^{n,m}$, complete orthonormal system $\Phi$, level of truncation $r\in \mathbb{Z}_{+}$, and tuning parameter $\lambda, \xi> 0$
        \State $(\mathcal{B}, \mathfrak{p}, M, k) \gets (\varnothing, \mathbf{1}_{n}, (-\xi, \infty, \infty, \ldots, \infty),n)$  \Comment{$M \in \mathbb{R}^{n+1}$ and $\mathbf{1}_{n} =(1,\ldots,1) \in \mathbb{R}^{n}$}
        \For{$e \in \{1, \ldots, n\}$}
        \For{$s \in \{1, \ldots, e\}$}
        \If{$|e-s+1| \geq \xi/m$}
        \State $L \gets M[s] + \xi + H(\widehat{C}_{[s,e]},[s,e])$ \Comment{$M[s]$ denotes the $s^{\mathrm{th}}$ entry of $M$}
        \Else 
        \State $L \gets M[s] + \xi$
        \EndIf
        \If{$L < M[e+1]$}
        \State $M[e+1] \gets L$, $\mathfrak{p}[e] \gets s$
        \EndIf
        \EndFor
        \EndFor
        \While{$\mathfrak{p}[k] > 1$}
        \State $b \gets \mathfrak{p}[k]-1$; $\mathcal{B} = \mathcal{B} \cup b$; $k \gets b$
        \EndWhile
        \Ensure $(\{\widehat{\eta}_\ell\}_{\ell=0}^{\widehat{K}+1} = \{0\} \cup \mathcal{B} \cup \{n\}, \widehat{K}= |\mathcal{B}|, \{\widehat{C}_{\ell}\}_{\ell=1}^{\widehat{K}+1})$\Comment{Set $\widehat{\eta}_{0} = 0$, and $\widehat{\eta}_{\widehat{K}+1} = n$}
    \end{algorithmic}
\end{algorithm}

Denote our preliminary estimators obtained from FFDP by $\{\widehat{\eta}_\ell\}_{\ell=1}^{\widehat{K}}$, and let the set of the corresponding estimators \eqref{loss_cov}  of covariance functions be $\{\widehat{C}_{k}\}_{k=1}^{\widehat{K}+1}$. Our final estimators are obtained via a local refinement step. For any $\ell \in \{1, \ldots, \widehat{K}\}$, let
\begin{align} \label{refine_interval}
    s_\ell = 3\widehat{\eta}_{\ell-1}/4+ \widehat{\eta}_\ell/4 \quad \text{and} \quad e_\ell = \widehat{\eta}_\ell/4+3\widehat{\eta}_{\ell+1}/4,
\end{align}
then the final estimators are proposed as
\begin{align}
    \label{loss_refine}
    \Tilde{\eta}_\ell &= \underset{s_\ell <\eta < e_\ell}{\arg \min} \; \mathcal{L}_\ell(\eta) \\ 
    \notag
    &= \underset{s_\ell <\eta < e_\ell}{\arg \min} \; \Bigg\{\sum_{t = s_\ell+1}^{\eta} \sum_{(j,k) \in \mathcal{O}} \big\{Y_{t,j}Y_{t,k} -\Phi_{r}^\top(X_{t,j})\widehat{C}_{\ell}\Phi_{r}(X_{t,k})\big\}^2 \\
    \notag
    & \hspace{2.5cm}+ \sum_{t = \eta+1}^{e_\ell}\sum_{(j,k) \in \mathcal{O}} \big\{Y_{t,j}Y_{t,k} -\Phi_{r}^\top(X_{t,j})\widehat{C}_{\ell+1}\Phi_{r}(X_{t,k})\big\}^2\Bigg\}.
\end{align}
Note that the choice of the constants $1/4$ and $3/4$ in \eqref{refine_interval} is arbitrary, and the intuition behind the step of local refinement is that, provided the initial estimators are good enough, i.e.~the localisation errors are a small fraction of the minimal spacing $\Delta$, each constructed interval $(s_\ell, e_\ell)$ contains one and only one true change point $\eta_\ell$ with large probability. Thus, performing an extra step on such an interval will result in a sharper localisation rate.

\section{Consistent estimation and limiting distributions} \label{section_theory}
In this section, we provide the theoretical guarantees for the change point estimators output by FFDP together with their limiting distributions after refinement. We begin by formulating the detailed setup of Model~\eqref{model_obs}. 
\begin{assumption}[Model] \label{a_model}
    The observed data $\{(X_{t,j},  Y_{t,j})\}_{t=1,j=1}^{n,m} \subseteq [0,1] \times \mathbb{R}$ are generated by Model~\eqref{model_obs} satisfying \eqref{model_cp}. 
    \begin{enumerate}[label=\textbf{\alph*.}]
        \item \label{a_model_frag}(Observed fragments). At the time $t \in [n]$, the noisy functional data are only observed on a subject-specific random interval $O_{t} = [a_{t}, a_{t}+\delta]$, where $\{a_{t}\}_{t=1}^{n}$ are independently and identically distributed (i.i.d.) Uniform random variables over the interval $[0,1-\delta]$, independent of $n$ random functions $\{f_t\}_{t=1}^n$, and $\delta \in (0,1]$ is an absolute constant.

        \item \label{a_model_grid}(Discrete grids). For any $t \in [n]$, and $j \in [m]$, let $X_{t,j} = a_t + U_{t,j}$, where $a_t$ is sampled under the \Cref{a_model}\ref{a_model_frag}, and $\{U_{t,j}\}_{t=1,j=1}^{n,m} \stackrel{\mathrm{i.i.d.}}{\sim} \mathrm{Uniform}[0,\delta]$. We further assume that $\{a_t\}_{t=1}^n$ and $\{U_{t,j}\}_{t=1,j=1}^{n,m}$ are mutually independent.

        \item \label{a_model_error} (Measurement error). Assume that $\{\varepsilon_{t,j}\}_{t=1,j=1}^{n,m}$ is a collection of i.i.d.~random variables with marginal mean zero and  $\|\varepsilon_{t,j}\|_{\psi_2} \leq C_{\varepsilon}$, where $C_\varepsilon >0$ is an absolute constant. 

        \item \label{a_model_function} (Random functions). Assume that the random functions are uniformly sub-Gaussian, i.e.~for any $t \in [n]$ and $x \in [0,1]$, it holds that $\|f_t(x)\|_{\psi_2} \leq C_{f}$, where $C_f >0$ is an absolute constant.

        \item \label{a_model_identifiable} (Covariance functions). For any $t \in [n]$, we assume that the covariance function $\Sigma^*_t$ is continuous and belongs to a $\mathcal{T}_{u, \delta}$-identifiable family $\mathcal{C}$ over $[0,1]^2$, which is defined as the class of continuous functions such that for any two functions $f_1, f_2 \in \mathcal{C}$, if $f_1(s_1,s_2) = f_2(s_1,s_2)$ for all $(s_1,s_2) \in [u, u+\delta]^2$ where $u \in [0, 1-\delta]$ is a fixed point, then $f_1(s_1,s_2) = f_2(s_1,s_2)$ for all $(s_1,s_2) \in [0,1]^2$.
    \end{enumerate}
\end{assumption}

Assumptions \ref{a_model}\ref{a_model_frag}~and \ref{a_model}\ref{a_model_grid}~provide details of the sampling mechanism in Model \eqref{model_obs}. Here, we assume that the starting point of each fragment is sampled uniformly over the interval $[0,1-\delta]$, and at each time $t$, conditioning on the starting point of the fragment $a_t$, the $m$ observation grids are sampled uniformly and independently over $[a_t, a_t+\delta]$. Several things could be relaxed in the setup, including the sampling distribution, homogeneous design among functions and random design assumptions.  

It is worth mentioning that under Model \eqref{model_obs}, for brevity, the mean function is assumed without loss of generality to be $0$. We remark that this assumption could be easily relaxed to accommodate unknown mean functions. To achieve this, note that under Assumptions \ref{a_model}\ref{a_model_frag}~and \ref{a_model}\ref{a_model_grid}, even though only short fragments of the functional data are observed, smoothing techniques such as \citet{yao2005functional} could still be applied to mean estimation since it holds that $\inf_{x \in (0,1)} f_X(x) >0$, where $f_X$ is the unconditional density of the observation grids $X$. Moreover, methods specifically tailored to the problem of mean function estimation for fragmented functional data could also be found in the existing literature \citep[e.g.][]{lin2021basis,lin2022mean}. The effect of estimating the mean function in establishing the convergence rate of our change point estimators is negligible compared to the effect introduced by the covariance estimation errors. 

In Assumptions \ref{a_model}\ref{a_model_error}~and~\ref{a_model}\ref{a_model_function}, the tail behaviour of both the measurement errors and the random functions are regulated. Similar sub-Gaussian assumptions on the tail behaviours could also be seen in the other literature on FDA \citep[e.g.][]{guo2023sparse, cai2024transfer}, and in our work, such conditions also help us to control the moments and unlock the functional central limit theorem in the inference step. Compared with the assumptions on the moments of random functions and measurement errors in \citet{lin2021basis}, our assumptions on the tail behaviours are stronger to derive non-asymptotic results. Following a similar treatment by using different concentration inequalities (e.g.~\Cref{t_concentration_subWeibull} in the Appendix), these tail behaviours could be relaxed to capture the scenarios with heavy tail distributions (e.g. sub-Weibull).

\Cref{a_model}\ref{a_model_identifiable}~is the identifiability assumption necessary to address the extrapolation nature of the covariance estimation problem in fragmented FDA aiming for a unique completion. For example, real analytic functions satisfy such assumption due to the analytic continuity \citep[Corollary 1.2.6 in][]{krantz2002primer}. A similar type of condition is commonly seen in the existing literature. For instance, \citet{descary2019recovering} assume that the covariance function possesses a finite number of real and analytic eigenfunctions; \citet{delaigle2021estimating} propose the concept of linear predictability, which assumes the values of the process $f_t$ on a sub-interval could be linearly predicted by solely using the information from the same process within another sub-interval; and \citet{lin2022mean} investigate the scenarios where the true covariance function is a member of a pre-specified semiparametric family. 

Our identifiability assumption here borrows the idea from the $\mathcal{T}_\delta$-identifiability assumption introduced in \citet{lin2021basis}. They assume if two functions $f_1, f_2 \in \mathcal{C}$ align in the diagonal band $\mathcal{T}_\delta = \{(s,t)\in [0,1]^2 : |s-t| \leq \delta\}$, then $f_1(s,t) = f_2(s,t)$ for all $(s,t) \in [0,1]^2$. It is worth noting that our assumption here is stronger for the sake of deriving the non-asymptotic result.

\subsection{Consistency of the preliminary estimators} \label{section_consistency} 

In this subsection, for the estimators output by our FFDP algorithm detailed in \Cref{DP}, we present high probability upper bounds on their
the localisation errors and demonstrate their consistency.

As discussed in \citet{lin2021basis}, the choice of the orthonormal system $\Phi$ in \eqref{loss_cov} and the true covariance functions' functional class result in varying performances of our estimators. The effect of the misspecification is captured by the approximation error $\rho_{r,t}$.  In this paper, we only focus on a special case and claim that theoretical justifications under other possible settings can also be established similarly with minor modifications.

\begin{assumption}[Covariance function] \label{a_localisation_cov}
In addition to \Cref{a_model}, we further assume that
for any $t \in [n]$, $\Sigma^*_t$ is periodic and belongs to a Sobolev space with parameter $q\geq 1$, $\mathcal{W}^{q}([0,1]^2)$, where
\begin{align*}
    \mathcal{W}^{q}([0,1]^2) = \{f \in L^2([0,1]^2): \ &\text{for each non-negative double index} \; \alpha \ \text{such that} \  |\alpha| \leq q, \\
    &D^{\alpha}f \in L^2([0,1]^2)\},
\end{align*}      
where for any $g:[0,1]^2 \rightarrow \mathbb{R}$ and any double index $\alpha = (\alpha_1, \alpha_2)$ of nonnegative integers, denote $|\alpha| = \alpha_1 + \alpha_2$, and $D^\alpha g$ is the weak derivative.
\end{assumption}

\Cref{a_localisation_cov} restricts our analysis to a special periodic Sobolev space defined in A.11.d in \citet{canuto2007spectral}. This is also one of the examples considered in \citet{lin2021basis}, Example 5 therein chooses~$\Phi$ as the Fourier basis, that
\begin{align}\label{eq_fourier}
    \phi_1(t) =1, \; \phi_{2k}(t) = \cos(2k\pi t),\; \phi_{2k+1}(t) =\sin(2k\pi t), \; \text{for} \; k \geq 1.
\end{align}
For $r \in \mathbb{N}_+$, using the first $r$ Fourier basis functions, the approximation error $\rho_{r,t}$ satisfies that $\rho_{r,t} = \|\Sigma^*_t - \Phi_r^\top C^*_{r,t}\Phi_r\|_{L^2} \leq C_{\rho}r^{-q}$, where $C^*_{r,t} \in \mathbb{R}^{r \times r}$ is a symmetric matrix with entries 
\[
    (C^*_{r,t})_{hl}= \int_{0}^1 \int_{0}^1 \Sigma_t^*(s,t)\phi_h(s)\phi_l(t) \,\mathrm{d}s\,\mathrm{d}t, \quad h,l \in [r],
\]
and $C_\rho >0$ is an absolute constant. We then detail our assumption on signal strength.
\begin{assumption} \leavevmode \label{a_localisation_jump} 
    \begin{enumerate}[label=\textbf{\alph*.}]
        \item \label{a_localisation_jump_snr} (Signal-to-noise ratio) Assume that there exists a sufficiently large absolute constant $C_{\mathrm{snr}} >0$ such that 
        \begin{align*}
            \Delta \kappa^2 \geq C_{\mathrm{snr}}\alpha_n K\Big\{\frac{r_{\opt}^4\log^2(n)}{\delta^2\zeta_{\delta}^2} \vee \frac{r_{\opt}^{6}\log(n)}{\delta^2\zeta_{\delta}^2 m}\Big\},
        \end{align*}
        where $\{\alpha_n\}_{n\in \mathbb{N}_{+}} > 0$ is any diverging sequence,
        \begin{align} \label{eq_r_opt}
            r_{\opt} = \begin{cases}
           \Big\lceil C_r\Big(\frac{nm}{\delta^4 \log(n)}\Big)^{\frac{1}{4+2q}}\Big\rceil, \quad & m \leq C_m n^{\frac{1}{1+q}}\delta^{-\frac{4}{1+q}}\log^{-\frac{3+q}{1+q}}(n),\\
           \Big\lceil C_r'\Big(\frac{n}{\delta^4 \log^2(n)}\Big)^{\frac{1}{2+2q}}\Big\rceil, &\text{otherwise},
       \end{cases}
        \end{align}
        and $C_r, C_r', C_m >0$ are absolute constants.

        \item \label{a_localisation_jump_lb} (Jump size) Assume that there exists an absolute constant $C_{\mathcal{K}} >0$ such that $\kappa \geq C_{\mathcal{K}} r_{\opt}^{-q}$, with $r_{\opt}$ in \eqref{eq_r_opt}.
    \end{enumerate}
\end{assumption}
It has been commonly seen in the existing literature \citep[e.g.][]{wang_univariate, verzelen2023optimal} that $\Delta\kappa^2$ is essentially the quantity capturing the fundamental difficulty of offline change point localisation problems.  Conditions in the form of \Cref{a_localisation_jump}\ref{a_localisation_jump_snr}~are ubiquitous in the change point literature and ensure theoretical guarantees of change point estimation. 

In \Cref{a_localisation_jump}\ref{a_localisation_jump_lb}, we require that $\kappa^2$ is lower bounded by the approximation error $C_{\mathcal{K}} r_{\opt}^{-q}$ when the level of truncation $r$ is chosen optimally. Intuitively speaking, instead of targeting the differences between two covariance functions, our proposed algorithm searches for the change points where the differences in the $r_{\opt} \times r_{\opt}$ coefficient matrices obtained by projecting the covariance function along the direction of orthonormal basis, are maximised. The minimum jump size, consequently, should not be too small to avoid being dominated by the information loss after truncating using the first $r_{\opt}$ basis.  The information loss is therefore approximately of the order of $r_{\opt}^{-q}$ under \Cref{a_localisation_cov} when $\Phi$ is chosen as the Fourier basis.

We are now ready to present our first theorem, showing the consistency of FFDP.
\begin{theorem} \label{t_localisation}
    Let the data $\{(X_{t,j},  Y_{t,j})\}_{t=1,j=1}^{n,m}$ be generated from Model \eqref{model_obs} satisfying Assumptions \ref{a_model}, \ref{a_localisation_cov} and \ref{a_localisation_jump}. Let $\{\widehat{\eta}_\ell\}_{\ell=1}^{\widehat{K}}$ be the estimated change points output by \textnormal{FFDP} with $\Phi$ being the set of Fourier basis functions defined in \eqref{eq_fourier},  
    \begin{align}\label{localisation_r}
         r=\begin{cases}
           \Big\lceil C_r\Big(\frac{nm}{\delta^4 \log(n)}\Big)^{\frac{1}{4+2q}} \Big\rceil, \quad & \text{when} \quad m \leq C_m n^{\frac{1}{1+q}}\delta^{-\frac{4}{1+q}}\log^{-\frac{3+q}{1+q}}(n),\\
           \Big\lceil C_r'\Big(\frac{n}{\delta^4 \log^2(n)}\Big)^{\frac{1}{2+2q}}\Big\rceil, &\text{otherwise},
       \end{cases}
    \end{align}
    $\lambda = C_{\lambda}\big\{r^{-11/2}\sqrt{m\log(n)} \vee r^{-9/2}\sqrt{\log(n)}\big\}$ and $\xi = C_{\xi}K\big\{\frac{r^4 m\log^2(n)}{\delta^2\zeta_{\delta}} \vee \frac{r^6\log(n)}{\delta^2 \zeta_{\delta}}\big\}$, with $C_r$, $C_r'$, $C_m$, $C_\lambda$, $C_{\xi} >0$ being absolute constants, and $\zeta_{\delta} >0$ being a constant depending on $\delta$ and $\Sigma^*$. It holds with probability at least $1-3n^{-3}$ that 
    \begin{align*}
        \widehat{K} = K, \; \mathrm{and} \; \max_{\ell=1, \dots, K} \kappa_{\ell}^2|\widehat{\eta}_\ell -\eta_\ell| \leq C_{\epsilon}K\Big\{\frac{r^4\log^2(n)}{\delta^2\zeta_{\delta}^2} \vee \frac{r^6\log(n)}{\delta^2 \zeta_{\delta}^2m}\Big\},
    \end{align*}
    where $C_{\varepsilon} >0$ is an absolute constant.
\end{theorem}
From \Cref{t_localisation}, we have that with properly chosen values of the tuning parameters and with probability tending to $1$ as $n\rightarrow \infty$,
\begin{align*}
     \max_{\ell=1, \dots, K} \frac{|\widehat{\eta}_\ell- \eta_\ell|}{\Delta} \leq \frac{C_{\varepsilon}K}{\Delta\kappa^2}\Big\{\frac{r^4\log^2(n)}{\delta^2\zeta_{\delta}^2} \vee \frac{r^6\log(n)}{\delta^2 \zeta_{\delta}^2m}\Big\} \leq \frac{C_\varepsilon}{C_{\mathrm{SNR}}\alpha_n} \rightarrow 0,
\end{align*}
where the second inequality follows from \Cref{a_localisation_jump}\ref{a_localisation_jump_snr}, hence the estimators are consistent. 

To further understand \Cref{t_localisation}, we discuss it in more detail in the following, focusing on the tuning parameter selection and the effects of model parameters.

\noindent \textbf{Choices of the tuning parameters.} Apart from the selected complete orthonormal system $\Phi$, our FFDP algorithm has three tuning parameters: $(1)$ roughness penalty $\lambda$, $(2)$ level of truncation~$r$ and $(3)$ penalty on the level of partition and interval cutoff $\xi$.

The parameter $\lambda$ affects the overall performances of the covariance function estimators defined in \eqref{loss_cov}. It controls the smoothness of the estimator and in particular, a proper choice of $\lambda$ helps to prevent over-fitting. The parameter $r$ is the selected number of leading terms and controls the approximation quality. Compared with the choices of $\lambda$ and $r$ in \Cref{t_estimation}, different values of the tuning parameters are used here, since in \Cref{t_localisation} a range of interval lengths are considered whereas a single interval is considered in \Cref{t_estimation}. We choose unified tuning parameter values in \Cref{t_localisation} to cover all intervals considered and to simplify arguments.

The tuning parameter $\xi$ prevents over-fitting while searching for the optimal interval partition. The order of $\xi$ can be intuitively explained as a high probability upper bound on the absolute difference between $H\big(\widehat{C}_{I_1}, I_1\big) + H\big(\widehat{C}_{I_2}, I_2\big)$ and $H\big(\widehat{C}_{I}, I\big)$, where $I_1$ and $I_2$ are two non-overlapping and adjacent intervals such that $I = I_1 \cup I_2$ containing no true change point. While in \eqref{loss_localisation}, $\xi$ also restricts interval splitting to the intervals with length at least $\xi/m$, but at the same time implies that the change point localisation errors are at least of order $\xi/m$. However, this still leads to the consistency since by \Cref{a_localisation_jump}\ref{a_localisation_jump_snr} and the fact that $\kappa^2 \lesssim 1$, it holds that $\Delta \gtrsim \alpha_n\xi/(m\zeta_{\delta}) \gg \xi/m$. The choice of $\xi$ also serves as an upper bound on the optimal choices of $r$ (see \Cref{r_choice_r} in the Appendix).

\noindent \textbf{Effects of $n$, $m$ and $\delta$.}
In our setup in \Cref{a_model}, a total of $n$ curves are observed at $m$ discrete points across a fragment of length $\delta$. We allow a full range of choices of $m \in \mathbb{Z}_+\setminus\{1\}$, reflecting scenarios ranging from sparse to dense observations over fragments. For a better illustration of the effects of $n$, $m$ and $\delta$ on the localisation errors, by plugging in the choices of $r$ in \eqref{localisation_r} into the localisation errors in \Cref{t_localisation}, we have the following holds with high probability.
\begin{itemize}
    \item When $m \lesssim n^{\frac{1}{1+q}}\delta^{-\frac{4}{1+q}}\log^{-\frac{3+q}{1+q}}(n)$, we select $r \asymp \Big(\frac{nm}{\delta^4 \log(n)}\Big)^{\frac{1}{4+2q}}$, and the localisation error is
        \begin{align*}
            |\widehat{\eta}_\ell -\eta_\ell| \leq \frac{C_{\varepsilon}Kn^{\frac{3}{2+q}}m^{-\frac{-1+q}{2+q}}\log^{\frac{-1+q}{2+q}}(n)\zeta_{\delta}^{-2}\delta^{-\frac{16+2q}{2+q}}}{\kappa^2}.
        \end{align*}

    \item When $m \gtrsim n^{\frac{1}{1+q}}\delta^{-\frac{4}{1+q}}\log^{-\frac{3+q}{1+q}}(n)$, we select $r \asymp \Big(\frac{n}{\delta^4 \log^2(n)}\Big)^{\frac{1}{2+2q}}$, and the localisation error is
        \begin{align*}
             |\widehat{\eta}_\ell -\eta_\ell| \leq \frac{C_{\varepsilon}Kn^{\frac{2}{1+q}}\log^{\frac{-2+2q}{1+q}}(n)\zeta_{\delta}^{-2}\delta^{-\frac{10+2q}{1+q}}}{\kappa^2}.
        \end{align*}
\end{itemize}

To further illustrate the above result, it is natural to compare our results with the existing change point literature on functional covariance changes. However, it is worth noting that there are arguably no existing results regarding localisation rates that we can use for direct comparisons. Most of the existing literature, to our best knowledge, focuses on change point testing problems and therefore is not comparable. 

Based on a range of different algorithms, high probability upper bounds on localisation errors can be seen in similar problems under different models, for example, functional mean changes and non-parametric model changes. We list results (up to logarithm and constant factors) here for completeness. For functional mean change problems. \citet{madrid2022change} investigate situations when the mean functions are Holder smooth, $\mathcal{H}_q(L)$ \footnote{For $q>0$, let $\mathcal{H}_q(L)$ be the set of functions $f:\mathbb{R}\rightarrow \mathbb{R}$ such that $f$ is $\lfloor q\rfloor$-times differentiable for all $x\in \mathbb{R}$, and satisfy $|f(x)-f_{x_0}^q(x)|\leq L|x-x_0|^q$ for all $x,x_0 \in \mathbb{R}$, where $f_{x_0}^q$ is its Taylor polynomial of degree $\lfloor q\rfloor$ at $x_0$ and is defined as $f_{x_0}^q(x) = \sum_{s\leq \lfloor q\rfloor}\frac{(x-x_0)^q}{q!}\partial^s f(x_0)$.}, and show that localisation errors are of the order of $(1+~n^{\frac{1}{2q+1}}m^{-\frac{2q}{2q+1}})\kappa^{-2}$; \citet{rice2022consistency} consider the problem under assumptions of lower bounded minimal jump size and $\Delta \geq Cn^{1-\omega}$ for $0 \leq \omega < 1/8$, and reach the result of $|\widehat{\eta}_\ell - \eta_\ell| \lesssim n^{2\omega+1/2}$. The problem concerning non-parametric model changes when the underlying density functions belong to the class of $\mathcal{H}_q(L)$ can be found in \citet{madrid2024change}, in which they prove that with their designed algorithm, it holds that $|\widehat{\eta}_\ell - \eta_\ell| \lesssim \kappa^{-2}n^{\frac{1}{2q+1}}$. Even though the above-mentioned rates are not directly comparable with our results in \Cref{t_localisation}, all of the rates follow the usual form of the localisation error rate in the offline change point literature, which is
\begin{align*}
    |\widehat{\eta}_\ell - \eta_\ell| \lesssim \frac{\text{variance of the noise}}{\kappa^2}.
\end{align*}

In our results in \Cref{t_localisation}, the smoothness parameter of the true covariance functions $q$ directly influences the rates of localisation, with smoother covariance functions leading to better localisation rates. Moreover, we observe that there are two regimes with a boundary at $m 
\asymp n^{\frac{1}{1+q}}\delta^{-\frac{4}{1+q}}\log^{-\frac{3+q}{1+q}}(n)$. When the sampling frequency $m$ is small, that is $m \lesssim n^{\frac{1}{1+q}}\delta^{-\frac{4}{1+q}}\log^{-\frac{3+q}{1+q}}(n)$, the localisation errors depend jointly on the value of $n$ and $m$. Provided that the true covariance functions are smooth enough $(q >1)$, having more observations inside each fragment always contributes to a better localisation rate in the sparse regime. As the sampling frequency increases above the boundary,  a transition to a dense regime occurs, where the localisation rate solely depends on the value of $n$. A similar phenomenon of phase transition is commonly observed in the FDA literature concerning sparsely observed functional data \citep[e.g.][]{cai2011optimal, zhang2016sparse}. However, we conjecture that the rate of phase transition in our work is not optimal and there is still potential for improvement. 

The length of the observed fragment $\delta$ controls the level of missingness of the observed functions, which further leads to the missingness of the design points in the off-diagonal region when estimating covariance functions. The parameter $\delta$ therefore reflects the complexity of the problem, with a larger value resulting in an easier problem and hence a better localisation rate. In the extreme case, when $\delta = 1$, we revert to the conventional setting of sparsely sampled functional data, with no missingness of design points in the off-diagonal regions. Furthermore, $\delta$ also captures the difficulty of the identifiability problem through the constant $\zeta_\delta$ as shown in \Cref{l_rec} in the Appendix. Denote $\delta_1$ the observed fragment length, and $\delta_2$ the parameter in the $\mathcal{T}_{u, \delta_2}$-identifiable family in \Cref{a_model}\ref{a_model_identifiable}. We conjecture that $\delta_1$ has to be greater than $\delta_2$, with a larger value of $\delta_2$ indicating an easier problem, but specifics of this relationship have yet to be covered. We also conjecture there is room for improvement in the exponent of $\delta$.

\subsection{Limiting distributions} \label{section_inference}
    Since inference is usually a more challenging task than deriving high probability upper bounds on localisation errors, we require stronger assumptions. This phenomenon is commonly seen in the change point literature on inference, see for example \citet{xu2022change} and \citet{madrid2024change}. A stronger assumption on the structure of the covariance functions is required compared to what is outlined in \Cref{a_localisation_cov}.

\begin{assumption}[Covariance function] \label{a_inference}
    Assume that there exists an absolute constant $\Tilde{r} \in \mathbb{N}_+$ such that for any $t \in [n]$, it holds that 
    \begin{align*}
        \Sigma_t^*(\cdot, \cdot) =\Phi^\top_{\Tilde{r}}(\cdot)C^*_t\Phi_{\Tilde{r}}(\cdot) = \sum_{1 \leq h, l \leq \Tilde{r}} (C^*_t)_{hl}\phi_h(\cdot)\phi_l(\cdot),
    \end{align*}
    where $(C^*_t)_{hl} = \int_{0}^1 \int_{0}^1 \Sigma_t^*(s,t)\phi_h(s)\phi_l(t) \,\mathrm{d}s\,\mathrm{d}t$, and $\Phi_{\Tilde{r}}(\cdot)$ is the vector formed by the first $\Tilde{r}$ Fourier basis as given in \eqref{eq_fourier}, with $(\Phi_{\Tilde{r}}(\cdot))_i = \phi_i(\cdot)$ for any $i \in [\Tilde{r}]$. We further assume that $m$ is fixed, and for any $\ell \in [K]$, there exists an absolute constant $\varpi_\ell >0$ such that $\sum_{(j,k) \in \mathcal{O}}\mathbb{E}[\Upsilon^2_\ell(X_{1,j}, X_{1,k})] \rightarrow \varpi_\ell$ as $n \rightarrow \infty$.
\end{assumption}
In \Cref{a_inference}, we impose a specific parametric structure on the true underlying covariance functions. This parametric structure is used to avoid selecting the tuning parameter $r$ in our FFDP algorithm (\Cref{DP}). A limiting drift coefficient $\varpi_\ell$ is also introduced for all $\ell \in [K]$ to characterize the limiting distributions. 

We state a signal-to-noise ratio below.
\begin{assumption}[Signal-to-noise ratio] \label{a_snr_strong}\leavevmode  
\begin{enumerate}[label=\textbf{\alph*.}]
    \item \label{a_snr_strong_1} Assume that there exists a sufficiently large absolute constant $C'_{\mathrm{snr}} >0$ and a diverging sequence $\{\alpha_n\}_{n\in \mathbb{N}_{+}} > 0$ such that
    \begin{align*}
        \Delta\kappa^2 \geq  \frac{C'_{\mathrm{snr}}\alpha_n K\Tilde{r}^4 \log^2(n)}{\delta^2\zeta_{\delta}^2}.
    \end{align*}

    \item \label{a_snr_strong_2} In addition, we further assume that $\{\alpha_n\}_{n\in \mathbb{N}_{+}}$ diverges faster than $\delta^{-2}\log(n)$, i.e.
    \begin{align*}
        \alpha_n \gg \delta^{-2}\log(n).
    \end{align*}
\end{enumerate}
\end{assumption}
Assumptions \ref{a_snr_strong}\ref{a_snr_strong_1}~and \ref{a_snr_strong}\ref{a_snr_strong_2}~impose lower bounds on $\Delta\kappa^2$ to ensure theoretical guarantees on inference procedures. This assumption is stronger compared to the signal-to-noise ratio stated in \Cref{a_localisation_jump}\ref{a_localisation_jump_snr}, and is required here to establish the uniform tightness of refined estimators in \Cref{t_inference}. 

For completeness, we present localisation errors of the FFDP algorithm under new Assumptions \ref{a_inference} and \ref{a_snr_strong}\ref{a_snr_strong_1}~in the next Corollary.
\begin{corollary} \label{c_inf_localisation}
    Suppose that Assumptions \ref{a_model}, \ref{a_inference} and \ref{a_snr_strong}\ref{a_snr_strong_1}~hold. Let $\{\widehat{\eta}_\ell\}_{\ell=1}^{\widehat{K}}$ be the estimated change points output by \textnormal{FFDP} in \Cref{DP}, with $r = \Tilde{r}$, $\lambda = C_{\lambda}\Tilde{r}^{-11/2}\sqrt{m\log(n)}$, $\xi = \frac{C_{\xi}K \Tilde{r}^4 m\log^2(n)}{\delta^2\zeta_{\delta}} $, with $C_{\lambda}$, $C_{\xi} >0$ being absolute constants, and $\zeta_{\delta} >0$ being a constant depending on $\delta$ and $\Sigma^*$. It holds with probability at least $1-3n^{-3}$ that $\widehat{K} = K$ and
    \begin{align*}
        \max_{\ell=1, \dots, K} \kappa_{\ell}^2|\widehat{\eta}_\ell -\eta_\ell| \leq \frac{C_{\varepsilon}K\Tilde{r}^4\log^2(n)}{\delta^2\zeta_{\delta}^2},
    \end{align*}
    where $C_{\varepsilon} >0 $ is an absolute constant.
\end{corollary}
Note that our rate of localisation errors outlined in \Cref{t_localisation} comprises two parts, where the first term arises from the estimation error of the coefficient matrix obtained by projecting the covariance functions along the direction of orthonormal basis functions, and the second term originates from the approximation error. In comparison to the result in \Cref{t_localisation}, only the term associated with the estimation error remains in \Cref{c_inf_localisation}, which is an immediate consequence of \Cref{a_inference}. Therefore, the proof of \Cref{c_inf_localisation} follows from a similar and simpler argument as the one used in the proof of \Cref{t_localisation}, hence is omitted here. 

We are ready to present the inference results.
\begin{theorem} \label{t_inference}
In addition to the same condition of \Cref{c_inf_localisation}, assume \Cref{a_snr_strong}\ref{a_snr_strong_2}~holds. Let $\{\Tilde{\eta}_\ell\}_{\ell =1} ^{\widehat{K}}$ be the change point estimators defined in \eqref{loss_refine}, with 
\begin{itemize}
    \item the intervals $\{(s_\ell, e_\ell)\}_{\ell=1}^{\widehat{K}}$ defined in \eqref{refine_interval};
    
    \item the preliminary estimators $\{\widehat{\eta}_\ell\}_{\ell=1}^{\widehat{K}}$ output by \textnormal{FFDP} in \Cref{DP} with tuning parameters $r = \Tilde{r}$, $\lambda = C_{\lambda}\Tilde{r}^{-11/2}\sqrt{m\log(n)}$, and $\xi = \frac{C_{\xi}K \Tilde{r}^4 m\log^2(n)}{\delta^2\zeta_{\delta}} $, with $C_{\lambda}$, $C_{\xi} >0$ being absolute constants; and

    \item the estimated covariance functions $\{\widehat{C}_{\ell}\}_{\ell=1}^{\widehat{K}+1}$ defined in \eqref{loss_cov} with the intervals constructed based on $\{\widehat{\eta}_\ell\}_{\ell=0}^{\widehat{K}+1}$.
\end{itemize}
For any $\ell \in [K+1]$, denote $\{Y_{t,j}^{(\ell)}\}_{t\in \mathbb{Z}, j\in [m]}$ an infinite process generated under \Cref{a_model} with $Y_{t,j}^{(\ell)}=f_{t}^{(\ell)}(X_{t,j})+\varepsilon_{t,j}$, and additionally we assume that $f_t^{(\ell)}$ is the random function with mean $0$ and $\mathbb{E}[f_t^{(\ell)}(s)f_t^{(\ell)}(k)] = \Sigma^*_{\eta_\ell}(s,t)$ for all $(s,k) \in [0,1]^2$. 

\begin{enumerate}[leftmargin=*, label=\textbf{\alph*.}]
    \item (Non-vanishing regime) For $\ell \in [K]$, if $\kappa_\ell \rightarrow \varrho_\ell$ as $n \rightarrow \infty$, with $\varrho_\ell > 0$ being an absolute constant, the following results hold.
    \begin{enumerate}[label={\textbf{a.\arabic*.}}]
        \item The estimation error satisfies that $|\Tilde{\eta}_\ell -\eta_\ell| = O_p(1)$, as $n \rightarrow \infty$.

        \item When $n \rightarrow \infty$, 
        \begin{align*}
            \Tilde{\eta}_\ell -\eta_\ell \stackrel{\mathcal{D}}{\longrightarrow} \underset{d \in \mathbb{Z}}{\arg\min}  \; P_\ell(d),
        \end{align*}
        where for $d \in \mathbb{Z}$, $P_\ell(d)$ is a double-sided random walk defined as
        \begin{align*}
            P_\ell(d) = \begin{cases}
            \sum_{t=1}^{d} \sum_{(j,k) \in \mathcal{O}} \big\{2\varrho_\ell\epsilon_{t,(j,k)}(\ell+1)\vartheta_{t,(j,k)}(\ell)+\varrho_\ell^2\vartheta_{t,(j,k)}^2(\ell)\big\}, & d>0,\\
            0, &d = 0,\\
            \sum_{t = d}^{-1} \sum_{(j,k) \in \mathcal{O}} \big\{-2\varrho_\ell\epsilon_{t,(j,k)}(\ell)\vartheta_{t,(j,k)}(\ell)+\varrho_\ell^2\vartheta_{t,(j,k)}^2(\ell)\big\}, &d<0,
    \end{cases}
        \end{align*}
        if we additionally assume that for any $(j,k) \in \mathcal{O}$ and any $t \in \mathbb{Z}$, 
        \begin{align} \label{a_inference_nonvanish}
            \{\Upsilon_\ell(X_{t,j},X_{t,k}),Y_{t,j}^{(\ell)}Y_{t,k}^{(\ell)}-\Phi_{\Tilde{r}}^\top(X_{t,j})C^*_{\Tilde{r},\eta_\ell}\Phi_{\Tilde{r}}(X_{t,k})\} \stackrel{\mathcal{D}}{\longrightarrow} \{\vartheta_{t,(j,k)}(\ell), \epsilon_{t,(j,k)}(\ell)\}.
        \end{align} 
    \end{enumerate}

    \item (Vanishing regime) For $\ell \in [K]$, if $\kappa_\ell \rightarrow 0$ as $n \rightarrow \infty$, the following results hold.
    \begin{enumerate}[label={\textbf{b.\arabic*.}}]
        \item The estimation error satisfies that $\kappa_\ell^2|\Tilde{\eta}_\ell -\eta_\ell| = O_p(1)$, as $n \rightarrow \infty$.

        \item When $n \rightarrow \infty$, if we additionally assume that $C^*_{\Tilde{r},\eta_{\ell}}$ and $C^*_{\Tilde{r},\eta_{\ell+1}}$ converge, then we have that
        \begin{align*}
            \kappa_\ell^2(\Tilde{\eta}_\ell -\eta_\ell) \stackrel{\mathcal{D}}{\longrightarrow} \underset{d \in \mathbb{R}}{\arg\min} \; \mathcal{B}_\ell(d),
        \end{align*}
        where $\mathcal{B}_\ell(d)$ is defined as 
        \begin{align*}
            \mathcal{B}_\ell(d) = \begin{cases}
            \sigma_{\ell}\mathbb{B}(d) + d\varpi_\ell, & d >0,\\
            0, &d =0,\\
            \sigma_{\ell}\mathbb{B}(-d)-d\varpi_\ell, & d<0,
            \end{cases}
        \end{align*}
        $\mathbb{B}(d)$ is the standard Brownian motion, and $\sigma_{\ell}$ is defined as
        \begin{align} \label{eq_var1}
            \sigma^2_{\ell} & = 4 \underset{n \rightarrow \infty}{\lim} \var\Big(\sum_{(j,k) \in \mathcal{O}} \big\{Y^{(\ell+1)}_{1,j}Y^{(\ell+1)}_{1,k}-\Phi_{\Tilde{r}}^\top(X_{1,j})C^*_{\Tilde{r},\eta_{\ell+1}}\Phi_{\Tilde{r}}(X_{1,k})\big\}\Upsilon_\ell(X_{1,j}, X_{1,k})\Big)\\ \label{eq_var2}
            & =4 \underset{n \rightarrow \infty}{\lim} \var\Big(\sum_{(j,k) \in \mathcal{O}} \big\{Y^{(\ell)}_{1,j}Y^{(\ell)}_{1,k}-\Phi_{\Tilde{r}}^\top(X_{1,j})C^*_{\Tilde{r},\eta_{\ell}}\Phi_{\Tilde{r}}(X_{1,k})\big\}\Upsilon_\ell(X_{1,j}, X_{1,k})\Big).
        \end{align}
        \end{enumerate}
\end{enumerate}
\end{theorem}
\Cref{t_inference} provides the localisation errors of the refined estimators $\{\Tilde{\eta}_\ell\}_{\ell=1}^{\widehat{K}}$, together with their limiting distributions in two regimes of jump sizes. Notably, for any $\ell \in [K]$, the upper bounds on the localisation errors in both the vanishing and non-vanishing regimes can be written as 
\begin{align*}
    \kappa_{\ell}^2|\Tilde{\eta}_\ell -\eta_\ell| = O_p(1).
\end{align*}
This uniform tightness condition of $\kappa_{\ell}^2|\widehat{\eta}_\ell -\eta_\ell|$ not only represents a significant improvement upon \Cref{c_inf_localisation}, in which the localisation errors depend both on the number of change points $K$ and a factor of $\log^2(n)$, but more importantly guarantees the existence of the limiting distributions. In the following, We further elaborate on \Cref{t_inference} by providing a sketch of the proof and offering insights into the intuition behind the two different forms of the limiting distributions.

\noindent \textbf{Sketch of the proof.}
The detailed proof can be found in Appendix \ref{pf_inference}, while here we offer some high-level insights. The proof could be roughly divided into two steps. 

In the first step, using the results of preliminary estimators output by FFDP in \Cref{c_inf_localisation}, we show that for $\ell \in [\widehat{K}]$, $\kappa_\ell^2|\Tilde{\eta}_\ell-\eta_\ell|$ is uniformly tight. In the second step, we establish the limiting distributions by applying the Argmax continuous mapping theorem \citep[e.g.~Theorem 3.2.2 in][]{vanderVaart1996}. Based on differences in the limits of the jump sizes $\{\kappa_\ell\}_{\ell=1}^{\widehat{K}}$, the limiting distributions of the refined estimators are derived individually under both the non-vanishing
and vanishing regimes.  In particular, in the vanishing case, an application of the functional central limit theorem yields the two-sided Brownian motion in the limit. 

\noindent \textbf{Limiting distributions.}
The limiting distributions of the change points for the fragmented functional data are the first time seen in \Cref{t_inference}. The results not only quantify the asymptotic behaviour of our refined estimators $\{\Tilde{\eta}_\ell\}_{\ell=1}^{\widehat{K}}$, but also facilitate the construction of confidence intervals, especially in the vanishing regimes when change point estimators have diverging localisation errors. In the vanishing regime, the limiting distribution provides us with an asymptotic framework to analyse the uncertainties. Note that $\sigma^2_\ell$ in \Cref{t_inference} has two expressions detailed in \eqref{eq_var1} and \eqref{eq_var2}, which equal to each other. This is due to the fact that in the jump size vanishing regime, the partial sums of data studied before and after each change point, in \eqref{eq_var1} and \eqref{eq_var2} respectively, have the same limiting distribution, hence the same limiting variance.  We will provide more discussions on the deviance between the theory and practice on this aspect in the sequel. 

Despite being novel for the problem setup and related limiting distributions, such a two-regime phenomenon described in \Cref{t_inference} is no stranger to the change point community.  Similar types of results have been discovered in various problems, for example, changes in high dimensional regression \citep[e.g.][]{xu2022change}, nonparametric statistics \citep[e.g.][]{madrid2024change}, and the mean functions of functional data \citep[e.g.][]{aue2018detecting}. The problem concerning the functional covariance changes could also be found in \citet{horvath2022change} and \citet{jiao2023break}, in which they derive the limiting distributions when there is only one change point in the covariance functions of fully observed and weakly dependent functional data. 

\subsubsection{Confidence intervals construction} \label{section_CI_construction}
An important application of the limiting distributions derived in \Cref{t_inference} is to construct confidence intervals for refined estimators $\{\Tilde{\eta}_\ell\}_{\ell=1}^{\widehat{K}}$. In this subsection, we provide a discussion on the estimation of the unknown parameters $\{\varpi_\ell\}_{\ell=1}^{K}$, $\{\kappa_\ell^2\}_{\ell =1}^{K}$ and {$\{\sigma^2_{\ell,+}, \sigma^2_{\ell,-}\}_{\ell=1}^{K}$}, and also demonstrate two practical inference procedures in the vanishing regime of \Cref{t_inference}.

For $\ell \in \{1, \ldots, \widehat{K}\}$, with the outputs of the FFDP algorithm: $\widehat{\eta}_\ell$, $\widehat{C}_\ell$ and $\widehat{C}_{\ell+1}$,
and outputs of local refinement $\Tilde{\eta}_\ell$, we construct $(1-\alpha)$-confidence intervals, $\alpha \in (0,1)$, of $\eta_\ell$ in the following ways.

\noindent \textbf{Step 1.}(Jump size estimation). Let the estimator of $\kappa_\ell^2$ - the squared $\ell^{\text{th}}$ jump size - be
\begin{align*}
    \widehat{\kappa}_\ell^2 = \|\widehat{C}_{\ell+1} - \widehat{C}_\ell\|_{\F}^2.
\end{align*}

\noindent \textbf{Step 2.}(Drift estimation). Construct the drift estimator $\widehat{\varpi}_\ell^2$ as
\begin{align} \label{CI_drift}
    \widehat{\varpi}_\ell^2 = \frac{1}{\widehat{\kappa}_\ell^2 (\widehat{\eta}_{\ell+1} - \widehat{\eta}_{\ell-1})}\sum_{t = \widehat{\eta}_{\ell-1}+1}^{\widehat{\eta}_{\ell+1}}\sum_{(j,k)\in \mathcal{O}}\big\{\Phi_r^\top(X_{t,j})(\widehat{C}_{\ell+1} - \widehat{C}_\ell)\Phi_r(X_{t,k})\big\}^2.
\end{align}

\noindent \textbf{Step 3.}(Variance estimation). We estimate $\sigma^2_{\ell,-}$ and $\sigma^2_{\ell,+}$ by sample variances of corresponding segments of data. Let 
\begin{align}\label{CI_var1}
    \widehat{\sigma}^2_{\ell, -} = \frac{1}{\widehat{\eta}_{\ell} - \widehat{\eta}_{\ell-1}-1}\sum_{t=\widehat{\eta}_{\ell-1}+1}^{\widehat{\eta}_{\ell}} (Z_{\ell, t} - \Bar{Z_\ell})^2,
\end{align}
where $\Bar{Z_\ell} = 1/(\widehat{\eta}_{\ell} - \widehat{\eta}_{\ell-1}) \sum_{t \in (\widehat{\eta}_{\ell-1},\widehat{\eta}_\ell]} Z_{\ell,t}$ and for $t \in (\widehat{\eta}_{\ell-1},\widehat{\eta}_\ell]$,
\begin{align*}
    Z_{\ell, t} = \frac{1}{\widehat{\kappa}_\ell}\sum_{(j,k) \in \mathcal{O}}\big\{Y_{t,j}Y_{t,k} -\Phi_{\Tilde{r}}^\top(X_{t,j})\widehat{C}_{\ell}\Phi_{\Tilde{r}}(X_{t,k})\big\}\big\{\Phi_{\Tilde{r}}^\top(X_{t,j})(\widehat{C}_{\ell+1} - \widehat{C}_\ell)\Phi_{\Tilde{r}}(X_{t,k})\big\}.
\end{align*}
Let
\begin{align}\label{CI_var2}
    \widehat{\sigma}^2_{\ell,+} = \frac{1}{\widehat{\eta}_{\ell+1} - \widehat{\eta}_{\ell}-1}\sum_{t = \widehat{\eta}_{\ell}+1}^{\widehat{\eta}_{\ell+1}} (Z'_{\ell, t} - \Bar{Z'_\ell})^2,
\end{align}
where $\Bar{Z'_\ell} = 1/(\widehat{\eta}_{\ell+1} - \widehat{\eta}_{\ell}) \sum_{t \in (\widehat{\eta}_{\ell},\widehat{\eta}_{\ell+1}]} Z'_{\ell,t}$ and for $t \in (\widehat{\eta}_{\ell},\widehat{\eta}_{\ell+1}]$, 
\begin{align*}
    Z'_{\ell, t} = \frac{1}{\widehat{\kappa}_\ell}\sum_{(j,k) \in \mathcal{O}}\big\{Y_{t,j}Y_{t,k} -\Phi_{\Tilde{r}}^\top(X_{t,j})\widehat{C}_{\ell+1}\Phi_{\Tilde{r}}(X_{t,k})\big\}\big\{\Phi_{\Tilde{r}}^\top(X_{t,j})(\widehat{C}_{\ell+1} - \widehat{C}_\ell)\Phi_{\Tilde{r}}(X_{t,k})\big\}.
\end{align*}

\noindent \textbf{Step 4.} (Limiting distribution simulation). Let $B \in \mathbb{Z}_{+}$, $M \in \mathbb{R}_{+}$ and $N \in \mathbb{Z}_{+}$. For $b \in \{1, \ldots, B\}$, let
\begin{align*}
    \widehat{u}^{(b)} = \underset{d \in (-M,M)}{\arg \min} \big\{\widehat{\varpi}_\ell^2|d| + \mathbbm{1}_{\{d<0\}}\widehat{\sigma}^2_{\ell,-}\mathbb{W}^{(b)}(d) + \mathbbm{1}_{\{d \geq 0\}}\widehat{\sigma}^2_{\ell,+}\mathbb{W}^{(b)}(d)\big\},
\end{align*}
where
\begin{align} \label{eq_w}
    \mathbb{W}^{(b)}(d) = \begin{cases}
        \frac{1}{\sqrt{N}}\sum_{i=\lceil Nd \rceil}^{-1} z_i^{(b)}, & d<0,\\
        0, & d=0,\\
        \frac{1}{\sqrt{N}}\sum_{i=1}^{\lfloor Nd \rfloor} z_i^{(b)}, & d>0,
    \end{cases}
\end{align}
$\{z_i^{(b)}\}_{i = -\lfloor NM \rfloor}^{\lceil NM \rceil}$ are independent standard Gaussian random variables, $\widehat{\varpi}_\ell^2$, $\widehat{\sigma}^2_{\ell,-}$ and $\widehat{\sigma}^2_{\ell,+}$ are constructed by \eqref{CI_drift}, \eqref{CI_var1} and \eqref{CI_var2} respectively.

\noindent \textbf{Step 4'.} (Alternative limiting distribution simulation). Let $B \in \mathbb{Z}_{+}$, $M \in \mathbb{R}_{+}$ and $N \in \mathbb{Z}_{+}$. For $b \in \{1, \ldots, B\}$, let
\begin{align}\label{eq_symmetric_BM}
    \widehat{u}^{(b)} = \underset{d \in (-M,M)}{\arg \min} \big\{\widehat{\varpi}_\ell^2|d| + \widehat{\Bar{\sigma}}^2_{\ell}\mathbb{W}^{(b)}(d)\big\},
\end{align}
where $\widehat{\Bar{\sigma}}^2_{\ell} = (\widehat{\sigma}^2_{\ell,-}+\widehat{\sigma}^2_{\ell,+})/2$,  $\widehat{\varpi}_\ell^2$, $\widehat{\sigma}^2_{\ell,-}$, $\widehat{\sigma}^2_{\ell,+}$ and $\mathbb{W}^{(b)}(d)$ are constructed by \eqref{CI_drift}, \eqref{CI_var1}, \eqref{CI_var2} and \eqref{eq_w} respectively.

\noindent \textbf{Step 5.}(Confidence interval construction). Let $\widehat{q}_u(1-\alpha/2)$ and $\widehat{q}_u(\alpha/2)$ be the $(1-\alpha/2)$- and $\alpha/2$-quantiles of the empirical distribution of $\{\widehat{u}^{(b)}\}_{b=1}^B$. A numerical $(1-\alpha)$-confidence interval of $\eta_\ell$ is constructed as
\begin{align} \label{eq_CI}
    \Big[\Tilde{\eta}_\ell - \frac{\widehat{q}_u(1-\alpha/2)}{\widehat{\kappa}_\ell^2}\mathbbm{1}_{\{\widehat{\kappa}_\ell^2 \neq 0\}} , \; \Tilde{\eta}_\ell -  \frac{\widehat{q}_u(\alpha/2)}{\widehat{\kappa}_\ell^2}\mathbbm{1}_{\{\widehat{\kappa}_\ell^2 \neq 0\}}  \Big].
\end{align}

\begin{remark}[Asymmetry of confidence intervals]\label{r_asymmetric}
    Despite that \eqref{eq_var1} and \eqref{eq_var2} - the limiting variances of partial sums of data before and after change points - are the same in the vanishing regime, in practice however this is never the case.  To deal with this inequality, following the suit in the existing literature \citep[e.g.][]{yau2016inference, ling2016estimation}, we consider an alternative step to enforce the equality of limiting variances.  In particular, \textbf{Step 4'} takes an average of $\widehat{\sigma}^2_{\ell,-}$ and $\widehat{\sigma}^2_{\ell,+}$ defined in \eqref{CI_var1} and \eqref{CI_var2} respectively.
\end{remark}

\subsection{Intermediate result} \label{section_byproduct}
In this section, we present an intermediate result on the estimation error of covariance functions under Model \eqref{model_obs} and \Cref{a_model}, which is of independent interest. 

We employ the analytic $(\alpha,\beta)$-basis framework for the chosen complete orthonormal system $\Phi$, which provides an ordering for selecting the basis functions. We say that a complete orthonormal system $\Phi = \{\phi_k\}_{k=1}^\infty$ is called an analytic $(\alpha,\beta)$-basis of $L^2([0,1])$ for some constants $\alpha, \beta \geq 0$ if there exists an absolute constant $C_{\phi} > 0$ such that $\|\phi_k\|_\infty \leq C_{\phi}k^{\alpha}$ and $\|\phi_k^{(d)}\|_{L^2} \leq C_{\phi}k^{\beta d}$ for $k \in \mathbb{N}$ and $d \in\{1,2\}$, where $\phi_k^{(d)}$ stands for the $d^{\mathrm{th}}$ order derivative of $\phi$.  This framework is also considered in \citet{lin2021basis}.

We are now ready to present our results.
\begin{theorem} \label{t_estimation} 
    For any integer interval $I\subseteq [n]$ such that $|I| \geq C\log(n)$ for an absolute constant $C>0$, let $\widehat{C}_I$ be defined as the solution to \eqref{loss_cov} with $\Phi$ being an analytic $(\alpha,\beta)$-basis of $L^2([0,1])$ for $\alpha, \beta \geq 0$ and 
    \begin{align} \label{t_estimation_lambda}
        \lambda = C_{\lambda}\Big\{r^{2\alpha-4\beta-3/2}\sqrt{m\log(n)} \vee r^{4\alpha-4\beta-1/2}\sqrt{\log(n)} \vee \frac{r^{-4\beta-5/2}\sqrt{m}\sum_{t \in I}\rho_{r,t}}{\sqrt{|I|}\delta^2}\Big\}, 
    \end{align}
    for some sufficiently large absolute constant $C_{\lambda} >0$. Then under \Cref{a_model}, for any positive integer $r < n$, it holds with probability at least $1-3n^{-3}$ that
    \begin{align*}
        \|\widehat{C}_I - C^*_{r,I}\|_{\F} \lesssim \frac{r^{2\alpha+1}}{\zeta_{\delta}}\sqrt{\frac{\log(n)}{|I|}}+\frac{r^{4\alpha+2}}{\zeta_{\delta}}\sqrt{\frac{\log(n)}{|I|m}}+\frac{\sum_{t \in I} \rho_{r,t}}{|I|\delta^2\zeta_{\delta}}, 
    \end{align*}
    and
    \begin{align} \label{t_estimation_result}
        \|\Phi_r^\top\widehat{C}_I\Phi_r - \Sigma^*_{I}\|_{L^2} \lesssim \frac{r^{2\alpha+1}}{\zeta_{\delta}}\sqrt{\frac{\log(n)}{|I|}}+\frac{r^{4\alpha+2}}{\zeta_{\delta}}\sqrt{\frac{\log(n)}{|I|m}}+ \frac{\sum_{t \in I} \rho_{r,t}}{|I|}\bigg(1 \vee \frac{1}{\delta^2\zeta_{\delta}}\bigg),
    \end{align}
    where $C^*_{r,t} \in \mathbb{R}^{r \times r}$ with entries $(C^*_{r,t})_{hl}= \int_{0}^1 \int_{0}^1 \Sigma_t^*(s,t)\phi_h(s)\phi_l(t) \,\mathrm{d}s\,\mathrm{d}t$ for $h,l \in [r]$, $C^*_{r, I} =~|I|^{-1}  \sum_{t \in I} C^*_{r,t}$, $\Sigma^*_{I} = |I|^{-1} \sum_{t \in I} \Sigma_t^*$, $\zeta_{\delta}>0$ is a constant depending on $\delta$ and $\{\Sigma^*_t\}_{t \in I}$, and $\rho_{r,t}= \|\Sigma^*_t - \Phi_r^\top C^*_{r,t} \Phi_r\|_{L^2}$ is the convergence rate of the approximation error for $\Sigma^*_t$ using the first $r$ basis of $\Phi$.
\end{theorem}

\Cref{t_estimation} shows that the solution to \eqref{loss_cov} serves as a good enough estimator for its population counterpart regardless of the choices of the interval $I$, even if $I$ is non-stationary. It is repeatedly used in the proof of Theorems~\ref{t_localisation} and \ref{t_inference} to establish high probability upper bounds on the localisation errors. Note that \Cref{t_estimation} presents non-asymptotic results.  We allow to a great extend of generality that, the sizes of sampling grids $m$, the length of the interval $I$ and the true covariance functions are all functions of the sample size $n$.

A similar result but in the asymptotic form and with different model assumptions is derived in \citet{lin2021basis}. \citet{lin2021basis} show that for any interval $I$ containing no change point, it holds that 
\[
    \|\Phi_r^\top\widehat{C}_I\Phi_r - \Sigma^*_{I}\|_{L^2} =  O_p(|I|^{-1/2}r^{2\alpha+1}+\rho_r), \quad \text{where} \quad \rho_r = \rho_{r,t}, \quad t\in I.
\]

Our model defined in \Cref{a_model} is a special case of the class of models defined in \citet{lin2021basis}. The first and third terms in \eqref{t_estimation_result} align with the results demonstrated in \citet{lin2021basis} up to a logarithm factor. Notably, the inclusion of the second term in \eqref{t_estimation_result} accounts for the effects of $m$, which comes from an application of Hoeffding’s inequality for general bounded random variables \citep[Theorem 2.2.6 in][]{vershynin2018high}. Moreover, our results in \Cref{t_estimation} can also be applied to those intervals containing change points to capture the deviation of the estimators to their population counterparts. 

In the literature of the conventional setup of sparse FDA where observation grids can be sampled over the entire domain, results of estimation errors of covariance function over intervals without change points can also be found. Denote $\Sigma^*$ the true covariance function and $\widehat{\Sigma}$ the corresponding estimator. \citet{cai2010nonparametric} utilise a reproducing kernel
Hilbert space framework in the case when the sample paths of the random functions belong to a Sobolev space $\mathcal{W}^q$ almost surely. They show that it holds with high probability $\|\widehat{\Sigma} - \Sigma^*\|^2_{L^2} \lesssim (mn)^{-\frac{2q}{2q+1}}+\frac{1}{n}$, which is also minimax optimal. \citet{zhang2016sparse} consider the local linear smoothers when $\partial^2\Sigma^*(s,t)/ \partial s^2$, $\partial^2\Sigma^*(s,t)/ \partial t^2$ and $\partial^2\Sigma^*(s,t)/ \partial s \partial t$ are bounded on $[0,1]^2$, and show that $\|\widehat{\Sigma}- \Sigma^*\|_{L^2}^2 = O_p(n^{-\frac{2}{3}}m^{-\frac{4}{3}}+n^{-1})$. Both the above results reveal a phase transition phenomenon, with the transition boundary at $m \asymp n^{\frac{1}{4}}$ \citep[when $q=2$ in][]{cai2010nonparametric}. Compared to aforementioned results, our results in \Cref{t_estimation} of the upper bounds on the estimation errors are looser. We conjecture that this phenomenon is partially due to a more challenging setup being considered. 
 
\section{Numerical experiments} \label{section_numerical}
In this section, we further investigate the performances of our FFDP algorithm by comparing with state-of-art methods through numerical experiments, with simulated data in \Cref{section_numerical_simulated} and two real data sets in \Cref{section_numerical_real}.~All the implementation of numerical experiments can be found at \url{https://github.com/GengyuXue/FragmentCP}.

\subsection{Simulated data analysis} \label{section_numerical_simulated}

\noindent \textbf{Settings.} We simulate data from the model described in \eqref{model_obs} with $\{f_t\}_{t=1}^n$ being Gaussian processes and $\{\varepsilon_{t,j}\}_{t=1,j=1}^{n,m}$ following Gaussian distributions with variance $\sigma_\epsilon^2$. We consider four different scenarios, with the first three showing the effectiveness of the FFDP algorithm and the last one focusing on inference. Each scenario consists of multiple setups. For each setup in the first three settings, we carry out 100 simulations, and for each setup in the fourth scenario, we carry out 200 simulations. These settings are designed to investigate the effects of $n$ the sample size, $m$ the grid size and the functional class of the true covariance functions. The details of each setting are as follows.
\begin{enumerate}[leftmargin=*, label =\textbf{(\roman*)}]
    \item \label{scenario_i} Single change point, varying $n$. Let $n \in \{200, 300, 400\}$, $m = 30$, $\delta = 0.6$ and $\sigma_\epsilon = 0.01$. The only change point occurs at $n/2$. The covariance functions are defined as
    \begin{align*}
        \Sigma^*_t(s,t)=
            \begin{cases}
            2\phi_1(s)\phi_1(t) + 1.4\phi_2(s)\phi_2(t) + 0.8\phi_3(s)\phi_3(t), \quad &t \in \{1, \ldots, n/2\},\\[5pt]
            5\phi_1(s)\phi_1(t) + 4.2\phi_2(s)\phi_2(t) + 3.4\phi_3(s)\phi_3(t), & t \in \{n/2+1, \ldots, n\},
            \end{cases}
    \end{align*}
    where $\phi_1(s) = 1$, $\phi_2(s) = \cos(2\pi s)$, and $\phi_3(s) = \sin(2\pi s)$.

    \item  \label{scenario_ii} Single change point, varying $m$. Let $m \in \{10,20,30,40\}$, $n =200$, $\delta = 0.6$ and $\sigma_\epsilon = 1$. The only change point occurs at $n/2$. The covariance functions are defined as
    \begin{align*}
        \Sigma^*_t(s,t)= 
            \begin{cases}
            \psi_1(s)\psi_1(t)+ 0.8\psi_2(s)\psi_2(t), \quad &t \in \{1, \ldots, 100\},\\[5pt]
            5\psi_1(s)\psi_1(t)+7.5\psi_3(s)\psi_3(t), & t \in \{101, \ldots, n\},
        \end{cases}
    \end{align*}
    where $\psi_1(s) = \sqrt{5}(6s^2-6s+1)$, $\psi_2(s)= \sqrt{2}\log(s+0.1)$, and $\psi_3(s) = 5(s-0.5)^2$.

    \item \label{scenario_iii} Multiple change points, varying $n$. Let $n \in \{150, 200\}$, $m = 15$, $\delta = 0.5$ and $\sigma_\epsilon = 0.01$. The unevenly spaced change points occur at $\{\lfloor n/4\rfloor, \lfloor 5n/8\rfloor\}$. The covariance functions are defined as
    \begin{align*}
        \Sigma^*_t(s,t)= 
            \begin{cases}
            1.8\phi_1(s)\phi_1(t)+ 1.3\phi_2(s)\phi_2(t), \quad &t \in \{1, \ldots, \lfloor n/4\rfloor\},\\[5pt]
            8\phi_1(s)\phi_1(t)+ 6.5\phi_2(s)\phi_2(t), & t \in \{\lfloor n/4\rfloor +1, \ldots, \lfloor 5n/8\rfloor\},\\[5pt]
            1.8\phi_1(s)\phi_1(t)+1.6\phi_2(s)\phi_2(t), & t \in \{ \lfloor 5n/8\rfloor + 1, \ldots n\},
        \end{cases}
    \end{align*}
    where $\phi_1(s) = 1$ and $\phi_2(s) = \cos(2\pi s)$.

    \item \label{scenario_iv} Single change point, varying both $n$ and $m$. Let $n \in \{200, 300, 400\}$ and $m \in \{10, 20\}$, $\delta = 0.5$ and $\sigma_\epsilon = 0.1$. The only change point occurs at $n/2$. The covariance functions are defined as
    \begin{align*}
        \Sigma^*_t(s,t)=
            \begin{cases}
            0.5\phi_1(s)\phi_1(t) + 2\phi_2(s)\phi_2(t) + 1.4\phi_3(s)\phi_3(t), \quad &t \in \{1, \ldots, n/2\},\\[5pt]
            3.5\phi_1(s)\phi_1(t) + 4\phi_2(s)\phi_2(t) + 3.5\phi_3(s)\phi_3(t), & t \in \{n/2+1, \ldots, n\},
            \end{cases}
    \end{align*}
    where $\phi_1(s) = 1$, $\phi_2(s) = \cos(2\pi s)$, and $\phi_3(s) = \sin(2\pi s)$.
\end{enumerate}

\noindent \textbf{Competitors.} Since the change point analysis problem under fragmented functional data is first studied here, we do not have any direct competitors.

For illustration purposes, we propose a modified version of the seeded binary segmentation (SBS) proposed in \citet{kovacs2023seeded}, detailed in Appendix \ref{appendix_SBS}. We also take into account the weighted cumulative sum (WCUSUM) approach in \citet{horvath2022change} and the functional seeded binary segmentation (FSBS) algorithm introduced in \citet{madrid2022change}. WCUSUM targets to detect a single covariance function change under the setting of fully observed and weakly dependent functional data, while FSBS specialises in managing discretely observed functional data and is designed to detect multiple change points by examining mean changes. In all simulated experiments, for implementations of WCUSUM and FSBS, we use the authors' provided code\footnote{The code for WCUSUM could be found at \url{https://github.com/yzhao7322/CovFun_Change}, and the code for FSBS could be found at \url{https://github.com/cmadridp/FSBS}.}. WCUSUM requires fully observed functional data, therefore, in practice, we convert discretely observed data to functional observations by B-splines of order four. For FSBS, the tuning parameters are chosen by the cross-validation method recommended by the authors.

\noindent \textbf{Tuning parameter selection and algorithm implementations.}
To compute the initial estimators $\{\widehat{\eta}_\ell\}_{\ell=1}^{\widehat{K}}$, we need to determine a set of basis functions $\Phi$. In our simulation experiments, we select the set of Fourier basis and incorporate the Fourier extension technique as recommended by \citet{lin2021basis} due to its computational stability. There are three more tuning parameters involved in our algorithm: $\lambda$, $r$ and $\xi$, and we adopt a cross-validation method to select these tuning parameters. Specifically, we first divide $\{X_{t,j}, Y_{t,j}\}_{t=1, j=1}^{n,m}$ into training and validation sets according to odd and even indices of $t$. For each candidate $(\lambda, r, \xi)$, we obtain change point estimators $\{\widehat{\eta}_\ell\}_{\ell=0}^{\widehat{K}+1}$ with $\widehat{\eta}_0 =0$ and $\widehat{\eta}_{\widehat{K}+1} =n/2$, and coefficients estimators $\{\widehat{C}_{(\widehat{\eta}_\ell,\widehat{\eta}_{\ell+1}]}\}_{\ell=0}^{\widehat{K}}$ on the training set. Using these outputs, we compute the validation loss
\begin{align*}
    \sum_{\ell = 0}^{\widehat{K}} \sum_{t = \widehat{\eta}_\ell+1}^{\widehat{\eta}_{\ell+1}} \sum_{(j,k)\in \mathcal{O}} \big\{Y_{t,j}Y_{t,k} - \Phi_r^\top(X_{t,j}) \widehat{C}_{(\widehat{\eta}_\ell,\widehat{\eta}_{\ell+1}]}\Phi_r(X_{t,k})\big\}^2
\end{align*}
with data in the validation set. The candidate $(\lambda, r, \xi)$ corresponding to the lowest validation loss is chosen as the final tuning parameter set. However, in practice, for computational efficiency, we pick the true $r$ in scenarios \ref{scenario_i}, \ref{scenario_iii} and \ref{scenario_iv} when covariance functions are constructed with Fourier basis, and select $r =3$ in scenario \ref{scenario_ii}. 

The constrained optimisation problem in \eqref{loss_cov} is tackled by the geometric Newton's method suggested in \citet{lin2021basis}, with all the tuning parameters chosen as suggested there.

\noindent \textbf{Evaluation metrics.} Let $\{\eta_\ell\}_{\ell =1}^{\widehat{K}}$ and $\{\Tilde{\eta}_\ell\}_{\ell =1}^{\widehat{K}}$ be the preliminary and final change point estimators defined in \eqref{loss_l0} and \eqref{loss_refine} respectively. 

To access the performance of localisation, we report (i) means and standard deviations of absolute errors $|\widehat{K}-K|$, (ii) the proportions of the repetitions with $\widehat{K} < K$ (i.e.~underestimating the number of change points), $\widehat{K} = K$ and $\widehat{K} >K$ (i.e.~overestimating the number of change points), and (iii) means and standard deviations of scaled Hausdorff distances of the preliminary and final estimators.  The scaled Hausdorff distance between $\{\widehat{\eta}_\ell\}_{\ell =1}^{\widehat{K}}$ and $\{\eta_\ell\}_{\ell =1}^{\widehat{K}}$ is defined as
\begin{align}\label{scaled_hausdorff}
    \mathcal{D}\big(\{\widehat{\eta}_\ell\}_{\ell=1}^{\widehat{K}},\{\eta_\ell\}_{\ell=1}^K\big)=\frac{d\big(\{\widehat{\eta}_\ell\}_{\ell=1}^{\widehat{K}},\{\eta_\ell\}_{\ell=1}^K\big)}{n},
\end{align}
where $d(\mathcal{A}, \mathcal{B})$ denotes the Hausdorff distance between two compact sets $\mathcal{A}, \mathcal{B}$ in $\mathbb{R}$, given by
\begin{align*}
    d(\mathcal{A},\mathcal{B})=\max \Big\{\max _{a \in \mathcal{A}} \min _{b \in \mathcal{B}}|a-b|, \; \max _{b \in \mathcal{B}} \min _{a \in \mathcal{A}}|a-b|\Big\}.
\end{align*}
To assess the performance of local refinement, we report the scaled Hausdorff distance $\mathcal{D}\big(\{\Tilde{\eta}_\ell\}_{\ell=1}^{\widehat{K}},\{\eta_\ell\}_{\ell=1}^{\widehat{K}}\big)$ as detailed in \eqref{scaled_hausdorff}. To assess the performance of inference, we report (i) the coverage rate of the constructed confidence interval detailed in \Cref{section_CI_construction}, i.e. we report the coverage probability $\text{cover}_\ell(1-\alpha)$, which is the proportion of times such that
\begin{align} \label{eq_cover}
   \eta_\ell \in \Bigg[ \Big\lfloor \Tilde{\eta}_\ell - \frac{\widehat{q}_u(1-\alpha/2)}{\widehat{\kappa}_\ell^2}\mathbbm{1}_{\{\widehat{\kappa}_\ell^2 \neq 0\}}\Big\rfloor,\; \Big\lceil\Tilde{\eta}_\ell - \frac{\widehat{q}_u(\alpha/2)}{\widehat{\kappa}_\ell^2}\mathbbm{1}_{\{\widehat{\kappa}_\ell^2 \neq 0\}}\Big\rceil \Bigg]
\end{align}
among all the repetitions with $\widehat{K} = K$, and (ii) medians and standard deviations of widths of numerical $(1-\alpha)$-confidence intervals, i.e.
\begin{align*}
    \text{width}_\ell(1-\alpha) = \Big\lceil\Tilde{\eta}_\ell - \frac{\widehat{q}_u(\alpha/2)}{\widehat{\kappa}_\ell^2}\mathbbm{1}_{\{\widehat{\kappa}_\ell^2 \neq 0\}}\Big\rceil  - \Big\lfloor \Tilde{\eta}_\ell - \frac{\widehat{q}_u(1-\alpha/2)}{\widehat{\kappa}_\ell^2}\mathbbm{1}_{\{\widehat{\kappa}_\ell^2 \neq 0\}}\Big\rfloor.
\end{align*}

\noindent \textbf{Results - localisation.}
We summarise the performances of FFDP, SBS, WCUSUM and FSBS under scenarios \ref{scenario_i},  \ref{scenario_ii} and \ref{scenario_iii} in Tables \ref{table_scenario_i}, \ref{table_scenario_ii} and \ref{table_scenario_iii} respectively. In Tables \ref{table_scenario_i} and \ref{table_scenario_ii}, we highlight the top two results in \textbf{bold}. Since WCUSUM is designed to detect one change point only, in scenario \ref{scenario_iii}, we only apply FSBS, SBS and FFDP, and the best result is highlighted in \textbf{bold}.

We can see that in nearly all of our settings, FFDP is among the top performers in terms of both absolute errors $|\widehat{K} - K|$ and scaled Hausdorff distances, and demonstrates robust behaviours in both sparsely and densely sampled data. Among different settings in three scenarios, we investigate the effect of $n$ and $m$ on localisation performances and our results indicate that better localisation performances are achieved as $n$ and $m$ increase.  This observation aligns with our theoretical findings. In scenario \ref{scenario_ii}, we further show that even when the true covariance functions are not periodic, our algorithm could still output satisfying estimators when $\Phi$ is chosen as the Fourier basis.

Note that scenario \ref{scenario_ii} has only one significant change point in each setting. This aligns perfectly with the strength of WCUSUM when the quality of the function smoothing is satisfactory. Therefore, WCUSUM also performs well in scenario \ref{scenario_ii} when  $m$ is relatively large. As we do not have any changes in mean functions, relatively poor performances of FSBS results are also observed as expected.

\noindent \textbf{Results - inference.} The performances of local refinement and inference are summarised in \Cref{table_inference1}. We can see that an extra step of local refinement can drastically improve localisation performances and better coverage is achieved as $n$ increases. However, we note that the coverage of constructed confidence intervals is relatively poor, especially when the sample size is small. 

To provide several insights behind this observation, we conjecture that the poor performance is attributed to the large jump size required for achieving satisfactory localisation performances. Even though we are already superior to all the competitors, change point localisation for fragmented functional data is in general a challenging task. To obtain reasonable localisation performance, a relatively large jump size is required. This requirement not only means that we are moving away from the vanishing regime but also increases the asymmetry of constructed confidence intervals if we follow \textbf{Step 4} in \Cref{section_CI_construction} during constructions, echoing \Cref{r_asymmetric}.  We have also provided simulation results based on the alternative procedure involving \textbf{Step 4'}, showing similar performances.  Note that even when substituting the estimators $\widehat{C}$ with population truth in all calculations in \Cref{section_CI_construction}, there is still no significant improvement in terms of both scaled Hausdorff distances and coverage probabilities. 

\begin{table}[!htbp]
\centering
\begin{tabular}{ccccccc}
\hline
Method & $n$                    & $|\widehat{K} - K|$   & $\widehat{K} <K$ & $\widehat{K} = K$ & $\widehat{K} > K$ & $\mathcal{D}$          \\ \hline
FFDP   & \multirow{4}{*}{100} & \textbf{0.270 (0.446)} & 0.230             & 0.730              & 0.040              & \textbf{0.204 (0.189)} \\
SBS    &                      & \textbf{0.600 (0.532)}  & 0.390             & 0.420              & 0.190              & \textbf{0.274 (0.204)}          \\
WCUSUM &                      & 0.700 (0.461)           & 0.700             & 0.300              & -                 & 0.371 (0.202) \\
FSBS   &                      & 4.100 (6.078)          & 0.370             & 0.120              & 0.510              & 0.324 (0.156)          \\ \hline
FFDP   & \multirow{4}{*}{300} & \textbf{0.370 (0.544)} & 0.160             & 0.660              & 0.180              & \textbf{0.172 (0.184)} \\
SBS    &                      & 0.550 (0.642)          & 0.13             & 0.520              & 0.350              & \textbf{0.188 (0.181)} \\
WCUSUM &                      & \textbf{0.490 (0.502)} & 0.490             & 0.510              & -                 & 0.260 (0.239)          \\
FSBS   &                      & 4.380 (8.208)          & 0.350             & 0.110              & 0.540              & 0.321 (0.147)          \\ \hline
FFDP   & \multirow{4}{*}{400} & \textbf{0.240 (0.429)} & 0.120             & 0.760              & 0.120              & \textbf{0.152 (0.172)} \\
SBS    &                      & 0.510 (0.659)          & 0.080             & 0.580              & 0.340              & \textbf{0.151 (0.170)} \\
WCUSUM &                      & \textbf{0.420 (0.496)} & 0.420             & 0.580              & -                 & 0.227 (0.237)          \\
FSBS   &                      & 5.090 (9.334)          & 0.350             & 0.150              & 0.500               & 0.431 (0.206)          \\ \hline
\end{tabular}
\caption{Localisation results for Scenario \ref{scenario_i}. Each cell is based on 100 repetitions. The columns $|\widehat{K} - K|$ and $\mathcal{D}$ are in the form of mean (standard deviation).  The columns $\widehat{K} < K$, $\widehat{K} = K$ and $\widehat{K} > K$ collect proportions of corresponding cases.}
\label{table_scenario_i}
\end{table}

\begin{table}[!htbp]
\centering
\begin{tabular}{ccccccc}
\hline
Method & $m$                 & $|\widehat{K} - K|$   & $\widehat{K} <K$ & $\widehat{K} = K$ & $\widehat{K} > K$ & $\mathcal{D}$          \\ \hline
FFDP   & \multirow{4}{*}{10} & \textbf{0.360 (0.482)} & 0.020            & 0.640              & 0.340              & \textbf{0.142 (0.135)} \\
SBS    &                     & \textbf{0.400 (0.550)} & 0.030             & 0.630              & 0.340              & \textbf{0.138 (0.149)} \\
WCUSUM &                     & 0.900 (0.302)          & 0.900             & 0.100              & -                 & 0.474 (0.090)          \\
FSBS   &                     & 2.120 (2.567)          & 0.310             & 0.210              & 0.480              & 0.306 (0.148)          \\ \hline
FFDP   & \multirow{4}{*}{20} & \textbf{0.150 (0.359)} & 0.020             & 0.850              & 0.130              & \textbf{0.101 (0.118)} \\
SBS    &                     & \textbf{0.300 (0.541)} & 0.050             & 0.740              & 0.210              & \textbf{0.131 (0.141)} \\
WCUSUM &                     & 0.450 (0.500)          & 0.450             & 0.550              & -                 & 0.239 (0.239)          \\
FSBS   &                     & 2.820 (4.174)          & 0.430             & 0.130              & 0.440              & 0.338 (0.155)          \\ \hline
FFDP   & \multirow{4}{*}{30} & \textbf{0.140 (0.349)} & 0.050             & 0.860              & 0.090              & \textbf{0.097 (0.125)} \\
SBS    &                     & 0.260 (0.463)          & 0.130             & 0.750              & 0.120              & 0.140 (0.169)          \\
WCUSUM &                     & \textbf{0.180 (0.386)} & 0.180             & 0.820              & -                 & \textbf{0.106 (0.187)} \\
FSBS   &                     & 3.08 (5.171)          & 0.520             & 0.120              & 0.360              & 0.381 (0.143)          \\ \hline
FFDP   & \multirow{4}{*}{40} & \textbf{0.170 (0.378)} & 0.060             & 0.830              & 0.110              & \textbf{0.094 (0.133)} \\
SBS    &                     & 0.260 (0.441)          & 0.140             & 0.740              & 0.120              & 0.139 (0.174)          \\
WCUSUM &                     & \textbf{0.110 (0.314)} & 0.110             & 0.890              & -                 & \textbf{0.077 (0.153)} \\
FSBS   &                     & 1.980 (3.629)          & 0.700             & 0.060              & 0.240              & 0.425 (0.123) \\        \hline
\end{tabular}
\caption{Localisation results for Scenario \ref{scenario_ii}. Each cell is based on 100 repetitions. The columns $|\widehat{K} - K|$ and $\mathcal{D}$ are in the form of mean (standard deviation).  The columns $\widehat{K} < K$, $\widehat{K} = K$ and $\widehat{K} > K$ collect proportions of corresponding cases.}
\label{table_scenario_ii}
\end{table}

\begin{table}[!htbp]
\centering
\begin{tabular}{ccccccc}
\hline
Method & $n$                    & $|\widehat{K} - K|$   & $\widehat{K} < K$ & $\widehat{K} = K$ & $\widehat{K} > K$ & $\mathcal{D}$          \\ \hline
FFDP   & \multirow{3}{*}{150} & \textbf{0.370 (0.580)} & 0.300              & 0.680              & 0.020              & \textbf{0.151 (0.106)} \\
SBS    &                      & 0.460 (0.784)          & 0.200              & 0.720              & 0.080              & 0.152 (0.122)          \\
FSBS   &                      & 1.690 (1.229)          & 0.720              & 0.150              & 0.130              & 0.309 (0.092)          \\ \hline
FFDP   & \multirow{3}{*}{200} & \textbf{0.260 (0.441)} & 0.110              & 0.740              & 0.150              & \textbf{0.121 (0.073)} \\
SBS    &                      & 0.330 (0.551)          & 0.070              & 0.710              & 0.220              & 0.131 (0.085)          \\
FSBS   &                      & 2.260 (2.809)          & 0.670              & 0.110              & 0.220              & 0.295 (0.093)          \\ \hline
\end{tabular}
\caption{Localisation results for Scenario \ref{scenario_iii}. Each cell is based on 100 repetitions. The columns $|\widehat{K} - K|$ and $\mathcal{D}$ are in the form of mean (standard deviation).  The columns $\widehat{K} < K$, $\widehat{K} = K$ and $\widehat{K} > K$ collect proportions of corresponding cases.}
\label{table_scenario_iii}
\end{table}

\begin{table}[htbp!]
\centering
\begin{tabular}{ccclccc}
\hline
\multirow{2}{*}{$m$} & \multirow{2}{*}{$\mathcal{D}(\widehat{\eta}, \eta)$} & \multirow{2}{*}{$\mathcal{D}(\Tilde{\eta}, \eta)$} & \multicolumn{2}{c}{$\alpha = 0.01$}                                           & \multicolumn{2}{c}{$\alpha = 0.05$}                                           \\
                     &                                                      &                                                    & \multicolumn{1}{l}{cover($1-\alpha$)} & \multicolumn{1}{l}{width($1-\alpha$)} & \multicolumn{1}{l}{cover($1-\alpha$)} & \multicolumn{1}{l}{width($1-\alpha$)} \\ \hline
\multicolumn{7}{c}{$n =200$}                                                                                                                                                                                                                                                                     \\
\multirow{2}{*}{10}  & \multirow{2}{*}{0.117 (0.102)}                       & \multirow{2}{*}{0.040 (0.062)}                     & (I) 0.805                                 & 22.000 (20.361)                       & 0.772                                 & 18.000 (17.049)                       \\
                     &                                                      &                                                    & (II) 0.837                                 & 28.000 (31.301)                       & 0.837                                 & 26.000 (20.029)                       \\
\multirow{2}{*}{20}  & \multirow{2}{*}{0.116 (0.097)}                       & \multirow{2}{*}{0.042 (0.064)}                     & (I) 0.818                                 & 24.000 (23.724)                       & 0.802                                 & 20.000 (18.796)                       \\
                     &                                                      &                                                    & (II) 0.835                                 & 32.000 (26.309)                       & 0.835                                 & 28.000 (22.094)                       \\ \hline
\multicolumn{7}{c}{$n =300$}                                                                                                                                                                                                                                                                     \\
\multirow{2}{*}{10}  & \multirow{2}{*}{0.107 (0.109)}                       & \multirow{2}{*}{0.029 (0.056)}                     & (I) 0.847                                 & 28.000 (21.886)                       & 0.824                                 & 23.000 (16.860)                       \\
                     &                                                      &                                                    & (II) 0.878                                 & 36.000 (24.701)                       & 0.863                                 & 33.000 (19.666)                       \\
\multirow{2}{*}{20}  & \multirow{2}{*}{0.102 (0.104)}                       & \multirow{2}{*}{0.028 (0.059)}                     & (I) 0.904                                 & 36.000 (42.547)                       & 0.856                                 & 30.000 (31.257)                       \\
                     &                                                      &                                                    &(II) 0.920                                 & 44.000 (38.389)                       & 0.896                                 & 39.000 (30.482)                       \\ \hline
\multicolumn{7}{c}{$n = 400$}                                                                                                                                                                                                                                                                    \\
\multirow{2}{*}{10}  & \multirow{2}{*}{0.098 (0.116)}                       & \multirow{2}{*}{0.030 (0.064)}                     & (I) 0.857                                 & 56.000 (40.122)                       & 0.846                                 & 45.000 (29.286)                       \\
                     &                                                      &                                                    & (II) 0.934                                 & 67.000 (34.848)                       & 0.923                                 & 49.000 (26.855)                       \\
\multirow{2}{*}{20}  & \multirow{2}{*}{0.094 (0.108)}                       & \multirow{2}{*}{0.022 (0.054)}                     & (I) 0.902                                 & 51.000 (45.011)                       & 0.856                                 & 41.500 (32.730)                       \\
                     &                                                      &                                                    & (II) 0.955                                 & 61.500 (37.110)                       & 0.947                                 & 45.000 (27.415)                       \\ \hline
\end{tabular}
\caption{Inference results for Scenario \ref{scenario_iv}. The columns of $\mathcal{D}(\widehat{\eta}, \eta)$ and $\mathcal{D}(\Tilde{\eta}, \eta)$ are in the form of mean (standard deviations). The columns of width$(1-\alpha)$ are in the form of median (standard deviation). For each setting, (I) corresponds to the construction following \textbf{Step 4} and (II) corresponds to the construction following \textbf{Step 4'}.}
\label{table_inference1}
\end{table}

\subsection{Real data example} \label{section_numerical_real}
We apply our algorithm to analyse the daily demand curve of German power market, which consists of two publicly available datasets: the average hourly wholesale electricity prices ($\{Y_{t,j}\}_{t=1,j=1}^{n,m}$) from \cite{electricityprice}  and hourly values of Germany's gross electricity demand ($\{X_{t,j}\}_{t=1,j=1}^{n,m}$) from \cite{electricitydemand} in 2022. By pairing the prices and demand of the same timestamp, we construct a daily price-versus-demand curve. Similar datasets from different years and the same functions of interest construction have also been previously analysed in the FDA literature in terms of prediction \citep[e.g.][]{liebl2013modeling}, testing \citep[e.g.][]{liebl2019nonparametric} and covariance function estimations \citep[e.g.][]{ Kneip2020optimal}. 

After removing days with missing data points, to construct the model as detailed in \eqref{model_obs}, we model daily demand curves as realisations of a sequence of random functions, indexed by days.  We are interested in localising the days where the covariance functions change.  In addition to carrying out analysis using all available data, we conduct a separate analysis using only daily peak time data (from 8 am to 8 pm). Such selected data inherit the largest variations in electricity prices and at the same time have enhanced fragmented natures as demonstrated in \Cref{figure_domain}, which collects all the observable electricity demand pairs.  To be specific, each dot corresponds to an observation pair within a day, with its $x$- and $y$-axis values being the two normalised observed electricity demand values - normalised by setting the minimum and maximum demand values to zero and one respectively. The shaded diagonal areas in \Cref{figure_domain} in fact resonate with the missing off-diagonal covariance matrices phenomena in the fragmented functional data analysis problems.

Accordingly, the data dimensions of two separate analyses are given by $(n, m_1)= (355, 24)$ and $(n, m_2) = (355, 12)$ respectively.  The lengths of observation fragments (on the normalised demand axis) are 0.714 and 0.384 respectively.

To remove the effect of the unknown mean function, we estimate it by the penalised least squares method introduced in \citet{lin2021basis}, which targets mean function estimation for fragmented functional data. We apply FFDP to detect potential change points in daily demand curves with $r=4$, and the tuning parameters $(\lambda, \xi)$ are selected following the same implementation as the one described in \Cref{section_numerical_simulated}. For comparison, we also apply the SBS and WCUSUM procedures. Results are collected in \Cref{table_real_data} and plots of correspondingly estimated covariance functions before the first change point, between two change points and after the second change point outputted by the FFDP algorithm can be found in \Cref{figure_covariance}.

For both the whole and peak time datasets, we identify two change points after applying our proposed FFDP algorithm, with common change points shared around August. When using peak time data only, we detect another change point in April, which is close to a change point detected by SBS, and a change point detected by WCUSUM when using the whole dataset. We conjecture that the change point on $17^{\text{th}}$ April is related to the energy crisis associated with the Russia-Ukraine war. According to \citet{ITA}, around $52 \%$ of the total volume of natural gas imported into Germany comes from Russia in 2021. On $23^{\text{rd}}$ March, Russian President Putin signed a decree requiring foreign buyers to pay for Russian gas in roubles from $1^{\text{st}}$ April. The April change point likely lines up with the issue of the decree. In addition, we reckon that the August change point is correlated with the observation that August sees the highest average electricity price throughout the year. As noted by \citet{VoxEU}, these extremely high prices might be attributed to the competition within EU countries driven by concerns about limited supplies of natural gases. 

After an extra step of local refinement, we obtain refined estimators $14^{\text{th}}$ August and $31^{\text{th}}$ August when using the whole dataset, and $16^{\text{th}}$ April and $14^{\text{th}}$ August when using only the peak time data. The confidence interval for each refined estimator could also be constructed following the same procedure in \Cref{section_CI_construction}. The widths of these constructed intervals are, however, less than 1.  We therefore omit the inference step.  We attribute this phenomenon to the substantial jump sizes. The minimum estimated jump size between two consecutive covariance functions exceeds $9.24 \times 10^7$. Therefore, we may well be in the non-vanishing jump size regime where we do not consider the numerical inference step in this paper. 

\begin{minipage}[!ht]{0.3\textwidth}
\centering
\vspace{0.3cm}
\includegraphics[width=1\textwidth]{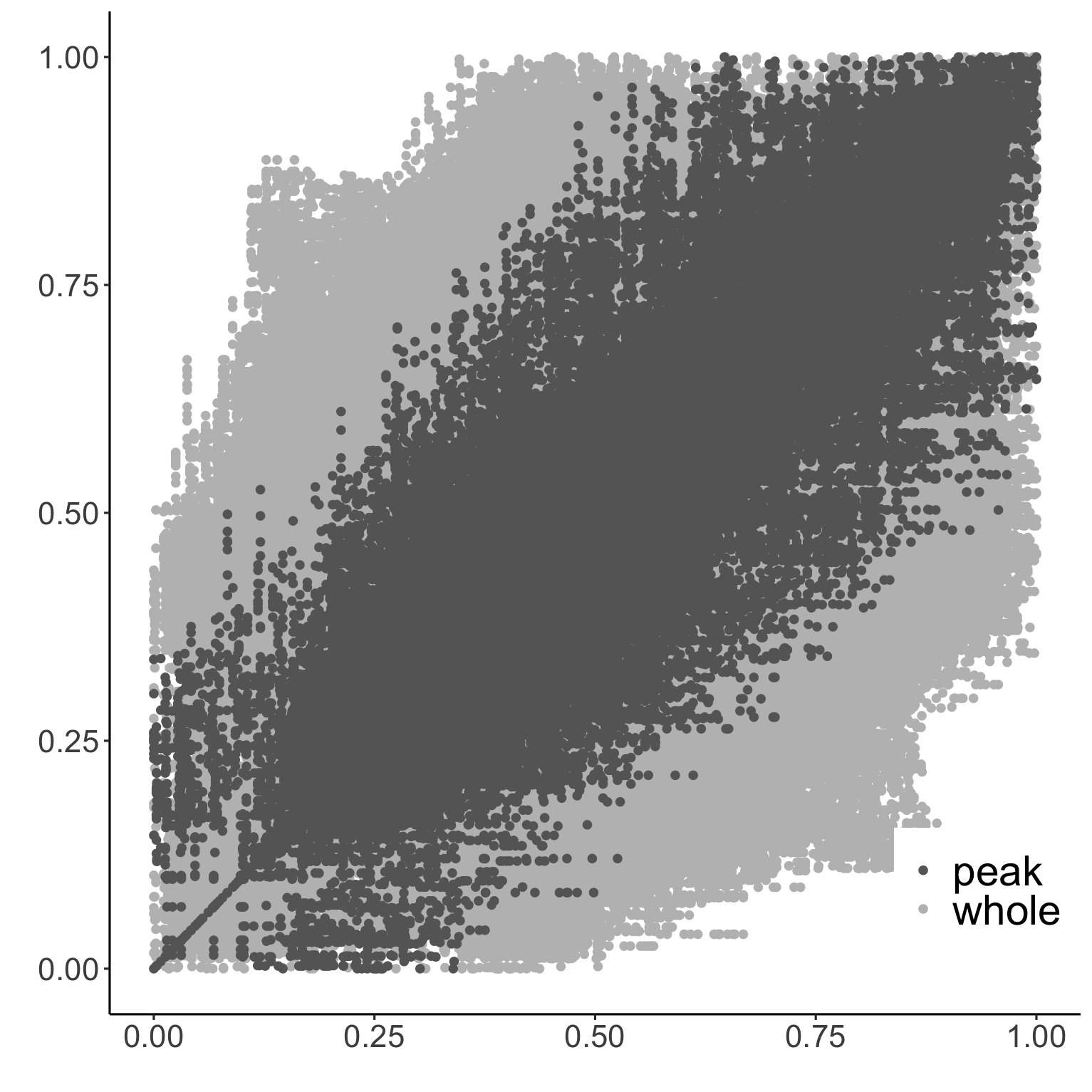}
\captionof{figure}{Observed domain of covariance functions using all observations in the real data example. The $x$-axis and $y$-axis are normalised demand.}
\label{figure_domain}
\end{minipage}
\hfill
\begin{minipage}[ht]{0.62\textwidth}
\centering
\begin{tabular}{ccc}
\hline
       & Whole data     & Peak time data \\ \hline
FFDP   & 26-Jul; 31-Aug & 17-Apr; 13-Aug \\
SBS    & 15-Aug; 12-Sep & 08-Apr; 15-Aug; 25-Nov \\
WCUSUM & 22-Apr         & 16-Aug         \\ \hline
\end{tabular}
\captionof{table}{Results of the estimated change points in the real data example.}
\label{table_real_data}
\end{minipage}

\begin{figure}[!htbp]
    \centering
    \includegraphics[scale = 0.34]{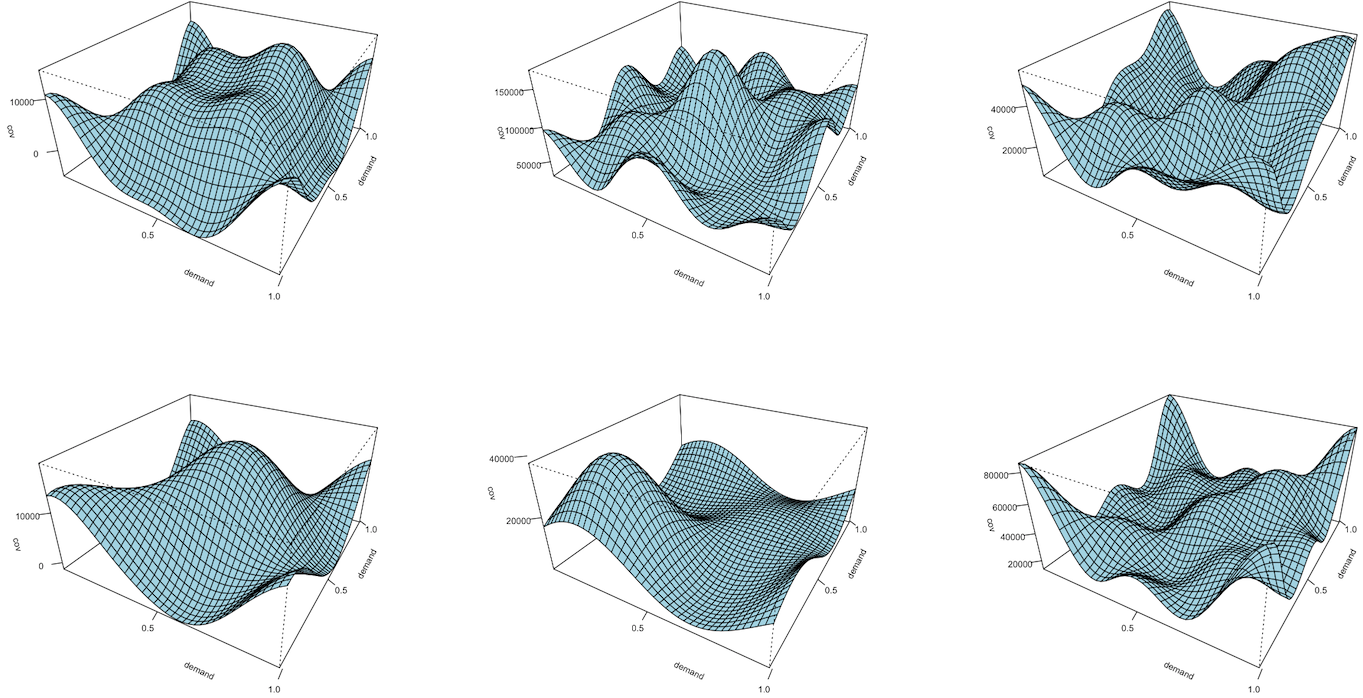}
    \caption{Estimated covariance functions. From left to right: estimated covariance functions before the first change point, between the two change points and after the second change point. The top and bottom rows correspond to analyses carried out with whole and peak-time data respectively.}
    \label{figure_covariance}
\end{figure}

\section{Conclusions} \label{section_conclusion} 
In this paper, we study the problems of change point localisation and inference in fragmented functional data. To the best of our knowledge, this study is the first time seen in the change point analysis literature. We propose a two-step method. In the first step, the localisation is efficiently implemented using the fragmented functional dynamic programming algorithm, whose outputs are also shown to have consistent localisation rates under common regularity conditions. In the second step, we perform local refined change point estimation and confidence construction based on our derived limiting distributions in different regimes of jump sizes.

We envisage several potential extensions. It will be interesting to further understand the role of the fragment length~$\delta$ on identifiability in covariance function estimation.  Lying in the core is a restricted eigenvalue type condition derived in \Cref{l_rec} in the Appendix. The quantity $\zeta_\delta$ - a lower bound on the restricted eigenvalue - is a function of $\delta$, but the dependence is yet known.  We conjecture some further structural assumptions on the covariance functions would be helpful in this task.

Another possible extension to our work is the implementation of limiting distributions when the jump size is non-vanishing. \Cref{t_inference} indicates that limiting distributions of change point estimators in the non-vanishing regime lack universality and depend heavily on the data generating mechanisms. While we focus on the implementation of the limiting distributions in the vanishing regime, there has been some recent work attempting to utilise limiting distributions in non-vanishing regime via bootstrap in univariate time series \citep[e.g.][]{chan2021optimal, ng2022bootstrap, cho2022bootstrap}. Deriving practical algorithms and theoretical properties for fragmented functional data remains an intriguing area for further investigation.

\section*{Acknowledgements}
Yu is partially supported by the Leverhulme Trust.

\newpage
\bibliographystyle{apalike}
\bibliography{reference}

\newpage
\begin{appendices}
\label{section_appendix}
All of the proofs and additional numerical details are collected in the Appendices. Details of the SBS algorithm in \Cref{section_numerical} can be found in Appendix \ref{appendix_SBS}. The proofs of Theorems \ref{t_localisation}, \ref{t_inference} and \ref{t_estimation} are presented in Appendices \ref{pf_localisation}, \ref{pf_inference} and \ref{pf_estimation} respectively. Necessary technical Lemmas can be found in Appendix \ref{pf_technical}.

Throughout the Appendix, with a slight abuse of notation, unless specifically stated otherwise, let~$C_1, C_2, \ldots > 0$ denote absolute constants whose values may vary from place to place. Also, denote the loss function $H(\cdot,\cdot)$ as
\begin{align*}
    H(C, I)= \begin{cases}
         0, \quad & |I| \leq \frac{\xi}{m} \; \text{and} \; C = \widehat{C}_I,\\
         \sum_{t \in I} \sum_{(j,k) \in \mathcal{O}} \big\{Y_{t,j}Y_{t,k} -\Phi_r^\top(X_{t,j}) C\Phi_r(X_{t,k})\big\}^2, \quad & \text{otherwise}.
     \end{cases}
\end{align*}

\section{Additional numerical details: Seeded Binary Segmentation (SBS)} \label{appendix_SBS}
In this section, we provide details of the SBS algorithm in \Cref{section_numerical}. Our algorithm is based on the SBS method introduced in \citet{kovacs2023seeded}. SBS is based on a collection of deterministic intervals defined in \Cref{def_seeded}. We further design a greedy selection method with our own constructed CUSUM statistics based on algorithms detailed in \Cref{section_algorithm}.

\begin{definition}[Seeded intervals] \label{def_seeded}
    Let $\mathcal{K}= \lceil C_\mathcal{K} \log_{2}(n)\rceil$, with some sufficiently large absolute constant $C_\mathcal{K} >0$. For $k \in \{1, \ldots, \mathcal{K}\}$, let $\mathcal{J}_k$ be the collection of $2^k-1$ intervals of length $l_k = n2^{-k+1}$ that that evenly shifted by $l_k/2 = n2^{-k}$, i.e.
    \begin{align*}
        \mathcal{J}_k = \{(\lfloor(i-1)n2^{-k}\rfloor, \lceil (i-1)n2^{-k}+n2^{-k+1}\rceil], \quad i =1, \ldots, 2^k-1\}.
    \end{align*}
    The overall collection of seeded intervals is denoted as $\mathcal{J} = \cup_{k=1}^\mathcal{K}\mathcal{J}_k$.
\end{definition}

\begin{algorithm}
    \caption{Seeded binary segmentation. SBS ($(s,e), \Phi, r,  \mathcal{J},\lambda,\tau)$}
    \begin{algorithmic}
        \Require Data $\{(X_{t,j},  Y_{t,j})\}_{t=1,j=1}^{n,m}$, complete orthonormal system $\Phi$, level of truncation $r\in \mathbb{Z}_{+}$, seeded intervals $\mathcal{J}$, and tuning parameter $\lambda, \tau> 0$
        \State \textbf{Initialization:} If $(s,e]=(0,n]$, set $\boldsymbol{S} \gets \varnothing$.
        \For{$\mathcal{I} = (\alpha,\beta] \in \mathcal{J}$}
        \State $b_{\mathcal{I}} \gets \underset{\alpha<t\leq \beta}{\arg\max} \; \big\{H(\widehat{C}_\mathcal{I}, \mathcal{I})-H(\widehat{C}_{(\alpha,t]}, (\alpha,t])-H(\widehat{C}_{(t,\beta]}, (t,\beta])\big\}$
        \State $a_{\mathcal{I}} \gets H(\widehat{C}_\mathcal{I}, \mathcal{I})-H(\widehat{C}_{(\alpha,b_{\mathcal{I}}]}, (\alpha,b_{\mathcal{I}}])-H(\widehat{C}_{(b_{\mathcal{I}},\beta]}, (b_{\mathcal{I}},\beta])$
        \EndFor
        \State$\mathcal{I}^* \gets \underset{\mathcal{I} \in \mathcal{J}}{\arg\max} \; a_\mathcal{I}$
        \If{$a_{\mathcal{I}^*} > \tau$}
        \State $\boldsymbol{S} \gets \boldsymbol{S} \cup \{b_{\mathcal{I}^*}\}$
        \State SBS$((s,b_{\mathcal{I}^*}),\Phi, r,  \mathcal{J},\lambda,\tau)$
        \State SBS$((b_{\mathcal{I}^*},e),\Phi, r,  \mathcal{J},\lambda,\tau)$
        \EndIf
        \Ensure The set of estimated change points $\boldsymbol{S}$
    \end{algorithmic}
\end{algorithm}

Apart from the orthonormal system $\Phi$, we have three more tuning parameters, $r,\tau, \lambda$. We select them by cross-validations. Specifically, we first divide $\{X_{t,j}, Y_{t,j}\}_{t=1, j=1}^{n,m}$ into training and validation sets according to odd and even indices of $t$. For each candidate $(\lambda, r, \tau)$, we obtain change point estimators $\{\widehat{\eta}^{(\text{t})}_\ell\}_{\ell=0}^{\widehat{K}+1}$ with $\widehat{\eta}^{(\text{t})}_0 =0$ and $\widehat{\eta}^{(\text{t})}_{\widehat{K}+1} =n/2$, and coefficients estimators $\{\widehat{C}_{(\widehat{\eta}^{(\text{t})}_\ell,\widehat{\eta}^{(\text{t})}_{\ell+1}]}\}_{\ell=0}^{\widehat{K}}$ on the training set. Using these outputs, we compute the validation loss
\begin{align*}
    \sum_{\ell = 0}^{\widehat{K}} \sum_{t = \widehat{\eta}^{(\text{t})}_\ell+1}^{\widehat{\eta}^{(\text{t})}_{\ell+1}} \sum_{(j,k)\in \mathcal{O}} \big\{Y_{t,j}Y_{t,k} - \Phi_r^\top(X_{t,j}) \widehat{C}_{(\widehat{\eta}^{(\text{t})}_\ell,\widehat{\eta}^{(\text{t})}_{\ell+1}]}\Phi_r(X_{t,k})\big\}^2
\end{align*}
with data in the validation set. The candidate $(\lambda, r, \tau)$ corresponding to the lowest validation loss is chosen as the final tuning parameter set.

\section{Proof of \texorpdfstring{\Cref{t_localisation}}{}} \label{pf_localisation}
\begin{proof}[Proof of \Cref{t_localisation}]
    It holds from \Cref{prop_localisation} that in two consecutive induced integer intervals formed by three estimated change points, there is at least one true change point, thus it follows that $|\widehat{\mathcal{P}}|-1 \leq 3K$. In addition, for any interval $I=(s,e] \in \widehat{\mathcal{P}}$, there are three possible scenarios: the interval $I$ contains no, one, or two true change points. If there is only one true change point inside $I$, then it is close to one of the endpoints. If there are two true change points, then the smaller true change point is closer to the left endpoint $s$, and the larger true change point is closer to the right endpoint $e$. Therefore, this shows that every true change point can be mapped to an estimated change point. Hence it holds that $|\widehat{\mathcal{P}}|-1> K$. Combined with the result in \Cref{prop_consistent_number}, we complete the proof of \Cref{t_localisation}.
\end{proof}

\begin{proposition} \label{prop_localisation}
    Under the same condition in \Cref{t_localisation} and letting $\widehat{\mathcal{P}}$ being the solution to \eqref{loss_l0}, the following events uniformly hold with probability at least $1-3n^{-3}$.
    \begin{enumerate}[label=(\roman*)]
        \item For all interval $I = (s,e] \in \widehat{\mathcal{P}}$ containing one and only one true change point $\eta_\ell$, it must be the case that for a sufficiently large constant $C_\varepsilon > 0$, 
        \begin{align*}
            \min\{\eta_\ell-s,e-\eta_\ell\} \leq C_{\epsilon}K\Big\{\frac{r^4\log^2(n)}{\kappa_\ell^2\delta^2\zeta_{\delta}^2} \vee \frac{r^6\log(n)}{\kappa_\ell^2\delta^2 \zeta_{\delta}^2m}\Big\};
        \end{align*}

        \item for all interval $I = (s,e] \in \widehat{\mathcal{P}}$ containing exactly two change points, say $\eta_\ell < \eta_{\ell+1}$, it must be the case that for a sufficiently large $C_\varepsilon > 0$,  
        \begin{align*}
            \eta_\ell-s \leq C_{\epsilon}K\Big\{\frac{r^4\log^2(n)}{\kappa_\ell^2\delta^2\zeta_{\delta}^2} \vee \frac{r^6\log(n)}{\delta^2 \zeta_{\delta}^2m}\Big\} \quad \text{and} \quad e-\eta_{\ell+1} \leq C_{\epsilon}K\Big\{\frac{r^4\log^2(n)}{\kappa_\ell^2\delta^2\zeta_{\delta}^2} \vee \frac{r^6\log(n)}{\kappa_\ell^2\delta^2 \zeta_{\delta}^2m}\Big\};
        \end{align*}

        \item for any two consecutive induced intervals $I_1$ and $I_2 \in \widehat{\mathcal{P}}$, the interval $I_1 \cup I_2$ contains at least one true change point; and,

        \item no interval $I \in \widehat{\mathcal{P}}$ contains strictly more than two true change points. 
    \end{enumerate}
\end{proposition}

\begin{proposition}\label{prop_consistent_number}
    Under the same condition in \Cref{t_localisation}, with $\widehat{\mathcal{P}}$ being the solution to \eqref{loss_l0}, satisfying $K \leq |\widehat{\mathcal{P}}|-1 \leq 3K$, then it holds with probability at least $1-3n^{-3}$ that $|\widehat{\mathcal{P}}|-1 = \widehat{K} = K$.
\end{proposition}

\subsection{Proof of \texorpdfstring{\Cref{prop_localisation}}{}}
In this subsection, we present the proofs of the four cases in \Cref{prop_localisation}. Throughout the whole section, we assume that all of the conditions in \Cref{t_localisation} hold.

\begin{lemma} \label{l_localisation_1true}
     Let $I = (s,e] \in \widehat{\mathcal{P}}$ be such that $I$ contains one and only one true change point $\eta_\ell$, then it holds with probability at least $1-3n^{-3}$ that
     \begin{align*}
         \min\{\eta_\ell-s,e-\eta_\ell\} \leq C_{\epsilon}K\Big\{\frac{r^4\log^2(n)}{\kappa_\ell^2\delta^2\zeta_{\delta}^2} \vee \frac{r^6\log(n)}{\kappa_\ell^2\delta^2 \zeta_{\delta}^2m}\Big\}.
     \end{align*}
\end{lemma}

\begin{proof}
For an interval $I = (s,e] \in \widehat{\mathcal{P}}$ containing only one true change point $\eta_\ell$, denote
\begin{align*}
    I_1 = (s,\eta_\ell] \quad \mathrm{and} \quad I_2 = (\eta_\ell,e].
\end{align*}
If $ |I| < \frac{\xi}{m\zeta_{\delta}\kappa_\ell^2}$, then there is nothing to show, and the result holds automatically. If $|I| \geq \frac{\xi}{m\zeta_{\delta}\kappa_\ell^2}$ and any of $|I_1|,|I_2|< \frac{\xi}{m\zeta_{\delta}\kappa_\ell^2}$, then the result also holds trivially. In the rest of the proof, we will prove the result by contradiction, and it suffices to assume that
\begin{align} \label{l_localisation_1true_eq1}
    \min\{|I_1|, |I_2|\} > \frac{\xi}{m\zeta_{\delta}\kappa_\ell^2} \geq C_{\kappa,\zeta} ^{-1} \frac{\xi}{m},
\end{align}
where $C_{\kappa,\zeta} > 0$ is an absolute constant, and the second inequality follows from the fact that $\kappa_\ell$ is bounded as $\Sigma^*_t \in L^2([0,1]^2)$ for any $t \in [n]$, and $\zeta_{\delta} \leq C_\zeta$ by \Cref{r_zeta_delta}. Since the algorithm does not split the induced interval $I$ further into two disjoint intervals, it must be the case that
\begin{align} \label{l_localisation_1true_eq2}
    H(\widehat{C}_I, I)=H(\widehat{C}_I, I_1)+ H(\widehat{C}_I, I_2) \leq H(\widehat{C}_{I_1}, I_1)+ H(\widehat{C}_{I_2}, I_2) + \xi.
\end{align}
Since for $\ell \in \{1,2\}$, it holds from \Cref{l_localisation_1true_eq1} that $\min\{|I_1|, |I_2|\} > C_{\kappa,\zeta}^{-1}\frac{\xi}{m}$. By \Cref{l_loss_long}, for any $\ell \in \{1,2\}$, it holds with probability at least $1-3n^{-3}$ that
\begin{align} \notag
    \Big|H(C^*_{r,I_\ell}, I_\ell)-H(\widehat{C}_{I_\ell}, I_\ell)\Big| &\leq C_1\Big\{\frac{r^4m\log(n)}{\delta^2\zeta_{\delta}} \vee \frac{r^6\log(n)}{\delta^2 \zeta_{\delta}} \vee \frac{r^{2-2q}|I_\ell|m}{\delta^6\zeta_{\delta}}\Big\}\\
    \label{l_localisation_1true_eq3}
    & \leq C_1\Big\{\frac{r^4m\log(n)}{\delta^2\zeta_{\delta}} \vee \frac{r^6\log(n)}{\delta^2 \zeta_{\delta}} \vee \frac{r^{2-2q}nm}{\delta^6\zeta_{\delta}}\Big\},
\end{align}
where $C^*_{r,I_\ell} = \frac{1}{|I_\ell|}\sum_{t \in I_\ell} C^*_{r,t}$. Plugging \Cref{l_localisation_1true_eq3} into \Cref{l_localisation_1true_eq2}, we have that 
\begin{align} \notag
    H(\widehat{C}_I, I_1)+ H(\widehat{C}_I, I_2) &\leq H(\widehat{C}_{I_1}, I_1)+ H(\widehat{C}_{I_2}, I_2) + \xi\\
    \label{l_localisation_1true_eq4}
    & \leq H(C^*_{r,I_1}, I_1)+H(C^*_{r,I_2}, I_2)+ \xi +C_2\Big\{\frac{r^4m\log(n)}{\delta^2\zeta_{\delta}} \vee \frac{r^6\log(n)}{\delta^2 \zeta_{\delta}} \vee \frac{r^{2-2q}nm}{\delta^6\zeta_{\delta}}\Big\}.
\end{align}
Rearranging \Cref{l_localisation_1true_eq4}, we have that
\begin{align}\label{l_localisation_1true_eq5}
    &\sum_{\ell=1}^2 \sum_{t\in I_\ell}\sum_{(j,k) \in \mathcal{O}} \big\{\Phi_r^\top(X_{t,j})(\widehat{C}_I-C^*_{r,I_\ell})\Phi_r(X_{t,k})\big\}^2\\
    \label{l_localisation_1true_eq6}
    \leq \;& 2 \sum_{\ell=1}^2 \sum_{t\in I_\ell} \sum_{(j,k) \in \mathcal{O}}\big\{Y_{t,j}Y_{t,k}-\Phi_r^\top(X_{t,j})C^*_{r,I_\ell}\Phi_r(X_{t,k})\big\}\big\{\Phi_r^\top(X_{t,j})(\widehat{C}_I-C^*_{r,I_\ell})\Phi_r(X_{t,k})\big\}\\
    \notag
    &+\xi+C_2\Big\{\frac{r^4m\log(n)}{\delta^2\zeta_{\delta}} \vee \frac{r^6\log(n)}{\delta^2 \zeta_{\delta}} \vee \frac{r^{2-2q}nm}{\delta^6\zeta_{\delta}}\Big\}.
\end{align}
\noindent \textbf{Step 1.} To find a lower bound on \Cref{l_localisation_1true_eq5}, denote $C^*_{r, I} = \frac{1}{|I|}\sum_{t \in I} C^*_{r,t}$, then \Cref{l_localisation_1true_eq5} could be rewritten as
\begin{align} \notag
    &\sum_{\ell=1}^2 \sum_{t\in I_\ell}\sum_{(j,k) \in \mathcal{O}} \big\{\Phi_r^\top(X_{t,j})(\widehat{C}_I-C^*_{r,I_\ell})\Phi_r(X_{t,k})\big\}^2\\
    \label{l_localisation_1true_eq7}
    =\;& \sum_{\ell=1}^2\sum_{t\in I_\ell}\sum_{(j,k) \in \mathcal{O}} \big\{\Phi_r^\top(X_{t,j})(\widehat{C}_I-C^*_{r,I})\Phi_r(X_{t,k}) + \Phi_r(X_{t,j})^{\top}(C^*_{r,I} - C^*_{r,I_\ell})\Phi_r(X_{t,k})\big\}^2.
\end{align}
To analyse \Cref{l_localisation_1true_eq7}, when $\ell =1$, we have that 
\begin{align} \notag
    &\sum_{t\in I_1}\sum_{(j,k) \in \mathcal{O}} \big\{\Phi_r^\top(X_{t,j})(\widehat{C}_I-C^*_{r,I})\Phi_r(X_{t,k}) + \Phi_r(X_{t,j})^{\top}(C^*_{r,I} - C^*_{r,I_1})\Phi_r(X_{t,k})\big\}^2\\
    \notag
    =\;&\sum_{t\in I_1}\sum_{(j,k) \in \mathcal{O}} \Big\{\Phi_r^\top(X_{t,j})(\widehat{C}_I-C^*_{r,I})\Phi_r(X_{t,k}) + \Phi_r^{\top}(X_{t,j})\frac{|I_2|(C^*_{r,I_2} - C^*_{r,I_1})}{|I|}\Phi_r(X_{t,k})\Big\}^2\\
    \label{l_localisation_1true_eq8}
    \geq\;& \frac{1}{2}\sum_{t\in I_1}\sum_{(j,k) \in \mathcal{O}}\Big\{\Phi_r^{\top}(X_{t,j})\frac{|I_2|(C^*_{r,I_2} - C^*_{r,I_1})}{|I|}\Phi_r(X_{t,k})\Big\}^2 -\sum_{t\in I_1}\sum_{(j,k) \in \mathcal{O}}\Big\{\Phi_r^{\top}(X_{t,j})(\widehat{C}_I-C^*_{r,I})\Phi_r(X_{t,k})\Big\}^2,
\end{align}
where the inequality follows from the fact that $(a+b)^2 \geq 1/2a^2- b^2$ for any $a,b \in \mathbb{R}$. To compute a lower bound on the first term in \Cref{l_localisation_1true_eq8}, firstly note that $C_{r, I_2}^* \in \mathcal{C}$ since a finite summation of analytic function is still analytic \citep[Proposition 1.1.7 in][]{krantz2002primer}. Therefore, by \Cref{l_localisation_1true_eq1} and the construction of $\mathcal{B}_{r,I_1}(\theta)$ in \Cref{t_estimation_eq_B}, we have that
\begin{align*}
    C^*_{r,I_2} - C^*_{r,I_1} \in \mathcal{B}_{r,I_1}(\|C^*_{r,I_2} - C^*_{r,I_1}\|_{\F})=\big\{\Delta \in \mathbb{R}^{r \times r}: \|\Delta\|_{\F}=\|C^*_{r,I_2} - C^*_{r,I_1}\|_{\F}, C^*_{r,I_1}+\Delta \in \mathfrak{C}(r)\big\},
\end{align*}
Therefore, by \Cref{l_rec}, we have that
\begin{align} \label{l_localisation_1true_eq9}
    \mathbb{E}[\{\Phi_r^\top(X_{1,1})(C^*_{r,I_2} - C^*_{r,I_1})\Phi_r(X_{1,2})\}^2|a_1] \geq \zeta_{\delta}\|C^*_{r,I_2} - C^*_{r,I_1}\|_{\F}^2
\end{align}
almost surely. Combining the result in \Cref{l_localisation_1true_eq9} with \Cref{l_op3}, we have that with probability at least $1-3n^{-3}$
\begin{align} \notag
    &\sum_{t\in I_1}\sum_{(j,k) \in \mathcal{O}}\Big\{\Phi_r^{\top}(X_{t,j})\frac{|I_2|(C^*_{r,I_2} - C^*_{r,I_1})}{|I|}\Phi_r(X_{t,k})\Big\}^2 \\
    \label{l_localisation_1true_eq10}
    \geq\;& \frac{\zeta_{\delta}|I_1||I_2|^2m}{|I|^2} \|C^*_{r,I_2} - C^*_{r,I_1}\|_{\F}^2 - \frac{C_3r^2|I_2|^2}{|I|^2}\sqrt{|I_1|m\log(n)}\|C^*_{r,I_2} - C^*_{r,I_1}\|_{\F}^2.
\end{align}
To find an upper bound on the second term in \Cref{l_localisation_1true_eq8}, using a similar argument as the one in the $\textbf{Step 1}$ in the proof of \Cref{l_loss_long}, we have that with probability at least $1-3n^{-3}$
\begin{align} \label{l_localisation_1true_eq11}
    \sum_{t\in I_1}\sum_{(j,k) \in \mathcal{O}}\Big\{\Phi_r^{\top}(X_{t,j})(\widehat{C}_I-C^*_{r,I})\Phi_r(X_{t,k})\Big\}^2 \leq \frac{C_4\zeta_{\delta}|I_1|mr^2}{\delta^2}\|\widehat{C}_I-C^*_{r,I}\|_{\F}^2,
\end{align}
if we select $r$ such that 
\begin{align} \label{l_localisation_1true_eq12}
     \frac{r^{3}\sqrt{\log(n)}}{\sqrt{|I_1|m}} \leq \frac{r^{3}\sqrt{\log(n)}}{\sqrt{\xi}} \lesssim \frac{\zeta_{\delta}}{2},
\end{align}
where the first inequality holds from \Cref{l_localisation_1true_eq1}. Therefore, with the condition in \Cref{l_localisation_1true_eq12}, combining \Cref{l_localisation_1true_eq10} and \Cref{l_localisation_1true_eq11}, we have with probability at least $1-3n^{-3}$ that 
\begin{align} \notag
    &\sum_{t\in I_1}\sum_{(j,k) \in \mathcal{O}} \big\{\Phi_r^\top(X_{t,j})(\widehat{C}_I-C^*_{r,I})\Phi_r(X_{t,k}) + \Phi_r(X_{t,j})^{\top}(C^*_{r,I} - C^*_{r,I_1})\Phi_r(X_{t,k})\big\}^2 \\
    \label{l_localisation_1true_eq13}
    \geq\;& \frac{\zeta_{\delta}|I_1||I_2|^2m}{4|I|^2} \|C^*_{r,I_2}-C^*_{r,I_1}\|_{\F}^2-\frac{\zeta_{\delta}|I_1|mr^2}{2\delta^2}\|\widehat{C}_I-C^*_{r,I}\|_{\F}^2.
\end{align}
By a similar justification, it also holds that for the other term in \Cref{l_localisation_1true_eq7} when $\ell=2$, 
\begin{align} \notag
    &\sum_{t\in I_2}\sum_{(j,k) \in \mathcal{O}} \big\{\Phi_r^\top(X_{t,j})(\widehat{C}_I-C^*_{r,I})\Phi_r(X_{t,k}) + \Phi_r(X_{t,j})^{\top}(C^*_{r,I} - C^*_{r,I_2})\Phi_r(X_{t,k})\big\}^2\\ \notag
    =\;&\sum_{t\in I_2}\sum_{(j,k) \in \mathcal{O}} \Big\{\Phi_r^\top(X_{t,j})(\widehat{C}_I-C^*_{r,I})\Phi_r(X_{t,k}) - \Phi_r^{\top}(X_{t,j})\frac{|I_1|(C^*_{r,I_2} - C^*_{r,I_1})}{|I|}\Phi_r(X_{t,k})\Big\}^2 \\
    \label{l_localisation_1true_eq14}
    \geq \;& \frac{\zeta_{\delta}|I_1|^2|I_2|m}{4|I|^2} \|C^*_{r,I_2}-C^*_{r,I_1}\|_{\F}^2 -\frac{\zeta_{\delta}|I_2|mr^2}{2\delta^2}\|\widehat{C}_I-C^*_{r,I}\|_{\F}^2.
\end{align}
Therefore, substituting \Cref{l_localisation_1true_eq13} and \Cref{l_localisation_1true_eq14} into \Cref{l_localisation_1true_eq7}, we have that 
\begin{align} \label{l_localisation_1true_eq15}
    \sum_{\ell=1}^2 \sum_{t\in I_\ell}\sum_{(j,k) \in \mathcal{O}} \big\{\Phi_r^\top(X_{t,j})(\widehat{C}_I-C^*_{r,I_\ell})\Phi_r(X_{t,k})\big\}^2 \geq &\frac{\zeta_{\delta}|I_1||I_2|m}{4|I|} \|C^*_{r,I_2}-C^*_{r,I_1}\|_{\F}^2 -\frac{\zeta_{\delta}|I|mr^2}{2\delta^2}\|\widehat{C}_I-C^*_{r,I}\|_{\F}^2.
\end{align}
\noindent \textbf{Step 2.} Next, we will find an upper bound on \Cref{l_localisation_1true_eq6}. For any $\ell \in \{1,2\}$, note that by \Cref{l_op1}, \Cref{l_op2} and \Cref{l_loss_long_eq9}, we have with probability at least $1-3n^{-3}$ that
\begin{align} \notag
    &\sum_{t\in I_\ell} \sum_{(j,k) \in \mathcal{O}}\big\{Y_{t,j}Y_{t,k}-\Phi_r^\top(X_{t,j})C^*_{r,I_\ell}\Phi_r(X_{t,k})\big\}\big\{\Phi_r^\top(X_{t,j})(\widehat{C}_I-C^*_{r,I_\ell})\Phi_r(X_{t,k})\big\}\\
    \notag
    \leq\;& C_5\Big\{r\sqrt{|I_\ell|m^2\log(n)} \vee r^{2}\sqrt{|I_\ell|m\log(n)} \vee \frac{r^{1-q}|I_\ell|m}{\delta^2} \Big\}\; \|\widehat{C}_I-C^*_{r,I_\ell}\|_{\F}\\
    \notag
    =\;& C_5\sqrt{\frac{\zeta_{\delta}|I_\ell|m}{32}}\|\widehat{C}_I-C^*_{r,I_\ell}\|_{\F} \cdot \Big\{r \sqrt{\frac{m\log(n)}{\zeta_{\delta}/32}} \vee r^{2} \sqrt{\frac{\log(n)}{\zeta_{\delta}/32}} \vee \frac{r^{1-q}}{\delta^2}\sqrt{\frac{|I_\ell|m}{\zeta_{\delta}/32}}\Big\}\\
    \notag
    \leq\;& C_5\sqrt{\frac{\zeta_{\delta}|I_\ell|m}{32}}\|\widehat{C}_I-C^*_{r,I_\ell}\|_{\F} \cdot \Big\{r \sqrt{\frac{m\log(n)}{\zeta_{\delta}/32}} \vee r^{2} \sqrt{\frac{\log(n)}{\zeta_{\delta}/32}} \vee \frac{r^{1-q}}{\delta^2}\sqrt{\frac{nm}{\zeta_{\delta}/32}}\Big\}\\
    \label{l_localisation_1true_eq16}
    \leq\;& \frac{\zeta_{\delta}|I_\ell|m}{32}\|\widehat{C}_I-C^*_{r,I_\ell}\|_{\F}^2+ C_6\Big\{\frac{r^2m\log(n)}{\zeta_{\delta}} \vee \frac{r^{4}\log(n)}{\zeta_{\delta}} \vee \frac{r^{2-2q}nm}{\delta^4\zeta_{\delta}}\Big\},
\end{align}
where the last inequality follows from the fact that $2ab \leq a^2 +b^2$ for all $a,b \in \mathbb{R}$. Also, we have that 
\begin{align} \notag
    \frac{\zeta_{\delta}|I_1|m}{32}\|\widehat{C}_I-C^*_{r,I_1}\|_{\F}^2 &= \frac{\zeta_{\delta}|I_1|m}{32}\|\widehat{C}_I-C^*_{r,I}+C^*_{r,I}-C^*_{r,I_1}\|_{\F}^2\\
    \notag
    & \leq \frac{\zeta_{\delta}|I_1|m}{16}\|\widehat{C}_I-C^*_{r,I}\|_{\F}^2+\frac{\zeta_{\delta}|I_1|m}{16}\|C^*_{r,I}-C^*_{r,I_1}\|_{\F}^2\\
    \label{l_localisation_1true_eq17}
    & = \frac{\zeta_{\delta}|I_1|m}{16}\|\widehat{C}_I-C^*_{r,I}\|_{\F}^2 + \frac{\zeta_{\delta}|I_1||I_2|^2m}{16|I|^2}\|C^*_{r,I_2}-C^*_{r,I_1}\|_{\F}^2.
\end{align}
Similarly, we also have that 
\begin{align}\label{l_localisation_1true_eq18}
    \frac{\zeta_{\delta}|I_2|m}{32}\|\widehat{C}_I-C^*_{r,I_2}\|_{\F}^2 \leq \frac{\zeta_{\delta}|I_2|m}{16}\|\widehat{C}_I-C^*_{r,I}\|_{\F}^2 + \frac{\zeta_{\delta}|I_1|^2|I_2|m}{16|I|^2}\|C^*_{r,I_2}-C^*_{r,I_1}\|_{\F}^2.
\end{align}
Therefore, by \Cref{l_localisation_1true_eq16}, \Cref{l_localisation_1true_eq17}, and \Cref{l_localisation_1true_eq18}, we have with probability at least $1-3n^{-3}$ that
\begin{align} \notag
    & \sum_{\ell=1}^2\sum_{t\in I_\ell} \sum_{(j,k) \in \mathcal{O}}\big\{Y_{t,j}Y_{t,k}-\Phi_r^\top(X_{t,j})C^*_{r,I_\ell}\Phi_r(X_{t,k})\big\}\big\{\Phi_r^\top(X_{t,j})(\widehat{C}_I-C^*_{r,I_\ell})\Phi_r(X_{t,k})\big\}\\
    \label{l_localisation_1true_eq19}
    \leq\;& \frac{\zeta_{\delta}|I|m}{16}\|\widehat{C}_I-C^*_{r,I}\|_{\F}^2+\frac{\zeta_{\delta}|I_1||I_2|m}{16|I|}\|C^*_{r,I_2}-C^*_{r,I_1}\|_{\F}^2+2C_6\Big\{\frac{r^2m\log(n)}{\zeta_{\delta}} \vee \frac{r^{4}\log(n)}{\zeta_{\delta}} \vee \frac{r^{2-2q}nm}{\delta^4\zeta_{\delta}}\Big\}.
\end{align}
\noindent \textbf{Step 3.} Substituting \Cref{l_localisation_1true_eq15} and \Cref{l_localisation_1true_eq19} into \Cref{l_localisation_1true_eq5} and \Cref{l_localisation_1true_eq6}, it holds with probability at least $1-3n^{-3}$ that
\begin{align*}
    \frac{\zeta_{\delta}|I_1||I_2|m}{8|I|}\|C^*_{r,I_2}-C^*_{r,I_1}\|_{\F}^2 \leq\;&  \frac{\zeta_{\delta}|I|mr^2}{2\delta^2}\|\widehat{C}_I-C^*_{r,I}\|_{\F}^2+\frac{\zeta_{\delta}|I|m}{8}\;\|\widehat{C}_I-C^*_{r,I}\|_{\F}^2 \\
    &+\xi+C_7\Big\{\frac{r^4m\log(n)}{\delta^2\zeta_{\delta}} \vee \frac{r^6\log(n)}{\delta^2 \zeta_{\delta}} \vee \frac{r^{2-2q}nm}{\delta^6\zeta_{\delta}}\Big\}\\
    \leq\;& \xi+C_8\Big\{\frac{r^4m\log(n)}{\delta^2\zeta_{\delta}} \vee \frac{r^6\log(n)}{\delta^2 \zeta_{\delta}} \vee \frac{r^{2-2q}nm}{\delta^6\zeta_{\delta}}\Big\},
\end{align*}
where the last inequality holds from \Cref{t_estimation}. Observe that
\begin{align*}
    \frac{\min\{|I_1|,|I_2|\}\|C^*_{r,I_1}-C^*_{r,I_2}\|_{\F}^2}{2} \leq  \frac{\zeta_{\delta}|I_1||I_2|m}{|I|}\|C^*_{r,I_2}-C^*_{r,I_1}\|_{\F}^2,
\end{align*}
therefore we have that 
\begin{align*}
    \frac{\zeta_{\delta}\min\{|I_1|,|I_2|\}m\|C^*_{r,I_1}-C^*_{r,I_2}\|_{\F}^2}{16} \leq \xi+C_8\Big\{\frac{r^4m\log(n)}{\delta^2\zeta_{\delta}} \vee \frac{r^6\log(n)}{\delta^2 \zeta_{\delta}} \vee \frac{r^{2-2q}nm}{\delta^6\zeta_{\delta}}\Big\}.
\end{align*}
Pick $\xi = C_{\xi}K\Big\{\frac{r^4 m\log^2(n)}{\delta^2\zeta_{\delta}} \vee \frac{r^6\log(n)}{\delta^2 \zeta_{\delta}} \vee \frac{r^{2-2q}nm}{\delta^6\zeta_{\delta}}\Big\}$ where $C_{\xi} >0$ is an absolute constant,  we have that 
\begin{align} \label{l_localisation_1true_eq20}
    \min\{|I_1|,|I_2|\} &\leq C_{9}K\Big\{\frac{r^4\log^2(n)}{\delta^2\zeta_{\delta}^2} \vee \frac{r^6\log(n)}{\delta^2 \zeta_{\delta}^2m} \vee \frac{r^{2-2q}n}{\delta^6\zeta_{\delta}^2}\Big\} \frac{1}{\|C^*_{r,I_1}-C^*_{r,I_2}\|_{\F}^2}.
\end{align}
\noindent \textbf{Step 4.} By \Cref{a_localisation_cov} and the triangle inequality, it holds that for any $\ell \in [K]$
\begin{align*}
    \kappa_\ell^2 &= \|\Sigma^*_{\eta_{\ell+1}}-\Sigma^*_{\eta_{\ell}}\|_{L^2}^2\\
    &=\|\Sigma^*_{\eta_{\ell+1}}-\Sigma^*_{r,\eta_{\ell+1}}-(\Sigma^*_{\eta_{\ell}}-\Sigma^*_{r,\eta_{\ell}})+\Sigma^*_{r,\eta_{\ell+1}}-\Sigma^*_{r,\eta_{\ell}}\|_{L^2}^2\\
    & \leq 3\Big\{\|\Sigma^*_{\eta_{\ell+1}}-\Sigma^*_{r,\eta_{\ell+1}}\|_{L^2}^2+\|\Sigma^*_{\eta_{\ell}}-\Sigma^*_{r,\eta_{\ell}}\|_{L^2}^2+\|\Sigma^*_{r,\eta_{\ell+1}}-\Sigma^*_{r,\eta_{\ell}}\|_{L^2}^2\Big\}\\
    & \leq 3\|C^*_{r,\eta_{\ell+1}}-C^*_{r,\eta_{\ell}}\|_{\F}^2 + 6C_{10}r^{-2q}.
\end{align*}
Therefore, by \Cref{a_localisation_jump}\ref{a_localisation_jump_lb}, we have that 
\begin{align}\label{l_localisation_1true_eq21}
    \|C^*_{r,\eta_{\ell+1}}-C^*_{r,\eta_{\ell}}\|_{\F}^2 \geq \frac{\kappa_\ell^2}{3}-2C_{10}r^{-2q} \geq \frac{\kappa_\ell^2}{6}.
\end{align}
Putting \Cref{l_localisation_1true_eq21} into \Cref{l_localisation_1true_eq20}, it holds that
\begin{align} \notag
    \min\{|I_1|,|I_2|\} &\leq C_{11}K\Big\{\frac{r^4\log^2(n)}{\delta^2\zeta_{\delta}^2} \vee \frac{r^6\log(n)}{\delta^2 \zeta_{\delta}^2m} \vee \frac{r^{2-2q}n}{\delta^6\zeta_{\delta}^2}\Big\} \frac{1}{\kappa_\ell^2}\\
    \label{l_localisation_1true_eq22}
    &= C_{11}\Big\{(I) \vee (II) \vee (III)\Big\} \frac{1}{\kappa_\ell^2},
\end{align}
which is a contradiction to \Cref{l_localisation_1true_eq1}. Hence, the Lemma holds.

\noindent \textbf{Step 5. (Choices of $r$)} In this step, we match the order of the three terms in \Cref{l_localisation_1true_eq22} to find the optimal rate of $r$. We will use $(I)$ and $(II)$ to derive a threshold of $m$ and then derive the optimal choice of $r$ by matching $(III)$ with $(I)$ and $(II)$ respectively. Matching $(I)$ with $(II)$, we have that $m = \frac{r^2}{\log(n)}$. Hence when $m < \frac{r^2}{\log(n)}$, $(II)$ dominates. By equating $(II)$ and $(III)$, we have that
\begin{align} \label{l_localisation_1true_eq23}
    r = \Big(\frac{nm}{\delta^4 \log(n)}\Big)^{\frac{1}{4+2q}},
\end{align}
and correspondingly, the rate for $m$ becomes
\begin{align*}
    m < n^{\frac{1}{1+q}}\delta^{-\frac{4}{1+q}}\log^{-\frac{3+q}{1+q}}(n).
\end{align*}
In the other case, when we have that $m > \frac{r^2}{\log(n)}$, $(I)$ dominates. By equating $(I)$ and $(III)$, we have that
\begin{align} \label{l_localisation_1true_eq24}
     r = \Big(\frac{n}{\delta^4 \log^2(n)}\Big)^{\frac{1}{2+2q}},
\end{align}
and correspondingly, the rate for $m$ becomes
\begin{align*}
    m > n^{\frac{1}{1+q}}\delta^{-\frac{4}{1+q}}\log^{-\frac{3+q}{1+q}}(n).
\end{align*}
\end{proof}

\begin{remark}[Choice of $\lambda$] \label{r_choice_lambda}
    Substituting the choice of $r$ given in \Cref{l_localisation_1true_eq23} and \Cref{l_localisation_1true_eq24} into \Cref{t_estimation_lambda} in \Cref{t_estimation}, we have that for any interval $|I| \subseteq [n]$ such that $|I| \geq \xi/m$, 
    \begin{align*}
        \Big\{r^{-11/2}\sqrt{m\log(n)} \vee r^{-9/2}\sqrt{\log(n)}\Big\} =  \frac{r^{-13/2}\sqrt{nm}r^{-q}}{\delta^2} \geq \frac{r^{-13/2}\sqrt{|I|m}r^{-q}}{\delta^2}.
    \end{align*}
    Hence, the value of $\lambda$ in \Cref{t_localisation} is picked at
    \begin{align*}
        \lambda = C_{\lambda}\Big\{r^{-11/2}\sqrt{m\log(n)} \vee r^{-9/2}\sqrt{\log(n)}\Big\}. 
    \end{align*}
\end{remark}

\begin{remark} \label{r_choice_r}
    Note that with the choices of $r$ and $\xi$ given in \Cref{t_localisation}, \Cref{l_localisation_1true_eq12} holds automatically if $\delta K^{-1/2} \lesssim \zeta_{\delta}$.
    \begin{itemize}
        \item When $m \leq n^{\frac{1}{1+q}}\delta^{-\frac{4}{1+q}}\log^{-\frac{3+q}{1+q}}(n)$, pick $\xi = \frac{C_{\xi}Kr^6\log(n)}{\delta^2\zeta_{\delta}}$, then with the optimal choice of $r$ in \Cref{l_localisation_1true_eq23}, we have that 
        \begin{align*}
            r^3\sqrt{\frac{\log(n)}{\xi}} = \frac{C_1r^3 \log^{1/2}(n)\delta\zeta_{\delta}^{1/2}}{r^3\log^{1/2}(n)K^{1/2}} = \frac{C_1\delta\zeta_{\delta}^{1/2}}{K^{1/2}} \lesssim \zeta_{\delta}.
        \end{align*}

        \item $m \geq n^{\frac{1}{1+q}}\delta^{-\frac{4}{1+q}}\log^{-\frac{3+q}{1+q}}(n)$, pick $\xi = \frac{C_{\xi}Kr^4m\log^2(n)}{\delta^2\zeta_{\delta}}$, then with the optimal choice of $r$ in \Cref{l_localisation_1true_eq24}, we have that 
        \begin{align*}
            r^3\sqrt{\frac{\log(n)}{\xi}} =r\delta\sqrt{\frac{C_2\zeta_{\delta}}{m\log(n)K}} \lesssim \Big(\frac{n}{
            \delta^4\log^2(n)}\Big)^{\frac{1}{2+2q}} \delta\sqrt{\frac{\zeta_{\delta}}{\log(n)K}} \frac{\delta^{\frac{4}{2+2q}}\log^{\frac{3+q}{2+2q}}(n)}{n^{\frac{1}{2+2q}}}= \frac{\delta\zeta_{\delta}^{1/2}}{K^{1/2}} \lesssim \zeta_{\delta}.
        \end{align*}
    \end{itemize}
\end{remark}

\begin{lemma} \label{l_localisation_2true}
    Let $I = (s,e] \in \widehat{\mathcal{P}}$ be such that $I$ contains two true change points $s \leq \eta_\ell < \eta_{\ell+1} \leq e$, then it holds with probability at least $1-3n^{-3}$ that for a sufficiently large constant $C_\epsilon > 0$,
    \begin{align*}
        \eta_\ell-s \leq C_{\epsilon}K\Big\{\frac{r^4\log^2(n)}{\kappa_\ell^2\delta^2\zeta_{\delta}^2} \vee \frac{r^6\log(n)}{\kappa_\ell^2\delta^2 \zeta_{\delta}^2m}\Big\} \quad \text{and} \quad e-\eta_{\ell+1} \leq C_{\epsilon}K\Big\{\frac{r^4\log^2(n)}{\kappa_\ell^2\delta^2\zeta_{\delta}^2} \vee \frac{r^6\log(n)}{\kappa_\ell^2\delta^2 \zeta_{\delta}^2m}\Big\}.
    \end{align*}
\end{lemma}

\begin{proof}
    For an interval $I = (s,e] \in \widehat{\mathcal{P}}$ containing two true change points $\eta_\ell$ and $\eta_{\ell+1}$, denote
    \begin{align*}
        I_1 = (s,\eta_\ell], \quad I_2 = (\eta_\ell,\eta_{\ell+1}] \quad \text{and} \quad I_3 = (\eta_{\ell+1},e].
    \end{align*}
    Without loss of generality, we will assume that $|I_1| \geq |I_3|$, and it suffices to derive an upper bound on $|I_1|$. Note that if $|I_1| \leq \frac{\xi}{m\zeta_{\delta}\kappa_\ell^2}$, then the result follows trivially. Hence, it suffices to assume that 
    \begin{align} \label{l_localisation_2true_eq1}
        |I_1| \geq \frac{\xi}{m\zeta_{\delta}\kappa_\ell^2} \geq C_{\kappa,\zeta}^{-1}\frac{\xi}{m},
    \end{align}
    where $C_{\kappa,\zeta} >0$ is a constant depending on the upper bound on $\kappa$ and the constant $C_\zeta$ in \Cref{r_zeta_delta}. Since $I$ contains two true change points, it follows that $|I| > |\eta_{\ell+1} - \eta_{\ell}| =|I_{2}| \geq \Delta \gtrsim \frac{\xi}{m}$, where the last inequality follows from \Cref{a_localisation_jump}\ref{a_localisation_jump_snr}. In the rest of the proof, there are two possible cases to consider, $|I_1| \geq |I_3| \gtrsim \frac{\xi}{m}$ or $|I_1| \gtrsim \frac{\xi}{m} > |I_3|$, and the two cases will be considered individually in $\textbf{Step 1}$ and $\textbf{Step 2}$. 

    \noindent \textbf{Step 1.} In the first case when $|I_1| \geq |I_3| \gtrsim \frac{\xi}{m}$, since $I \in \widehat{\mathcal{P}}$ and the algorithm does not split $I$ into three disjoint intervals, it must be the case that with probability at least $1-3n^{-3}$,
    \begin{align} \notag
        H\big(\widehat{C}_{I},  I\big)
        &\leq H(\widehat{C}_{I_1}, I_1)+ H(\widehat{C}_{I_2}, I_2)+H\big(\widehat{C}_{I_3}, I_3\big) +2 \xi\\
        \label{l_localisation_2true_eq2}
        & \leq H(C^*_{r,I_1}, I_1)+ H(C^*_{r,I_2}, I_2)+H\big(C^*_{r,I_3}, I_3\big) +2 \xi+ C_1\Big\{\frac{r^4m\log(n)}{\delta^2\zeta_{\delta}} \vee \frac{r^6\log(n)}{\delta^2 \zeta_{\delta}} \vee \frac{r^{2-2q}nm}{\delta^6\zeta_{\delta}}\Big\},
    \end{align}
    where the second inequality follows from \Cref{l_loss_long}. Rearranging \Cref{l_localisation_2true_eq2}, we have that 
    \begin{align} \label{l_localisation_2true_eq3}
        &\sum_{\ell=1}^3 \sum_{t\in I_\ell}\sum_{(j,k) \in \mathcal{O}} \big\{\Phi_r^\top(X_{t,j})(\widehat{C}_I-C^*_{r,I_\ell})\Phi_r(X_{t,k})\big\}^2\\
        \label{l_localisation_2true_eq4}
        \leq \;& \sum_{\ell=1}^3\sum_{t\in I_\ell} \sum_{(j,k) \in \mathcal{O}}\big\{Y_{t,j}Y_{t,k}-\Phi_r^\top(X_{t,j})C^*_{r,I_\ell}\Phi_r(X_{t,k})\big\}\big\{\Phi_r^\top(X_{t,j})(\widehat{C}_I-C^*_{r,I_\ell})\Phi_r(X_{t,k})\big\}\\
        \notag
        &+2\xi+C_1\Big\{\frac{r^4m\log(n)}{\delta^2\zeta_{\delta}} \vee \frac{r^6\log(n)}{\delta^2 \zeta_{\delta}} \vee \frac{r^{2-2q}nm}{\delta^6\zeta_{\delta}}\Big\}.
    \end{align}
    \noindent \textbf{Step 1.1.} To find a lower bound on \Cref{l_localisation_2true_eq3}, denote $C^*_{r,I} = \frac{1}{|I|}\sum_{t\in I} C^*_t$, then for any $\ell \in \{1, 2, 3\}$, using a similar argument as the one in \Cref{l_localisation_1true_eq13} and \Cref{l_localisation_1true_eq14}, we have with probability at least $1-3n^{-3}$ that
    \begin{align*}
        &\sum_{t\in I_\ell}\sum_{(j,k) \in \mathcal{O}} \big\{\Phi_r^\top(X_{t,j})(\widehat{C}_I-C^*_{r,I_\ell})\Phi_r(X_{t,k})\big\}^2\\
        =\; & \sum_{t\in I_\ell}\sum_{(j,k) \in \mathcal{O}} \big\{\Phi_r^\top(X_{t,j})(\widehat{C}_I-C^*_{r,I})\Phi_r(X_{t,k}) + \Phi_r(X_{t,j})^{\top}(C^*_{r,I} - C^*_{r,I_\ell})\Phi_r(X_{t,k})\big\}^2\\
        \geq \;&\frac{1}{2}\sum_{t\in I_\ell}\sum_{(j,k) \in \mathcal{O}}\big\{\Phi_r^{\top}(X_{t,j})(C^*_{r,I} - C^*_{r,I_\ell})\Phi_r(X_{t,k})\big\}^2- \sum_{t\in I_\ell}\sum_{(j,k) \in \mathcal{O}}\big\{\Phi_r^\top(X_{t,j})(\widehat{C}_I-C^*_{r,I})\Phi_r(X_{t,k}) \big\}^2\\
        \geq \;& \frac{\zeta_{\delta}|I_\ell|m}{4}\|C^*_{r,I} - C^*_{r,I_\ell}\|_{\F} - \sum_{t\in I_\ell}\sum_{(j,k) \in \mathcal{O}}\big\{\Phi_r^\top(X_{t,j})(\widehat{C}_I-C^*_{r,I})\Phi_r(X_{t,k}) \big\}^2\\
        \geq \;& \frac{\zeta_{\delta}|I_\ell|m}{4}\|C^*_{r,I} - C^*_{r,I_\ell}\|_{\F}^2 - \frac{\zeta_{\delta}|I_\ell|mr^2}{2\delta^2}\|\widehat{C}_I-C^*_{r,I}\|_{\F}^2.
    \end{align*}
    Therefore, we have that 
    \begin{align} \notag
        &\sum_{\ell=1}^3 \sum_{t\in I_\ell}\sum_{(j,k) \in \mathcal{O}} \big\{\Phi_r^\top(X_{t,j})(\widehat{C}_I-C^*_{r,I_\ell})\Phi_r(X_{t,k})\big\}^2\\
        \notag
        \geq\;& \sum_{\ell=1}^3 \Big\{\frac{\zeta_{\delta}|I_\ell|m}{4}\|C^*_{r,I} - C^*_{r,I_\ell}\|_{\F}^2 - \frac{\zeta_{\delta}|I_\ell|mr^2}{2\delta^2}\|\widehat{C}_I-C^*_{r,I}\|_{\F}^2\Big\}\\
        \label{l_localisation_2true_eq5}
        =\;& \sum_{\ell=1}^3 \frac{\zeta_{\delta}|I_\ell|m}{4}\|C^*_{r,I} - C^*_{r,I_\ell}\|_{\F}^2 - \frac{\zeta_{\delta}|I|mr^2}{2\delta^2}\|\widehat{C}_I-C^*_{r,I}\|_{\F}^2.
    \end{align}

    \noindent \textbf{Step 1.2.} To find an upper bound on \Cref{l_localisation_2true_eq4}, using a similar argument as the one in deriving \Cref{l_localisation_1true_eq16} and \Cref{l_localisation_1true_eq19}, we have that with probability at least $1-3n^{-3}$ that, for any $\ell \in \{1,2,3\}$,
    \begin{align*}
        &\sum_{t\in I_\ell} \sum_{(j,k) \in \mathcal{O}}\big\{Y_{t,j}Y_{t,k}-\Phi_r^\top(X_{t,j})C^*_{r,I_\ell}\Phi_r(X_{t,k})\big\}\big\{\Phi_r^\top(X_{t,j})(\widehat{C}_I-C^*_{r,I_\ell})\Phi_r(X_{t,k})\big\}\\
        \leq \;& C_2\Big\{r\sqrt{|I_\ell|m^2\log(n)} \vee r^{2}\sqrt{|I_\ell|m\log(n)} \vee \frac{r^{1-q}|I_\ell|m}{\delta^2} \Big\}\; \|\widehat{C}_I-C^*_{r,I_\ell}\|_{\F}\\
        \leq \;& C_2 \sqrt{\frac{\zeta_{\delta}|I_\ell|m}{32}}\|\widehat{C}_I-C^*_{r,I_\ell}\|_{\F} \cdot \Big\{r \sqrt{\frac{m\log(n)}{\zeta_{\delta}/32}} \vee r^{2} \sqrt{\frac{\log(n)}{\zeta_{\delta}/32}} \vee \frac{r^{1-q}}{\delta^2}\sqrt{\frac{nm}{\zeta_{\delta}/32}}\Big\}\\
        \leq \;& \frac{\zeta_{\delta}|I_\ell|m}{32}\|\widehat{C}_I-C^*_{r,I_\ell}\|_{\F}^2+ C_3\Big\{\frac{r^2m\log(n)}{\zeta_{\delta}} \vee \frac{r^{4}\log(n)}{\zeta_{\delta}} \vee \frac{r^{2-2q}nm}{\delta^4\zeta_{\delta}}\Big\}\\
        \leq \;& \frac{\zeta_{\delta}|I_\ell|m}{16}\|\widehat{C}_I-C^*_{r,I}\|_{\F}^2+\frac{\zeta_{\delta}|I_\ell|m}{16}\|C^*_{r,I}-C^*_{r,I_\ell}\|_{\F}^2 +C_3\Big\{\frac{r^2m\log(n)}{\zeta_{\delta}} \vee \frac{r^{4}\log(n)}{\zeta_{\delta}} \vee \frac{r^{2-2q}nm}{\delta^4\zeta_{\delta}}\Big\}.
    \end{align*}
    Therefore, we have that 
    \begin{align} \notag
        &\sum_{\ell=1}^3\sum_{t\in I_\ell} \sum_{(j,k) \in \mathcal{O}}\big\{Y_{t,j}Y_{t,k}-\Phi_r^\top(X_{t,j})C^*_{r,I_\ell}\Phi_r(X_{t,k})\big\}\big\{\Phi_r^\top(X_{t,j})(\widehat{C}_I-C^*_{r,I_\ell})\Phi_r(X_{t,k})\big\}\\
        \label{l_localisation_2true_eq6}
        \leq \;& \sum_{\ell=1}^3\frac{\zeta_{\delta}|I_\ell|m}{16}\|C^*_{r,I}-C^*_{r,I_\ell}\|_{\F}^2  + \frac{\zeta_{\delta}|I|m}{16}\|\widehat{C}_I-C^*_{r,I}\|_{\F}^2+ C_4\Big\{\frac{r^2m\log(n)}{\zeta_{\delta}} \vee \frac{r^{4}\log(n)}{\zeta_{\delta}} \vee \frac{r^{2-2q}nm}{\delta^4\zeta_{\delta}}\Big\}.
    \end{align}
    \noindent \textbf{Step 1.3.} Putting the results in \Cref{l_localisation_2true_eq5} and \Cref{l_localisation_2true_eq6} into \Cref{l_localisation_2true_eq3} and \Cref{l_localisation_2true_eq4}, we have that 
    \begin{align*}
        \sum_{\ell=1}^3 \frac{\zeta_{\delta}|I_\ell|m}{8}\|C^*_{r,I} - C^*_{r,I_\ell}\|_{\F}^2 \leq \;&\Big\{\frac{\zeta_{\delta}|I|mr^2}{2\delta^2}+\frac{\zeta_{\delta}|I|m}{16}\Big\}\|\widehat{C}_I-C^*_{r,I}\|_{\F}^2\\
        &+ 2\xi+C_5\Big\{\frac{r^4m\log(n)}{\delta^2\zeta_{\delta}} \vee \frac{r^6\log(n)}{\delta^2 \zeta_{\delta}} \vee \frac{r^{2-2q}nm}{\delta^6\zeta_{\delta}}\Big\}.
    \end{align*}
    Using \Cref{t_estimation} and the fact that $\frac{\zeta_{\delta}|I_3|m}{8}\|C^*_{r,I} - C^*_{r,I_3}\|_{\F}^2 \geq 0$, it holds with probability at least $1-3n^{-3}$ that 
    \begin{align} \label{l_localisation_2true_eq7}
         \sum_{\ell=1}^2 \frac{\zeta_{\delta}|I_\ell|m}{8}\|C^*_{r,I} - C^*_{r,I_\ell}\|_{\F}^2 \leq 2\xi+C_6\Big\{\frac{r^4m\log(n)}{\delta^2\zeta_{\delta}} \vee \frac{r^6\log(n)}{\delta^2 \zeta_{\delta}} \vee \frac{r^{2-2q}nm}{\delta^6\zeta_{\delta}}\Big\}.
    \end{align}
    Observe that 
    \begin{align} \notag
        \sum_{\ell=1}^2 \frac{\zeta_{\delta}|I_\ell|m}{8}\|C^*_{r,I} - C^*_{r,I_\ell}\|_{\F}^2 &\geq \underset{C \in \mathbb{R}^{r\times r}}{\min} \Big\{ \frac{\zeta_{\delta}|I_1|m}{8}\|C - C^*_{r,I_1}\|_{\F}^2+\frac{\zeta_{\delta}|I_2|m}{8}\|C - C^*_{r,I_2}\|_{\F}^2\Big\}\\
        \notag
        & = \frac{\zeta_{\delta}|I_1||I_2|^2m}{8|I|^2}\|C^*_{r,I_1}-C^*_{r,I_2}\|_{\F}^2 + \frac{\zeta_{\delta}|I_1|^2|I_2|m}{8|I|^2}\|C^*_{r,I_1}-C^*_{r,I_2}\|_{\F}^2\\
        \label{l_localisation_2true_eq8}
        &\geq \frac{\zeta_{\delta}\min\{|I_1|,|I_2|\}m\|C^*_{r,I_1}-C^*_{r,I_2}\|_{\F}^2}{16}.
    \end{align}
    Put \Cref{l_localisation_2true_eq8} into \Cref{l_localisation_2true_eq7}, and pick $\xi = C_{\xi}K\Big\{\frac{r^4m\log^2(n)}{\delta^2\zeta_{\delta}} \vee \frac{r^6\log(n)}{\delta^2 \zeta_{\delta}} \vee \frac{r^{2-2q}nm}{\delta^6\zeta_{\delta}}\Big\}$, we have that
    \begin{align*}
        \min\{|I_1|,|I_2|\} &\leq C_7K\Big\{\frac{r^4\log^2(n)}{\delta^2\zeta_{\delta}^2} \vee \frac{r^6\log(n)}{\delta^2 \zeta_{\delta}^2m} \vee \frac{r^{2-2q}n}{\delta^6\zeta_{\delta}^2}\Big\} \frac{1}{\|C^*_{r,I_1}-C^*_{r,I_2}\|_{\F}^2}\\
        & \leq C_8K\Big\{\frac{r^4\log^2(n)}{\delta^2\zeta_{\delta}^2} \vee \frac{r^6\log(n)}{\delta^2 \zeta_{\delta}^2m} \vee \frac{r^{2-2q}n}{\delta^6\zeta_{\delta}^2}\Big\} \frac{1}{\kappa_\ell^2},
    \end{align*}
    where the second inequality follows from \Cref{l_localisation_1true_eq21}. Finally, since $|I_2| \geq \Delta$, then it must be the case that
    \begin{align*}
        |I_1| \leq C_8K\Big\{\frac{r^4\log^2(n)}{\delta^2\zeta_{\delta}^2} \vee \frac{r^6\log(n)}{\delta^2 \zeta_{\delta}^2m} \vee \frac{r^{2-2q}n}{\delta^6\zeta_{\delta}^2}\Big\} \frac{1}{\kappa_k^2},
    \end{align*}
    which is a contradiction to \Cref{l_localisation_2true_eq1}.

    \noindent \textbf{Step 2.} In the second case, when we have $|I_1| \gtrsim \frac{\xi}{m} > |I_3|$, note that with probability at least $1-3n^{-3}$, 
    \begin{align} \notag
        &H\big(\widehat{C}_{I}, I_3\big) - H(C^*_{r,I_3}, I_3)\\
        \notag
        =\;& \sum_{t\in I_3} \sum_{(j,k) \in \mathcal{O}} \big\{Y_{t,j}Y_{t,k} - \Phi_r^\top(X_{t,j}) \widehat{C}_{I} \Phi_r(X_{t,k})\big\}^2- \sum_{t\in I_3} \sum_{(j,k) \in \mathcal{O}} \big\{Y_{t,j}Y_{t,k} - \Phi_r^\top(X_{t,j}) C^*_{r,I_3} \Phi_r(X_{t,k})\big\}^2\\
        \notag
        =\;& \sum_{t\in I_3} \sum_{(j,k) \in \mathcal{O}} \{\Phi_r^\top(X_{t,j})(\widehat{C}_{I} - C^*_{r,I_3})\Phi_r(X_{t,k})\}^2\\
        \notag
        &-2\sum_{t\in I_3} \sum_{(j,k) \in \mathcal{O}}\{Y_{t,j}Y_{t,k}-\Phi_r^\top(X_{t,j})C^*_{r,I_3}\Phi_r(X_{t,k})\}\{\Phi_r^\top(X_{t,j})(\widehat{C}_{I} - C^*_{r,I_3})\Phi_r(X_{t,k})\}\\
        \notag
        \geq\;&  \sum_{t\in I_3} \sum_{(j,k) \in \mathcal{O}} \{\Phi_r^\top(X_{t,j})(\widehat{C}_{I} - C^*_{r,I_3})\Phi_r(X_{t,k})\}^2 -\frac{1}{2}\sum_{t\in I_3} \sum_{(j,k) \in \mathcal{O}} \{\Phi_r^\top(X_{t,j})(\widehat{C}_{I} - C^*_{r,I_3})\Phi_r(X_{t,k})\}^2 \\
        \notag
        & - 2\sum_{t\in I_3} \sum_{(j,k) \in \mathcal{O}} \Big\{Y_{t,j}Y_{t,k} - \Phi_r^\top(X_{t,j}) C^*_{r,I_3} \Phi_r(X_{t,k})\Big\}^2\\
        \label{l_localisation_2true_eq9}
        \geq\;& \frac{1}{2}\sum_{t\in I_3} \sum_{(j,k) \in \mathcal{O}} \{\Phi_r^\top(X_{t,j})(\widehat{C}_{I} - C^*_{r,I_3})\Phi_r(X_{t,k})\}^2 - \frac{\xi}{2} \geq -\frac{\xi}{2},
    \end{align}
    where the second inequality follows from \Cref{l_loss_short}, and the last inequality follows from the fact that $\frac{1}{2}\sum_{t\in I_3} \sum_{(j,k) \in \mathcal{O}} \{\Phi_r^\top(X_{t,j})(\widehat{C}_{I} - C^*_{r,I_3})\Phi_r(X_{t,k})\}^2 \geq 0$. Since the algorithm does not split $I$ into three disjoint intervals and by \Cref{l_localisation_2true_eq9} and using the result in \Cref{l_loss_long} for interval $I_1$ and $I_2$, it holds with probability at least $1-3n^{-3}$ that 
    \begin{align*}
         H(\widehat{C}_{I}, I_1)+ H(\widehat{C}_{I}, I_2) \leq H(C^*_{r,I_1}, I_1)+ H(C^*_{r,I_2}, I_2) +\frac{5}{2} \xi+ C_9\Big\{\frac{r^4m\log(n)}{\delta^2\zeta_{\delta}} \vee \frac{r^6\log(n)}{\delta^2 \zeta_{\delta}} \vee \frac{r^{2-2q}nm}{\delta^6\zeta_{\delta}}\Big\}.
    \end{align*}
    Rearranging similarly as \Cref{l_localisation_2true_eq5} and \Cref{l_localisation_2true_eq6} in \textbf{Step 1}, it holds with probability at least $1-3n^{-3}$ that
    \begin{align*}
        &\frac{\zeta_{\delta}|I_1|m}{4}\|C^*_{r,I} - C^*_{r,I_1}\|_{\F}^2+\frac{\zeta_{\delta}|I_2|m}{4}\|C^*_{r,I} - C^*_{r,I_2}\|_{\F}^2\\
        \leq\;& \frac{\zeta_{\delta}(|I_1|+|I_2|)mr^2}{2\delta^2}\|\widehat{C}_I-C^*_{r,I}\|_{\F}^2+ \sum_{\ell=1}^2 \Bigg\{\frac{\zeta_{\delta}|I_\ell|m}{8}\|\widehat{C}_I-C^*_{r,I}\|_{\F}^2+\frac{\zeta_{\delta}|I_\ell|m}{8}\|C^*_{r,I}-C^*_{r,I_\ell}\|_{\F}^2\Bigg\}\\
        &+\frac{5}{2}\xi+C_{10}\Big\{\frac{r^4m\log(n)}{\delta^2\zeta_{\delta}} \vee \frac{r^6\log(n)}{\delta^2 \zeta_{\delta}} \vee \frac{r^{2-2q}nm}{\delta^6\zeta_{\delta}}\Big\}.
    \end{align*}
    By \Cref{t_estimation} and rearranging the above equation, it holds that 
    \begin{align} \notag
        &\frac{\zeta_{\delta}|I_1|m}{8}\|C^*_{r,I} - C^*_{r,I_1}\|_{\F}^2+\frac{\zeta_{\delta}|I_2|m}{8}\|C^*_{r,I} - C^*_{r,I_2}\|_{\F}^2 \\\notag
        \leq \;&\Big\{\frac{\zeta_{\delta}|I|mr^2}{2\delta^2} +\frac{\zeta_{\delta}|I|m}{8}\Big\}\|\widehat{C}_I-C^*_{r,I}\|_{\F}^2+\frac{5}{2}\xi+C_{10}\Big\{\frac{r^4m\log(n)}{\delta^2\zeta_{\delta}} \vee \frac{r^6\log(n)}{\delta^2 \zeta_{\delta}} \vee \frac{r^{2-2q}nm}{\delta^6\zeta_{\delta}}\Big\}\\
        \label{l_localisation_2true_eq10}
        \leq \;&\frac{5}{2}\xi+C_{11}\Big\{\frac{r^4m\log(n)}{\delta^2\zeta_{\delta}} \vee \frac{r^6\log(n)}{\delta^2 \zeta_{\delta}} \vee \frac{r^{2-2q}nm}{\delta^6\zeta_{\delta}}\Big\},
    \end{align}
    where the first inequality holds since $\frac{\zeta_{\delta}|I_3|mr^2}{2\delta^2}\|\widehat{C}_I-C^*_{r,I}\|_{\F}^2, \frac{\zeta_{\delta}|I_3|m}{8}\|\widehat{C}_I-C^*_{r,I}\|_{\F}^2 \geq 0$. By \Cref{l_localisation_2true_eq8} and a similar argument in \textbf{Step 1.3}, we reach again a contradiction to \Cref{l_localisation_2true_eq1}, which immediately gives the desired result.
\end{proof}

\begin{lemma} \label{l_localisation_3true}
    With probability at least $1-3n^{-3}$, there is no interval $I \in \widehat{\mathcal{P}}$ containing three or more true change points.
\end{lemma}

\begin{proof}
    We will prove this by contradiction. Suppose $I =(s,e] \in \widehat{\mathcal{P}}$ is an induced interval containing $L \geq 3$ true change points $\{\eta_1, \eta_2, \dots, \eta_L\}$. Denote
    \begin{align*} 
        I_1 = (s,\eta_1], \; I_{L+1} = (\eta_L,e],\; \text{and}\; \;I_\ell = (\eta_{\ell-1}, \eta_\ell] \; \; \text{for}\;\; \ell \in \{2,3,\dots,L\}.
    \end{align*}
    Since $I$ contains $L$ true change points, by \Cref{a_localisation_jump}\ref{a_localisation_jump_snr}, it must hold that $|I| > \Delta \gtrsim \frac{\xi}{m}$, and $|I_\ell| >\Delta \gtrsim \frac{\xi}{m}$ for each $\ell \in \{2,3,\dots, L\}$. Since the algorithm does not split $I$ into $L+1$ disjoint intervals, it must be the case that
    \begin{align} \label{l_localisation_3true_eq1}
        H(\widehat{C}_{I},I) \leq \sum_{\ell =1}^{L+1} H(\widehat{C}_{I_\ell},I_\ell) + L\xi.
    \end{align}
    In the rest of the proof, it suffices to consider four cases: $\min\{|I_1|, |I_{L+1}|\} \geq \frac{\xi}{m}$; $\max \{|I_1|, |I_{L+1}|\} \leq \frac{\xi}{m}$; $|I_1| \geq \frac{\xi}{m} > |I_{L+1}|$; and $|I_{L+1}| \geq \frac{\xi}{m} > |I_1|$.

    \noindent \textbf{Step 1.} Denote $C^*_{r,I} = \frac{1}{|I|}\sum_{t \in I}C^*_{r,t}$. In the first case when $\min\{|I_1|, |I_{L+1}|\} \geq \frac{\xi}{m}$, using a similar argument that leads to \Cref{l_localisation_2true_eq5} and \Cref{l_localisation_2true_eq6}, we have that with probability at least $1-3n^{-3}$ that for $\ell \in \{1, L+1\}$,
    \begin{align}\notag
        H(\widehat{C}_{I},I_\ell) - H(C^*_{r,I_\ell},I_\ell) \geq \;& \frac{\zeta_{\delta}|I_\ell|m}{8}\|C^*_{r,I} - C^*_{r,I_\ell}\|_{\F}^2 - \frac{\zeta_{\delta}|I_\ell|mr^2}{2\delta^2}\|\widehat{C}_I-C^*_{r,I}\|_{\F}^2 \\
        \notag
        &-\frac{\zeta_{\delta}|I_\ell|m}{8}\|\widehat{C}_I-C^*_{r,I}\|_{\F}^2-C_1\Big\{\frac{r^2m\log(n)}{\zeta_{\delta}} \vee \frac{r^{4}\log(n)}{\zeta_{\delta}} \vee \frac{r^{2-2q}nm}{\delta^4\zeta_{\delta}}\Big\} \\
        \label{l_localisation_3true_eq2}
        \geq\;& -C_2\Big\{\frac{r^4m\log(n)}{\delta^2\zeta_{\delta}} \vee \frac{r^6\log(n)}{\delta^2 \zeta_{\delta}} \vee \frac{r^{2-2q}nm}{\delta^6\zeta_{\delta}}\Big\} \geq -\frac{\xi}{2},
    \end{align}
    where the second inequality follows from \Cref{t_estimation}. In the second case when $\max \{|I_1|, |I_{L+1}|\} \leq \frac{\xi}{m}$, by \Cref{l_loss_short}, we have that for $\ell \in \{1, L+1\}$
    \begin{align} \label{l_localisation_3true_eq3}
        H(\widehat{C}_{I},I_\ell) - H(C^*_{r,I_\ell},I_\ell) \geq - H(C^*_{r,I_\ell},I_\ell) \geq -\frac{\xi}{2}.
    \end{align}
    In the third case when $|I_1| \geq \frac{\xi}{m} > |I_{L+1}|$, by the arguments that lead to \Cref{l_localisation_2true_eq5} and \Cref{l_localisation_2true_eq6}, and \Cref{l_loss_short}, we have that 
    \begin{align}\notag
        H(\widehat{C}_{I},I_1) - H(C^*_{r,I_1},I_1) \geq \;& \frac{\zeta_{\delta}|I_1|m}{8}\|C^*_{r,I} - C^*_{r,I_1}\|_{\F}^2 - \frac{\zeta_{\delta}|I_1|mr^2}{2\delta^2}\|\widehat{C}_I-C^*_{r,I}\|_{\F}^2 \\
        \notag
        &-\frac{\zeta_{\delta}|I_1|m}{8}\|\widehat{C}_I-C^*_{r,I}\|_{\F}^2-C_3\Big\{\frac{r^2m\log(n)}{\zeta_{\delta}} \vee \frac{r^{4}\log(n)}{\zeta_{\delta}} \vee \frac{r^{2-2q}nm}{\delta^4\zeta_{\delta}}\Big\} \\
        \label{l_localisation_3true_eq4}
        \geq\;& -C_4\Big\{\frac{r^4m\log(n)}{\delta^2\zeta_{\delta}} \vee \frac{r^6\log(n)}{\delta^2 \zeta_{\delta}} \vee \frac{r^{2-2q}nm}{\delta^6\zeta_{\delta}}\Big\} \geq -\frac{\xi}{2},
    \end{align}
    and 
    \begin{align} \label{l_localisation_3true_eq5}
        H(\widehat{C}_{I},I_{L+1}) - H(C^*_{r,I_{L+1}},I_{L+1}) \geq -\frac{\xi}{2}.
    \end{align}
    By a symmetric argument, similar results as the third case could also be obtained for the fourth case. Therefore, in all of the four cases above, we have that for $\ell \in \{1, L+1\}$, 
    \begin{align*}
        H(\widehat{C}_{I},I_\ell) - H(C^*_{r,I_\ell},I_\ell) \geq -\frac{\xi}{2}.
    \end{align*}

    \noindent \textbf{Step 2.} Using \Cref{l_localisation_3true_eq2}, \Cref{l_localisation_3true_eq3}, \Cref{l_localisation_3true_eq4} and \Cref{l_localisation_3true_eq5}, together with the argument that leads to \Cref{l_localisation_2true_eq10}, it holds with probability at least $1-3n^{-3}$ that
    \begin{align}\notag
        &\sum_{\ell=2}^{L} \frac{\zeta_{\delta}|I_\ell|m}{8}\|C^*_{r,I} - C^*_{r,I_\ell}\|_{\F}^2\\
        \notag
        \leq\;& \Big\{\frac{\zeta_{\delta}|I|mr^2}{2\delta^2} +\frac{\zeta_{\delta}|I|m}{8}\Big\}\|\widehat{C}_I-C^*_{r,I}\|_{\F}^2+(L+1)\xi +C_5(L-1)\Big\{\frac{r^4m\log(n)}{\delta^2\zeta_{\delta}} \vee \frac{r^6\log(n)}{\delta^2 \zeta_{\delta}} \vee \frac{r^{2-2q}nm}{\delta^6\zeta_{\delta}}\Big\}\\
        \label{l_localisation_3true_eq6}
        \leq \;&(L+1)\xi +C_6L\Big\{\frac{r^4m\log(n)}{\delta^2\zeta_{\delta}} \vee \frac{r^6\log(n)}{\delta^2 \zeta_{\delta}} \vee \frac{r^{2-2q}nm}{\delta^6\zeta_{\delta}}\Big\}.
    \end{align}
    \noindent \textbf{Step 3.} Observe that for all $\ell \in \{3, \ldots, L\}$,
    \begin{align*}
        & \frac{\zeta_{\delta}|I_{\ell-1}|m}{8}\|C^*_{r,I} - C^*_{r,I_{\ell-1}}\|_{\F}^2+ \frac{\zeta_{\delta}|I_\ell|m}{8}\|C^*_{r,I} - C^*_{r,I_\ell}\|_{\F}^2\\
        \geq\;& \underset{C \in \mathbb{R}^{r\times r}}{\min} \Big\{ \frac{\zeta_{\delta}|I_{\ell-1}|m}{8}\|C - C^*_{r,I_{\ell-1}}\|_{\F}^2+\frac{\zeta_{\delta}|I_\ell|m}{8}\|C - C^*_{r,I_\ell}\|_{\F}^2\Big\}\\
        \geq\;& \frac{\zeta_{\delta}\min\{|I_{\ell-1}|, |I_\ell|\}m\|C^*_{r,I_{\ell-1}}-C^*_{r,I_\ell}\|_{\F}^2}{16}\\
        \geq\;& \frac{\zeta_{\delta}\Delta m \underset{\ell \in \{3, \ldots, L\}}{\min} \|C^*_{r,I_{\ell-1}}-C^*_{r,I_\ell}\|_{\F}^2}{16}.
    \end{align*}
    where the second inequality follows from the fact that $|I_\ell| \geq \Delta$. Also, note that
    \begin{align} \notag
        2\sum_{\ell=2}^{L} \frac{\zeta_{\delta}|I_\ell|m}{8}\|C^*_{r,I} - C^*_{r,I_\ell}\|_{\F}^2 &\geq \sum_{\ell=3}^{L}\Big\{\frac{\zeta_{\delta}|I_{\ell-1}|m}{8}\|C^*_{r,I} - C^*_{r,I_{\ell-1}}\|_{\F}^2+ \frac{\zeta_{\delta}|I_\ell|m}{8}\|C^*_{r,I} - C^*_{r,I_\ell}\|_{\F}^2\Big\}\\
        \notag
        & \geq \frac{(L-2)\zeta_{\delta}\Delta m \underset{\ell \in \{3, \ldots, L\}}{\min} \|C^*_{r,I_{\ell-1}}-C^*_{r,I_\ell}\|_{\F}^2}{16}\\
        \notag
        & \geq \frac{(L+1)\zeta_{\delta}\Delta m \underset{\ell \in \{3, \ldots, L\}}{\min} \|C^*_{r,I_{\ell-1}}-C^*_{r,I_\ell}\|_{\F}^2}{64}\\
        \label{l_localisation_3true_eq7}
        & \geq \frac{(L+1)\zeta_{\delta}\Delta m \kappa^2}{384},
    \end{align}
    where the third inequality follows from the fact that $L \geq 3$, and the last inequality holds from \Cref{l_localisation_1true_eq21} and the fact that 
    \begin{align*}
        \underset{\ell \in \{3, \ldots, L\}}{\min} \|C^*_{r,I_{\ell-1}}-C^*_{r,I_\ell}\|_{\F}^2 \geq \underset{\ell \in \{3, \ldots, L\}}{\min} \frac{\kappa_{\ell-1}^2}{6} \geq \frac{\kappa^2}{6}.
    \end{align*}
    Combine \Cref{l_localisation_3true_eq6} and \Cref{l_localisation_3true_eq7} together, and pick $\xi = C_{\xi}K\Big\{\frac{r^4m\log^2(n)}{\delta^2\zeta_{\delta}} \vee \frac{r^6\log(n)}{\delta^2 \zeta_{\delta}} \vee \frac{r^{2-2q}nm}{\delta^6\zeta_{\delta}}\Big\}$, we have that
    \begin{align*}
        \Delta \kappa^2 \leq C_7K\Big\{\frac{r^4\log^2(n)}{\delta^2\zeta_{\delta}^2} \vee \frac{r^6\log(n)}{\delta^2 \zeta_{\delta}^2m} \vee \frac{r^{2-2q}n}{\delta^6\zeta_{\delta}^2}\Big\},
    \end{align*}
    which is a contradiction to \Cref{a_localisation_jump}\ref{a_localisation_jump_snr}.
\end{proof}

\begin{lemma} \label{l_localisation_2contain1}
    With probability at least $1-3n^{-3}$, there are no two consecutive induced intervals $I_1 = [s,t) \in \mathcal{\widehat{P}}$ and $I_2 = [t,e) \in \mathcal{\widehat{P}}$ such that $I_1 \cup I_2$ contains no true change point.
\end{lemma}

\begin{proof}
     We will prove by contradiction. Assume that there is no true change point inside the interval $I =I_1\cup I_2$, consequently, it holds that $C^*_{r, I} =C^*_{r, I_1}=C^*_{r, I_2}$. There are four cases to consider: $\min\{|I_1|, |I_2|\} \geq \frac{\xi}{m}$; $|I_1| \geq \frac{\xi}{m} > |I_2|$; $|I_2| \geq \frac{\xi}{m} > |I_1|$; and $\max\{|I_1|, |I_2|\} < \frac{\xi}{m}$. All of the four cases will be discussed individually in \textbf{Step 1-4}.
     
     \noindent \textbf{Step 1.} In the first case, when $\min\{|I_1|, |I_2|\} \geq \frac{\xi}{m}$, we have that $|I| \geq \frac{\xi}{m}$. By \Cref{l_loss_long}, for $\ell \in \{1,2\}$, we have that with probability at least $1-3n^{-3}$ that
     \begin{align*}
         \big|H(C^*_{r,I_\ell},I_\ell)-H(\widehat{C}_{I_\ell},I_\ell)\big| \leq C_1\Big\{\frac{r^4m\log(n)}{\delta^2\zeta_{\delta}} \vee \frac{r^6\log(n)}{\delta^2 \zeta_{\delta}} \vee \frac{r^{2-2q}|I_\ell|m}{\delta^6\zeta_{\delta}}\Big\},
     \end{align*}
     and also
     \begin{align*}
         \big|H(C^*_{r,I},I)-H(\widehat{C}_{I},I)\big| \leq C_2\Big\{\frac{r^4m\log(n)}{\delta^2\zeta_{\delta}} \vee \frac{r^6\log(n)}{\delta^2 \zeta_{\delta}} \vee \frac{r^{2-2q}|I|m}{\delta^6\zeta_{\delta}}\Big\}.
     \end{align*}
     Therefore, we have that
     \begin{align*}
         H(\widehat{C}_{I},I) &\leq H(C^*_{r,I},I)+C_2\Big\{\frac{r^4m\log(n)}{\delta^2\zeta_{\delta}} \vee \frac{r^6\log(n)}{\delta^2 \zeta_{\delta}} \vee \frac{r^{2-2q}|I|m}{\delta^6\zeta_{\delta}}\Big\} \\
         &= H(C^*_{r,I_1},I_1)+H(C^*_{r,I_2},I_2)+C_2\Big\{\frac{r^4m\log(n)}{\delta^2\zeta_{\delta}} \vee \frac{r^6\log(n)}{\delta^2 \zeta_{\delta}} \vee \frac{r^{2-2q}|I|m}{\delta^6\zeta_{\delta}}\Big\}\\
         & \leq H(\widehat{C}_{I_1},I_1) +H(\widehat{C}_{I_2},I_2) +C_3\Big\{\frac{r^4m\log(n)}{\delta^2\zeta_{\delta}} \vee \frac{r^6\log(n)}{\delta^2 \zeta_{\delta}} \vee \frac{r^{2-2q}nm}{\delta^6\zeta_{\delta}}\Big\}\\
         & \leq H(\widehat{C}_{I_1},I_1) +H(\widehat{C}_{I_2},I_2) +\xi,
     \end{align*}
     which is a contradiction.

     \noindent \textbf{Step 2.} In the second case, it still holds that $|I| \geq \frac{\xi}{m}$, therefore by \Cref{l_loss_long}, it holds with probability at least $1-3n^{-3}$ that
     \begin{align*}
         H(\widehat{C}_{I},I) &\leq H(C^*_{r,I},I)+C_2\Big\{\frac{r^4m\log(n)}{\delta^2\zeta_{\delta}} \vee \frac{r^6\log(n)}{\delta^2 \zeta_{\delta}} \vee \frac{r^{2-2q}|I|m}{\delta^6\zeta_{\delta}}\Big\} \\
         &= H(C^*_{r,I_1},I_1)+H(C^*_{r,I_2},I_2)+C_2\Big\{\frac{r^4m\log(n)}{\delta^2\zeta_{\delta}} \vee \frac{r^6\log(n)}{\delta^2 \zeta_{\delta}} \vee \frac{r^{2-2q}|I|m}{\delta^6\zeta_{\delta}}\Big\}\\
         & \leq H(\widehat{C}_{I_1},I_1) + C_4\Big\{\frac{r^4m\log(n)}{\delta^2\zeta_{\delta}} \vee \frac{r^6\log(n)}{\delta^2 \zeta_{\delta}} \vee \frac{r^{2-2q}nm}{\delta^6\zeta_{\delta}}\Big\} + \frac{\xi}{2}\\
         & \leq H(\widehat{C}_{I_1},I_1) +H(\widehat{C}_{I_2},I_2) +\xi,
     \end{align*}
     where the second inequality also follows from \Cref{l_loss_short}, which is a contradiction.

     \noindent \textbf{Step 3.} The proof of the third case is similar to the second case.

     \noindent \textbf{Step 4.} When $|I| \leq \frac{\xi}{m}$, by \Cref{loss_localisation}, we have that $H(\widehat{C}_{I},I)= H(\widehat{C}_{I_1},I_1)= H(\widehat{C}_{I_2},I_2)=0$, therefore, it holds trivially that $H(\widehat{C}_{I},I) \leq  H(\widehat{C}_{I_1},I_1) +H(\widehat{C}_{I_2},I_2) +\xi$. When $|I| \geq \frac{\xi}{m}$, then using \Cref{l_loss_long}, we have that with probability at least $1-3n^{-3}$ that
     \begin{align*}
         H(\widehat{C}_{I},I) &\leq H(C^*_{r,I_1},I_1)+H(C^*_{r,I_2},I_2)+C_2\Big\{\frac{r^4m\log(n)}{\delta^2\zeta_{\delta}} \vee \frac{r^6\log(n)}{\delta^2 \zeta_{\delta}} \vee \frac{r^{2-2q}|I|m}{\delta^6\zeta_{\delta}}\Big\} \\
    & \leq H(\widehat{C}_{I_1},I_1) +H(\widehat{C}_{I_2},I_2) +\xi,
     \end{align*}
     where the second inequality follows from \Cref{l_loss_short}, which is again a contradiction.
\end{proof}

\subsection{Proof of \texorpdfstring{\Cref{prop_consistent_number}}{}}
\begin{proof}
    Denote
    \begin{align*}
        S_n^* = \sum_{k=0}^K H\big(C^*_{r,(\eta_k,\eta_{k+1}]},(\eta_k,\eta_{k+1}]\big).
    \end{align*}
    Given any collection of times $\{t_1,\ldots, t_m\}$, where $t_1 < \cdots < t_m$, and $t_0 = 0$, and $t_{m+1} = n$, let
    \begin{align*}
        S_{n}(t_1, \ldots, t_m) = \sum_{k=0}^m H\big(\widehat{C}_{(t_k,t_{k+1}]},(t_k,t_{k+1}]\big).
    \end{align*}
    Let $\{\widehat{\eta}_{k}\}_{k=1}^{\widehat{K}}$ denote the change points induced by $\widehat{\mathcal{P}}$, then it suffices to justify that
    \begin{align}
        \label{prop_consistent_number_eq1}
        S_n^* +K\xi \geq\;& S_n(\eta_1,\ldots, \eta_K) + K\xi - C_1(K+1)\Big\{\frac{r^4m\log(n)}{\delta^2\zeta_{\delta}} \vee \frac{r^6\log(n)}{\delta^2 \zeta_{\delta}} \vee \frac{r^{2-2q}nm}{\delta^6\zeta_{\delta}}\Big\}\\ \label{prop_consistent_number_eq2}
        \geq \;& S_{n}(\widehat{\eta}_1, \ldots, \widehat{\eta}_{\widehat{K}})+\widehat{K}\xi - C_1(K+1)\Big\{\frac{r^4m\log(n)}{\delta^2\zeta_{\delta}} \vee \frac{r^6\log(n)}{\delta^2 \zeta_{\delta}} \vee \frac{r^{2-2q}nm}{\delta^6\zeta_{\delta}}\Big\}\\ \notag
        \geq\;& S_n\big(\mathrm{sort}(\eta_1, \ldots, \eta_K,\widehat{\eta}_1, \ldots, \widehat{\eta}_{\widehat{K}})\big)+ \widehat{K}\xi \\ \label{prop_consistent_number_eq3}
        &- C_2(K+\widehat{K}+1)\Big\{\frac{r^4m\log(n)}{\delta^2\zeta_{\delta}} \vee \frac{r^6\log(n)}{\delta^2 \zeta_{\delta}} \vee \frac{r^{2-2q}nm}{\delta^6\zeta_{\delta}}\Big\},
    \end{align}
    and that
    \begin{align}
        \label{prop_consistent_number_eq4}
        S_n^* - S_n(\mathrm{sort}(\eta_1,\ldots, \eta_K,\widehat{\eta}_1, \ldots, \widehat{\eta}_{\widehat{K}})) \leq C_3(K+\widehat{K}+1)\Big\{\frac{r^4m\log(n)}{\delta^2\zeta_{\delta}} \vee \frac{r^6\log(n)}{\delta^2 \zeta_{\delta}} \vee \frac{r^{2-2q}nm}{\delta^6\zeta_{\delta}}\Big\}.
    \end{align}
    If both sets of the equations are true, then it must hold that $\widehat{K}=K$, as otherwise if $\widehat{K} \geq K+1$, then
    \begin{align}
        \notag
        &C_3(4K+1)\Big\{\frac{r^4m\log(n)}{\delta^2\zeta_{\delta}} \vee \frac{r^6\log(n)}{\delta^2 \zeta_{\delta}} \vee \frac{r^{2-2q}nm}{\delta^6\zeta_{\delta}}\Big\}\\
        \notag
        \geq\;& C_3(K+\widehat{K}+1)\Big\{\frac{r^4m\log(n)}{\delta^2\zeta_{\delta}} \vee \frac{r^6\log(n)}{\delta^2 \zeta_{\delta}} \vee \frac{r^{2-2q}nm}{\delta^6\zeta_{\delta}}\Big\}\\
        \notag
        \geq\;& S_n^* - S_n(\mathrm{sort}(\eta_1,\ldots, \eta_K,\widehat{\eta}_1, \ldots, \widehat{\eta}_{\widehat{K}}))\\
        \notag
        \geq\;& -C_2(K+\widehat{K}+1)\Big\{\frac{r^4m\log(n)}{\delta^2\zeta_{\delta}} \vee \frac{r^6\log(n)}{\delta^2 \zeta_{\delta}} \vee \frac{r^{2-2q}nm}{\delta^6\zeta_{\delta}}\Big\} + (\widehat{K}-K)\xi\\
        \label{prop_consistent_number_eq5}
        \geq\;& -C_2(4K+1)\Big\{\frac{r^4m\log(n)}{\delta^2\zeta_{\delta}} \vee \frac{r^6\log(n)}{\delta^2 \zeta_{\delta}} \vee \frac{r^{2-2q}nm}{\delta^6\zeta_{\delta}}\Big\} + \xi,
    \end{align}
    where the first inequality follows from the fact that $|\widehat{\mathcal{P}}|-1= \widehat{K} \leq 3K$, the second inequality follows from Equation \eqref{prop_consistent_number_eq4}, the third inequality follows from Equation \eqref{prop_consistent_number_eq3} and the last inequality holds since $\widehat{K} \geq K+1$. However, if Equation \eqref{prop_consistent_number_eq5} holds, this would imply
    \begin{align*}
        \xi \leq C_4(8K+2)\Big\{\frac{r^4m\log(n)}{\delta^2\zeta_{\delta}} \vee \frac{r^6\log(n)}{\delta^2 \zeta_{\delta}} \vee \frac{r^{2-2q}nm}{\delta^6\zeta_{\delta}}\Big\} \leq C_5 K\Big\{\frac{r^4m\log^2(n)}{\delta^2\zeta_{\delta}} \vee \frac{r^6\log(n)}{\delta^2 \zeta_{\delta}} \vee \frac{r^{2-2q}nm}{\delta^6\zeta_{\delta}}\Big\},
    \end{align*}
    which is a contradiction to the choice of the tuning parameter $\xi$. To show the above two sets of equations hold, note that Equation \eqref{prop_consistent_number_eq1} is implied by the fact that $\Delta \gtrsim \frac{\xi}{m}$ and by \Cref{l_loss_long}. Therefore, it holds with probability at least $1-3n^{-3}$ that
    \begin{align}
        \label{prop_consistent_number_eq6}
        |S_n^*- S_n(\eta_1,\ldots, \eta_K)| \leq C_1(K+1)\Big\{\frac{r^4m\log(n)}{\delta^2\zeta_{\delta}} \vee \frac{r^6\log(n)}{\delta^2 \zeta_{\delta}} \vee \frac{r^{2-2q}nm}{\delta^6\zeta_{\delta}}\Big\}.
    \end{align}
    Equation \eqref{prop_consistent_number_eq2} holds from the fact that $S_{n}(\widehat{\eta}_1, \ldots, \widehat{\eta}_{\widehat{K}})+\widehat{K}\xi$ is the minimal value of \Cref{loss_l0}. To show \Cref{prop_consistent_number_eq3} holds, for every $I \in (s,e] \in \widehat{\mathcal{P}}$, denote
    \begin{align*}
        I = (s,\eta_{\ell+1}] \cup \cdots \cup (\eta_{\ell+q},e] = I_1 \cup \cdots \cup I_{q+1},
    \end{align*}
    where $\{\eta_{\ell+m}\}_{m=1}^q = I \cap \{\eta_k\}_{k=1}^K$, then in the rest of the proof, it suffices to show that with probability at least $1-3n^{-3}$ that
    \begin{align}
        \label{prop_consistent_number_eq7}
        H(\widehat{C}_{I}, I) \geq \sum_{\ell=1}^{q+1} H(\widehat{C}_{I_\ell}, I_\ell)-C_6(q+1)\Big\{\frac{r^4m\log(n)}{\delta^2\zeta_{\delta}} \vee \frac{r^6\log(n)}{\delta^2 \zeta_{\delta}} \vee \frac{r^{2-2q}nm}{\delta^6\zeta_{\delta}}\Big\}.
    \end{align}
    Note that, if $|I| \leq \frac{\xi}{m}$, then $H(\widehat{C}_{I}, I) =H(\widehat{C}_{I_\ell}, I_\ell)=0$ for all $\ell \in \{1, \ldots, q+1\}$, and Equation \eqref{prop_consistent_number_eq7} trivially holds. Therefore, it suffices to assume that $|I| \geq \frac{\xi}{m}$. For any $\ell \in \{1, \ldots, q+1\}$, there are two cases to consider. 

    \noindent \textbf{Case 1.} When $|I_\ell| \geq \frac{\xi}{m}$, following a similar argument as the one leading to \Cref{l_localisation_3true_eq4}, it holds with probability at least $1-3n^{-3}$ that
    \begin{align*}
        H(\widehat{C}_I, I_\ell)- H(C^*_{r, I_\ell}, I_\ell) =\;& \sum_{t\in I_\ell}\sum_{(j,k) \in \mathcal{O}} \big\{\Phi_r^\top(X_{t,j})(\widehat{C}_I-C^*_{r,I_\ell})\Phi_r(X_{t,k})\big\}^2\\
        & - 2\sum_{t\in I_\ell} \sum_{(j,k) \in \mathcal{O}}\big\{Y_{t,j}Y_{t,k}-\Phi_r^\top(X_{t,j})C^*_{r,I_\ell}\Phi_r(X_{t,k})\big\}\big\{\Phi_r^\top(X_{t,j})(\widehat{C}_I-C^*_{r,I_\ell})\Phi_r(X_{t,k})\big\}\\
        \geq \;& \frac{\zeta_{\delta}|I_\ell|m}{8}\|C^*_{r,I} - C^*_{r,I_\ell}\|_{\F}^2 - C_7\Big\{\frac{r^4m\log(n)}{\delta^2\zeta_{\delta}} \vee \frac{r^6\log(n)}{\delta^2 \zeta_{\delta}} \vee \frac{r^{2-2q}nm}{\delta^6\zeta_{\delta}}\Big\}\\
        \geq\;& - C_7\Big\{\frac{r^4m\log(n)}{\delta^2\zeta_{\delta}} \vee \frac{r^6\log(n)}{\delta^2 \zeta_{\delta}} \vee \frac{r^{2-2q}nm}{\delta^6\zeta_{\delta}}\Big\},
    \end{align*}
    where the last inequality holds since $\frac{\zeta_{\delta}|I_\ell|m}{8}\|C^*_{r, I} - C^*_{r, I_\ell}\|_{\F}^2 >0$. Hence, we have that
    \begin{align*}
        H(\widehat{C}_{I_\ell}, I_\ell) &\leq H(C^*_{r, I_\ell}, I_\ell) +C_8\Big\{\frac{r^4m\log(n)}{\delta^2\zeta_{\delta}} \vee \frac{r^6\log(n)}{\delta^2 \zeta_{\delta}} \vee \frac{r^{2-2q}nm}{\delta^6\zeta_{\delta}}\Big\}\\
        & \leq H(\widehat{C}_I, I_\ell) + C_9\Big\{\frac{r^4m\log(n)}{\delta^2\zeta_{\delta}} \vee \frac{r^6\log(n)}{\delta^2 \zeta_{\delta}} \vee \frac{r^{2-2q}nm}{\delta^6\zeta_{\delta}}\Big\},
    \end{align*}
    where the first inequality follows from \Cref{l_loss_long}.

    \noindent \textbf{Case 2.} When$|I_\ell| \lesssim \frac{\xi}{m}$, by \Cref{loss_localisation}, we have that 
    \begin{align*}
        H(\widehat{C}_{I_\ell}, I_\ell) = 0 \leq  H(\widehat{C}_I, I_\ell)+ C_{10}\Big\{\frac{r^4m\log(n)}{\delta^2\zeta_{\delta}} \vee \frac{r^6\log(n)}{\delta^2 \zeta_{\delta}} \vee \frac{r^{2-2q}nm}{\delta^6\zeta_{\delta}}\Big\}.
    \end{align*}
    Therefore, in both cases, it holds that
    \begin{align*}
        H(\widehat{C}_{I}, I) = \sum_{\ell=1}^{q+1} H(\widehat{C}_I, I_\ell)\geq \sum_{\ell=1}^{q+1} H(\widehat{C}_{I_\ell}, I_\ell)-C_{11}(q+1)\Big\{\frac{r^4m\log(n)}{\delta^2\zeta_{\delta}} \vee \frac{r^6\log(n)}{\delta^2 \zeta_{\delta}} \vee \frac{r^{2-2q}nm}{\delta^6\zeta_{\delta}}\Big\},
    \end{align*}
    which is exactly Equation \eqref{prop_consistent_number_eq7}. Finally, to show Equation \eqref{prop_consistent_number_eq4}, note that 
    \begin{align*}
        \big| S_n^* - S_n\big(\mathrm{sort}(\eta_1,\ldots, \eta_K,\widehat{\eta}_1, \ldots, \widehat{\eta}_{\widehat{K}})\big)|\leq\;& |S_n^*- S_n\big(\eta_1,\ldots, \eta_K\big)\big| \\
        &+ \big|S_n(\eta_1,\ldots, \eta_K) -S_n(\mathrm{sort}(\eta_1,\ldots, \eta_K,\widehat{\eta}_1, \ldots, \widehat{\eta}_{\widehat{K}}))\big|\\
        \leq \;& C_{12}(K+1)\Big\{\frac{r^4m\log(n)}{\delta^2\zeta_{\delta}} \vee \frac{r^6\log(n)}{\delta^2 \zeta_{\delta}} \vee \frac{r^{2-2q}nm}{\delta^6\zeta_{\delta}}\Big\} \\
        &+ \big|S_n(\eta_1,\ldots, \eta_K) -S_n\big(\mathrm{sort}(\eta_1,\ldots, \eta_K,\widehat{\eta}_1, \ldots, \widehat{\eta}_{\widehat{K}})\big)\big|,
    \end{align*}
    where the second inequality follows from Equation \eqref{prop_consistent_number_eq6}, and it suffices to show that
    \begin{align*}
        \Big|S_n(\eta_1,\ldots, \eta_K) -S_n\big(\mathrm{sort}(\eta_1,\ldots, \eta_K,\widehat{\eta}_1, \ldots, \widehat{\eta}_{\widehat{K}})\big)\Big| \leq C_{13}(K+\widehat{K})\Big\{\frac{r^4m\log(n)}{\delta^2\zeta_{\delta}} \vee \frac{r^6\log(n)}{\delta^2 \zeta_{\delta}} \vee \frac{r^{2-2q}nm}{\delta^6\zeta_{\delta}}\Big\},
    \end{align*}
    where the analysis follows from a similar argument as above.
\end{proof}

\subsection{Additional results}
\begin{lemma} \label{l_loss_long}
    Under the setup in \Cref{t_localisation}, let $I = [s,e) \subseteq [n]$ be any integer interval such that $|I| \geq \xi/m$. If $I$ contains no true change point, then it holds that 
    \begin{align*}
        \mathbb{P}\Bigg\{\big|H(C^*_{r,I},I)-H(\widehat{C}_I,I)\big| \leq C\Big\{\frac{r^4m\log(n)}{\delta^2\zeta_{\delta}} \vee \frac{r^6\log(n)}{\delta^2 \zeta_{\delta}} \vee \frac{r^{2-2q}|I|m}{\delta^6\zeta_{\delta}}\Big\}\Bigg\} \geq 1-3n^{-3},
    \end{align*}
    where $C^*_{r,I} = \frac{1}{|I|}\sum_{t \in I}C^*_{r,t}$, and $C >0$ is an absolute constant.
\end{lemma}

\begin{proof}
    Denote $\Delta_I = \widehat{C}_I - C^*_{r,I}$, and note that 
    \begin{align} \notag
        &\big|H(C^*_{r,I},I)-H(\widehat{C}_I,I)\big|\\
        \notag
        =& \Big|\sum_{t=s+1}^e\sum_{(j,k) \in \mathcal{O}} \big\{Y_{t,j}Y_{t,k} - \Phi_r^\top(X_{t,j}) C^*_{r,I}\Phi_r(X_{t,k})\big\}^2 - \sum_{t=s+1}^e\sum_{(j,k) \in \mathcal{O}} \big\{Y_{t,j}Y_{t,k} - \Phi_r^\top(X_{t,j}) \widehat{C}_I\Phi_r(X_{t,k})\big\}^2\Big|\\
        \label{l_loss_long_eq1}
        \leq &\Big|\sum_{t=s+1}^e \sum_{(j,k) \in \mathcal{O}} \big\{\Phi_r^\top(X_{t,j})\Delta_I\Phi_r(X_{t,k})\big\}^2\Big|\\
        \label{l_loss_long_eq2}
        &+2\Big|\sum_{t=s+1}^e \sum_{(j,k) \in \mathcal{O}}\big\{Y_{t,j}Y_{t,k}-\Phi_r^\top(X_{t,j})C^*_r\Phi_r(X_{t,k})\big\}\big\{\Phi_r^\top(X_{t,j})\Delta_I\Phi_r(X_{t,k})\big\}\Big|.
    \end{align}
    Next, we find upper bounds on \Cref{l_loss_long_eq1} and \Cref{l_loss_long_eq2} in \textbf{Step 1} and \textbf{Step 2} respectively. 

    \noindent \textbf{Step 1.} Denote 
    \begin{align*}
        T= \sum_{t=s+1}^e \sum_{(j,k) \in \mathcal{O}} \Psi(X_{t,j},X_{t,k})\Psi^T(X_{t,j},X_{t,k}),
    \end{align*}
    where $\Psi(X_{t,j}, X_{t,k}) = \Phi_r(X_{t,j}) \odot \Phi_r(X_{t,k})\in \mathbb{R}^{r^2}$. 
    
    \noindent Therefore,  we have that for \Cref{l_loss_long_eq1},
    \begin{align}
        \notag
        &\Big|\sum_{t=s+1}^e \sum_{(j,k) \in \mathcal{O}} \big\{\Phi_r^\top(X_{t,j})\Delta_I\Phi_r(X_{t,k})\big\}^2\Big| = \Big|\mathrm{vec(\Delta_I)^{\top}}T\;\mathrm{vec(\Delta_I)}\Big|\\
        \notag
        \leq\; &  \Big|\mathrm{vec}(\Delta_I)^{\top}\mathbb{E}[T|a]\mathrm{vec}(\Delta_I)\Big| + \Big|\mathrm{vec}(\Delta_I)^{\top}(T-\mathbb{E}[T|a])\mathrm{vec}(\Delta_I)\Big|\\
        \notag
        \leq \; & |I|m \; \Big|\mathrm{vec}(\Delta_I)^{\top}\mathbb{E}\big[\Psi(X_{s+1,1},X_{s+1,2})\Psi^T(X_{s+1,1},X_{s+1,2})\big|a\big]\mathrm{vec}(\Delta_I)\Big|+C_1r^{3}\sqrt{|I|m\log(n)}\|\Delta_I\|_{\F}^2\\
        \label{l_loss_long_eq3}
        \leq \;& |I|m\;\Lambda_{\max}\Big(\mathbb{E}\big[\Psi(X_{s+1,1},X_{s+1,2})\Psi^T(X_{s+1,1},X_{s+1,2})\big|a\big]\Big)\|\Delta_{I}\|_{\F}^2+C_1r^{3}\sqrt{|I|m\log(n)}\|\Delta_I\|_{\F}^2,
    \end{align}
    where $\Lambda_{\max}(\cdot)$ stands for the largest eigenvalue, and the second inequality follows from \Cref{l_op4}. To find an upper bound on $\Lambda_{\max}\big(\mathbb{E}\big[\Psi(X_{s+1,1},X_{s+1,2})\Psi^T(X_{s+1,1},X_{s+1,2})\big|a\big]\big)$, note that
    \begin{align}
        \notag&\Lambda_{\max}\Big(\mathbb{E}\big[\Psi(X_{s+1,1},X_{s+1,2})\Psi^T(X_{s+1,1},X_{s+1,2})\big|a\big]\Big)
        \leq \mathrm{Tr}\Big(\mathbb{E}\big[\Psi(X_{s+1,1},X_{s+1,2})\Psi^T(X_{s+1,1},X_{s+1,2})\big|a\big]\Big)\\
        \notag
        \leq \; & \sum_{i=1}^{r^2} \mathbb{E}\big[\phi^2_{\lceil\frac{i}{r}\rceil}(X_{s+1,1})\phi^2_{i-(\lceil\frac{i}{r}\rceil-1)r}(X_{s+1,2})\big|a\big]= \sum_{i=1}^{r^2}\mathbb{E}\big[\phi^2_{\lceil\frac{i}{r}\rceil}(X_{s+1,1})\big|a\big]\;\mathbb{E}\big[\phi^2_{i-(\lceil\frac{i}{r}\rceil-1)r}(X_{s+1,2})\big|a\big]\\
        \label{l_loss_long_eq4}
        \leq \; &  \sum_{i=1}^{r^2} \frac{1}{\delta^2} \int_{0}^1  \phi^2_{\lceil\frac{i}{r}\rceil}(s)\,\mathrm{d}s \int_{0}^1 \phi^2_{i-(\lceil\frac{i}{r}\rceil-1)r}(l) \,\mathrm{d}l = \sum_{i=1}^{r^2}\frac{1}{\delta^2} = \frac{r^2}{\delta^2},
    \end{align}
    where the first inequality follows from the fact that $\mathbb{E}\big[\Psi(X_{s+1,1}, X_{s+1,2})\Psi^T(X_{s+1,1}, X_{s+1,2})\big|a\big]$ is positive semi-definite almost surely, and the first equality holds since $X_{s+1,1}$ and $X_{s+1,2}$ are conditionally independent given $a$. Plugging \Cref{l_loss_long_eq4} into \Cref{l_loss_long_eq3}, it holds that
    \begin{align} \notag
        & \Big|\sum_{t=s+1}^e \sum_{(j,k) \in \mathcal{O}} \big\{\Phi_r^\top(X_{t,j})\Delta_I\Phi_r(X_{t,k})\big\}^2\Big| \leq \frac{|I|mr^2}{\delta^2}\|\Delta_{I}\|_{\F}^2+C_1r^{3}\sqrt{|I|m\log(n)}\|\Delta_I\|_{\F}^2\\
        \notag
        \leq\;& C_2\Big\{\frac{|I|mr^2}{\delta^2} \vee r^{3}\sqrt{|I|m\log(n)}\Big\} \cdot \Big\{\frac{r}{\zeta_{\delta}}\sqrt{\frac{\log(n)}{|I|}} \vee \frac{r^{2}}{\zeta_{\delta}}\sqrt{\frac{\log(n)}{|I|m}} \vee \frac{r^{-q}}{\delta^2\zeta_{\delta}}\Big\}^2\\
        \notag
        \leq\;& \frac{C_3\zeta_{\delta}|I|mr^2}{\delta^2}\cdot \Big\{\frac{r^2\log(n)}{\zeta_{\delta}^2|I|} \vee \frac{r^{4}\log(n)}{\zeta_{\delta}^2|I|m} \vee \frac{r^{-2q}}{\delta^4\zeta_{\delta}^2}\Big\}\\
        \label{l_loss_long_eq5}
        \leq\;&  C_4\Big\{\frac{r^4m\log(n)}{\delta^2\zeta_{\delta}} \vee \frac{r^6\log(n)}{\delta^2 \zeta_{\delta}} \vee \frac{r^{2-2q}|I|m}{\delta^6\zeta_{\delta}}\Big\},
    \end{align}
    where the third inequality follows from the fact that $\delta <1$, and the assumption that with the choice of the $r$ given in \Cref{t_localisation}, by \Cref{r_choice_r}, we have that
    \begin{align*}
        \frac{r^{3}\sqrt{\log(n)}}{\sqrt{|I|m}} \leq \frac{r^{3}\sqrt{\log(n)}}{\sqrt{\xi}} \leq \frac{\zeta_{\delta}}{2},
    \end{align*}
    hence
    \begin{align*}
        r^{3}\sqrt{|I|m\log(n)} &= r^2\cdot r\sqrt{|I|m\log(n)} \leq r^2 \Big(\frac{\zeta_{\delta}^2|I|m}{4\log(n)}\Big)^{1/6} (|I|m\log(n))^{1/2}\\
        & \leq r^2 \Big(\frac{\zeta_{\delta}^2|I|m}{4\log(n)}\Big)^{1/2} (|I|m\log(n))^{1/2}=\frac{\zeta_{\delta}|I|mr^2}{2} \leq \frac{\zeta_{\delta}|I|mr^2}{2\delta^2}.
    \end{align*}

    \noindent \textbf{Step 2.} We first find an upper bound on the term 
    \begin{align*}
        \Bigg\|\sum_{t \in I} \sum_{(j,k) \in \mathcal{O}} \mathbb{E}\Big[\big\{\Sigma_I^*(X_{t,j},X_{t,k})-\Phi_r^\top(X_{t,j})C^*_{r,I}\Phi_r(X_{t,k})\big\}\Phi_r^\top(X_{t,j})\odot \Phi^{\top}_r(X_{t,k})\Big|a\Big]\Bigg\|_{\op}.
    \end{align*}
    Note that for any $h, l \in [r]$, $t \in I$ and $(j,k)\in \mathcal{O}$,
    \begin{align}
        \notag
        &\mathbb{E}\Big[\big\{\Sigma_I^*(X_{t,j},X_{t,k})-\Phi_r^\top(X_{t,j})C^*_{r,I}\Phi_r(X_{t,k})\big\}\phi_h(X_{t,j})\phi_l(X_{t,k})\Big|a\Big]\\
        \notag
        \leq\;& \Big\{\mathbb{E}\Big[\big\{\Sigma_I^*(X_{t,j},X_{t,k})-\Phi_r^\top(X_{t,j})C^*_{r,I}\Phi_r(X_{t,k})\big\}^2\Big|a\Big]\Big\}^{\frac{1}{2}}\Big\{\mathbb{E}\Big[\phi_h^2(X_{t,j})\phi_l^2(X_{t,k})\Big|a\Big]\Big\}^{\frac{1}{2}}\\
        \label{l_loss_long_eq6}
        =\;& \Big\{\mathbb{E}\Big[\big\{\Sigma_I^*(X_{t,j},X_{t,k})-\Phi_r^\top(X_{t,j})C^*_{r,I}\Phi_r(X_{t,k})\big\}^2\Big|a_t\Big]\Big\}^{\frac{1}{2}}\Big\{\mathbb{E}\Big[\phi_h^2(X_{t,j})\Big|a\Big]\Big\}^{\frac{1}{2}}\Big\{\mathbb{E}\Big[\phi_l^2(X_{t,k})\Big|a\Big]\Big\}^{\frac{1}{2}},
    \end{align}
    where the first inequality follows from the Cauchy-Schwarz inequality, and the equality follows from the fact that $X_{t,j}$ and $X_{t,k}$ are conditionally independent given $a$. By the orthonormal property of the function $\phi_h(\cdot)$ and $\phi_l(\cdot)$, we have that 
    \begin{align} \notag
        \Big\{\mathbb{E}\Big[\big\{\Sigma_I^*(X_{t,j},X_{t,k})-\Phi_r^\top(X_{t,j})C^*_{r,I}\Phi_r(X_{t,k})\big\}^2\Big|a\Big]\Big\}^{\frac{1}{2}} &= \frac{1}{\delta}\Big[\int_{a_t}^{a_t+\delta} \int_{a_t}^{a_t+\delta}\big\{\Sigma_I^*(s,l)-\Sigma^*_{r,I}(s,l)\big\}^2 \,\mathrm{d}s\,\mathrm{d}l\Big]^{\frac{1}{2}}\\
        \notag
        & \leq \frac{1}{\delta}\Big[\int_0^1 \int_0^1 \big\{\Sigma_I^*(s,l)-\Sigma^*_{r,I}(s,l)\big\}^2 \,\mathrm{d}s\,\mathrm{d}l\Big]^{\frac{1}{2}}\\
        \label{l_loss_long_eq7}
        & \leq \frac{C_5r^{-q}}{\delta},
    \end{align}
    and 
    \begin{align} \label{l_loss_long_eq8}
        \Big\{\mathbb{E}\Big[\phi_h^2(X_{t,j})\Big|a_t\Big]\Big\}^{\frac{1}{2}}
        & = \Big\{\frac{1}{\delta} \int_{a_t}^{a_t+\delta} \phi_h^2(s)\,\mathrm{d}s\Big\}^{\frac{1}{2}} \leq \Big\{\frac{1}{\delta} \int_{0}^{1} \phi_h^2(s)\,\mathrm{d}s\Big\}^{\frac{1}{2}}= \sqrt{\frac{1}{\delta}}.
    \end{align}
    Plugging \Cref{l_loss_long_eq7} and \Cref{l_loss_long_eq8} into \Cref{l_loss_long_eq6}, and using a similar argument as \Cref{l_op1_eq1}, we have that 
    \begin{align} \label{l_loss_long_eq9}
         \Bigg\|\sum_{t \in I} \sum_{(j,k) \in \mathcal{O}} \mathbb{E}\Big[\big\{Y_{t,j}Y_{t,k}-\Sigma^*(X_{t,j},X_{t,k})\big\}\Phi_r^\top(X_{t,j})\odot \Phi^{\top}_r(X_{t,k})\Big|a\Big]\Bigg\|_{\op} &\leq C_5r \underset{1\leq h,l\leq r}{\max} \frac{r^{-q}|I|m}{\delta^2} \leq \frac{C_5r^{1-q}|I|m}{\delta^2}.
    \end{align}
    Hence, using \Cref{l_op1}, \Cref{l_op2} and \Cref{l_loss_long_eq9}, we have that 
    \begin{align} \notag
        &\Big|\sum_{t=s+1}^e \sum_{(j,k) \in \mathcal{O}}\{Y_{t,j}Y_{t,k}-\Phi_r^\top(X_{t,j})C^*_r\Phi_r(X_{t,k})\}\{\Phi_r^\top(X_{t,j})\Delta_I\Phi_r(X_{t,k})\}\Big|\\
        \notag
        \leq\;& C_6\Big\{r\sqrt{|I|m^2\log(n)} \vee r^{2}\sqrt{|I|m\log(n)} \vee \frac{r^{1-q}|I|m}{\delta^2} \Big\}\; \|\Delta_I\|_{\F}\\
        \notag
        \leq\;& C_7\Big\{r\sqrt{|I|m^2\log(n)} \vee r^{2}\sqrt{|I|m\log(n)} \vee \frac{r^{1-q}|I|m}{\delta^2} \Big\}\; \Big\{\frac{r}{\zeta_{\delta}}\sqrt{\frac{\log(n)}{|I|}} \vee \frac{r^{2}}{\zeta_{\delta}}\sqrt{\frac{\log(n)}{|I|m}} \vee \frac{r^{-q}}{\zeta_{\delta}\delta^2}\Big\}\\
        \label{l_loss_long_eq10}
        \leq\;& C_8\Big\{\frac{r^2m\log(n)}{\zeta_{\delta}} \vee \frac{r^4\log(n)}{\zeta_{\delta}} \vee \frac{r^{2-2q}|I|m}{\delta^4\zeta_{\delta}}\Big\}.
    \end{align}

    \noindent \textbf{Step 3.} Substituting the results in \Cref{l_loss_long_eq5} and \Cref{l_loss_long_eq10} into \Cref{l_loss_long_eq1} and \Cref{l_loss_long_eq2} respectively, the result follows. 
\end{proof}

\begin{lemma} \label{l_loss_short}
    Under the setup in \Cref{t_localisation}, let $I = [s,e)\subseteq [n]$ be any integer interval. If $I$ contains no change point, and $|I| \leq \xi/m$, then it holds that 
    \begin{align*}
        \mathbb{P}\Big\{\big|H(C^*_{r,I},I)-H(\widehat{C}_I,I)\big| \leq C\xi \Big\} \geq 1-3n^{-3},
    \end{align*}
    where $C^*_{r,I} = \frac{1}{|I|}\sum_{t \in I}C^*_{r,t}$, and $C >0$ is an absolute constant.
\end{lemma}
\begin{proof}
    In the case when $|I| \leq \frac{\xi}{m}$, we have that $H(\widehat{C}_I, I) = 0$ by the construction of \Cref{loss_localisation}, hence finding an upper bound on the term $|H(C^*_{r, I}, I)-H(\widehat{C}_I, I)|$ is equivalent to finding an upper bound on $|H(C^*_{r,I}, I)|$. Note that
    \begin{align}  \notag
         H(C^*_{r,I}, I) &= \sum_{t\in I} \sum_{(j,k) \in \mathcal{O}} \big\{Y_{t,j}Y_{t,k} - \Phi_r^\top(X_{t,j}) C^*_{r,I} \Phi_r(X_{t,k})\big\}^2\\
         \notag
         & \leq 2 \sum_{t\in I} \sum_{(j,k) \in \mathcal{O}} \big\{Y_{t,j}Y_{t,k} - \Sigma^*_{I}(X_{t,j}, X_{t,k})\big\}^2 + 2\sum_{t\in I} \sum_{(j,k) \in \mathcal{O}} \big\{\Sigma^*_{I}(X_{t,j}, X_{t,k}) - \Phi_r^\top(X_{t,j}) C^*_{r,I} \Phi_r(X_{t,k})\big\}^2\\
         \label{l_loss_short_eq1}
         & = 2(A) + 2(B).
    \end{align}
    In the rest of the proof, we will find upper bounds on $(A)$ and $(B)$ respectively in $\textbf{Step 1}$ and $\textbf{Step 2}$.

    \noindent \textbf{Step 1.} By the model setup in \Cref{model_obs}, we have that 
    \begin{align*}
        (A) = \sum_{t\in I} \sum_{(j,k) \in \mathcal{O}} \big\{f_{t}(X_{t,j})f_{t}(X_{t,k}) + f_{t}(X_{t,j})\varepsilon_{t,k} + f_{t}(X_{t,k})\varepsilon_{t,j} +\varepsilon_{t,j}\varepsilon_{t,k} - \Sigma^*_{I}(X_{t,j}, X_{t,k})\big\}^2.
    \end{align*}
    Following from Lemma 26 in \citet{xu2022change} and \Cref{l_addition_subWeibull}, we have that conditioning on the observation grids $X = \{X_{t,j}\}_{t \in I, j\in [m]}$, $\sum_{(j,k) \in \mathcal{O}}\big\{Y_{t,j}Y_{t,k}-\Sigma^*_{I}(X_{t,j}, X_{t,k})\big\}^2$ are mutually independent across $t \in I$, and each follows a sub-Weibull distribution with parameter $1/2$ with
    \begin{align*}
        \Big\|\sum_{(j,k) \in \mathcal{O}}\big\{Y_{t,j}Y_{t,k}-\Sigma^*_{I}(X_{t,j}, X_{t,k})\big\}^2\Big\|_{\psi_{1/2}} \leq \sum_{(j,k) \in \mathcal{O}}\Big\|\big\{Y_{t,j}Y_{t,k}-\Sigma^*_{I}(X_{t,j}, X_{t,k})\big\}^2\Big\|_{\psi_{1/2}} \leq  mC_{f,\varepsilon},
    \end{align*}
    where the first inequality follows from the triangle inequality, and $C_{f,\varepsilon} >0$ is a constant depending only on the absolute constants $C_f, C_\varepsilon > 0$ given in \Cref{a_model}. Therefore, it holds from \Cref{t_concentration_subWeibull} that for any $\tau > 0$, and an interval $I \subset [1,n]$,
    \begin{align*}
        \mathbb{P}\Big\{\Big|(A)-\mathbb{E}[(A)|X]\Big| \geq \tau \Big|X \Big\} \leq e\exp\Big(-C_1\min\Big\{\frac{\tau^2}{|I|m^2}, \Big(\frac{\tau}{m}\Big)^{1/2}\Big\}\Big).
    \end{align*}
    By an union bound argument on the choices of interval $I$, and pick $\tau = C_2\{m\sqrt{|I|\log(n)}\vee m\log^2(n)\}$, we have that
    \begin{align} \label{l_loss_short_eq2}
        \mathbb{P}\Big\{\Big|(A)-\mathbb{E}[(A)|X]\Big| \geq \tau, \; \text{for any} \; I\Big\} = \mathbb{E}_X\Big[\mathbb{P}\Big\{\Big|(A)-\mathbb{E}[(A)|X]\Big| \geq \tau,  \; \text{for any} \; I \Big|X\Big\}\Big] \leq 3n^{-3}.
    \end{align}
    In addition, note that
    \begin{align} \notag
        &\mathrm{E}\Big[\big\{f_{t}(X_{t,j})f_{t}(X_{t,k}) + f_{t}(X_{t,j})\varepsilon_{t,k} + f_{t}(X_{t,k})\varepsilon_{t,j} +\varepsilon_{t,j}\varepsilon_{t,k} - \Sigma^*_{I}(X_{t,j}, X_{t,k})\big\}^2\Big|X\Big]\\
        \notag
        =\;&\mathrm{Var}\big[f_{t}(X_{t,j})f_{t}(X_{t,k}) + f_{t}(X_{t,j})\varepsilon_{t,k} + f_{t}(X_{t,k})\varepsilon_{t,j} +\varepsilon_{t,j}\varepsilon_{t,k}\big|X\big]\\
        \notag
        = \;&\mathrm{Var}\big[f_{t}(X_{t,j})f_{t}(X_{t,k})\big|X\big] +\mathrm{Var}\big[f_{t}(X_{t,j})\varepsilon_{t,k}\big|X\big]+\mathrm{Var}\big[f_{t}(X_{t,k})\varepsilon_{t,j}\big|X\big]+\mathrm{Var}\big[\varepsilon_{t,j}\varepsilon_{t,k}\big|X\big]\\
        \label{l_loss_short_eq3}
        \leq\;& C_{f,\epsilon},
    \end{align}
    where the first equality follows from the fact that
    \begin{align*}
        \mathrm{E}\big[f_{t}(X_{t,j})f_{t}(X_{t,k}) + f_{t}(X_{t,j})\varepsilon_{t,k} + f_{t}(X_{t,k})\varepsilon_{t,j} +\varepsilon_{t,j}\varepsilon_{t,k}\big|X\big] =  \Sigma^*_{I}(X_{t,j}, X_{t,k}),
    \end{align*}
    the second equality follows since all the covariance terms equal to zero using the independent conditions between $f$ and $\varepsilon$, i.e.
    \begin{align*}
        &\mathrm{Cov}\big(f_{t}(X_{t,j})f_{t}(X_{t,k}),f_{t}(X_{t,j})\varepsilon_{t,k}\big|X\big)\\
        =\;& \mathbb{E}\big[f_{t}^2(X_{t,j})f_{t}(X_{t,k})\varepsilon_{t,k}\big|X\big] - \mathbb{E}\big[f_{t}(X_{t,j})f_{t}(X_{t,k})\big|X\big]\mathbb{E}\big[f_{t}(X_{t,j})\varepsilon_{t,k}\big|X\big]\\
        =\;& \mathbb{E}\big[f_{t}^2(X_{t,j})f_{t}(X_{t,k})\big|X\big]\mathbb{E}\big[\varepsilon_{t,k}\big|X\big] -\mathbb{E}\big[f_{t}(X_{t,j})f_{t}(X_{t,k})\big|X\big] \mathbb{E}\big[f_{t}(X_{t,j})\big|X\big]\mathbb{E}\big[\varepsilon_{t,k}\big|X\big]\\
        =\;& 0,
    \end{align*}
    and the last inequality follows from the Cauchy-Schwarz inequality and the sub-Gaussian properties of $\varepsilon$ and $f_t(x)$ for any $t \in I$ and $x \in [0,1]$ in \Cref{a_model}, i.e.
    \begin{align*}
        \var\big[f_{t}(X_{t,j})f_{t}(X_{t,k})\big|X\big] &= \mathbb{E}\big[f_{t}^2(X_{t,j})f_{t}^2(X_{t,k})\big|X\big] - \mathbb{E}\big[f_{t}(X_{t,j})f_{t}(X_{t,k})\big|X\big]^2\\
        &\leq \mathbb{E}\big[f_{t}^4(X_{t,j})\big|X\big]^{1/2}\mathbb{E}\big[f_{t}^4(X_{t,k})\big|X\big]^{1/2} - \Sigma^*_I(X_{t,j}, X_{t,k})^2 \\
        &\leq C_3C_f,
    \end{align*}
    \begin{align*}
        \var\big[f_{t}(X_{t,k})\varepsilon_{t,j}\big|X\big] &= \mathbb{E}\big[f_{t}^2(X_{t,k})\varepsilon_{t,j}^2\big|X\big] - \mathbb{E}\big[f_{t}(X_{t,k})\varepsilon_{t,j}\big|X\big]^2\\
        & =\mathbb{E}\big[f_{t}^2(X_{t,k})\big|X\big]\mathbb{E}\big[\varepsilon_{t,j}^2\big|X\big]-\mathbb{E}\big[f_{t}(X_{t,k})\big|X\big]^2\mathbb{E}\big[\varepsilon_{t,j}\big|X\big]^2\\
        & \leq C_4C_fC_\varepsilon,
    \end{align*}
    and
    \begin{align*}
        \var\big[\varepsilon_{t,j}\varepsilon_{t,k}\big|X\big] &= \mathbb{E}\big[\varepsilon_{t,j}^2\varepsilon_{t,k}^2\big|X\big] - \mathbb{E}\big[\varepsilon_{t,j}\varepsilon_{t,k}\big|X\big]^2\\
        &\leq C_5C_\varepsilon^2.
    \end{align*}
    Therefore, combining \Cref{l_loss_short_eq2} and \Cref{l_loss_short_eq3}, we have that with the probability of at least $1 -n^{-3}$
    \begin{align} \label{l_loss_short_eq4}
        (A) \leq \big|(A) - \mathbb{E}[(A)|X]\big|+ |\mathbb{E}[(A)|X]| \lesssim m\sqrt{|I|\log(n)} + m\log^2(n)+ |I|m \lesssim \xi,
    \end{align}
    where the last inequality follows from the fact that $|I| \lesssim \frac{\xi}{m}$, hence
    \begin{align*}
        m\sqrt{|I|\log(n)} \lesssim \sqrt{\xi}\sqrt{m\log(n)} \lesssim \xi, \; \mathrm{and} \quad m\log^2(n) \lesssim \xi.
    \end{align*}

    \noindent \textbf{Step 2.} By Cauchy-Schwarz inequality, we have that
    \begin{align*}
        \Big|\Big\{\Sigma^*_{I}(X_{t,j}, X_{t,k}) - \Phi_r^\top(X_{t,j})C^*_{r,I}\Phi_r(X_{t,k})\Big\}^2\Big| & \leq C_6\Big|\sum_{1\leq h',l' \leq r} (C^*_{r,I})_{h'l'} \phi_{h'}(X_{t,j})\phi_{l'}(X_{t,k})\Big|^2 \\
        & \leq C_6\Big\{\sum_{1\leq h',l' \leq r} |(C^*_{r,I})_{h'l'}|^2\Big\} \Big\{\sum_{1\leq h',l'\leq r} \|\phi_{h'}\|^2_{\infty}\|\phi_{l'}\|^2_{\infty}\Big\}\\
       & \leq C_6\|\Sigma^*_I\|_{L^2}^2 r^{4\alpha+2} \leq C_7r^{2},
    \end{align*}
    where the first inequality holds since $\Sigma^*_I$ is bounded, and the last inequality holds since $\Sigma^*_I$ is square-integrable and $\alpha =0$ for Fourier basis. Then by Hoeffding’s inequality for general bounded random variables \citep[Theorem 2.2.6 in][]{vershynin2018high}, conditioning on the set of the starting points on the fragments $a = \{a_t\}_{t \in I}$, we have that for any $\tau > 0$,
    \begin{align*}
        \mathbb{P}\Big\{\Big|(B)-\mathbb{E}[(B)|a]\Big| \geq \tau \Big|a \Big\} \leq \exp\Big(-\frac{C_8\tau^2}{r^4|I|m}\Big).
    \end{align*}
    By a union bound argument on the interval $I$, and pick $\tau = C_9r^2\sqrt{|I|m\log(n)}$, we have that
    \begin{align}\label{l_loss_short_eq5}
        \mathbb{P}\Big\{\Big|(B)-\mathbb{E}[(B)|a]\Big| \geq \tau, \; \text{for any} \; I\Big\} = \mathbb{E}_a\Big[\mathbb{P}\Big\{\Big|(B)-\mathbb{E}[(B)|a]\Big| \geq \tau, \; \text{for any} \; I\Big|a \Big\}\Big] \leq 3n^{-3}.
    \end{align}
    Under \Cref{a_model}\ref{a_model_grid}, we also have that for any $t \in I$, and $(j,k) \in \mathcal{O}$,
    \begin{align} \notag
        \mathbb{E}\Big[\big\{\Sigma^*(X_{t,j},X_{t,k})-\Phi_r^\top(X_{t,j})C^*_{r,I}\Phi_r(X_{t,k})\big\}^2\Big|a\Big] &= \frac{1}{\delta^2}\int_{a_t}^{a_t+\delta} \int_{a_t}^{a_t+\delta}\big\{\Sigma^*_I(s,l)-\Sigma^*_{r,I}(s,l)\big\}^2 \,\mathrm{d}s\,\mathrm{d}l\\
        \notag
        & \leq \frac{1}{\delta^2}\int_0^1 \int_0^1 \big\{\Sigma^*_I(s,l)-\Sigma^*_{r,I}(s,l)\big\}^2 \,\mathrm{d}s\,\mathrm{d}l\\
        \label{l_loss_short_eq6}
        & = \frac{1}{\delta^2}\|\Sigma^*_I-\Sigma^*_{r,I}\|_{L^2}^2 \leq \frac{C_{10}r^{-2q}}{\delta^2}.
    \end{align}
    Hence, \Cref{l_loss_short_eq5} and \Cref{l_loss_short_eq6} entail that, with probability at least $1-3n^{-3}$ that
    \begin{align} \label{l_loss_short_eq7}
        (B) \leq \big|(B) - \mathbb{E}[(B)|a]\big|+ \mathbb{E}[(B)|a] \lesssim r^2\sqrt{|I|m\log(n)} + \frac{|I|mr^{-2q}}{\delta^2} \lesssim \xi,
    \end{align}
    where the last inequality follows the fact that
    \begin{align*}
        r^2\sqrt{|I|m\log(n)} \lesssim \sqrt{\xi}r^2\sqrt{\log(n)}\lesssim \xi, \; \mathrm{and} \quad  \frac{|I|mr^{-2q}}{\delta^2} \lesssim \frac{\xi r^{-2q}}{\delta^2} \lesssim \xi.
    \end{align*}

    \noindent \textbf{Step 3.} Substituting the results in \Cref{l_loss_short_eq4} and \Cref{l_loss_short_eq7} into \Cref{l_loss_short_eq1}, the final result follows accordingly.
\end{proof}

\section{Proof of \texorpdfstring{\Cref{t_inference}}{}}
\label{pf_inference}
\begin{proof}[Proof of \Cref{t_inference}]
    Define the event 
    \begin{align*}
        \mathcal{E} = \Big\{K = \widehat{K},\; \mathrm{and}\;\max_{\ell=1, \dots, K} \kappa_{\ell}^2|\widehat{\eta}_\ell -\eta_\ell| \leq \frac{C_{\varepsilon}K\Tilde{r}^4\log^2(n)}{\delta^2\zeta_{\delta}^2}\Big\}.
    \end{align*}
    By \Cref{c_inf_localisation}, we have that the event $\mathcal{E}$ holds with probability at least $1-3n^{-3}$. In the rest of the proof, all of the analyses are carried out conditioning on the event $\mathcal{E}$. 
    
    \noindent \textbf{Preliminary.} For $\ell \in \{0, \ldots, K\}$, let $I_{\ell} = (\widehat{\eta}_{\ell-1},\widehat{\eta}_{\ell}]$, $\widehat{\eta}_{0} = 0$, and $\widehat{\eta}_{K+1} =n+1$. By \Cref{prop_localisation}, the interval $(\widehat{\eta}_{\ell-1}, \widehat{\eta}_{\ell+1}] = I_{\ell} \cup I_{\ell+1}$ can contain 1, 2, or 3 true change points. The most completed case when $(\widehat{\eta}_{\ell-1}, \widehat{\eta}_{\ell+1}]$ contains 3 true change points is analysed and the other two cases are similar and simpler. Without loss of generality, assume that $\widehat{\eta}_{\ell-1} < \eta_{\ell-1} < \eta_\ell < \widehat{\eta}_{\ell} < \eta_{\ell+1} < \widehat{\eta}_{\ell+1}$. In the rest of the proof, denote 
    \begin{align*}
       C^*_{\Tilde{r},I_\ell} = \frac{1}{|\widehat{\eta}_{\ell} - \widehat{\eta}_{\ell-1}|}\sum_{t = \widehat{\eta}_{\ell-1} +1}^{\widehat{\eta}_{\ell}} C^*_{\Tilde{r},t}, \; \mathrm{and}\quad C^*_{\Tilde{r},\eta_\ell} = C^*_{\Tilde{r},t},\; \mathrm{for}\;  t\in (\eta_{\ell-1}, \eta_\ell],
    \end{align*}
    where $ C^*_{\Tilde{r},t} \in \mathbb{R}^{\Tilde{r} \times \Tilde{r}}$ is the matrix with entries $(C^*_{\Tilde{r},t})_{hl}=\int_{0}^1 \int_{0}^1 \Sigma_t^*(s,t)\phi_h(s)\phi_l(t) \,\mathrm{d}s\,\mathrm{d}t$ for $h,l \in [\Tilde{r}]$. Since $\Sigma^*_t$ is square-integrable by \Cref{a_model}, it holds that for any $\ell \in \{1, \ldots, K\}$, the jump size is bounded by an absolute constant $C_\kappa >0$:
    \begin{align} \label{t_inference_eq1}
        \kappa_\ell = \|\Sigma^*_{\eta_{\ell+1}}-\Sigma^*_{\eta_{\ell}}\|_{L^2} \leq \|\Sigma^*_{\eta_{\ell+1}}\|_{L^2} +\|\Sigma^*_{\eta_{\ell}}\|_{L^2} \leq C_\kappa.
    \end{align}
    By \Cref{c_inf_localisation}, we have that 
    \begin{align*}
        |I_\ell| = \widehat{\eta}_\ell - \widehat{\eta}_{\ell-1} \geq \eta_{\ell+1} - \eta_\ell - C_{\varepsilon} \frac{K\Tilde{r}^4\log^2(n)}{\delta^2\zeta_{\delta}^2}  \Big(\frac{1}{\kappa_{\ell}^2}+\frac{1}{\kappa_{\ell-1}^2}\Big) \geq \frac{\Delta}{2} \gtrsim \frac{\xi}{m},
    \end{align*}
    where the second inequality follows from \Cref{a_snr_strong}\ref{a_snr_strong_1} Therefore following a similar and simpler argument as the one in the proof of \Cref{t_estimation}, we have that under \Cref{a_inference}, for any $\ell \in \{1, \ldots, K\}$
    \begin{align}\label{t_inference_eq2}
        \|\widehat{C}_\ell-C^*_{\Tilde{r},I_\ell}\|_{\F} \leq \frac{C_1\Tilde{r}}{\zeta_{\delta}} \sqrt{\frac{\log(n)}{|I_\ell|}} \leq \frac{C_2\Tilde{r}}{\zeta_{\delta}} \sqrt{\frac{\log(n)}{\Delta}}.
    \end{align}
    Moreover, note we have by construction that
    \begin{align*}
        C^*_{\Tilde{r},I_\ell}-C^*_{\Tilde{r},\eta_{\ell+1}} &= \frac{(\eta_{\ell-1}-\widehat{\eta}_{\ell-1})C^*_{\Tilde{r},\eta_{\ell-1}}}{\widehat{\eta}_{\ell}-\widehat{\eta}_{\ell-1}}+\frac{(\eta_{\ell}-\eta_{\ell-1})C^*_{\Tilde{r},\eta_{\ell}}}{\widehat{\eta}_{\ell}-\widehat{\eta}_{\ell-1}}+\frac{(\widehat{\eta}_{\ell-1}-\eta_{\ell})C^*_{\Tilde{r},\eta_{\ell+1}}}{\widehat{\eta}_{\ell}-\widehat{\eta}_{\ell-1}}\\
        & = \frac{(\eta_{\ell-1}-\widehat{\eta}_{\ell-1})(C^*_{\Tilde{r},\eta_{\ell-1}}-C^*_{\Tilde{r},\eta_{\ell+1}})}{\widehat{\eta}_{\ell}-\widehat{\eta}_{\ell-1}} + \frac{(\eta_{\ell}-\eta_{\ell-1})(C^*_{\Tilde{r},\eta_{\ell}}-C^*_{\Tilde{r},\eta_{\ell+1}})}{\widehat{\eta}_{\ell}-\widehat{\eta}_{\ell-1}},
    \end{align*}
    \begin{align*}
        C^*_{\Tilde{r},\eta_{\ell+1}}-C^*_{\Tilde{r},I_{\ell+1}} &= C^*_{\Tilde{r},\eta_{\ell+1}}-\frac{(\widehat{\eta}_{\ell+1}-\eta_{\ell+1})C^*_{\Tilde{r},\eta_{\ell+2}}}{\widehat{\eta}_{\ell+1}-\widehat{\eta}_{\ell}} - \frac{(\eta_{\ell+1}-\widehat{\eta}_{\ell})C^*_{\Tilde{r},\eta_{\ell+1}}}{\widehat{\eta}_{\ell+1}-\widehat{\eta}_{\ell}}\\
        & = \frac{(\widehat{\eta}_{\ell+1}-\eta_{\ell+1})C^*_{\Tilde{r},\eta_{\ell+1}}-(\widehat{\eta}_{\ell+1}-\eta_{\ell+1})C^*_{\Tilde{r},\eta_{\ell+2}}}{\widehat{\eta}_{\ell+1}-\widehat{\eta}_{\ell}} = \frac{(\widehat{\eta}_{\ell+1}-\eta_{\ell+1})(C^*_{\Tilde{r},\eta_{\ell+1}}-C^*_{\Tilde{r},\eta_{\ell+2}})}{\widehat{\eta}_{\ell+1}-\widehat{\eta}_{\ell}},
    \end{align*}
    and
    \begin{align*}
        C^*_{\Tilde{r},I_\ell}-C^*_{\Tilde{r},\eta_\ell} &= \frac{(\eta_{\ell-1}-\widehat{\eta}_{\ell-1})C^*_{\Tilde{r},\eta_{\ell-1}}}{\widehat{\eta}_\ell - \widehat{\eta}_{\ell-1}} + \frac{(\eta_\ell - \eta_{\ell-1})C^*_{\Tilde{r},\eta_{\ell}}}{\widehat{\eta}_\ell - \widehat{\eta}_{\ell-1}}+ \frac{(\widehat{\eta}_{\ell}-\eta_{\ell})C^*_{\Tilde{r} ,\eta_{\ell+1}}}{\widehat{\eta}_\ell - \widehat{\eta}_{\ell-1}}-C^*_{\Tilde{r},\eta_\ell}\\
        & = \frac{(\eta_{\ell-1}-\widehat{\eta}_{\ell-1})(C^*_{\Tilde{r},\eta_{\ell-1}}-C^*_{\Tilde{r},\eta_{\ell}})}{\widehat{\eta}_\ell - \widehat{\eta}_{\ell-1}} + \frac{(\eta_{\ell}-\widehat{\eta}_{\ell})(C^*_{\Tilde{r},\eta_{\ell}}-C^*_{\Tilde{r},\eta_{\ell+1}})}{\widehat{\eta}_\ell - \widehat{\eta}_{\ell-1}}.
    \end{align*}
    Therefore, taking Frobenius norms on both sides of the three equations above, we have that 
    \begin{align} \label{t_inference_eq3}
        \|C^*_{\Tilde{r},I_\ell}-C^*_{\Tilde{r},\eta_{\ell+1}}\|_{\F} \leq \frac{C_3(\eta_{\ell-1}-\widehat{\eta}_{\ell-1})}{\Delta}(\kappa_{\ell-1}+\kappa_\ell) + C_4\kappa_\ell \leq C_5\kappa_\ell,
    \end{align}
    \begin{align}\label{t_inference_eq4}
        \|C^*_{\Tilde{r},\eta_{\ell+1}}-C^*_{\Tilde{r},I_{\ell+1}}\|_{\F} &= \frac{(\widehat{\eta}_{\ell+1}-\eta_{\ell+1})\kappa_{\ell+1}}{\widehat{\eta}_{\ell+1}-\widehat{\eta}_{\ell}} \leq \frac{C_6(\widehat{\eta}_{\ell+1}-\eta_{\ell+1})\kappa_{\ell+1}^2}{\Delta\kappa_{\ell+1}} \leq C_7\alpha_n^{-1}\frac{\kappa^2}{\kappa_{\ell+1}} \leq C_7\alpha_n^{-1}\kappa_{\ell}, 
    \end{align}
    and
    \begin{align} \label{t_inference_eq5}
        \|C^*_{\Tilde{r},I_\ell}-C^*_{\Tilde{r},\eta_\ell}\|_{\F} \leq \frac{(\eta_{\ell-1}-\widehat{\eta}_{\ell-1})\kappa_{\ell-1}}{\Delta} + \frac{(\eta_{\ell}-\widehat{\eta}_{\ell})\kappa_{\ell}}{\Delta} \leq C_8\alpha_n^{-1}\frac{\kappa^2}{\kappa_{\ell-1}} + C_9\alpha_n^{-1}\frac{\kappa^2}{\kappa_{\ell}}\leq C_{10}\alpha_n^{-1}\kappa_{\ell}.
    \end{align}
    
    \noindent \textbf{Uniform tightness} Denote $d = \Tilde{\eta}_\ell - \eta_\ell$. Without loss of generality, suppose that $d \geq 0$. Since from \Cref{loss_refine}, $\Tilde{\eta}_\ell = \eta_\ell+d$ is the minimizer of $\mathcal{L}_\ell(\eta)$, it holds that $\mathcal{L}_\ell(\eta_\ell+d) -\mathcal{L}_\ell(\eta_\ell) \leq 0$.
    Rearrange and observe that 
    \begin{align*}
        & \mathcal{L}_\ell(\eta_\ell+d) -\mathcal{L}_\ell(\eta_\ell)\\
        =\;& \sum_{t = \eta_\ell+1}^{\eta_\ell+d}\sum_{(j,k) \in \mathcal{O}}\Big\{\big\{Y_{t,j}Y_{t,k} -\Phi_{\Tilde{r}}^\top(X_{t,j})\widehat{C}_{\ell}\Phi_{\Tilde{r}}(X_{t,k})\big\}^2 -\big\{Y_{t,j}Y_{t,k} -\Phi_{\Tilde{r}}^\top(X_{t,j})\widehat{C}_{\ell+1}\Phi_{\Tilde{r}}(X_{t,k})\big\}^2 \Big\}\\
        =\;& \sum_{t = \eta_\ell+1}^{\eta_\ell+d}\sum_{(j,k) \in \mathcal{O}}\Big\{\big\{Y_{t,j}Y_{t,k} -\Phi_{\Tilde{r}}^\top(X_{t,j})\widehat{C}_{\ell}\Phi_{\Tilde{r}}(X_{t,k})\big\}^2 -\big\{Y_{t,j}Y_{t,k} -\Phi_{\Tilde{r}}^\top(X_{t,j})C^*_{\Tilde{r},I_\ell}\Phi_{\Tilde{r}}(X_{t,k})\big\}^2 \Big\}\\
        &- \sum_{t = \eta_\ell+1}^{\eta_\ell+d}\sum_{(j,k) \in \mathcal{O}}\Big\{\big\{Y_{t,j}Y_{t,k} -\Phi_{\Tilde{r}}^\top(X_{t,j})\widehat{C}_{\ell+1}\Phi_{\Tilde{r}}(X_{t,k})\big\}^2 -\big\{Y_{t,j}Y_{t,k} -\Phi_{\Tilde{r}}^\top(X_{t,j})C^*_{\Tilde{r},I_{\ell+1}}\Phi_{\Tilde{r}}(X_{t,k})\big\}^2\Big\}\\
        &+\sum_{t = \eta_\ell+1}^{\eta_\ell+d}\sum_{(j,k) \in \mathcal{O}} \Big\{\big\{Y_{t,j}Y_{t,k} -\Phi_{\Tilde{r}}^\top(X_{t,j})C^*_{\Tilde{r},I_\ell}\Phi_{\Tilde{r}}(X_{t,k})\big\}^2 -\big\{Y_{t,j}Y_{t,k} -\Phi_{\Tilde{r}}^\top(X_{t,j})C^*_{\Tilde{r},\eta_\ell}\Phi_{\Tilde{r}}(X_{t,k})\big\}^2 \Big\}\\
        &- \sum_{t = \eta_\ell+1}^{\eta_\ell+d}\sum_{(j,k) \in \mathcal{O}} \Big\{\big\{Y_{t,j}Y_{t,k} -\Phi_{\Tilde{r}}^\top(X_{t,j})C^*_{\Tilde{r},I_{\ell+1}}\Phi_{\Tilde{r}}(X_{t,k})\big\}^2 -\big\{Y_{t,j}Y_{t,k} -\Phi_{\Tilde{r}}^\top(X_{t,j})C^*_{\Tilde{r},\eta_{\ell+1}}\Phi_{\Tilde{r}}(X_{t,k})\big\}^2 \Big\}\\
        &+\sum_{t = \eta_\ell+1}^{\eta_\ell+d}\sum_{(j,k) \in \mathcal{O}} \Big\{\big\{Y_{t,j}Y_{t,k} -\Phi_r^\top(X_{t,j})C^*_{\Tilde{r},\eta_{\ell}}\Phi_r(X_{t,k})\big\}^2 -\big\{Y_{t,j}Y_{t,k} -\Phi_{\Tilde{r}}^\top(X_{t,j})C^*_{\Tilde{r},\eta_{\ell+1}}\Phi_{\Tilde{r}}(X_{t,k})\big\}^2 \Big\}\\
        =\;& (I) -(II) +(III) -(IV) + (V) \leq 0.
    \end{align*}
    Therefore, we have that
    \begin{align} \label{t_inference_eq6}
        (V) \leq -(I) + (II) - (III) +(IV) \leq |(I)| + |(II)| + |(III)| + |(IV)|.
    \end{align}
    \noindent \textbf{Step 1: Order of magnitude of $(I)$}. We have that
    \begin{align*}
        (I) =\;& \sum_{t = \eta_\ell+1}^{\eta_\ell+d}\sum_{(j,k) \in \mathcal{O}} \big\{\Phi_{\Tilde{r}}^\top(X_{t,j})(\widehat{C}_\ell-C^*_{\Tilde{r},I_\ell})\Phi_{\Tilde{r}}(X_{t,k})\big\}^2\\
        & - 2\sum_{t = \eta_\ell+1}^{\eta_\ell+d} \sum_{(j,k) \in \mathcal{O}}\big\{Y_{t,j}Y_{t,k}-\Phi_{\Tilde{r}}^\top(X_{t,j})C^*_{\Tilde{r},I_\ell}\Phi_{\Tilde{r}}(X_{t,k})\big\}\big\{\Phi_{\Tilde{r}}^\top(X_{t,j})(\widehat{C}_\ell-C^*_{\Tilde{r},I_\ell})\Phi_{\Tilde{r}}(X_{t,k})\big\}\\
        =\;& \sum_{t = \eta_\ell+1}^{\eta_\ell+d}\sum_{(j,k) \in \mathcal{O}} \big\{\Phi_{\Tilde{r}}^\top(X_{t,j})(\widehat{C}_\ell-C^*_{\Tilde{r},I_\ell})\Phi_{\Tilde{r}}(X_{t,k})\big\}^2\\
        &- 2\sum_{t = \eta_\ell+1}^{\eta_\ell+d} \sum_{(j,k) \in \mathcal{O}} \big\{Y_{t,j}Y_{t,k}-\Phi_{\Tilde{r}}^\top(X_{t,j})C^*_{\Tilde{r},\eta_{\ell+1}}\Phi_{\Tilde{r}}(X_{t,k})\big\}\big\{\Phi_{\Tilde{r}}^\top(X_{t,j})(\widehat{C}_\ell-C^*_{\Tilde{r},I_\ell})\Phi_{\Tilde{r}}(X_{t,k})\big\}\\
        &-2\sum_{t = \eta_\ell+1}^{\eta_\ell+d} \sum_{(j,k) \in \mathcal{O}}\big\{\Phi_{\Tilde{r}}^\top(X_{t,j})(C^*_{\Tilde{r},\eta_{\ell+1}}-C^*_{\Tilde{r},I_\ell})\Phi_{\Tilde{r}}(X_{t,k})\big\}\big\{\Phi_{\Tilde{r}}^\top(X_{t,j})(\widehat{C}_\ell-C^*_{\Tilde{r},I_\ell})\Phi_{\Tilde{r}}(X_{t,k})\big\}\\
        =\;& (I_1) - 2(I_2) - 2(I_3) \leq |(I_1)| + 2|(I_2)| + 2|(I_3)|.
    \end{align*}
    Note that for $(I_1)$, we have that
    \begin{align*}
        |(I_1)| & = \text{vec}(\widehat{C}_\ell-C^*_{\Tilde{r},I_\ell})^\top \sum_{t = \eta_\ell+1}^{\eta_\ell+d}\sum_{(j,k) \in \mathcal{O}}\Psi(X_{t,j}, X_{t,k})\Psi^\top(X_{t,j},X_{t,k}) \text{vec}(\widehat{C}_\ell-C^*_{\Tilde{r},I_\ell})\\
        &\leq  \frac{dm\Tilde{r}^2}{\delta^2}\|\widehat{C}_\ell-C^*_{\Tilde{r},I_\ell}\|_{\F}^2+C_{11}\Tilde{r}^{3}\sqrt{dm\log(n)}\|\widehat{C}_\ell-C^*_{\Tilde{r},I_\ell}\|_{\F}^2\\
        & \leq \frac{C_{12}\Tilde{r}^2\log(n)}{\zeta_{\delta}^2\Delta} \Big\{\frac{dm\Tilde{r}^2}{\delta^2} + \Tilde{r}^{3}\sqrt{dm\log(n)}\Big\} \\
        & = \frac{C_{12}\Tilde{r}^4\log^2(n)}{\delta^2\zeta_{\delta}^2\Delta} \Big(dm\log^{-1}(n) + \sqrt{d}\Tilde{r}\sqrt{m}\log^{-1/2}(n)\delta^2\Big) \\
        & = O_p\Big(\alpha_n^{-1}K^{-1}\Big\{d\kappa_\ell^2 m\log^{-1}(n)+\sqrt{d\kappa_\ell^2} \Tilde{r}\sqrt{m}\log^{-1/2}(n)\delta^2\Big\}\Big),
    \end{align*}
    where the first inequality follows from a similar argument as the one leading to \Cref{l_loss_long_eq3} and \Cref{l_loss_long_eq4}, the second inequality follows from \Cref{t_inference_eq2}, and the last equality follows from \Cref{a_snr_strong}\ref{a_snr_strong_1}~and \Cref{t_inference_eq1}. For $(I_2)$, we have that
    \begin{align*}
        |(I_2)| &\leq  \Big\|\sum_{t = \eta_\ell+1}^{\eta_\ell+d} \sum_{(j,k) \in \mathcal{O}} \big\{Y_{t,j}Y_{t,k}-\Phi_{\Tilde{r}}^\top(X_{t,j})C^*_{\Tilde{r},\eta_{\ell+1}}\Phi_{\Tilde{r}}(X_{t,k})\big\}\Phi_{\Tilde{r}}^\top(X_{t,j}) \odot \Phi_{\Tilde{r}}^\top(X_{t,k})\Big\|_{\op}\Big\|\widehat{C}_\ell-C^*_{\Tilde{r},I_\ell}\Big\|_{\F}\\
        & \leq  \frac{C_{13}\Tilde{r}}{\zeta_{\delta}} \sqrt{\frac{\log(n)}{\Delta}} \Big\{\sqrt{d}\Tilde{r}m\sqrt{\log(n)} \vee \Tilde{r}m\log(n)\Big\} = \frac{C_{13}\delta}{\Tilde{r}}\Big\{\frac{\Tilde{r}^2\log(n)}{\delta\zeta_{\delta}\sqrt{\Delta}}\Big\}\Big\{\sqrt{d}\Tilde{r}m  \vee \Tilde{r}m\sqrt{\log(n)}\Big\}\\
        & = O_p\Big(\alpha_n^{-1/2}K^{-1/2}\delta\Big\{\sqrt{d\kappa_\ell^2}m + m\sqrt{\log(n)}\Big\}\Big),
    \end{align*}
    where the second inequality follows from \Cref{l_op1_eq5} in \Cref{l_op1}, and the last equality follows from \Cref{a_snr_strong}\ref{a_snr_strong_1} To find an upper bound on $(I_3)$, denote $\Lambda_{tjk} = \mathbb{E}\big[\Psi(X_{t,j}, X_{t,k})\Psi^\top(X_{t,j}, X_{t,k})\big|a\big]$, then we have that 
    \begin{align*}
        |(I_3)| =\;& \Big|\text{vec}(C^*_{\Tilde{r},\eta_{\ell+1}}-C^*_{\Tilde{r},I_\ell})^\top\sum_{t = \eta_\ell+1}^{\eta_\ell+d} \sum_{(j,k) \in \mathcal{O}}\Psi(X_{t,j}, X_{t,k})\Psi^\top(X_{t,j}, X_{t,k}) \text{vec}(\widehat{C}_\ell-C^*_{\Tilde{r},I_\ell})\Big|\\
        \leq\;& \Big|\text{vec}(C^*_{\Tilde{r},\eta_{\ell+1}}-C^*_{\Tilde{r},I_\ell})^\top\Big\{\sum_{t = \eta_\ell+1}^{\eta_\ell+d} \sum_{(j,k) \in \mathcal{O}}\Psi(X_{t,j}, X_{t,k})\Psi^\top(X_{t,j}, X_{t,k}) - \Lambda_{ijk}\Big\} \text{vec}(\widehat{C}_\ell-C^*_{\Tilde{r},I_\ell})\Big|\\
        &+dm\Big|\text{vec}(C^*_{\Tilde{r},\eta_{\ell+1}}-C^*_{\Tilde{r},I_\ell})^\top\Lambda_{\eta_\ell+1,1,2}\text{vec}(\widehat{C}_\ell-C^*_{\Tilde{r},I_\ell}) \Big|\\
        \leq \;& C_{14} \Tilde{r}^{2}\sqrt{dm\log(n)} \|C^*_{\Tilde{r},\eta_{\ell+1}}-C^*_{\Tilde{r},I_\ell}\|_{\F}\|\widehat{C}_\ell-C^*_{\Tilde{r},I_\ell}\|_{\F} + \frac{dm\Tilde{r}^2}{\delta^2}\|C^*_{\Tilde{r},\eta_{\ell+1}}-C^*_{\Tilde{r},I_\ell}\|_{\F}\|\widehat{C}_\ell-C^*_{\Tilde{r},I_\ell}\|_{\F}\\
        \leq \;& \frac{C_{15}\kappa_\ell \delta}{\Tilde{r}\sqrt{\log(n)}} \cdot \frac{\Tilde{r}^2\log(n)}{\delta\zeta_{\delta}\sqrt{\Delta}} \Big\{\frac{dm\Tilde{r}^2}{\delta^2} + \Tilde{r}^{2}\sqrt{dm\log(n)}\Big\}\\
        = \;& O_p\Big(\alpha_n^{-1/2}K^{-1/2} \Big\{d\kappa_\ell^2 \Tilde{r}m \log^{-1/2}(n)\delta^{-1}+\sqrt{d\kappa_\ell^2}\Tilde{r}\sqrt{m}\delta\Big\}\Big),
    \end{align*}
    where the second inequality follows from \Cref{l_op5} and \Cref{l_loss_long_eq4}, the third inequality follows from \Cref{t_inference_eq2} and \Cref{t_inference_eq3}, and the last equality follows from \Cref{a_snr_strong}\ref{a_snr_strong_1}~and \Cref{t_inference_eq1}. Therefore, 
    \begin{align} \notag
        |(I)| &=  O_p\Big(\alpha_n^{-1}K^{-1}\Big\{d\kappa_\ell^2 m\log^{-1}(n)+\sqrt{d\kappa_\ell^2} \Tilde{r}\sqrt{m}\log^{-1/2}(n)\delta^2\Big\} + \alpha_n^{-1/2}K^{-1/2}\delta\Big\{\sqrt{d\kappa_\ell^2}m + m\sqrt{\log(n)}\Big\}\\ \notag
        & \hspace{1.2cm}+\alpha_n^{-1/2}K^{-1/2} \Big\{d\kappa_\ell^2 \Tilde{r}m \log^{-1/2}(n)\delta^{-1}+\sqrt{d\kappa_\ell^2}\Tilde{r}\sqrt{m}\delta\Big\}\Big)\\ \label{t_inference_eq7}
        &= o_p(d\kappa_\ell^2m) +o_p(\sqrt{d\kappa_\ell^2}m) + o_p(m),
    \end{align}
    where the last inequality holds from \Cref{a_snr_strong}\ref{a_snr_strong_2}

    \noindent \textbf{Step 2: Order of magnitude of $(II)$}. We have that
    \begin{align*}
        (II) =\;& \sum_{t = \eta_\ell+1}^{\eta_\ell+d}\sum_{(j,k) \in \mathcal{O}} \big\{\Phi_{\Tilde{r}}^\top(X_{t,j})(\widehat{C}_{\ell+1}-C^*_{\Tilde{r},I_{\ell+1}})\Phi_{\Tilde{r}}(X_{t,k})\big\}^2\\
        & - 2\sum_{t = \eta_\ell+1}^{\eta_\ell+d} \sum_{(j,k) \in \mathcal{O}} \big\{Y_{t,j}Y_{t,k}-\Phi_{\Tilde{r}}^\top(X_{t,j})C^*_{\Tilde{r},\eta_{\ell+1}}\Phi_{\Tilde{r}}(X_{t,k})\big\}\big\{\Phi_{\Tilde{r}}^\top(X_{t,j})(\widehat{C}_{\ell+1}-C^*_{\Tilde{r},I_{\ell+1}})\Phi_{\Tilde{r}}(X_{t,k})\big\}\\
        & -2\sum_{t = \eta_\ell+1}^{\eta_\ell+d} \sum_{(j,k) \in \mathcal{O}}\big\{\Phi_{\Tilde{r}}^\top(X_{t,j})(C^*_{\Tilde{r},\eta_{\ell+1}}-C^*_{\Tilde{r},I_{\ell+1}})\Phi_{\Tilde{r}}(X_{t,k})\big\}\big\{\Phi_{\Tilde{r}}^\top(X_{t,j})(\widehat{C}_{\ell+1}-C^*_{\Tilde{r},I_{\ell+1}})\Phi_{\Tilde{r}}(X_{t,k})\big\}\\
        =\;& (II_1)-2(II_2) -2 (II_3) \leq |(II_1)|+2|(II_2)| +2|(II_3)|.
    \end{align*}
    Using a similar argument as the one in \textbf{Step 1}, it holds that
    \begin{align*}
        (II_1) = O_p\Big(\alpha_n^{-1}K^{-1}\Big\{d\kappa_\ell^2 m\log^{-1}(n)+\sqrt{d\kappa_\ell^2} \Tilde{r}\sqrt{m}\log^{-1/2}(n)\delta^2\Big\}\Big),
    \end{align*}
    and
    \begin{align*}
        (II_2) = O_p\Big(\alpha_n^{-1/2}K^{-1/2}\delta\Big\{\sqrt{d\kappa_\ell^2}m + m\sqrt{\log(n)}\Big\}\Big).
    \end{align*}
    For term $(II_3)$, using a similar argument as the one used to derive an upper bound on $(I_3)$, we have that
    \begin{align*}
        (II_3) &\leq C_{16}\Big\{\frac{dm\Tilde{r}^2}{\delta^2} + \Tilde{r}^{2}\sqrt{dm\log(n)}\Big\} \|C^*_{\Tilde{r},\eta_{\ell+1}}-C^*_{\Tilde{r},I_{\ell+1}}\|_{\F}\|\widehat{C}_{\ell+1}-C^*_{\Tilde{r},I_{\ell+1}}\|_{\F}\\
        &\leq C_{17}\alpha_n^{-1}\kappa_{\ell}\Big\{\frac{dm\Tilde{r}^2}{\delta^2} + \Tilde{r}^{2}\sqrt{dm\log(n)}\Big\} \|\widehat{C}_{\ell+1}-C^*_{\Tilde{r},I_{\ell+1}}\|_{\F}\\
        & \leq C_{18}\alpha_n^{-1}\frac{\kappa_{\ell}\delta}{\Tilde{r}\sqrt{\log(n)}}\frac{\Tilde{r}^2\log(n)}{\delta\zeta_{\delta}\sqrt{\Delta}}\Big\{\frac{dm\Tilde{r}^2}{\delta^2} + \Tilde{r}^{2}\sqrt{dm\log(n)}\Big\}\\
        & = O_p\Big(\alpha_n^{-3/2}K^{-1/2}\Big\{d\kappa_\ell^2\Tilde{r}m\log^{-1/2}(n)\delta^{-1}+\sqrt{d\kappa_{\ell}^2}\Tilde{r}\sqrt{m}\delta\Big\}\Big),
    \end{align*}
    where the second inequality follows from \Cref{t_inference_eq4}, the third inequality follows from \Cref{t_inference_eq2}, and the last inequality follows from \Cref{a_snr_strong}\ref{a_snr_strong_1}~and \Cref{t_inference_eq1}. Therefore, we have that
    \begin{align}\notag
        |(II)| &= O_p\Big(\alpha_n^{-1}K^{-1}\Big\{d\kappa_\ell^2 m\log^{-1}(n)+\sqrt{d\kappa_\ell^2} \Tilde{r}\sqrt{m}\log^{-1/2}(n)\delta^2\Big\} + \alpha_n^{-1/2}K^{-1/2}\delta\Big\{\sqrt{d\kappa_\ell^2}m + m\sqrt{\log(n)}\Big\} \\
        \notag
        & \hspace{1.2cm}+\alpha_n^{-3/2}K^{-1/2}\Big\{d\kappa_\ell^2\Tilde{r}m\log^{-1/2}(n)\delta^{-1}+\sqrt{d\kappa_{\ell}^2}\Tilde{r}\sqrt{m}\delta\Big\}\Big)\\ \label{t_inference_eq8}
        & =o_p(d\kappa_\ell^2m) +o_p(\sqrt{d\kappa_\ell^2}m) + o_p(m),
    \end{align}
    where the last inequality follows from \Cref{a_snr_strong}\ref{a_snr_strong_2}

    \noindent \textbf{Step 3: Order of magnitude of $(III)$}. We have that
    \begin{align*}
        (III) =\;& \sum_{t = \eta_\ell+1}^{\eta_\ell+d}\sum_{(j,k) \in \mathcal{O}} \big\{\Phi_{\Tilde{r}}^\top(X_{t,j})(C^*_{\Tilde{r},I_\ell}-C^*_{\Tilde{r},\eta_\ell})\Phi_{\Tilde{r}}(X_{t,k})\big\}^2\\
        &- 2\sum_{t = \eta_\ell+1}^{\eta_\ell+d} \sum_{(j,k) \in \mathcal{O}}\big\{Y_{t,j}Y_{t,k}-\Phi_{\Tilde{r}}^\top(X_{t,j})C^*_{\Tilde{r},\eta_\ell}\Phi_{\Tilde{r}}(X_{t,k})\big\}\big\{\Phi_{\Tilde{r}}^\top(X_{t,j})(C^*_{\Tilde{r},I_\ell}-C^*_{\Tilde{r},\eta_\ell})\Phi_{\Tilde{r}}(X_{t,k})\big\}\\
        =\;& \sum_{t = \eta_\ell+1}^{\eta_\ell+d}\sum_{(j,k) \in \mathcal{O}} \big\{\Phi_{\Tilde{r}}^\top(X_{t,j})(C^*_{\Tilde{r},I_\ell}-C^*_{\Tilde{r},\eta_\ell})\Phi_{\Tilde{r}}(X_{t,k})\big\}^2\\
        & - 2\sum_{t = \eta_\ell+1}^{\eta_\ell+d} \sum_{(j,k) \in \mathcal{O}}\big\{Y_{t,j}Y_{t,k}-\Phi_{\Tilde{r}}^\top(X_{t,j})C^*_{\Tilde{r},\eta_{\ell+1}}\Phi_{\Tilde{r}}(X_{t,k})\big\}\big\{\Phi_{\Tilde{r}}^\top(X_{t,j})(C^*_{\Tilde{r},I_\ell}-C^*_{\Tilde{r},\eta_\ell})\Phi_{\Tilde{r}}(X_{t,k})\big\}\\
        & - 2\sum_{t = \eta_\ell+1}^{\eta_\ell+d} \sum_{(j,k) \in \mathcal{O}}\big\{\Phi_{\Tilde{r}}^\top(X_{t,j})(C^*_{\Tilde{r},\eta_{\ell+1}}-C^*_{\Tilde{r},\eta_\ell})\Phi_{\Tilde{r}}(X_{t,k})\big\}\big\{\Phi_{\Tilde{r}}^\top(X_{t,j})(C^*_{\Tilde{r},I_\ell}-C^*_{\Tilde{r},\eta_\ell})\Phi_{\Tilde{r}}(X_{t,k})\big\}\\
        =\;& (III_1)-2(III_2)-2(III_3) \leq |(III_1)|+2|(III_2)| +2|(III_3)|.
    \end{align*}
    For the term $(III_1)$, \Cref{t_inference_eq5} and a similar argument used to derive an upper bound on $(I_1)$ together lead to 
    \begin{align*}
        |(III_1)| &= \text{vec}(C^*_{\Tilde{r},I_\ell}-C^*_{\Tilde{r},\eta_\ell})^\top \sum_{t = \eta_\ell+1}^{\eta_\ell+d}\sum_{(j,k) \in \mathcal{O}}\Psi(X_{t,j}, X_{t,k})\Psi^\top(X_{t,j},X_{t,k}) \text{vec}(C^*_{\Tilde{r},I_\ell}-C^*_{\Tilde{r},\eta_\ell})\\
        &\leq  \frac{dm\Tilde{r}^2}{\delta^2}\|C^*_{\Tilde{r},I_\ell}-C^*_{\Tilde{r},\eta_\ell}\|_{\F}^2+C_{19}\Tilde{r}^{2}\sqrt{dm\log(n)}\|C^*_{\Tilde{r},I_\ell}-C^*_{\Tilde{r},\eta_\ell}\|_{\F}^2\\
        & \leq C_{20}\alpha_n^{-2}\kappa_{\ell}^2\Big\{\frac{dm\Tilde{r}^2}{\delta^2} + \Tilde{r}^{2}\sqrt{dm\log(n)}\Big\}\\
        & = O_p\Big(\alpha_n^{-2}\Big\{d\kappa_{\ell}^2\Tilde{r}^2 m\delta^{-2} + \sqrt{d\kappa_\ell^2}\Tilde{r}^2\sqrt{m\log(n)}\Big\}\Big).
    \end{align*}
    For the term $(III_2)$, by \Cref{t_inference_eq5} and \Cref{l_op1_eq5} in \Cref{l_op1}, we have that
    \begin{align*}
        |(III_2)| &\leq  \Big\|\sum_{t = \eta_\ell+1}^{\eta_\ell+d} \sum_{(j,k) \in \mathcal{O}} \big\{Y_{t,j}Y_{t,k}-\Phi_{\Tilde{r}}^\top(X_{t,j})C^*_{\Tilde{r},\eta_{\ell+1}}\Phi_{\Tilde{r}}(X_{t,k})\big\}\Phi_{\Tilde{r}}^\top(X_{t,j}) \odot \Phi_{\Tilde{r}}^\top(X_{t,k})\Big\|_{\op}\Big\|C^*_{\Tilde{r},I_\ell}-C^*_{\Tilde{r},\eta_\ell}\Big\|_{\F}\\
        & \leq C_{21}\alpha_n^{-1}\kappa_{\ell}\Big\{\sqrt{d}\Tilde{r}m\sqrt{\log(n)} \vee \Tilde{r}m\log(n)\Big\}\\
        & = O_p\Big(\alpha_n^{-1}\Big\{\sqrt{d\kappa_\ell^2}\Tilde{r}m\sqrt{\log(n)}+ \Tilde{r}m\log(n)\Big\}\Big).
    \end{align*}
    Moreover, for the term $(III_3)$, \Cref{t_inference_eq5} and a similar argument used to derive an upper bound on $(I_3)$ lead to 
    \begin{align*}
        |(III_3)| &\leq C_{22}\Big\{\frac{dm\Tilde{r}^2}{\delta^2} + \Tilde{r}^{2}\sqrt{dm\log(n)}\Big\} \|C^*_{\Tilde{r},\eta_{\ell+1}}-C^*_{\Tilde{r},\eta_\ell}\|_{\F}\|C^*_{\Tilde{r},I_\ell}-C^*_{\Tilde{r},\eta_\ell}\|_{\F}\\
        & \leq C_{23}\alpha_n^{-1}\kappa_{\ell}^2 \Big\{\frac{dm\Tilde{r}^2}{\delta^2} + \Tilde{r}^{2}\sqrt{dm\log(n)}\Big\}\\
        & = O_p\Big(\alpha_n^{-1}\Big\{ d\kappa_\ell^2 \Tilde{r}^{2}m\delta^{-2}+\sqrt{d\kappa_\ell^2}\Tilde{r}^{2}\sqrt{m\log(n)}\Big\}\Big).
    \end{align*}
    Therefore, we have that
    \begin{align}\notag
        |(III)| =\;& O_p\Big(\alpha_n^{-2}\Big\{d\kappa_{\ell}^2\Tilde{r}^2 m\delta^{-2} + \sqrt{d\kappa_\ell^2}\Tilde{r}^2\sqrt{m\log(n)}\Big\}+\alpha_n^{-1}\Big\{\sqrt{d\kappa_\ell^2}\Tilde{r}m\sqrt{\log(n)}+ \Tilde{r}m\log(n)\Big\}\\ \notag
        & +\alpha_n^{-1}\Big\{ d\kappa_\ell^2 \Tilde{r}^{2}m\delta^{-2}+\sqrt{d\kappa_\ell^2}\Tilde{r}^{2}\sqrt{m\log(n)}\Big\}\Big)\\ \label{t_inference_eq9}
        = \;& o_p(d\kappa_\ell^2m) +o_p(\sqrt{d\kappa_\ell^2}m) + o_p(m),
    \end{align}
    where the last equality follows from \Cref{a_snr_strong}\ref{a_snr_strong_2}

    \noindent \textbf{Step 4: Order of magnitude of $(IV)$.} Following a similar argument as the one in \textbf{Step 3}, we have that
    \begin{align} \notag
        |(IV)|=\;& O_p\Big(\alpha_n^{-2}\Big\{d\kappa_{\ell}^2\Tilde{r}^2m\delta^{-2} + \sqrt{d\kappa_\ell^2}\Tilde{r}^2\sqrt{m\log(n)}\Big\}+\alpha_n^{-1}\Big\{\sqrt{d\kappa_\ell^2}\Tilde{r}m\sqrt{\log(n)}+ \Tilde{r}m\log(n) \Big\}\Big)\\ \label{t_inference_eq10}
        =\;& o_p(d\kappa_\ell^2m) +o_p(\sqrt{d\kappa_\ell^2}m) + o_p(m).
    \end{align}

    \noindent \textbf{Step 5: Lower bound on $(V)$.} With the same notation of $\Lambda_{ijk}$ in \textbf{Step 1}, we have that 
    \begin{align*}
        (V) =\;& \sum_{t = \eta_\ell+1}^{\eta_\ell+d}\sum_{(j,k) \in \mathcal{O}} \Big\{\big\{Y_{t,j}Y_{t,k} -\Phi_r^\top(X_{t,j})C^*_{\Tilde{r},\eta_{\ell}}\Phi_r(X_{t,k})\big\}^2 -\big\{Y_{t,j}Y_{t,k} -\Phi_{\Tilde{r}}^\top(X_{t,j})C^*_{\Tilde{r},\eta_{\ell+1}}\Phi_{\Tilde{r}}(X_{t,k})\big\}^2 \Big\}\\
        =\;& \text{vec}(C^*_{\Tilde{r},\eta_{\ell}}-C^*_{\Tilde{r},\eta_{\ell+1}})^\top\Big\{\sum_{t = \eta_\ell+1}^{\eta_\ell+d} \sum_{(j,k) \in \mathcal{O}}\Psi(X_{t,j}, X_{t,k})\Psi^\top(X_{t,j}, X_{t,k}) - \Lambda_{ijk}\Big\}\text{vec}(C^*_{\Tilde{r},\eta_{\ell}}-C^*_{\Tilde{r},\eta_{\ell+1}})\\
        & + dm\; \text{vec}(C^*_{\Tilde{r},\eta_{\ell}}-C^*_{\Tilde{r},\eta_{\ell+1}})^\top\Lambda_{\eta_\ell+1,1,2}\;\text{vec}(C^*_{\Tilde{r},\eta_{\ell}}-C^*_{\Tilde{r},\eta_{\ell+1}})\\
        & - 2\sum_{t = \eta_\ell+1}^{\eta_\ell+d} \sum_{(j,k) \in \mathcal{O}}\big\{Y_{t,j}Y_{t,k}-\Phi_{\Tilde{r}}^\top(X_{t,j})C^*_{\Tilde{r},\eta_{\ell+1}}\Phi_{\Tilde{r}}(X_{t,k})\big\}\big\{\Phi_{\Tilde{r}}^\top(X_{t,j})(C^*_{\Tilde{r},\eta_\ell} - C^*_{\Tilde{r},\eta_{\ell+1}})\Phi_{\Tilde{r}}(X_{t,k})\big\}\\
        =\;& (V_1)+(V_2) -2(V_3),
    \end{align*}
    thus we have that
    \begin{align*}
        (V) \geq (V_2) - |(V_1)| -2|(V_3)|.
    \end{align*}
    For $(V_2)$, by \Cref{l_rec}, we have that
    \begin{align*}
        (V_2) \geq dm\zeta_{\delta}\|C^*_{\Tilde{r},\eta_{\ell+1}}-C^*_{\Tilde{r},\eta_\ell}\|_{\F}^2 = dm\kappa_{\ell}^2 \zeta_{\delta}.
    \end{align*}
    To find an upper bound on $(V_1)$, for any $t \in \{\eta_\ell+1, \ldots, \eta_\ell+d\}$, consider the random variable
    \begin{align*}
        V_{1,t} = \frac{1}{\Tilde{r}^2m}\sum_{(j,k) \in \mathcal{O}} v^\top\Psi(X_{t,j},X_{t,k})\Psi^T(X_{t,j},X_{t,k})v.
    \end{align*}
    In order to use \Cref{l_inequality_partial_sum}, we are about to show that conditioning on $a= \{a_t\}_{t=\eta_\ell+1}^{\eta_\ell+d}$, it holds that $\|V_{1,t}/r^2m\|_2 \leq C_{25}$ for all $t \in \{\eta_\ell+1, \ldots, \eta_\ell+d\}$. Firstly observe from \Cref{l_op3_eq1} that for any deterministic vector $v \in \mathbb{R}^{\Tilde{r}^2}$ such that $\|v\|_2=1$, we have $\big|v^\top\Psi(X_{t,j},X_{t,k})\Psi^T(X_{t,j},X_{t,k})v\big|\leq  C_{24}\Tilde{r}^{2}$ for any $t \in \{\eta_\ell+1, \ldots, \eta_\ell+d\}$ and $j,k \in \{1, \ldots, m\}$. Therefore, it holds that
    \begin{align*}
        & \underset{t \in \{\eta_\ell+1, \ldots, \eta_\ell+d\}}{\max}  \Big\|\sum_{(j,k) \in \mathcal{O}} v^\top\Psi(X_{t,j},X_{t,k})\Psi^T(X_{t,j},X_{t,k})v\Big\|_2\\
        \leq\;& \underset{t \in \{\eta_\ell+1, \ldots, \eta_\ell+d\}}{\max}  \sum_{(j,k) \in \mathcal{O}} \|v^\top\Psi(X_{t,j},X_{t,k})\Psi^T(X_{t,j},X_{t,k})v\|_2\leq C_{25}\Tilde{r}^{2}m,
    \end{align*}
    which gives $\|V_{1,t}/r^2m\|_2 \leq C_{25}$. Next, note that the terms $\{V_{1,t}\}_{t=\eta_\ell+1}^{\eta_\ell+d}$ are mutually independent conditioning on $a$, therefore by \Cref{l_inequality_partial_sum}, we have that conditioning on $a$
    \begin{align*}
        |(V_1)| &=  \text{vec}(C^*_{\Tilde{r},\eta_{\ell}}-C^*_{\Tilde{r},\eta_{\ell+1}})^\top\Big\{\sum_{t = \eta_\ell+1}^{\eta_\ell+d} \sum_{(j,k) \in \mathcal{O}}\Psi(X_{t,j}, X_{t,k})\Psi^\top(X_{t,j}, X_{t,k}) - \Lambda_{ijk}\Big\}\text{vec}(C^*_{\Tilde{r},\eta_{\ell}}-C^*_{\Tilde{r},\eta_{\ell+1}})\\
        &= O_p\Big( \Tilde{r}^{2}\sqrt{d}m\{\log(d\kappa_\ell^2)+1\}\|C^*_{\Tilde{r},\eta_{\ell}}-C^*_{\Tilde{r},\eta_{\ell+1}}\|_{\F}^2\Big)\\
        &= O_p\Big(\sqrt{d\kappa_\ell^2}  \Tilde{r}^{2}m\{\log(d\kappa_\ell^2)+1\}\Big).
    \end{align*}
    Taking another expectation with respect to $a$, we have that
    \begin{align*}
        |(V_1)| = O_p\Big(\sqrt{d\kappa_\ell^2}  \Tilde{r}^{2}\sqrt{m}\{\log(d\kappa_\ell^2)+1\}\Big).
    \end{align*}
    To find an upper bound on the term $|(V_{3})|$, for any $t \in \{\eta_\ell+1, \ldots, \eta_\ell+d\}$, denote
    \begin{align*}
        V_{3,t} = \frac{1}{\Tilde{r}m}\sum_{(j,k) \in \mathcal{O}}\big\{Y_{t,j}Y_{t,k}-\Phi_{\Tilde{r}}^\top(X_{t,j})C^*_{\Tilde{r},\eta_{\ell+1}}\Phi_{\Tilde{r}}(X_{t,k})\big\}\big\{\Phi_{\Tilde{r}}^\top(X_{t,j})v\Phi_{\Tilde{r}}(X_{t,k})\big\},
    \end{align*}
    where $v \in \mathbb{R}^{\Tilde{r} \times \Tilde{r}}$ is  any deterministic matrix such that $\|v\|_{\F} =1$. We will show that $\|V_{3,t}/rm\|_2 \leq C_{29}$. Firstly note that conditioning on all of the observed grids $X=\{X_{t,j}\}_{t= \eta_\ell+1,j=1}^{\eta_\ell+d, m}$, for any $t \in \{\eta_\ell+1, \ldots, \eta_\ell+d\}$,, we have that
    \begin{align*}
        & \Big\|\sum_{(j,k) \in \mathcal{O}}\big\{Y_{t,j}Y_{t,k}-\Phi_{\Tilde{r}}^\top(X_{t,j})C^*_{\Tilde{r},\eta_{\ell+1}}\Phi_{\Tilde{r}}(X_{t,k})\big\}\big\{\Phi_{\Tilde{r}}^\top(X_{t,j})v\Phi_{\Tilde{r}}(X_{t,k})\big\}\Big\|_{\psi_1}\\
        \leq \;& \sum_{(j,k) \in \mathcal{O}} \Big\|\big\{Y_{t,j}Y_{t,k}-\Phi_{\Tilde{r}}^\top(X_{t,j})C^*_{\Tilde{r},\eta_{\ell+1}}\Phi_{\Tilde{r}}(X_{t,k})\big\}\big\{\Phi_{\Tilde{r}}^\top(X_{t,j})v\Phi_{\Tilde{r}}(X_{t,k})\big\}\Big\|_{\psi_1}\\
        \leq\;& C_{26}\Tilde{r}m,
    \end{align*}
    where $C_{26} >0$ is a constant depending on $C_f$ and $C_\varepsilon$ defined in \Cref{a_model}, and the last inequality holds from the standard property of sub-Exponential random variables \citep[Lemma 2.7.7 in][]{vershynin2018high} and Cauchy-Schwarz inequality, i.e. 
    \begin{align*}
        \Big|\Phi_{\Tilde{r}}^\top(X_{t,j})v\Phi_{\Tilde{r}}(X_{t,k})\Big| & \leq \Big\{\sum_{1\leq a,b \leq \Tilde{r}} |(v)_{ab}|^2\Big\}^{1/2} \Big\{\sum_{1\leq a,b \leq \Tilde{r}} \|\phi_{a}\|^2_{\infty}\|\phi_{b}\|^2_{\infty}\Big\}^{1/2} \leq C_{27}\Tilde{r}.
    \end{align*}
    Therefore, by the sub-Exponential properties \citep[Proposition 2.7.1 in][]{vershynin2018high}, we have that for any $t \in \{\eta_\ell+1, \ldots, \eta_\ell+d\}$, conditioning on $X$
    \begin{align*}
        & \Big\|\sum_{(j,k) \in \mathcal{O}}\big\{Y_{t,j}Y_{t,k}-\Sigma^*_{\eta_{\ell+1}}(X_{t,j}, X_{t,k})\big\}\big\{\Phi_{\Tilde{r}}^\top(X_{t,j})v\Phi_{\Tilde{r}}(X_{t,k})\big\}\Big\|_{2} \\
        \leq\;& 2C_{28}\Big\|\sum_{(j,k) \in \mathcal{O}}\big\{Y_{t,j}Y_{t,k}-\Sigma^*_{\eta_{\ell+1}}(X_{t,j}, X_{t,k})\big\}\big\{\Phi_{\Tilde{r}}^\top(X_{t,j})v\Phi_{\Tilde{r}}(X_{t,k})\big\}\Big\|_{\psi_1} \leq C_{29}\Tilde{r}m,
    \end{align*}
    which gives $\|V_{3,t}/rm\|_2 \leq C_{29}$. Also, note that conditioning on $X$, the terms $\{V_{3,t}\}_{t=\eta_\ell+1}^{\eta_\ell+d}$ are mutually independent across $t$ with mean $0$, therefore applying \Cref{l_inequality_partial_sum} and taking another expectation with respect to $X$, we have that
    \begin{align*}
        |(V_{3})| =  O_p\Big(\sqrt{d\kappa_\ell^2}\Tilde{r}m\{\log(d\kappa_\ell^2)+1\}\Big).
    \end{align*}
    Thus, we have that 
    \begin{align} \notag
        |(V)| &\geq  d\kappa_{\ell}^2 m\zeta_{\delta} - O_p\Big(\sqrt{d\kappa_\ell^2}\{\Tilde{r}m+ \Tilde{r}^{2}m\}\{\log(d\kappa_\ell^2)+1\}\Big)\\ \label{t_inference_eq11}
        & \geq d\kappa_{\ell}^2m\zeta_{\delta} - O_p\Big(\sqrt{d\kappa_\ell^2}m\{\log(d\kappa_\ell^2)+1\}\Big).
    \end{align}
    Taking the results in \Cref{t_inference_eq7}, \eqref{t_inference_eq8}, \eqref{t_inference_eq9}, \eqref{t_inference_eq10} and \eqref{t_inference_eq11} into \Cref{t_inference_eq6}, we have that 
    \begin{align*}
        d\kappa_{\ell}^2 m\zeta_{\delta} \leq O_p\Big(\sqrt{d\kappa_\ell^2}m\{\log(d\kappa_\ell^2)+1\}\Big) + o_p(d\kappa_\ell^2m) +o_p(\sqrt{d\kappa_\ell^2}m) + o_p(m),
    \end{align*}
    which implies that 
    \begin{align} \label{t_inference_eq12}
        d\kappa_{\ell}^2 = O_p(1).
    \end{align}
    
    \noindent \textbf{Limiting distributions.}
    For any $\ell \in \{1, \ldots, K\}$, given the endpoints $s_\ell$ and $e_\ell$ of the interval in local refinement  and the coefficients after projecting the true covariance functions to the selected $\Tilde{r}$ basis before and after the true change point $\eta_\ell$,  $C^*_{\Tilde{r},\eta_{\ell}}$ and $C^*_{\Tilde{r},\eta_{\ell+1}}$,  we define the function $\mathcal{L}^*_\ell(\eta)$ as 
    \begin{align*} 
        \mathcal{L}^*_\ell(\eta) =\;& \sum_{t = s_\ell+1}^{\eta} \sum_{(j,k) \in \mathcal{O}} \big\{Y_{t,j}Y_{t,k} -\Phi_{\Tilde{r}}^\top(X_{t,j})C^*_{\Tilde{r},\eta_{\ell}}\Phi_{\Tilde{r}}(X_{t,k})\big\}^2\\
        &+ \sum_{t = \eta+1}^{e_\ell}\sum_{(j,k) \in \mathcal{O}} \big\{Y_{t,j}Y_{t,k} -\Phi_{\Tilde{r}}^\top(X_{t,j})C^*_{\Tilde{r},\eta_{\ell+1}}\Phi_{\Tilde{r}}(X_{t,k})\big\}^2.
    \end{align*}
    Note that in the case when $d>0$,
    \begin{align*}
        (V)&= \sum_{t = \eta_\ell+1}^{\eta_\ell+d}\sum_{(j,k) \in \mathcal{O}} \Big\{\big\{Y_{t,j}Y_{t,k} -\Phi_r^\top(X_{t,j})C^*_{\Tilde{r},\eta_{\ell}}\Phi_r(X_{t,k})\big\}^2 -\big\{Y_{t,j}Y_{t,k} -\Phi_{\Tilde{r}}^\top(X_{t,j})C^*_{\Tilde{r},\eta_{\ell+1}}\Phi_{\Tilde{r}}(X_{t,k})\big\}^2 \Big\}\\
        & = \mathcal{L}^*_\ell(\eta_\ell+d) - \mathcal{L}^*_\ell(\eta_\ell).
    \end{align*}
    Since $m \asymp 1$ and by the uniform tightness property in \Cref{t_inference_eq12}, we have that when $n \rightarrow \infty$, 
    \begin{align*}
        |\mathcal{L}(\eta_\ell+d)-\mathcal{L}(\eta_\ell)-\{\mathcal{L}^*(\eta_\ell+d) - \mathcal{L}^*(\eta_\ell)\}|\leq |(I)|+|(II)|+|(III)|+|(IV)| = o_p(1).
    \end{align*}
    A similar result when $d <0$ can also be obtained. Therefore, it suffices to find the limiting distribution of $\mathcal{L}^*(\eta_\ell+d) - \mathcal{L}^*(\eta_\ell)$. 

    \noindent \textbf{Non-vanishing regime.} Note that it holds from \Cref{a_inference} that 
    \begin{align*}
        \Upsilon_\ell(\cdot, \cdot) = \big\{\Sigma^*_{\eta_{\ell+1}}(\cdot, \cdot)-\Sigma^*_{\eta_{\ell}}(\cdot, \cdot)\big\}/\kappa_\ell = \big\{\Phi_{\Tilde{r}}^\top (\cdot)(C^*_{\Tilde{r},\eta_{\ell+1}} - C^*_{\Tilde{r},\eta_{\ell}})\Phi_{\Tilde{r}}(\cdot)\big\}/\kappa_\ell.
    \end{align*}
    Observe that when $d >0$, we have that 
    \begin{align*}
        &\mathcal{L}^*(\eta_\ell+d) - \mathcal{L}^*(\eta_\ell)\\
        =\;& \sum_{t = \eta_\ell+1}^{\eta_\ell+d}\sum_{(j,k) \in \mathcal{O}}\Big\{\big\{Y_{t,j}Y_{t,k} -\Phi_{\Tilde{r}}^\top(X_{t,j})C^*_{\Tilde{r},\eta_{\ell}}\Phi_{\Tilde{r}}(X_{t,k})\big\}^2 -\big\{Y_{t,j}Y_{t,k} -\Phi_{\Tilde{r}}^\top(X_{t,j})C^*_{\Tilde{r},\eta_{\ell+1}}\Phi_{\Tilde{r}}(X_{t,k})\big\}^2 \Big\}\\
        =\;& 2\sum_{t = \eta_\ell+1}^{\eta_\ell+d} \sum_{(j,k) \in \mathcal{O}}\big\{Y_{t,j}Y_{t,k}-\Phi_{\Tilde{r}}^\top(X_{t,j})C^*_{\Tilde{r},\eta_{\ell+1}}\Phi_{\Tilde{r}}(X_{t,k})\big\}\big\{\Phi_{\Tilde{r}}^\top(X_{t,j})(C^*_{\Tilde{r},\eta_{\ell+1}}-C^*_{\Tilde{r},\eta_\ell})\Phi_{\Tilde{r}}(X_{t,k})\big\}\\
        & + \sum_{t = \eta_\ell+1}^{\eta_\ell+d}\sum_{(j,k) \in \mathcal{O}} \big\{\Phi_{\Tilde{r}}^\top(X_{t,j})(C^*_{\Tilde{r},\eta_{\ell+1}}-C^*_{\Tilde{r},\eta_{\ell}})\Phi_{\Tilde{r}}(X_{t,k})\big\}^2\\
        \stackrel{\mathcal{D}}{=} \;&\sum_{t = 1}^{d} \sum_{(j,k) \in \mathcal{O}} \Big\{2\varrho_\ell \big\{Y_{t,j}^{(\ell+1)}Y_{t,k}^{(\ell+1)}-\Phi_{\Tilde{r}}^\top(X_{t,j})C^*_{\Tilde{r},\eta_{\ell+1}}\Phi_{\Tilde{r}}(X_{t,k})\big\}\Upsilon_\ell(X_{t,j}, X_{t,k})+ \varrho_\ell^2\Upsilon_\ell^2(X_{t,j}, X_{t,k}) \Big\}\\
        \stackrel{\mathcal{D}}{\rightarrow} \;&\sum_{t=1}^{d} \sum_{(j,k) \in \mathcal{O}} \big\{2\varrho_\ell\epsilon_{t,(j,k)}(\ell+1)\vartheta_{t,(j,k)}(\ell)+\varrho_\ell^2\vartheta_{t,(j,k)}^2(\ell)\big\}.
    \end{align*}
    For $d<0$, we have that
    \begin{align*}
        &\mathcal{L}^*(\eta_\ell+d) - \mathcal{L}^*(\eta_\ell)\\
        =\;& \sum_{t = \eta_\ell+d+1}^{\eta_\ell}\sum_{(j,k) \in \mathcal{O}}\Big\{\big\{Y_{t,j}Y_{t,k} -\Phi_{\Tilde{r}}^\top(X_{t,j})C^*_{\Tilde{r},\eta_{\ell+1}}\Phi_{\Tilde{r}}(X_{t,k})\big\}^2 -\big\{Y_{t,j}Y_{t,k} -\Phi_{\Tilde{r}}^\top(X_{t,j})C^*_{\Tilde{r},\eta_{\ell}}\Phi_{\Tilde{r}}(X_{t,k})\big\}^2 \Big\}\\
        =\;& -2\sum_{t = \eta_\ell+d+1}^{\eta_\ell} \sum_{(j,k) \in \mathcal{O}}\big\{Y_{t,j}Y_{t,k}-\Phi_{\Tilde{r}}^\top(X_{t,j})C^*_{\Tilde{r},\eta_{\ell}}\Phi_{\Tilde{r}}(X_{t,k})\big\}\big\{\Phi_{\Tilde{r}}^\top(X_{t,j})(C^*_{\Tilde{r},\eta_{\ell+1}}-C^*_{\Tilde{r},\eta_\ell})\Phi_{\Tilde{r}}(X_{t,k})\big\}\\
        & + \sum_{t = \eta_\ell+d+1}^{\eta_\ell}\sum_{(j,k) \in \mathcal{O}} \big\{\Phi_{\Tilde{r}}^\top(X_{t,j})(C^*_{\Tilde{r},\eta_{\ell+1}}-C^*_{\Tilde{r},\eta_{\ell}})\Phi_{\Tilde{r}}(X_{t,k})\big\}^2\\
        \stackrel{\mathcal{D}}{=}\;& \sum_{t = d}^{-1} \sum_{(j,k) \in \mathcal{O}} \Big\{-2\varrho_\ell \big\{Y_{t,j}^{(\ell)}Y_{t,k}^{(\ell)}-\Phi_{\Tilde{r}}^\top(X_{t,j})C^*_{\Tilde{r},\eta_{\ell}}\Phi_{\Tilde{r}}(X_{t,k})\big\}\Upsilon_\ell(X_{t,j}, X_{t,k}) + \varrho_\ell^2\Upsilon^2_\ell(X_{t,j}, X_{t,k})\Big\}\\
        \stackrel{\mathcal{D}}{\rightarrow}\;& \sum_{t = d}^{-1} \sum_{(j,k) \in \mathcal{O}} \big\{-2\varrho_\ell\epsilon_{t,(j,k)}(\ell)\vartheta_{t,(j,k)}(\ell)+\varrho_\ell^2\vartheta_{t,(j,k)}^2(\ell)\big\}.
    \end{align*}
    Denote $P_\ell(d)$ a two-sided random walk defined as
    \begin{align*}
        P_\ell(d) = \begin{cases}
           \sum_{t=1}^{d} \sum_{(j,k) \in \mathcal{O}} \big\{2\varrho_\ell\epsilon_{t,(j,k)}(\ell+1)\vartheta_{t,(j,k)}(\ell)+\varrho_\ell^2\vartheta_{t,(j,k)}^2(\ell)\big\}, & d>0,\\
           0, &d = 0,\\
           \sum_{t = d}^{-1} \sum_{(j,k) \in \mathcal{O}} \big\{-2\varrho_\ell\epsilon_{t,(j,k)}(\ell)\vartheta_{t,(j,k)}(\ell)+\varrho_\ell^2\vartheta_{t,(j,k)}^2(\ell)\big\}, &d<0.
        \end{cases}
    \end{align*}
    Note that the path of $P_\ell(d)$ is continuous, and the minimize of $P_k(d)$ is also unique by the fact that the path of $\underset{d \in \mathbb{Z}}{\min}P_\ell(d)$ is continuous, so the probability of it attaining two identical values is zero. Hence all regularity conditions of the Argmax continuous mapping theorem are satisfied. So it follows from Slutsky's theorem and the Argmax continuous mapping theorem \citep[Theorem 3.2.2 in][]{vanderVaart1996} that
    \begin{align*}
        \Tilde{\eta_\ell} - \eta_\ell \stackrel{\mathcal{D}}{\longrightarrow} \underset{d \in \mathbb{Z}}{\arg\min} \; P_\ell(d).
    \end{align*}

    \noindent \textbf{Vanishing regime.} Let $h = \kappa_\ell^{-2}$, and we have that $h \rightarrow \infty$ as $n \rightarrow \infty$. Observe that when $d >0$, we have that
    \begin{align*}
        &\mathcal{L}^*(\eta_\ell+dh) - \mathcal{L}^*(\eta_\ell)\\
        =\;& \sum_{t = \eta_\ell+1}^{\eta_\ell+dh}\sum_{(j,k) \in \mathcal{O}}\Big\{\big\{Y_{t,j}Y_{t,k} -\Phi_{\Tilde{r}}^\top(X_{t,j})C^*_{\Tilde{r},\eta_{\ell}}\Phi_{\Tilde{r}}(X_{t,k})\big\}^2 -\big\{Y_{t,j}Y_{t,k} -\Phi_{\Tilde{r}}^\top(X_{t,j})C^*_{\Tilde{r},\eta_{\ell+1}}\Phi_{\Tilde{r}}(X_{t,k})\big\}^2 \Big\}\\
        =\;& 2\sum_{t = \eta_\ell+1}^{\eta_\ell+dh} \sum_{(j,k) \in \mathcal{O}}\big\{Y_{t,j}Y_{t,k}-\Phi_{\Tilde{r}}^\top(X_{t,j})C^*_{\Tilde{r},\eta_{\ell+1}}\Phi_{\Tilde{r}}(X_{t,k})\big\}\big\{\Phi_{\Tilde{r}}^\top(X_{t,j})(C^*_{\Tilde{r},\eta_{\ell+1}}-C^*_{\Tilde{r},\eta_\ell})\Phi_{\Tilde{r}}(X_{t,k})\big\}\\
        & + \sum_{t = \eta_\ell+1}^{\eta_\ell+dh}\sum_{(j,k) \in \mathcal{O}} \big\{\Phi_{\Tilde{r}}^\top(X_{t,j})(C^*_{\Tilde{r},\eta_{\ell+1}}-C^*_{\Tilde{r},\eta_{\ell}})\Phi_{\Tilde{r}}(X_{t,k})\big\}^2\\
        =\;& \frac{1}{\sqrt{h}}\sum_{t = \eta_\ell+1}^{\eta_\ell+dh} \sum_{(j,k) \in \mathcal{O}} 2\big\{Y_{t,j}Y_{t,k}-\Phi_{\Tilde{r}}^\top(X_{t,j})C^*_{\Tilde{r},\eta_{\ell+1}}\Phi_{\Tilde{r}}(X_{t,k})\big\}\Upsilon_\ell(X_{t,j}, X_{t,k})\\
        &+ \frac{1}{h}\sum_{t = \eta_\ell+1}^{\eta_\ell+dh} \sum_{(j,k) \in \mathcal{O}}\Upsilon^2_\ell(X_{t,j}, X_{t,k}) - \mathbb{E}\big[\Upsilon^2_\ell(X_{t,j}, X_{t,k})\big]\Big\}+ d\sum_{(j,k) \in \mathcal{O}}\mathbb{E}\big[\Upsilon^2_\ell(X_{1,j}, X_{1,k})\big].
    \end{align*}
    For any $s \in \{\ell, \ell+1\}$, denote 
    \begin{align*}
        \sigma^2_\ell(s) = 4 \underset{n \rightarrow \infty}{\lim} \var\Big(\sum_{(j,k) \in \mathcal{O}} \big\{Y^{(s)}_{1,j}Y^{(s)}_{1,k}-\Phi_{\Tilde{r}}^\top(X_{1,j})C^*_{\Tilde{r},\eta_{s}}\Phi_{\Tilde{r}}(X_{1,k})\big\}\Upsilon_\ell(X_{1,j}, X_{1,k}) \Big).
    \end{align*}
    To verify $\sigma_\ell^2(s)$ is well defined, firstly note that 
    \begin{align} \notag
        & \var\Big(\sum_{(j,k) \in \mathcal{O}} \big\{Y^{(s)}_{1,j}Y^{(s)}_{1,k}-\Phi_{\Tilde{r}}^\top(X_{1,j})C^*_{\Tilde{r},\eta_{s}}\Phi_{\Tilde{r}}(X_{1,k})\big\}\Upsilon_\ell(X_{1,j}, X_{1,k}) \Big)\\ \notag
        =\;& \mathbb{E}\Big[\Big\{\sum_{(j,k) \in \mathcal{O}} \big\{Y^{(s)}_{1,j}Y^{(s)}_{1,k}-\Phi_{\Tilde{r}}^\top(X_{1,j})C^*_{\Tilde{r},\eta_{s}}\Phi_{\Tilde{r}}(X_{1,k})\big\}\Upsilon_\ell(X_{1,j}, X_{1,k})\Big\}^2\Big]\\ \notag
        \leq \;& m \sum_{(j,k) \in \mathcal{O}} \Big\|Y^{(s)}_{1,j}Y^{(s)}_{1,k}-\Phi_{\Tilde{r}}^\top(X_{1,j})C^*_{\Tilde{r},\eta_{s}}\Phi_{\Tilde{r}}(X_{1,k})\Upsilon_\ell(X_{1,j}, X_{1,k})\Big\|_2^2\\\label{t_inference_eq13}
        \leq \;& m \sum_{(j,k) \in \mathcal{O}}\Big\|Y^{(s)}_{1,j}Y^{(s)}_{1,k}-\Phi_{\Tilde{r}}^\top(X_{1,j})C^*_{\Tilde{r},\eta_{s}}\Phi_{\Tilde{r}}(X_{1,k})\Big\|_4^2\Big\|\Upsilon_\ell(X_{1,j}, X_{1,k})\Big\|_4^2,
    \end{align}
    where the first equality follows from the fact that
    \begin{align*}
        &\mathbb{E}\Big[\sum_{(j,k) \in \mathcal{O}} \big\{Y^{(s)}_{1,j}Y^{(s)}_{1,k}-\Phi_{\Tilde{r}}^\top(X_{1,j})C^*_{\Tilde{r},\eta_{s}}\Phi_{\Tilde{r}}(X_{1,k})\big\}\Upsilon_\ell(X_{1,j}, X_{1,k})\Big]\\
        =\;&\mathbb{E}_X\Bigg[\mathbb{E}_{f|X}\Big[\sum_{(j,k) \in \mathcal{O}} \big\{Y^{(s)}_{1,j}Y^{(s)}_{1,k}-\Phi_{\Tilde{r}}^\top(X_{1,j})C^*_{\Tilde{r},\eta_{s}}\Phi_{\Tilde{r}}(X_{1,k})\big\}\Upsilon_\ell(X_{1,j}, X_{1,k})\Big]\Bigg] = 0.
    \end{align*}
    Moreover, by Assumptions \ref{a_model}\ref{a_model_error}~and \ref{a_model}\ref{a_model_function}, following a similar argument as the one leads to \Cref{l_loss_short_eq3}, we could prove that 
    \begin{align} \label{t_inference_eq14}
        \Big\|Y^{(s)}_{1,j}Y^{(s)}_{1,k}-\Phi_{\Tilde{r}}^\top(X_{1,j})C^*_{\Tilde{r},\eta_{s}}\Phi_{\Tilde{r}}(X_{1,k})\Big\|_4^2 \leq C_{30}.
    \end{align}
    Also, 
    \begin{align} \label{t_inference_eq15}
        \Big\|\Upsilon_\ell(X_{1,j}, X_{1,k})\Big\|_4 &= \frac{1}{\kappa_\ell}\Big\|\Phi_{\Tilde{r}}^\top(X_{t,j})(C^*_{\Tilde{r},\eta_{\ell+1}}-C^*_{\Tilde{r},\eta_\ell})\Phi_{\Tilde{r}}(X_{t,k})\Big\|_4 \leq \underset{\|\nu\|_2=1}{\sup} \|\Phi_r^\top(X_{t,j})\odot \Phi^{\top}_r(X_{t,k})\nu\|_4 \leq C_{31},
    \end{align}
    where the last inequality holds since each entry of $\Phi_r$ is bounded. 
    Putting results in \Cref{t_inference_eq14} and \Cref{t_inference_eq15} into \Cref{t_inference_eq13}, we have that
    \begin{align*}
        \var\Big(\sum_{(j,k) \in \mathcal{O}} \big\{Y^{(s)}_{1,j}Y^{(s)}_{1,k}-\Phi_{\Tilde{r}}^\top(X_{1,j})C^*_{\Tilde{r},\eta_{s}}\Phi_{\Tilde{r}}(X_{1,k})\big\}\Upsilon_\ell(X_{1,j}, X_{1,k}) \Big) \leq C_{32}m^2,
    \end{align*}
    which suggest that $\sigma_\ell^2(s)< \infty$ for $s \in \{\ell, \ell+1\}$. Hence, by the Functional central limit theorem, we have that when $n \rightarrow \infty$, 
    \begin{align*}
        \frac{1}{h}\sum_{t = \eta_\ell+1}^{\eta_\ell+dh}\Bigg\{ \sum_{(j,k) \in \mathcal{O}}\Big\{\Upsilon^2_\ell(X_{t,j}, X_{t,k}) - \mathbb{E}\big[\Upsilon^2_\ell(X_{t,j}, X_{t,k})\big]\Big\}\Bigg\} \stackrel{p}{\rightarrow} 0,
    \end{align*}
    and
    \begin{align*}
        &\frac{1}{\sqrt{h}}\sum_{t = \eta_\ell+1}^{\eta_\ell+dh} \Bigg\{\sum_{(j,k) \in \mathcal{O}} 2\big\{Y_{t,j}Y_{t,k}-\Phi_{\Tilde{r}}^\top(X_{t,j})C^*_{\Tilde{r},\eta_{\ell+1}}\Phi_{\Tilde{r}}(X_{t,k})\big\}\Upsilon_\ell(X_{t,j}, X_{t,k}) \Bigg\}\\
        \stackrel{\mathcal{D}}{=} \;& \frac{1}{\sqrt{h}}\sum_{t = 1}^{dh} \Bigg\{\sum_{(j,k) \in \mathcal{O}} 2\big\{Y^{(\ell+1)}_{t,j}Y^{(\ell+1)}_{t,k}-\Phi_{\Tilde{r}}^\top(X_{t,j})C^*_{\Tilde{r},\eta_{\ell+1}}\Phi_{\Tilde{r}}(X_{t,k})\big\}\Upsilon_\ell(X_{t,j}, X_{t,k}) \Bigg\}\\
        \stackrel{\mathcal{D}}{\rightarrow}\;& \sigma_\ell(\ell+1)\mathbb{B}(d),
    \end{align*}
    where $\mathbb{B}(d)$ is a standard Brownian motion. Therefore, it holds that when $n \rightarrow \infty$, 
    \begin{align*}
        \mathcal{L}^*(\eta_\ell+dh) - \mathcal{L}^*(\eta_\ell) \stackrel{\mathcal{D}}{\rightarrow}\sigma_\ell(\ell+1)\mathbb{B}(d) + d\varpi_\ell.
    \end{align*}
    Similarly, if $d < 0$, we have that
    \begin{align*}
        &\mathcal{L}^*(\eta_\ell+dh) - \mathcal{L}^*(\eta_\ell)\\
        =\;& \sum_{t = \eta_\ell+dh+1}^{\eta_\ell}\sum_{(j,k) \in \mathcal{O}}\Big\{\big\{Y_{t,j}Y_{t,k} -\Phi_{\Tilde{r}}^\top(X_{t,j})C^*_{\Tilde{r},\eta_{\ell+1}}\Phi_{\Tilde{r}}(X_{t,k})\big\}^2 -\big\{Y_{t,j}Y_{t,k} -\Phi_{\Tilde{r}}^\top(X_{t,j})C^*_{\Tilde{r},\eta_{\ell}}\Phi_{\Tilde{r}}(X_{t,k})\big\}^2 \Big\}\\
        =\;& -\frac{1}{\sqrt{h}}\sum_{t = \eta_\ell+dh+1}^{\eta_\ell} \sum_{(j,k) \in \mathcal{O}} 2\big\{Y_{t,j}Y_{t,k}-\Phi_{\Tilde{r}}^\top(X_{t,j})C^*_{\Tilde{r},\eta_{\ell}}\Phi_{\Tilde{r}}(X_{t,k})\big\}\Upsilon_\ell(X_{t,j}, X_{t,k})\\
        & + \frac{1}{h}\sum_{t = \eta_\ell+dh+1}^{\eta_\ell} \sum_{(j,k) \in \mathcal{O}}\Big\{\Upsilon^2_\ell(X_{t,j}, X_{t,k}) - \mathbb{E}\big[\Upsilon^2_\ell(X_{t,j}, X_{t,k})\big]\Big\}- d\sum_{(j,k) \in \mathcal{O}}\mathbb{E}\big[\Upsilon^2_\ell(X_{1,j}, X_{1,k})\big]\\
        \stackrel{\mathcal{D}}{\rightarrow}\;& \sigma_\ell(\ell)\mathbb{B}(-d)-d\varpi_\ell.
    \end{align*}
    Note that in the vanishing regime, we have $\kappa_\ell^2 \rightarrow 0$ when $n \rightarrow \infty$. Therefore, it holds that
    \begin{align*}
        \sum_{(j,k) \in \mathcal{O}} \big\{Y^{(\ell)}_{1,j}Y^{(\ell)}_{1,k}-\Phi_{\Tilde{r}}^\top(X_{1,j})C^*_{\Tilde{r},\eta_{\ell}}\Phi_{\Tilde{r}}(X_{1,k})\big\}\Upsilon_\ell(X_{1,j}, X_{1,k})
    \end{align*}
    and 
    \begin{align*}
        \sum_{(j,k) \in \mathcal{O}} \big\{Y^{(\ell+1)}_{1,j}Y^{(\ell+1)}_{1,k}-\Phi_{\Tilde{r}}^\top(X_{1,j})C^*_{\Tilde{r},\eta_{\ell+1}}\Phi_{\Tilde{r}}(X_{1,k})\big\}\Upsilon_\ell(X_{1,j}, X_{1,k}).
    \end{align*}
    are equal in distribution. Thus, it holds that 
    \begin{align*}
        \sigma_\ell(\ell) =\sigma_\ell(\ell+1) = \sigma_\ell.
    \end{align*} 
    Moreover, since Brownian motion is a continuous process with a unique minimizer, all regularity conditions for the Argmax continuous mapping theorem are satisfied. Thus, the Slutsky's theorem and the Argmax continuous mapping theorem \citep[Theorem 3.2.2 in][]{vanderVaart1996} lead to
    \begin{align*}
        \kappa_\ell^2(\Tilde{\eta}_\ell - \eta_\ell) \stackrel{\mathcal{D}}{\longrightarrow} \underset{d \in \mathbb{Z}}{\arg\min} \; \mathcal{B}_\ell(d),
    \end{align*}
    where
    \begin{align*}
        \mathcal{B}_\ell(d) = \begin{cases}
        \sigma_\ell\mathbb{B}(d) + d\varpi_\ell, & d >0,\\
        0, &d =0,\\
        \sigma_\ell\mathbb{B}(-d)-d\varpi_\ell, & d<0 .
        \end{cases}
    \end{align*}
\end{proof}

\section{Proof of \texorpdfstring{\Cref{t_estimation}}{}}
\label{pf_estimation}
\begin{proof}[Proof of \Cref{t_estimation}]
    For any $\theta >0$ and a given integer $0<r<n$, consider the set of deterministic matrix 
    \begin{align} \label{t_estimation_eq_B}
        \mathcal{B}_{r,I}(\theta)=\{\Delta \in \mathbb{R}^{r \times r}: \|\Delta\|_{\F}=\theta, C^*_{r,I}+\Delta \in \mathfrak{C}(r), \; \text{for all} \;|I|\in [n] \;\text{such that}\; |I| \gtrsim \log(n)\},
    \end{align} 
    where $\mathfrak{C}(r)$ denotes the collection of $r \times r$ symmetric and positive semidefinite matrix $G$ with elements $g_{hl}$ such that $\Sigma_G = \sum_{1\leq h, l \leq r}g_{hl}\phi_{h}\phi_{l} \in \mathcal{C}$. 
    
    \noindent \textbf{Step 1.} Consider the matrix $\Delta \in \mathcal{B}_{r,I}(\theta)$, and a rate $\omega_{I,m}$ depending on both $n$ and $m$, we would like to find the rate of $\omega_{I,m}$ such that 
    \begin{align} \label{t_estimation_eq1}
        Q(C^*_{r,I} + \omega_{I,m}\Delta)- Q(C^*_{r,I}) > 0.
    \end{align}
    Suppose it is the case that $\|\widehat{C} - C^*_{r,I}\|_{\F} > \omega_{I,m}\theta$, then there exists a quantity $\alpha_{n,m} \in (0, 1)$ depending on $n,m$ such that $\|\widehat{C} - C^*_{r,I}\|_{\F} =\frac{\omega_{I,m}}{1-\alpha_{n,m}}\theta$, which further suggests that there exists a $\Delta \in \mathcal{B}_{r,I}(\theta)$ such that $\widehat{C}$ could be written as $\widehat{C} =C^*_{r,I} +\frac{\omega_{I,m}}{1-\alpha_{n,m}}\Delta$. By convexity of the function $Q(\cdot)$, it holds that
    \begin{align*}
        \alpha_{n,m}Q(C^*_{r,I}) +(1-\alpha_{n,m})Q(\widehat{C}) \geq Q(\alpha_{n,m}C^*_{r,I}+(1-\alpha_{n,m})\widehat{C})= Q(C^*_{r,I} + \omega_{I,m}\Delta) > Q(C^*_r),
    \end{align*}
    where the last inequality follows from \Cref{t_estimation_eq1}. Hence we have that $Q(\widehat{C}) > Q(C^*_{r,I})$, which is a contradiction to the fact that $\widehat{C}$ is the minimizer of $Q(\cdot)$. Therefore, if \Cref{t_estimation_eq1} holds, we could conclude that $\|C^*_{r,I}-\widehat{C}\|_{\F} \leq \omega_{I,m}\theta$. Using the above idea and rearrange \Cref{t_estimation_eq1}, it holds that
    \begin{align}
        \notag
        Q(C^*_r + \omega_{I,m}\Delta)- Q(C^*_r) \geq & -\frac{2\omega_{I,m}}{|I|\lfloor \frac{m}{2} \rfloor}\sum_{t\in I} \sum_{(j,k) \in \mathcal{O}}\big\{Y_{t,j}Y_{t,k}-\Phi_r^\top(X_{t,j})C^*_{r,I}\Phi_r(X_{t,k})\big\}\big\{\Phi_r^\top(X_{t,j})\Delta\Phi_r(X_{t,k})\big\}\\
        \notag
        &+\frac{\omega_{I,m}^2}{|I|\lfloor \frac{m}{2} \rfloor}\sum_{t\in I} \sum_{(j,k) \in \mathcal{O}} \big\{\Phi_r^\top(X_{t,j})\Delta\Phi_r(X_{t,k})\big\}^2\\
        \notag
        &- \frac{2\lambda \omega_{I,m}}{\sqrt{|I|\lfloor \frac{m}{2} \rfloor}} \Big\{|\mathrm{Tr}(C^*_{r,I}U\Delta W)| + |\mathrm{Tr}(C^*_{r,I}V\Delta V)|\Big\}\\
        \notag
        &- \frac{\lambda \omega_{I,m}^2}{\sqrt{|I|\lfloor \frac{m}{2} \rfloor}}  \Big\{|\mathrm{Tr}(\Delta U\Delta W)| + |\mathrm{Tr}(\Delta V\Delta V)|\Big\}\\
        \label{t_estimation_eq2}
        &= -2\omega_{I,m}(A) + \omega^2_{n,m}(B)- \lambda(C),
    \end{align}
    where $W,U,V \in \mathbb{R}^{r\times r}$ are matrices with entries $W_{kl} =\int_{0}^1 \phi_k^{(2)}(t)\phi_l^{(2)}(t)\,\mathrm{d}t$, $V_{kl} =\int_{0}^1 \phi_k^{(1)}(t)\phi_l^{(1)}(t)\,\mathrm{d}t$, and $U_{kl} =\int_{0}^1 \phi_k(t)\phi_l(t)\,\mathrm{d}t =\mathbbm{1}_{\{k=l\}}$ respectively, where $\mathbbm{1}$ is the indicator function. In the rest of the proof, we find an upper bound on each of the three terms individually.
    
    \noindent \textbf{Step 2 Upper bound on Term ($A$).} Note that
    \begin{align*}
        (A) =\;&\frac{1}{|I|\lfloor \frac{m}{2} \rfloor}\sum_{t \in I} \sum_{(j,k) \in \mathcal{O}}\big\{Y_{t,j}Y_{t,k}-\Phi_r^\top(X_{t,j})C^*_{r,I}\Phi_r(X_{t,k})\big\}\big\{\Phi_r^\top(X_{t,j})\Delta\Phi_r(X_{t,k})\big\} \\
        =\; & \frac{1}{|I|\lfloor \frac{m}{2} \rfloor}\sum_{t\in I} \sum_{(j,k) \in \mathcal{O}}\big\{Y_{t,j}Y_{t,k}-\Phi_r^\top(X_{t,j})C^*_{r,t}\Phi_r(X_{t,k})\big\}\big\{\Phi_r^\top(X_{t,j})\Delta\Phi_r(X_{t,k})\big\}\\
        & + \frac{1}{|I|\lfloor \frac{m}{2} \rfloor}\sum_{t\in I} \sum_{(j,k) \in \mathcal{O}}\big\{\Phi_r^\top(X_{t,j})(C^*_{r,t}-C^*_{r,I})\Phi_r(X_{t,k})\big\}\big\{\Phi_r^\top(X_{t,j})\Delta\Phi_r(X_{t,k})\big\}\\
        =\;& \frac{1}{|I|\lfloor \frac{m}{2} \rfloor}\sum_{t \in I}\sum_{(j,k) \in \mathcal{O}}\big\{Y_{t,j}Y_{t,k}-\Sigma_t^*(X_{t,j},X_{t,k})\big\}\big\{\Phi_r^\top(X_{t,j})\Delta\Phi_r(X_{t,k})\big\}\\
        &+\frac{1}{|I|\lfloor \frac{m}{2} \rfloor}\sum_{t \in I}\sum_{(j,k) \in \mathcal{O}}\big\{\Sigma_t^*(X_{t,j},X_{t,k})-\Phi_r^\top(X_{t,j})C^*_{r,t}\Phi_r(X_{t,k})\big\}\big\{\Phi_r^\top(X_{t,j})\Delta\Phi_r(X_{t,k})\big\}\\
        &+\frac{1}{|I|\lfloor \frac{m}{2} \rfloor}\sum_{t\in I} \sum_{(j,k) \in \mathcal{O}}\big\{\Phi_r^\top(X_{t,j})(C^*_{r,t}-C^*_{r,I})\Phi_r(X_{t,k})\big\}\big\{\Phi_r^\top(X_{t,j})\Delta\Phi_r(X_{t,k})\big\}\\
        =\; & (A_1) + (A_2) + (A_3).
    \end{align*}
    To analyze $(A_1)$, rewrite it as 
    \begin{align}
        \notag
        (A_1) &= \frac{1}{|I|\lfloor \frac{m}{2} \rfloor}\sum_{t \in I}\sum_{(j,k) \in \mathcal{O}}\{Y_{t,j}Y_{t,k}-\Sigma^*_t(X_{t,j},X_{t,k})\}\Phi_r^\top(X_{t,j})\odot \Phi^{\top}_r(X_{t,k})\; \mathrm{vec}(\Delta)\\
        \label{t_estimation_eq3}
        & \leq \Bigg\|\frac{1}{|I|\lfloor \frac{m}{2} \rfloor}\sum_{t \in I}\sum_{(j,k) \in \mathcal{O}}\big\{Y_{t,j}Y_{t,k}-\Sigma_t^*(X_{t,j},X_{t,k})\big\}\Phi_r^\top(X_{t,j})\odot \Phi^{\top}_r(X_{t,k})\Bigg\|_{\op} \|\Delta\|_{\F},
    \end{align}
    where the inequality holds from the definition of the operator norm. Hence, by \Cref{l_op1}, it holds with probability at least $1-n^{-3}$ that
    \begin{align} \label{t_estimation_eq4}
        (A_1) \leq C_1r^{2\alpha+1} \sqrt{\frac{\log(n)}{|I|}}\|\Delta\|_{\F} = C_1\theta r^{2\alpha+1} \sqrt{\frac{\log(n)}{|I|}}.
    \end{align}
    For $(A_2)$, following a similar argument as $(A_1)$, denote
    \begin{align*}
        S_{A_2} = \frac{1}{|I|\lfloor \frac{m}{2} \rfloor}\sum_{t \in I} \sum_{(j,k) \in \mathcal{O}}\{\Sigma^*_t(X_{t,j},X_{t,k})-\Phi_r^\top(X_{t,j})C^*_{r,t}\Phi_r(X_{t,k})\}\Phi_r^\top(X_{t,j})\odot \Phi^{\top}_r(X_{t,k}),
    \end{align*}
    and rewrite $(A_2)$ as
    \begin{align*}
        (A_2) &= \frac{1}{I||\lfloor \frac{m}{2} \rfloor}\sum_{t \in I}  \sum_{(j,k) \in \mathcal{O}}\{\Sigma^*_t(X_{t,j},X_{t,k})-\Phi_r^\top(X_{t,j})C^*_{r,t}\Phi_r(X_{t,k})\}\Phi_r^\top(X_{t,j})\odot \Phi^{\top}_r(X_{t,k})\;\mathrm{vec}(\Delta)\\
        & = (S_{A_2} - \mathbb{E}[S_{A_2}|a] +\mathbb{E}[S_{A_2}|a])\;\mathrm{vec}(\Delta)\\
        & \leq \big\|S_{A_2} - \mathbb{E}[S_{A_2}|a]\big\|_{\op}\|\Delta\|_{\F} + \mathbb{E}[S_{A_2}\mathrm{vec}(\Delta)|a]
    \end{align*}
    where the last inequality follows from the definition of the operator norm and fact that $\Delta \in \mathcal{B}_{r,I}(\theta)$ is deterministic, and $a = \{a_t\}_{t\in I}$ is the collection of all the fragment starting points inside the time interval $I$, thus bounding $(A_2)$ is equivalent to find upper bounds on each term individually. For the first term, by \Cref{l_op2}, it holds with probability at least $1-3n^{-3}$ that
    \begin{align} \label{t_estimation_eq5}
        \big\|S_{A_2} - \mathbb{E}[S_{A_2}|a]\big\|_{\op}\|\Delta\|_{\F} \leq C_2r^{4\alpha+2}\sqrt{\frac{\log(n)}{|I|m}} \|\Delta\|_{\F} = C_2\theta r^{4\alpha+2}\sqrt{\frac{\log(n)}{|I|m}}.
    \end{align}
    To find an upper bound on $\mathbb{E}[S_{A_2}\mathrm{vec}(\Delta)|a]$, note that
    \begin{align} \notag
        &\mathbb{E}\Big[\frac{1}{|I|\lfloor \frac{m}{2} \rfloor}\sum_{t \in I} \sum_{(j,k) \in \mathcal{O}}\big\{\Sigma_t^*(X_{t,j},X_{t,k})-\Phi_r^\top(X_{t,j})C^*_{r,t}\Phi_r(X_{t,k})\big\}\Phi_r^\top(X_{t,j})\Delta \Phi_r(X_{t,k})\Big|a\Big]\\ 
        \notag
        =\;& \frac{1}{|I|}\sum_{t \in I}\mathbb{E}\Big[\big\{\Sigma_t^*(X_{t,1},X_{t,2})-\Phi_r^\top(X_{t,1})C^*_{r,t}\Phi_r(X_{t,2})\big\}\Phi_r^\top(X_{t,1})\Delta \Phi_r(X_{t,2})\Big|a\Big]\\
        \notag
        \leq\;& \frac{1}{|I|}\sum_{t \in I}\mathbb{E}\Big[\big\{\Sigma_t^*(X_{t,1},X_{t,2})-\Phi_r^\top(X_{t,1})C^*_{r,t}\Phi_r(X_{t,2})\big\}^2\Big|a\Big]^{\frac{1}{2}}\mathbb{E}\Big[\big\{\Phi_r^\top(X_{t,1})\Delta \Phi_r(X_{t,2})\big\}^2\Big|a\Big]^{\frac{1}{2}}\\
        \notag
        =\;& \frac{1}{|I|\delta^2} \sum_{t \in I}\Big\{\Big(\int_{a_1}^{a_1+\delta}\int_{a_1}^{a_1+\delta}\big\{\Sigma_t^*(s,l)-\Phi_r^\top(s)C^*_{r,t}\Phi_r(l)\big\}^2 \,\mathrm{d}s\,\mathrm{d}l\Big)^{\frac{1}{2}}\Big(\int_{a_1}^{a_1+\delta}\int_{a_1}^{a_1+\delta}\big\{\Phi_r^\top(s)\Delta \Phi_r(l)\big\}^2 \,\mathrm{d}s\,\mathrm{d}l\Big)^{\frac{1}{2}}\Big\}\\
        \notag
        \leq\;& \frac{1}{|I|\delta^2}\sum_{t \in I} \Big(\int_0^1\int_0^1\big\{\Sigma_t^*(s,l)-\Phi_r^\top(s)C^*_{r,t}\Phi_r(l)\big\}^2 \,\mathrm{d}s\,\mathrm{d}l\Big)^{\frac{1}{2}}\Big(\int_0^1 \int_0^1\big\{\Phi_r^\top(s)\Delta \Phi_r(l)\big\}^2 \,\mathrm{d}s\,\mathrm{d}l\Big)^{\frac{1}{2}}\\
        \label{t_estimation_eq6}
        =\;& \frac{\sum_{t \in I} \rho_{r,t}}{|I|\delta^2}\|\Delta \|_{\F} =\frac{\theta\sum_{t \in I} \rho_{r,t}}{|I|\delta^2},
    \end{align}
    where the first equality follows from the fact that conditioning on $a$, under \Cref{a_model}\ref{a_model_grid}, all of the terms are independently and identically distributed across $(j,k) \in \mathcal{O}$ for any $t \in I$, the first inequality follows from the Cauchy-Schwarz inequality, the second equality follows from \Cref{a_model}\ref{a_model_grid}, and the third equality follows from the definition of $\rho_{r,t}$, and the fact that $\int_0^1 \int_0^1\big\{\Phi_r^\top(s)\Delta \Phi_r(t)\big\}^2 \,\mathrm{d}s\,\mathrm{d}t = \|\Delta\|_{\F}^2$ since $\Phi$ is a complete orthonormal system. Hence, combine the results in \Cref{t_estimation_eq5} and \Cref{t_estimation_eq6}, we conclude that with probability at least $1-3n^{-3}$ that
    \begin{align} \label{t_estimation_eq7}
        (A_2) \leq C_2\theta r^{4\alpha+2}\sqrt{\frac{\log(n)}{|I|m}} +\frac{\theta\sum_{t \in I} \rho_{r,t}}{|I|\delta^2}.
    \end{align}
    Next, we are ready to bound $(A_3)$. Using a similar argument as the one in \Cref{t_estimation_eq3}, bounding $(A_3)$ is equivalent to bounding the operator norm of $S_{A_3}$, where $S_{A_3}$ is defined as 
    \begin{align*}
        S_{A_3} = \frac{1}{|I|\lfloor \frac{m}{2} \rfloor}\sum_{t\in I}  \sum_{(j,k) \in \mathcal{O}}\big\{\Phi_r^\top(X_{t,j})(C^*_{r,t}-C^*_{r,I})\Phi_r(X_{t,k})\big\}\Phi_r^\top(X_{t,j})\odot \Phi^{\top}_r(X_{t,k}).
    \end{align*}
    Using a similar argument as the one in \Cref{l_op1_eq1}, we have that
    \begin{align*}
        \|S_{A_3}\|_{\op} &\leq \frac{r}{|I|\lfloor \frac{m}{2} \rfloor} \; \underset{1\leq h,l \leq r}{\max}\Bigg|\sum_{t\in I} \sum_{(j,k) \in \mathcal{O}}\big\{\Phi_r^\top(X_{t,j})(C^*_{r,t}-C^*_{r,I})\Phi_r(X_{t,k})\big\}\phi_h(X_{t,j})\phi_l(X_{t,k})\Bigg|\\
        &\leq \frac{r}{|I|\lfloor \frac{m}{2} \rfloor} \; \underset{1\leq h,l \leq r}{\max} (S_{A_3.1}).
    \end{align*}
    Note that, by the definition of $C^*_{r,I}$, we have that
    \begin{align*}
        & \mathbb{E}\Big[\sum_{t\in I} \sum_{(j,k) \in \mathcal{O}}\big\{\Phi_r^\top(X_{t,j})(C^*_{r,t}-C^*_{r,I})\Phi_r(X_{t,k})\big\}\phi_h(X_{t,j})\phi_l(X_{t,k})\Big|a\Big]\\
        =\;& \sum_{(j,k)\in \mathcal{O}}\sum_{1 \leq a,b \leq r}\mathbb{E}\Big[\sum_{t\in I}  (C^*_{r,t}-C^*_{r,I})_{ab}\phi_a(X_{t,j})\phi_b(X_{t,k})\phi_h(X_{t,j})\phi_l(X_{t,k})\Big|a\Big]\\
        =\;& \sum_{(j,k)\in \mathcal{O}}\sum_{1 \leq a,b \leq r}\Big\{\mathbb{E}\Big[\phi_a(X_{1,j})\phi_b(X_{1,k})\phi_h(X_{1,j})\phi_l(X_{1,k})\Big|a\Big] \Big(\sum_{t\in I}  (C^*_{r,t}-C^*_{r,I})_{ab}\Big)\Big\}=0,
    \end{align*}
    where the last inequality holds since $\{X_{t,j}\}_{t\in I, j \in [m]}$ are i.i.d. given $a$. Also, using a similar argument as the one in \Cref{l_op2_eq1} with the Cauchy-Schwarz inequality, we have that
    \begin{align*}
        &\Big|\sum_{1 \leq a,b \leq r}(C^*_{r,t}-C^*_{r,I})_{a,b}\phi_a(X_{t,j})\phi_b(X_{t,k})\phi_h(X_{t,j})\phi_l(X_{t,k})\Big|\\
        \leq\;& \Big\{\sum_{1\leq a,b \leq r} |(C^*_{r,t}-C^*_{r,I})_{a,b}|^2\Big\}^{1/2} \Big\{\sum_{1\leq a,b \leq r} \|\phi_a\|^2_{\infty}\|\phi_b\|^2_{\infty}\Big\}^{1/2}\Big|\phi_h(X_{t,j})\phi_l(X_{t,k})\Big|\\
        \leq\;& C_3\|\Sigma^*_t- \Sigma^*_I\|_{L^2} \; r^{2\alpha+1}(hl)^{\alpha} \leq C_4r^{2\alpha+1}(hl)^{\alpha},
    \end{align*}
    where the last inequality holds from the fact that $\Sigma^*_t, \Sigma^*_I \in L([0,1]^2)$. By Hoeffding’s inequality for general bounded random variables \citep[Theorem 2.2.6 in][]{vershynin2018high}, it holds that
    \begin{align*}
        \mathbb{P}\Big\{(S_{A_3.1}) \geq \tau\Big\}&= \mathbb{E}_a\Big[\mathbb{P}\big\{(S_{A_3.1}) \geq \tau |a\big\}\Big]\\
         & \leq \exp\Bigg\{-\frac{C_5\tau^2}{|I|m(hl)^{2\alpha}r^{4\alpha+2}}\Bigg\}.
    \end{align*}
    Hence, by taking two union bounds arguments individually on $(h,l)$ and the choice of the integer interval $I$, and pick
    \begin{align*}
        \tau = C_6\sqrt{|I|m\log(n \vee r)}r^{4\alpha+1},
    \end{align*}
    we have that with probability at least $1-3n^{-3}$ that for any integer interval $|I|$, 
    \begin{align*}
        \|S_{A_3}\|_{\op} \leq C_6r^{4\alpha+2}\sqrt{\frac{\log(n)}{|I|m}},
    \end{align*}
    therefore, it holds that 
    \begin{align} \label{t_estimation_eq8}
        (A_3) \leq C_6r^{4\alpha+2}\sqrt{\frac{\log(n)}{|I|m}}\|\Delta\|_{\F} =  C_6\theta r^{4\alpha+2}\sqrt{\frac{\log(n)}{|I|m}}.
    \end{align}
    Combine results in \Cref{t_estimation_eq4}, \Cref{t_estimation_eq7}, and \Cref{t_estimation_eq8}, we can conclude that with probability at least $1-3n^{-3}$ that
    \begin{align} \label{t_estimation_eq9}
        (A) \leq C_7\theta \Bigg\{r^{2\alpha+1} \sqrt{\frac{\log(n)}{|I|}} +r^{4\alpha+2}\sqrt{\frac{\log(n)}{|I|m}} +\frac{\sum_{t \in I} \rho_{r,t}}{|I|\delta^2}\Bigg\}.
    \end{align}

    \noindent \textbf{Step 3 Lower bound on Term ($B$).}
    Write $\Psi(X_{t,j}, X_{t,k}) = \Phi_r(X_{t,j}) \odot \Phi_r(X_{t,k})\in \mathbb{R}^{r^2}$ which is the column vector of $(\phi_h\phi_l)_{1\leq h,l\leq r}$ and define
    \begin{align*}
        S_{B} = \frac{1}{|I|\lfloor \frac{m}{2} \rfloor}\sum_{t \in I}\sum_{(j,k) \in \mathcal{O}} \Psi(X_{t,j},X_{t,k})\Psi^T(X_{t,j},X_{t,k}) \in \mathbb{R}^{r^2 \times r^2},
    \end{align*}
    then $(B)$ could be written as
    \begin{align*}
        (B)= \frac{1}{|I|\lfloor \frac{m}{2} \rfloor}\sum_{t \in I} \sum_{(j,k) \in \mathcal{O}} \{\Phi_r^\top(X_{t,j})\Delta\Phi_r(X_{t,k})\}^2 = \mathrm{vec}(\Delta)^{\top}S_{B}\;\mathrm{vec}(\Delta).
    \end{align*}
    Therefore, by \Cref{l_rec} and \Cref{l_op3}, it holds with probability at least $1-3n^{-3}$ that
    \begin{align} \notag
        (B)&= \mathrm{vec}(\Delta)^{\top}\mathbb{E}[S_{B}|a]\mathrm{vec}(\Delta)+\mathrm{vec}(\Delta)^{\top}(S_{B}-\mathbb{E}[S_{B}|a])\mathrm{vec}(\Delta)\\
        \notag
        &\geq \mathbb{E}[\mathrm{vec}(\Delta)^{\top}S_{B}\mathrm{vec}(\Delta)|a] - \Big|\mathrm{vec}(\Delta)^{\top}(S_{B}-\mathbb{E}[S_{B}|a])\mathrm{vec}(\Delta)\Big|\\
        \label{t_estimation_eq10}
        & \geq C_8\Big\{\zeta_{\delta}- r^{4\alpha+2}\sqrt{\frac{\log(n)}{|I|m}}\Big\}\|\Delta \|_{\F}^2 =C_8\Big\{\zeta_{\delta}- r^{4\alpha+2}\sqrt{\frac{\log(n)}{|I|m}}\Big\}\theta^2 ,
    \end{align}
    where the first inequality follows from the triangle inequality, and the fact that $\Delta$ is deterministic.
    
    \noindent \textbf{Step 4 Upper bound on Term ($C$).} Since by assumption on $\Phi$, we have that $\|\phi_k^{(d)}\|_{L^2} \leq C_{\phi}k^{\beta d}$ for $d \in \{1,2\}$, therefore it holds that $\|W\|^2_{\F} \leq C_9r^{8\beta+2}$, $\|V\|^2_{\F} \leq C_{10}r^{4\beta+2}$, and $\|U\|^2_{\F} \leq r$. Also, we have that
    \begin{align*}
        \|C^*_{r,I}\|_{\F} \leq \|\Sigma^*_{I}\|_{L^2} \leq C_{11},
    \end{align*}
    where the last inequality holds since
    \begin{align*}
        \|\Sigma^*_{I}\|_{L^2} \leq \frac{1}{|I|}\sum_{t \in I}\|\Sigma^*_t\|_{\F} \leq C_{11}.
    \end{align*}
    Therefore, we have that
    \begin{align*}
        &|\mathrm{Tr}(C^*_{r,I}U\Delta W)| \leq \sqrt{r}\|C^*_{r,I}U\Delta W\|_{\F} \leq \sqrt{r}\|C^*_{r,I}\|_{\F}\|U\|_{\F}\|\Delta \|_{\F}\|W\|_{\F} = C_{12}r^{4\beta+2}\|\Delta \|_{\F},\\
        &|\mathrm{Tr}(C^*_{r,I}V\Delta V)| \leq \sqrt{r} \|C^*_rV\Delta V\|_{\F} \leq \sqrt{r}\|C^*_{r,I}\|_{\F}\|V\|_{\F}\|\Delta \|_{\F}\|V\|_{\F} = C_{13}r^{4\beta+5/2}\|\Delta \|_{\F},\\
        &|\mathrm{Tr}(\Delta U\Delta W)| \leq \sqrt{r}\|\Delta U\Delta W\|_{\F} \leq \sqrt{r} \|\Delta\|_{\F}\|U\|_{\F}\|\Delta \|_{\F}\|W\|_{\F} = C_{14}r^{4\beta+2}\|\Delta \|_{\F}^2,\\
        & |\mathrm{Tr}(\Delta V\Delta V)| \leq \sqrt{r} \|\Delta V\Delta V\|_{\F}  \leq \sqrt{r} \|\Delta\|_{\F}\|V\|_{\F}\|\Delta \|_{\F}\|V\|_{\F}= C_{15}r^{4\beta+5/2}\|\Delta \|_{\F}^2,
    \end{align*}
    where the inequalities follow from the facts that for any matrix $A, B \in \mathbb{R}^{r\times r}$, it holds that $\mathrm{Tr}(A) \leq \sqrt{r}\|A\|_{\F}$, and $\|AB\|_{\F} \leq \|A\|_{\F}\|B\|_{\F}$. Hence,
    \begin{align} \label{t_estimation_eq11}
        (C) \leq \frac{C_{16}\{\omega_{I,m} \vee \omega_{I,m}^2\}r^{4\beta+5/2}}{\sqrt{|I|\lfloor \frac{m}{2} \rfloor}}\{\|\Delta \|_{\F} \vee \|\Delta \|_{\F}^2\} = \frac{C_{16}\{\omega_{I,m} \vee \omega_{I,m}^2\}r^{4\beta+5/2}}{\sqrt{|I|\lfloor \frac{m}{2} \rfloor}}\{\theta \vee \theta^2\}.
    \end{align}

    \noindent \textbf{Step 5.} Plugging the results of \Cref{t_estimation_eq9}, \Cref{t_estimation_eq10}, and \Cref{t_estimation_eq11} into \Cref{t_estimation_eq2}, we have that
    \begin{align*}
        &Q(C^*_{r,I} + \omega_{I,m}\Delta)- Q(C^*_{r,I})\\
        \geq& -2C_7\omega_{I,m}\theta\Big\{r^{2\alpha+1} \sqrt{\frac{\log(n)}{|I|}} \vee r^{4\alpha+2}\sqrt{\frac{\log(n)}{|I|m}} \vee \frac{\sum_{t \in I} \rho_{r,t}}{|I|\delta^2}\Big\}\\
        & + C_8\omega_{I,m}^2\theta^2\Big\{\zeta_{\delta}- r^{4\alpha+2}\sqrt{\frac{\log(n)}{|I|m}}\Big\}-\frac{C_{16}\lambda\{\omega_{I,m} \vee \omega_{I,m}^2\}\{\theta \vee \theta^2\}r^{4\beta+5/2}}{\sqrt{|I|\lfloor \frac{m}{2} \rfloor}}.
    \end{align*}
    Pick
    \begin{align*}
        \omega_{I,m} \asymp \frac{r^{2\alpha+1}}{\zeta_{\delta}}\sqrt{\frac{\log(n)}{|I|}}+\frac{r^{4\alpha+2}}{\zeta_{\delta}}\sqrt{\frac{\log(n)}{|I|m}}+\frac{\sum_{t \in I} \rho_{r,t}}{|I|\zeta_{\delta}\delta^2},
    \end{align*}
    and 
    \begin{align*}
        \lambda = C_{\lambda}\Big\{r^{2\alpha-4\beta-3/2}\sqrt{m\log(n)} \vee r^{4\alpha-4\beta-1/2}\sqrt{\log(n)} \vee \frac{r^{-4\beta-5/2}\sqrt{m}\sum_{t \in I}\rho_{r,t}}{\sqrt{|I|}\delta^2}\Big\},
    \end{align*}
    we have $Q(C^*_{r,I} + \omega_{I,m}\Delta)- Q(C^*_{r,I}) > 0$.
    
    \noindent \textbf{Step 6.} To further compute an upper bound on the estimation error of the covariance function, $\|\Phi_r^\top\widehat{C}_I\Phi_r - \Sigma^*_{I}\|_{L^2}$, we have that
    \begin{align*}
        \|\Phi_r^\top\widehat{C}_I\Phi_r - \Sigma^*_{I}\|_{L^2} &\leq \|\Phi_r^\top\widehat{C}_I\Phi_r - \Phi_r^\top C^*_{r,I}\Phi_r\|_{L^2} + \| \Phi_r^\top C^*_{r,I}\Phi_r - \Sigma^*_{I}\|_{L^2}\\
        & =  \|\widehat{C}_I - C^*_{r,I}\|_{\F}+\| \Phi_r^\top C^*_{r,I}\Phi_r - \Sigma^*_{I}\|_{L^2}\\
        & \lesssim \frac{r^{2\alpha+1}}{\zeta_{\delta}}\sqrt{\frac{\log(n)}{|I|}}+\frac{r^{4\alpha+2}}{\zeta_{\delta}}\sqrt{\frac{\log(n)}{|I|m}} + \frac{\sum_{t \in I} \rho_{r,t}}{|I|}\{1 \vee \frac{1}{\zeta_{\delta}\delta^2}\},
    \end{align*}
    where the equality follows from the fact that $\Phi_r$ are orthonormal hence $\int_{0}^1\int_{0}^1 \phi_h(s)\phi_l(t)\,\mathrm{d}s\,\mathrm{d}t = \mathbbm{1}_{\{h=l\}}$ fir $h,l \in [r]$, and $\mathbbm{1}$ is the indicator function, thus the result follows.
\end{proof}

\section{Technical Lemmas} \label{pf_technical}
We collect all the necessary technical Lemmas in this section. Deviation bounds used in the proof of Theorems \ref{t_localisation}, \ref{t_inference} and \ref{t_estimation} are collected in Appendix \ref{pf_TL_deviation}. Sub-Weibull properties used in the proof of \Cref{t_localisation} can be found in Appendix \ref{pf_TL_subWeibull}, and a maximal inequality used in the proof of \Cref{t_inference} are collected in Appendix \ref{pf_TL_maximal}.

\subsection{Deviation bounds} \label{pf_TL_deviation}
\begin{lemma} \label{l_op1}
    Under the condition of \Cref{t_estimation}, with a given integer $0 <r <n$, it holds with probability at least $1-3n^{-3}$ that for any consecutive interval $I \subseteq [n]$ such that $|I| \geq C_1\log(n)$,
    \begin{align*}
        \Bigg\|\frac{1}{|I|\lfloor \frac{m}{2} \rfloor}\sum_{t \in I}\sum_{(j,k) \in \mathcal{O}}\big\{Y_{t,j}Y_{t,k}-\Sigma_t^*(X_{t,j},X_{t,k})\big\}\Phi_r^\top(X_{t,j})\odot \Phi^{\top}_r(X_{t,k})\Bigg\|_{\op} \leq C_2r^{2\alpha+1} \sqrt{\frac{\log(n)}{|I|}},
    \end{align*}
    where $C_2 >0$ is a sufficiently large constant depending only on $C_f$ and $C_\varepsilon$.
\end{lemma}

\begin{proof}
    Note that
    \begin{align}
        \notag
        &\Bigg\|\frac{1}{|I|\lfloor \frac{m}{2} \rfloor}\sum_{t \in I}\sum_{(j,k) \in \mathcal{O}}\big\{Y_{t,j}Y_{t,k}-\Sigma_t^*(X_{t,j},X_{t,k})\big\}\Phi_r^\top(X_{t,j})\odot \Phi^{\top}_r(X_{t,k})\Bigg\|_{\op}\\
        \notag
        =\;  &\frac{1}{|I|\lfloor \frac{m}{2} \rfloor} \sqrt{\sum_{1\leq h,l \leq r}\Bigg\{\sum_{t \in I} \sum_{(j,k) \in \mathcal{O}}\big\{Y_{t,j}Y_{t,k}-\Sigma_t^*(X_{t,j},X_{t,k})\big\}\phi_h(X_{t,j})\phi_l(X_{t,k})\Bigg\}^2}\\
        \label{l_op1_eq1}
        \leq \; &\frac{r}{|I|\lfloor \frac{m}{2} \rfloor}\; \underset{1\leq h,l \leq r}{\max} \Bigg|\sum_{t \in I} \sum_{(j,k) \in \mathcal{O}}\big\{Y_{t,j}Y_{t,k}-\Sigma_t^*(X_{t,j},X_{t,k})\big\}\phi_h(X_{t,j})\phi_l(X_{t,k})\Bigg|,
    \end{align}
    where the last inequality follows from the fact that for any matrix $A \in \mathbb{R}^{1\times r^2}$ $\|A\|_{2} \leq r\|A\|_{\infty}$. Under \Cref{a_model}, by the property of sub-Exponential random variables \citep[Lemma 2.7.7 in][]{vershynin2018high}, conditioning on the set of observation grids $X = \{X_{t,j}\}_{t\in I, j \in [m]}$, it holds that
    \begin{align}
        \notag
        &\Big\|\sum_{(j,k) \in \mathcal{O}} Y_{t,j}Y_{t,k}\phi_h(X_{t,j})\phi_l(X_{t,k})\Big\|_{\psi_1}\\
        \notag
        =\; & \Big\|\sum_{(j,k) \in \mathcal{O}}\big\{f_t(X_{t,j})f_t(X_{t,k})+ \varepsilon_{t,j}f_t(X_{t,k})+\varepsilon_{t,k}f_t(X_{t,j})+\varepsilon_{t,j}\varepsilon_{t,k}\big\}\phi_h(X_{t,j})\phi_l(X_{t,k})\Big\|_{\psi_1}\\
        \notag
        \leq & \; \sum_{(j,k) \in \mathcal{O}} \Big\|\big\{f_t(X_{t,j})f_t(X_{t,k})+ \varepsilon_{t,j}f_t(X_{t,k})+\varepsilon_{t,k}f_t(X_{t,j})+\varepsilon_{t,j}\varepsilon_{t,k}\big\}\phi_h(X_{t,j})\phi_l(X_{t,k})\Big\|_{\psi_1}\\
        \notag
        \leq & \; C_{\phi}^2C_f^2m(hl)^{\alpha}+ 2C_{\phi}^2C_\varepsilon C_f m(hl)^\alpha +C_{\phi}^2C_{\varepsilon}^2m(hl)^\alpha\\
        \label{l_op1_eq2}
        \leq & \; C_1m(hl)^\alpha,
    \end{align}
    where the first inequality follows from the triangle inequality. Also, note that for any $t \in I$, 
    \begin{align}
        \notag
        &\mathbb{E}\Big[\sum_{(j,k) \in \mathcal{O}}\big\{Y_{t,j}Y_{t,k}-\Sigma^*(X_{t,j},X_{t,k})\big\}\phi_h(X_{t,j})\phi_l(X_{t,k})\Big|X\Big]\\
        \notag
        = \; &  \sum_{(j,k) \in \mathcal{O}} \mathbb{E}\Big[\big\{f_t(X_{t,j})f_t(X_{t,k})-\Sigma^*(X_{t,j},X_{t,k})\big\}\phi_h(X_{t,j})\phi_l(X_{t,k})\Big|X\Big]\\
        \notag
        &+\sum_{(j,k) \in \mathcal{O}}\mathbb{E}\Big[\big\{\varepsilon_{t,j}f_t(X_{t,k})+\varepsilon_{t,k}f_t(X_{t,j})+\varepsilon_{t,j}\varepsilon_{t,k}\big\}\phi_h(X_{t,j})\phi_l(X_{t,k})\Big|X\Big]\\
        \label{l_op1_eq3}
        =\; & 0,
    \end{align}
    where the last equality follows from the independence between $\{\varepsilon_{t,j}\}_{t \in I, j \in [m]}$, $\{f_t\}_{t\in I}$ and $X$, and the fact that $\mathbb{E}[\epsilon_{t,j}|X]=\mathbb{E}[\epsilon_{t,k}|X] =0$, and $\mathbb{E}[f_t(X_{t,j})f_t(X_{t,k})|X] = \Sigma^*_t(X_{t,j}, X_{t,k})$. Hence, using \Cref{l_op1_eq2} and \Cref{l_op1_eq3}, by the Bernstein inequality for the sum of independent sub-Exponential random variables \citep[Theorem 2.8.1 in][]{vershynin2018high}, we have that for any $\tau > 0$,  
    \begin{align*}
        \mathbb{P}\Bigg\{\Big|\sum_{t \in I} \sum_{(j,k) \in \mathcal{O}}\big\{Y_{t,j}Y_{t,k}-\Sigma_t^*&(X_{t,j},X_{t,k})\big\}\phi_h(X_{t,j})\phi_l(X_{t,k})\Big| \geq \tau \Bigg| X\Bigg\}\\
        &\leq 2\exp\Bigg(-C_2 \min\Big\{\frac{\tau^2}{|I|m^2(hl)^{2\alpha}}, \frac{\tau}{m(hl)^{\alpha}}\Big\}\Bigg).
    \end{align*}
    Consequently, we have that
    \begin{align} \notag
         &\mathbb{P}\Bigg\{\Big|\sum_{t \in I} \sum_{(j,k) \in \mathcal{O}}\big\{Y_{t,j}Y_{t,k}-\Sigma_t^*(X_{t,j},X_{t,k})\big\}\phi_h(X_{t,j})\phi_l(X_{t,k})\Big| \geq \tau \Bigg\}\\
         \label{l_op1_eq4}
         =\;& \mathbb{E}\Bigg[\mathbb{P}\Bigg\{\Big|\sum_{t \in I} \sum_{(j,k) \in \mathcal{O}}\big\{Y_{t,j}Y_{t,k}-\Sigma_t^*(X_{t,j},X_{t,k})\big\}\phi_h(X_{t,j})\phi_l(X_{t,k})\Big| \geq \tau \Bigg| X\Bigg\}\Bigg]\\
         \notag
         \leq \; & 2\exp\Bigg(-C_2 \min\Big\{\frac{\tau^2}{|I|m^2(hl)^{2\alpha}}, \frac{\tau}{m(hl)^{\alpha}}\Big\}\Bigg).
    \end{align}
    Applying two union-bound arguments on $(h,l)$ and the choices of the interval $I$ respectively, and pick 
    \begin{align*} 
         \tau = C_3\Big\{r^{2\alpha}m\sqrt{|I|\log(n\vee r)} \vee r^{2\alpha}m\log(n\vee r)\Big\},
    \end{align*}
    we have that conditioning on $X$, with probability at least $1-3n^{-3}$ that for any chosen interval $I \subseteq [n]$, 
    \begin{align} \notag
        &\Bigg\|\frac{1}{|I|\lfloor \frac{m}{2} \rfloor}\sum_{t \in I}\sum_{(j,k) \in \mathcal{O}}\big\{Y_{t,j}Y_{t,k}-\Sigma_t^*(X_{t,j},X_{t,k})\big\}\Phi_r^\top(X_{t,j})\odot \Phi^{\top}_r(X_{t,k})\Bigg\|_{\op}\\ \label{l_op1_eq5}
        \leq \; &\frac{C_3 r}{|I|\lfloor \frac{m}{2} \rfloor}\Big\{r^{2\alpha}m\sqrt{|I|\log(n\vee r)} \vee r^{2\alpha}m\log(n\vee r)\Big\}\\ \notag
        \leq \; & C_4r^{2\alpha+1} \sqrt{\frac{\log(n)}{|I|}},
    \end{align}
    where the last inequality follows from the fact that $|I|\gtrsim \log(n)$.
\end{proof}

\begin{lemma} \label{l_op2}
    Denote 
    \begin{align*}
        S_1 = \frac{1}{|I|\lfloor \frac{m}{2} \rfloor}\sum_{t \in I} \sum_{(j,k) \in \mathcal{O}}\{\Sigma^*_t(X_{t,j},X_{t,k})-\Phi_r^\top(X_{t,j})C^*_{r,t}\Phi_r(X_{t,k})\}\Phi_r^\top(X_{t,j})\odot \Phi^{\top}_r(X_{t,k}).
    \end{align*}
    Under the condition of \Cref{t_estimation}, with a given integer $0 < r< n$, it holds with probability at least $1-3n^{-3}$ that for any integer interval $I \subseteq [n]$,
    \begin{align*}
        \big\|S_1- \mathbb{E}[S_1|a]\big\|_{\op} \leq Cr^{4\alpha+2}\sqrt{\frac{\log(n)}{|I|m}},
    \end{align*}
    for a sufficiently large constant $C >0$, and $a = \{a_t\}_{t\in I}$ is the collection of the fragment starting points inside the time interval $I$.
\end{lemma}

\begin{proof}
    Following a similar argument as the one in \Cref{l_op1_eq1}, it holds that
    \begin{align*}
        & \big\|S_1- \mathbb{E}[S_1|a]\big\|_{\op}\\
        \leq \; & \frac{r}{|I|\lfloor \frac{m}{2} \rfloor} \; \underset{1\leq h,l \leq r}{\max}\Bigg|\sum_{t \in I} \sum_{(j,k) \in \mathcal{O}}\big\{\Sigma_t^*(X_{t,j},X_{t,k})-\Phi_r^\top(X_{t,j})C^*_{r,t}\Phi_r(X_{t,k})\big\}\phi_h(X_{t,j})\phi_l(X_{t,k})\\
        &\hspace{6cm}-\mathbb{E}\Big[\big\{\Sigma_t^*(X_{t,j},X_{t,k})-\Phi_r^\top(X_{t,j})C^*_{r,t}\Phi_r(X_{t,k})\big\}\phi_h(X_{t,j})\phi_l(X_{t,k})\Big|a\Big]\Bigg|\\
       = \; & \frac{r}{|I|\lfloor \frac{m}{2} \rfloor} \; \underset{1\leq h,l \leq r}{\max}(S_{1.1}).
    \end{align*}
    By the Cauchy-Schwarz inequality, it holds that
    \begin{align} \notag
        \big|\Phi_r^\top(X_{t,j})C^*_{r,t}\Phi_r(X_{t,k})\big| &= \Big|\sum_{1\leq h',l' \leq r} (C^*_{r,t})_{h'l'} \phi_{h'}(X_{t,j})\phi_{l'}(X_{t,k})\Big| \\
        \notag
        & \leq \Big\{\sum_{1\leq h',l' \leq r} |(C^*_{r,t})_{h'l'}|^2\Big\}^{1/2} \Big\{\sum_{1\leq h',l'\leq r} \|\phi_{h'}\|^2_{\infty}\|\phi_{l'}\|^2_{\infty}\Big\}^{1/2}\\
        \label{l_op2_eq1}
        & \leq \|\Sigma^*_t\|_{L^2} \; r^{2\alpha+1} \leq C_1r^{2\alpha+1},
    \end{align}
    where the last inequality holds since $\Sigma^*_t 
    \in L([0,1]^2)$. Hence we have that
    \begin{align*}
        \Big|\big\{\Sigma_t^*(X_{t,j},X_{t,k})-\Phi_r^\top(X_{t,j})C^*_{r,t}\Phi_r(X_{t,k})\big\}\phi_h(X_{t,j})\phi_l(X_{t,k})\Big| \leq C_2r^{2\alpha+1}(hl)^{\alpha},
    \end{align*}
    where the last inequality holds since by \Cref{a_model}\ref{a_model_identifiable}, $\Sigma_t^*$ is continuous over a compact interval $[0,1]^2$, hence bounded. Moreover, by \Cref{a_model}\ref{a_model_grid}, conditioning on $a$, all of the terms inside the double summation of $(S_{1.1})$ are independent. By Hoeffding’s inequality for general bounded random variables \citep[Theorem 2.2.6 in][]{vershynin2018high}, we have that for any $\tau >0$,
    \begin{align*}
        \mathbb{P}\Big\{(S_{1.1}) \geq \tau\Big\}&= \mathbb{E}\Big[\mathbb{P}\big\{(S_{1.1}) \geq \tau |a\big\}\Big] \leq \exp\Bigg(-\frac{C_3\tau^2}{|I|m(hl)^{2\alpha}r^{4\alpha+2}}\Bigg).
    \end{align*}
    The final result follows by taking two union bounds arguments individually on $(h,l)$ and the choice of the integer interval $I$, and pick
    \begin{align*}
        \tau = C_4\sqrt{|I|m\log(n \vee r)}r^{4\alpha+1}.
    \end{align*}  
\end{proof}

\begin{lemma} \label{l_op3}
    Write $\Psi(X_{t,j}, X_{t,k}) = \Phi_r(X_{t,j}) \odot \Phi_r(X_{t,k})\in \mathbb{R}^{r^2}$, and define 
    \begin{align*}
        S_2 = \frac{1}{|I|\lfloor \frac{m}{2} \rfloor}\sum_{t \in I}\sum_{(j,k) \in \mathcal{O}} \Psi(X_{t,j},X_{t,k})\Psi^T(X_{t,j},X_{t,k}) \in \mathbb{R}^{r^2 \times r^2}.
    \end{align*}
   Let $0 <r <n$ be a given integer and $v \in \mathbb{R}^{r^2}$ be any deterministic vector satisfying $\|v\|_2 =1$. Under the condition of \Cref{t_estimation}, it holds with probability at least $1-3n^{-3}$ that for any integer interval $I \subseteq [n]$,
    \begin{align*}
        \Big|v^\top S_2 v - v^\top\mathbb{E}[S_2|a]v \Big| \leq Cr^{4\alpha+2}\sqrt{\frac{\log(n)}{|I|m}},
    \end{align*}
    for a sufficiently large constant $C >0$, and $a = \{a_t\}_{t\in I}$ is the collection of all the fragment starting points within the time interval $I$.
\end{lemma}

\begin{proof}
     Denote $\Lambda_{tjk} = \mathbb{E}[\Psi(X_{t,j},X_{t,k})\Psi^T(X_{t,j},X_{t,k})|a]$, then it holds that
     \begin{align*}
         \Big|v^\top(S_2-\mathbb{E}[S_2|a])v\Big| &= \frac{1}{|I|\lfloor \frac{m}{2} \rfloor}\Bigg|\sum_{t\in I} \sum_{(j,k) \in \mathcal{O}} v^\top\big(\Psi(X_{t,j},X_{t,k})\Psi^T(X_{t,j},X_{t,k})-\Lambda_{tjk}\big)v\Bigg|\\
         & \leq \frac{1}{n\lfloor \frac{m}{2} \rfloor}\; (S_{2.1}).
     \end{align*}
     Note that by Cauchy-Schwarz inequality, and following a similar argument as the one in \Cref{l_op2_eq1}, it holds that
     \begin{align} \notag
         &\Big|v^\top\Psi(X_{t,j},X_{t,k})\Psi^T(X_{t,j},X_{t,k})v\Big|\\ \notag
         =\;& \Big|\sum_{1\leq a, b\leq r^2} v_av_b\phi_{\lceil\frac{a}{r}\rceil}(X_{t,j})\phi_{\lceil\frac{b}{r}\rceil}(X_{t,j})\phi_{a-(\lceil\frac{a}{r}\rceil-1)r}(X_{t,k})\phi_{b-(\lceil\frac{b}{r}\rceil-1)r}(X_{t,k})\Big|\\ \notag
         \leq\;& \Big(\sum_{1\leq a, b\leq r^2} v_a^2 v_b^2\Big)^{1/2}\Big(\sum_{1\leq a, b\leq r^2}\phi^2_{\lceil\frac{a}{r}\rceil}(X_{t,j})\phi^2_{\lceil\frac{b}{r}\rceil}(X_{t,j})\phi^2_{a-(\lceil\frac{a}{r}\rceil-1)r}(X_{t,k})\phi^2_{b-(\lceil\frac{b}{r}\rceil-1)r}(X_{t,k})\Big)^{1/2}\\ \label{l_op3_eq1}
         \leq\;& C_1\Big(\sum_{1\leq a \leq r^2} v_a^2\Big)r^{4\alpha+2} = C_1 r^{4\alpha+2}.
     \end{align}
     Moreover, by \Cref{a_model}\ref{a_model_grid}, conditioning on $a$, all terms inside the double summation of $(S_{2.1})$ are independent. By Hoeffding’s inequality for general bounded random variables \citep[Theorem 2.2.6 in][]{vershynin2018high}, we have that for any $\tau >0$,
    \begin{align*}
        \mathbb{P}\Big\{(S_{2.1}) \geq \tau\Big\}&= \mathbb{E}\Big[\mathbb{P}\big\{(S_{2.1}) \geq \tau |a\big\}\Big] \leq \exp\Bigg(-\frac{C_2\tau^2}{|I|mr^{8\alpha+4}}\Bigg).
    \end{align*}
    Applying a union bound argument on the choices of the interval $I$, and picking $\tau = \sqrt{|I|m\log(n)}r^{4\alpha+2}$, the result follows.
\end{proof}

\begin{lemma} \label{l_op4}
    Under the same setup of \Cref{l_op3}, let $0 < r < n$ be a given integer and $v \in \mathbb{R}^{r^2}$ be any vector satisfying $\|v\|_2 =1$, then it holds with probability at least $1-3n^{-3}$ that for any integer interval $I \subseteq [n]$,
    \begin{align*}
        \Big|v^\top S_2 v - v^\top\mathbb{E}[S_2|a]v \Big| \leq Cr^{4\alpha+3}\sqrt{\frac{\log(n)}{|I|m}},
    \end{align*}
    for a sufficiently large constant $C >0$, and $a = \{a_t\}_{t\in I}$ is the collection of all the fragment starting points within the time interval $I$.
\end{lemma}
\begin{proof}
    The proof follows from \Cref{l_op3}, and a similar argument as the one in Lemma 35 in \citet{xu2022change} by taking $s = r^2$ in their proof. Using standard techniques of covering Lemma and a union-bound argument, we have that for any $\tau >0$
    \begin{align*}
        \mathbb{P}\Big\{\underset{v \in \mathbb{R}^{r^2}:\|v\|_2=1}{\sup}\sum_{t \in I}\sum_{(j,k) \in \mathcal{O}}\Big|v^\top\big(\Psi(X_{t,j},X_{t,k})\Psi^T(X_{t,j},X_{t,k})-\Lambda_{t,j,k}\big)v\Big| \geq \tau\Big\} \leq \exp\Big(r^2\log(21)- \frac{C_1\tau^2}{|I|mr^{8\alpha+4}}\Big).
    \end{align*}
    Hence picking $\tau = C_2\sqrt{|I|m\log(n)}r^{4\alpha+3}$, the result follows accordingly.
\end{proof}

\begin{lemma} \label{l_op5}
    Under the condition of \Cref{t_estimation}, write $\Psi(X_{t,j}, X_{t,k}) = \Phi_r(X_{t,j}) \odot \Phi_r(X_{t,k})$ and $\Lambda_{tjk} =~\mathbb{E}\big[\Psi(X_{t,j}, X_{t,k})\Psi^\top(X_{t,j}, X_{t,k})\big|a\big]$, where $a = \{a_{t}\}_{t\in I}$. Let $0 < r < n$ be a given integer and $V \in \mathbb{R}^{r^2}$ be any deterministic vector satisfying $\|V\|_2 =1$, then it holds with probability at least $1-3n^{-3}$ that for any integer interval $I \subseteq [n]$,
    \begin{align*}
        \Big\|V^\top\Big\{\sum_{t \in I} \sum_{(j,k) \in \mathcal{O}}\Psi(X_{t,j}, X_{t,k})\Psi^\top(X_{t,j}, X_{t,k}) - \Lambda_{tjk}\Big\} \Big\|_{\op} \leq Cr^{4\alpha+2}\sqrt{|I|m\log(n)},
    \end{align*}  
    for a sufficiently large constant $C >0$.
    
\end{lemma}
\begin{proof}
    Following a similar argument as the one that leads to \Cref{l_op1_eq1}, we have that
    \begin{align*}
        &\Big\|V^\top\Big\{\sum_{t \in I} \sum_{(j,k) \in \mathcal{O}}\Psi(X_{t,j}, X_{t,k})\Psi^\top(X_{t,j}, X_{t,k}) - \Lambda_{ijk}\Big\}\Big\|_{\op} \\
        \leq \;& r\underset{1\leq a \leq r^2}{\max} \Big|\sum_{t \in I}\sum_{(j,k) \in \mathcal{O}} \sum_{1\leq b\leq r^2} V_b\phi_{\lceil\frac{a}{r}\rceil}(X_{t,j})\phi_{\lceil\frac{b}{r}\rceil}(X_{t,j})\phi_{a-(\lceil\frac{a}{r}\rceil-1)r}(X_{t,k})\phi_{b-(\lceil\frac{b}{r}\rceil-1)r}(X_{t,k}) -V_b\big(\Lambda_{tjk}\big)_{ba}\Big|.
    \end{align*}
    By Cauchy-Schwarz inequality, we have that
    \begin{align*}
        &\Big|\sum_{1\leq b\leq r^2} V_b\phi_{\lceil\frac{a}{r}\rceil}(X_{t,j})\phi_{\lceil\frac{b}{r}\rceil}(X_{t,j})\phi_{a-(\lceil\frac{a}{r}\rceil-1)r}(X_{t,k})\phi_{b-(\lceil\frac{b}{r}\rceil-1)r}(X_{t,k})\Big|\\
        \leq\;& \Big(\sum_{1\leq b\leq r^2} V_b^2\Big)^{1/2}\Big(\sum_{1\leq b\leq r^2}\phi^2_{\lceil\frac{a}{r}\rceil}(X_{t,j})\phi^2_{\lceil\frac{b}{r}\rceil}(X_{t,j})\phi^2_{a-(\lceil\frac{a}{r}\rceil-1)r}(X_{t,k})\phi^2_{b-(\lceil\frac{b}{r}\rceil-1)r}(X_{t,k})\Big)^{1/2}\\
        \leq\;& C_1\Big(\sum_{1\leq b \leq r^2} V_b^2\Big)r^{2\alpha+1}\Big\lceil\frac{a}{r}\Big\rceil^\alpha\Big(a-\Big(\Big\lceil\frac{a}{r}\Big\rceil-1\Big)r\Big)^\alpha\\
        =\;& C_1r^{2\alpha+1}\Big\lceil\frac{a}{r}\Big\rceil^\alpha\Big(a-\Big(\Big\lceil\frac{a}{r}\Big\rceil-1\Big)r\Big)^\alpha \leq C_1r^{4\alpha+1}.
    \end{align*}
   Thus, conditioning on the starting points $a = \{a_t\}_{t\in I}$, applying Hoeffding’s inequality for general bounded random variables \citep[Theorem 2.2.6 in][]{vershynin2018high} and two union-bound arguments on the choices of $r$ and $I$, it holds that for any $\tau >0$, with probability at most
    \begin{align*}
        \exp\Big(4\log(n\vee r)-\frac{C_2\tau^2}{|I|mr^{8\alpha+2}}\Big),
    \end{align*}
    that for any integer interval $I$,
    \begin{align*}
        \underset{1\leq a \leq r^2}{\max} \Big|\sum_{t \in I}\sum_{(j,k) \in \mathcal{O}} \sum_{1\leq b\leq r^2} V_b\phi_{\lceil\frac{a}{r}\rceil}(X_{t,j})\phi_{\lceil\frac{b}{r}\rceil}(X_{t,j})\phi_{a-(\lceil\frac{a}{r}\rceil-1)r}(X_{t,k})\phi_{b-(\lceil\frac{b}{r}\rceil-1)r}(X_{t,k}) -V_b\big(\Lambda_{ijk}\big)_{b,a}\Big| \geq \tau.
    \end{align*}
    Pick $\tau = C_3r^{4\alpha+1}\sqrt{|I|m\log(n)}$, and taking another expectation with respect to $a$ similar to the argument in \Cref{l_op1_eq4}, the result follows.
\end{proof}

\begin{lemma}[Restricted eigenvalue condition] \label{l_rec}
For any $\theta >0$, integer $0<r<n$ and~$\Delta \in~\mathcal{B}_{r,I}(\theta)$, there exists some constants $\zeta_{\delta} > 0$ depending on $\delta$ and $\Sigma^*$, independently of $\Delta$, such that
\begin{align*}
    \mathbb{E}[\{\Phi_r^\top(X_{1,1})\Delta\Phi_r(X_{1,2})\}^2|a_1] \geq \zeta_{\delta}\|\Delta\|_{\F}^2
\end{align*}
almost surely, where $a_1 \in [0,1-\delta]$ is a given starting point of the observing fragment at $t =1$.
\end{lemma}
\begin{proof}
    The proof follows a similar argument as the proof of Lemma S4 in \citet{lin2021basis}. We will prove by contradiction, and the following proof is constructed conditioning on $a_1$. Suppose that there exists a $\Delta^* \in \mathcal{B}_{r,I}(\theta)$ such that $\mathbb{E}[\{\Phi_r^\top(X_{1,1})\Delta^*\Phi_r(X_{1,2})\}^2] = 0 $. Since $\mathcal{B}_{r,I}(\theta)$ is closed, therefore, we can find a sequence $\xi_k \rightarrow 0$ and a convergent sequence $\Delta_k \in \mathcal{B}_{r,I}(\theta)$ with $\underset{k \rightarrow \infty}{\lim} \Delta_k =\Delta^*$ such that
    \begin{align*}
        \mathbb{E}[\{\Phi_r^\top(X_{1,1})\Delta_k\Phi_r(X_{1,2})\}^2] \leq \xi_k \|\Delta_k\|_{\F}^2 = \xi_k\theta^2 \rightarrow 0.
    \end{align*}
    By Fatou's Lemma, we have that
    \begin{align*}
        \mathbb{E}[\underset{k \rightarrow \infty}{\lim}\{\Phi_r^\top(X_{1,1})\Delta_k\Phi_r(X_{1,2})\}^2] \leq \underset{k \rightarrow \infty}{\lim}\mathbb{E}[\{\Phi_r^\top(X_{1,1})\Delta_k\Phi_r(X_{1,2})\}^2] = 0,
    \end{align*}
    which implies that there exists a subsequence $\{k_\ell\}_{\ell=1}^{\infty}$ such that
    \begin{align*}
        \underset{\ell \rightarrow \infty}{\lim} \{\Phi_r^\top(s)\Delta_{k_\ell}\Phi_r(t)\} =0, \quad a.e. \; \mathrm{on} \; [a_1, a_1+\delta]^2,
    \end{align*}
    or equivalently
    \begin{align} \label{l_rec_eq1}
        \underset{\ell \rightarrow \infty}{\lim} \{\Phi_r^\top(s)(\Delta_{k_\ell}+C_{r,I}^*)\Phi_r(t)\} = \Phi_r^\top(s)C_{r,I}^*\Phi_r(t) = \Sigma^*_{r,I}(s,t), \quad a.e. \; \mathrm{on} \; [a_1, a_1+\delta]^2.
    \end{align}
    Since a finite sum of analytic functions is still analytic \citep[by Proposition 1.1.7 in][]{krantz2002primer}, we have that $\Sigma^*_{r, I}$ is analytic, hence $\Sigma^*_{r, I} \in \mathcal{C}$ by analytic continuity \citep[Corollary 1.2.6 in][]{krantz2002primer}. Also, since $\Delta_{k_\ell}\in \mathcal{B}_{r,I}(\theta)$, we have that $\Phi_r^\top(\cdot)(\Delta_{k_{\ell}}+C_{r, I}^*)\Phi_r(\cdot) \in \mathcal{C}$. 
    
    Furthermore, by Proposition 1(iii) in \citet{lin2021basis} and \Cref{a_model}\ref{a_model_identifiable}, the fact that all functions in $\mathcal{C}$ are continuous over a compact interval implies that $\mathcal{C}$ is a bounded sequentially compact family \citep[Definition 2 in][]{lin2021basis}, hence we can find a further subsequence $\{k_{\ell_h}\}_{h=1}^{\infty}$ such that $\Phi_r^\top(\cdot)(\Delta_{k_{\ell_h}}+C_{r, I}^*)\Phi_r(\cdot)$ converges pointwise over $[0,1]^2$ to some limit $\psi(\cdot, \cdot) \in \mathcal{C}$. This $\psi$ should satisfy $\psi(s,t) = \Sigma_{r, I}^*(s,t)$ for all $(s,t) \in [a_1, a_1+\delta]^2$ by \Cref{l_rec_eq1} and the fact that every subsequence of a convergent sequence will converge to the same limit. Moreover, by the $\mathcal{T}_{u,\delta}$-identifiability in \Cref{a_model}\ref{a_model_identifiable}, this further implies that $\psi(s,t) = \Sigma^*_{r, I}(s,t)$ for all $(s,t) \in [0,1]^2$. Therefore, using the above argument, we have that
    \begin{align*}
        \theta^2 & = \big\|\Delta^*\big\|_{\F}^2= \big\|\underset{h \rightarrow \infty}{\lim} \Delta_{k_{\ell_h}} \big\|_{\F}^2 \\
        &= \int_{0}^1\int_{0}^1 \underset{h \rightarrow \infty}{\lim} \{\Phi_r^\top(s)\Delta_{k_{\ell_h}}\Phi_r(t)\}^2 \,\mathrm{d}s\,\mathrm{d}t\\
        & = \int_{0}^1\int_{0}^1 \underset{h \rightarrow \infty}{\lim} \{\Phi_r^\top(s)(\Delta_{k_{\ell_h}}+C^*_{r, I})\Phi_r(t) - \Phi_r^\top(s)C^*_{r, I}\Phi_r(t)\}^2 \,\mathrm{d}s\,\mathrm{d}t\\
        & =\int_{0}^1\int_{0}^1 \{\psi(s,t) - \Sigma^*_{r,I}(s,t)\}^2 \,\mathrm{d}s\,\mathrm{d}t\\
        & = 0,
    \end{align*}
    where the first equality follows from the fact that $\Delta^* \in \mathcal{B}_{r,I}(\theta)$, and the third equality follows since $\Phi$ is a complete orthonormal system. Hence we have reached a contradiction to the fact that $\theta >0$. 
\end{proof}

\begin{remark} \label{r_zeta_delta}
Under \Cref{a_model}, it holds that
\begin{align*}
    \zeta_\delta \|\Delta\|_{\F}^2 &\leq \mathbb{E}[\{\Phi_r^\top(X_{1,1})\Delta\Phi_r(X_{1,2})\}^2|a_1] = \frac{1}{\delta^2}\int_{a_t}^{a_t+\delta} \int_{a_t}^{a_t+\delta}\{\Phi_r^\top(s)\Delta\Phi_r(t)\}^2 \,\mathrm{d}s\,\mathrm{d}t\\
    &\leq \frac{1}{\delta^2} \int_{0}^{1} \int_{0}^{1}\{\Phi_r^\top(s)\Delta\Phi_r(t)\}^2 \,\mathrm{d}s\,\mathrm{d}t = \frac{1}{\delta^2}\|\Delta\|_{\F}^2,
\end{align*}
where the first inequality follows from \Cref{l_rec}, and the last equality follows from the orthonormal properties of $\Phi$. Hence we have that $\zeta_\delta \leq C_\zeta$, where $C_\zeta > 0$ is an absolute constant.
\end{remark}

\subsection{Sub-Weibull property} \label{pf_TL_subWeibull}
\begin{lemma}[\citealp{wong2020lasso}, Sub-Weibull properties]
\label{l_subWeibuill_property}
Let $X$ be a random variable. Then the following statements are equivalent for every $\alpha>0$. The constants $C_1, C_2, C_3 >0$ differ at most by a constant depending only on $\alpha$.
\begin{enumerate}
    \item The tail of $X$ satisfies 
    \begin{align*}
        \mathbb{P}\big\{|X|>t\big\} \leq 2 \exp \big(-(t / C_1)^\alpha\big), \; \text{for all} \;\; t \geq 0.
    \end{align*}
    \item The moments of $X$ satisfy
    \begin{align*}
        \|X\|_p:=\big(\mathbb{E}\big[|X|^p\big]\big)^{1 / p} \leq C_2 p^{1 / \alpha}, \; \text{for all} \;\; p \geq 1 \wedge \alpha.
    \end{align*}
    \item  The moment generating function of $|X|^\alpha$ is finite at some point; namely
    \begin{align*}
        \mathbb{E}\big[\exp (|X| / C_3)^\alpha\big] \leq 2.
    \end{align*}
\end{enumerate}
We further call a random variable $X$ which satisfies any of the properties above a sub-Weibull random variable with parameter $\alpha$.
\end{lemma}

\begin{theorem}[Theorem 3.1 in \citealp{kuchibhotla2022moving}]
\label{t_concentration_subWeibull}
    Let $X_1, \ldots, X_n$ be independent mean zero sub-Weibull random variables with parameter $\alpha$, where $0 < \alpha \leq 1$. If it follows that $\|X_i\|_{\psi_\alpha} \leq K$ for all $1 \leq i \leq n$ and $K >0$, then there exists a constant $C_\alpha >0$ depending only on $\alpha$ such that the following holds
    \begin{align*}
        \mathbb{P}\Big\{\Big|\sum_{i=1}^n X_i\Big| \geq \tau \Big\} \leq e\exp\Big(-C_\alpha \min\Big\{\frac{\tau^2}{nK^2},  \Big(\frac{\tau}{K}\Big)^\alpha\Big\} \Big),
    \end{align*}
    for any $\tau >0$.
\end{theorem}

\begin{lemma}
\label{l_addition_subWeibull}
    Let $X_1, \ldots, X_n$ be a finite collection of sub-Weibull random variables with parameters $\alpha_1, \ldots,\alpha_n$ respectively. Then for any deterministic $v_1, \ldots, v_n \in \mathbb{R}$, it holds that $\sum_{i=1}^n v_iX_i$ is a sub-Weibull random variable with parameter $\min\{\alpha_1, \ldots, \alpha_n\}$.
\end{lemma}

\begin{proof}
    The proof of the case when $n=2$ is given here and a similar proof of the general case could be constructed iteratively. To prove the closure under addition, Property 2 in \Cref{l_subWeibuill_property} entails that there exist constants $C_1, C_2 > 0$ such that for all $p \geq 1 \wedge \alpha_1 \wedge \alpha_2$,
    \begin{align*}
        \|v_1X_1+v_2X_2\|_p &\leq v_1\|X_1\|_p+v_2\|X_2\|_p \leq v_1C_1 p^{1 / \alpha_1}+v_2C_2 p^{1 / \alpha_2}\\
        &\leq (v_1C_1+v_2C_2)p^{\max\{1/\alpha_1,1/\alpha_2\}},
    \end{align*}
    where the first inequality follows from the triangle inequality. Hence, Property 2 in \Cref{l_subWeibuill_property} suggests that $v_1X_1+v_2X_2$ follows a sub-Weibull distribution with parameters $\min\{\alpha_1,\alpha_2\}$.
\end{proof}

\subsection{A maximal inequality for weighted partial sums} \label{pf_TL_maximal}
\begin{theorem}[Theorem 3 in \citealp{rosenthal1970subspaces}]
\label{t_rosenthal}
    Let $X_1, \ldots, X_n$ be mean $0$ independent random variables, and we further assume that for all $i\in\{1, \ldots, n\}$, $\|X_i\|_p = \big(\mathbb{E}\big[|X_i|^p\big]\big)^{1/p} < \infty$ for $p \geq 2$. Denote $S_n = \sum_{i=1}^n X_i$, then it holds that 
    \begin{align*}
        \mathbb{E}\big[|S_n|^p\big] \leq K_p \max \Big\{\sum_{i=1}^n \mathbb{E}|X_i|^p, \Big(\sum_{i=1}^n \mathbb{E}|X_i|^2\Big)^{p/2}\Big\},
    \end{align*}
    where $K_p > 0$ is a constant depending only on the value of the $p$.
\end{theorem}

\begin{lemma}
\label{l_maximal_inequality}
    Under the same assumption of \Cref{t_rosenthal}, if we further assume that $\|X_i\|_2 \leq C_1$ for all $i \in [n]$ and $C_1>0$ is an absolute constant, then it follows that for any $m \in \mathbb{Z}_{+}$, $\alpha >0$, and $\tau>0$,
    \begin{align*}
        \mathbb{P}\Bigg\{\underset{k \in \{m, \ldots, \lfloor(\alpha+1)m \rfloor\}}{\max} \frac{\Big|\sum_{i=1}^k X_i\Big|}{\sqrt{k}} \geq \tau \Bigg\} \leq C_2\tau^{-2},
    \end{align*}
    where $C_2>0$ is a constant depending on $C_1$.
\end{lemma}

\begin{proof}
    Let $N=\lfloor(\alpha+1)m \rfloor$, and note that $\big\{\big(\sum_{i=1}^k X_i\big)^2\big\}_{k=1}^N$ is a sub-martingale relative to a filtration $\{\mathcal{F}_{k}\}_{k=1}^N = \{\sigma(X_1, \ldots, X_k)\}_{k=1}^N$ since
    \begin{align*}
        \mathbb{E}\Big[\Big(\sum_{i=1}^{k+1} X_i\Big)^2\Big|X_1, \ldots, X_k\Big] \geq \Big(\mathbb{E}\Big[\sum_{i=1}^{k+1} X_i\Big|X_1, \ldots, X_k\Big]\Big)^2=\Big(\sum_{i=1}^{k} X_i\Big)^2,
    \end{align*}
    where the first inequality follows from Jensen's inequality and the equality follows from the fact that $\mathbb{E}[X_{k+1}]=0$. Hence, following from Doob's martingale inequality, we have that
    \begin{align*}
        \mathbb{P}\Big\{\underset{k \in [N]}{\max}\Big|\sum_{i=1}^{k} X_i\Big| \geq \tau \Big\} = \mathbb{P}\Big\{\underset{k \in [N]}{\max}\Big(\sum_{i=1}^{k} X_i\Big)^2 \geq \tau^2 \Big\} \leq \frac{\mathbb{E}\Big[\Big(\sum_{i=1}^{N} X_i\Big)^2\Big]}{\tau^2}.
    \end{align*}
    By \Cref{t_rosenthal}, it holds that 
    \begin{align*}
        \mathbb{E}\Big[\Big(\sum_{i=1}^{N} X_i\Big)^2\Big] \leq C_1\sum_{i=1}^N \mathbb{E}[X_i^2] \leq C_2N,
    \end{align*}
    hence we have that
    \begin{align*}
        \mathbb{P}\Bigg\{\frac{1}{\sqrt{N}} \; \underset{k \in [N]}{\max}\Big|\sum_{i=1}^{k} X_i\Big| \geq \tau \Bigg\} \leq C_2\tau^{-2}.
    \end{align*}
    Observe that
    \begin{align*}
         \underset{k \in [N]}{\max} \; \frac{\Big|\sum_{i=1}^{k} X_i\Big|}{\sqrt{N} } \geq \underset{k \in \{m,\ldots, \lfloor(\alpha+1)m\rfloor\}}{\max} \; \frac{\Big|\sum_{i=1}^{k} X_i\Big|}{\sqrt{N} } \geq \underset{k \in \{m,\ldots, \lfloor(\alpha+1)m\rfloor\}}{\max} \;\frac{\Big|\sum_{i=1}^{k} X_i\Big|}{\sqrt{(\alpha+1)k} }.
    \end{align*}
    Therefore, 
    \begin{align*}
        \mathbb{P}\Bigg\{\underset{k \in \{m,\ldots, \lfloor(\alpha+1)m\rfloor\}}{\max} \;\frac{\Big|\sum_{i=1}^{k} X_i\Big|}{\sqrt{k} } \geq \sqrt{\alpha+1}\tau\Bigg\} \leq \mathbb{P}\Bigg\{\underset{k \in [N]}{\max}\frac{\Big|\sum_{i=1}^{k} X_i\Big|}{\sqrt{N}}  \geq \tau \Bigg\} \leq C_2\tau^{-2},
    \end{align*}
    which gives
    \begin{align*}
          \mathbb{P}\Bigg\{\underset{k \in \{m,\ldots, \lfloor(\alpha+1)m\rfloor\}}{\max} \;\frac{\Big|\sum_{i=1}^{k} X_i\Big|}{\sqrt{k} } \geq \tau\Bigg\} \leq C_3\tau^{-2}.
    \end{align*}
\end{proof}

\begin{lemma}
\label{l_inequality_partial_sum}
    Let $\nu >0$ be given. Under the same assumptions as the one in \Cref{l_maximal_inequality}, for any $0 <a<1$ it holds that
    \begin{align*}
        \mathbb{P}\Bigg\{\Big|\sum_{i=1}^d X_i\Big| \leq \frac{C}{a}\sqrt{d}\{\log(d\nu)+1\} \quad \text{for all} \quad d \geq 1/\nu \Bigg\} \geq 1-a^2,
    \end{align*}
    where $C>0$ is an absolute constant.
\end{lemma}

\begin{proof}
    Without loss of generality, assume that $1/\nu \geq 1$. Let $s \in \mathbb{Z}^+$, and $\mathcal{T}_s = [2^s/\nu, 2^{s+1}/\nu] \cap \mathbb{Z}$. Then following form \Cref{l_maximal_inequality}, we have that for any $\tau >0$,
    \begin{align*}
        \mathbb{P}\Bigg\{\underset{d \in \mathcal{T}_s}{\sup} \frac{\Big|\sum_{i=1}^d X_i\Big|}{\sqrt{d}} \geq \tau \Bigg\} \leq C_1\tau^{-2}.
    \end{align*}
    Therefore, by a union-bound argument, for any $0 < a <1$, we have that
    \begin{align*}
        \mathbb{P}\Big\{\exists s \in \mathbb{Z}^+:\underset{d \in \mathcal{T}_s}{\sup}\frac{\Big|\sum_{i=1}^d X_i\Big|}{\sqrt{d}} \geq \frac{\sqrt{C_1}}{a}(s+1) \Big\} \leq \sum_{s=0}^\infty \frac{a^2}{(s+1)^2}=\frac{a^2\pi^2}{6}.
    \end{align*}
    For any $d \in \mathcal{T}_s$, it holds that $s \leq \log(d\nu)/\log(2)$, and therefore we have that
    \begin{align*}
         \mathbb{P}\Bigg\{\underset{d \in \mathcal{T}_s}{\sup}\frac{\Big|\sum_{i=1}^d X_i\Big|}{\sqrt{d}} \geq \frac{\sqrt{C_1}}{a}\Big\{\frac{\log(d\nu)}{\log(2)}+1\Big\} \; \text{for all} \; s \leq \log(d\nu)/\log(2)\Bigg\} \leq \frac{a^2\pi^2}{6},
    \end{align*}
    which directly leads to
    \begin{align*}
        \mathbb{P}\Bigg(\Big|\sum_{i=1}^d X_i\Big| \leq \frac{C_2}{a}\sqrt{d}\{\log(d\nu)+1\} \; \text{for all} \; d \geq 1/\nu \Bigg) \geq 1-a^2.
    \end{align*}
\end{proof}

\end{appendices}

\end{document}